\documentclass[a4paper]{article}
\usepackage{fullpage}
\usepackage{mathtools}

\usepackage{amsmath}
\DeclareMathOperator{\ADV}{ADV}
\DeclareMathOperator{\diag}{diag}

\DeclareMathOperator{\poly}{poly}
\DeclareMathOperator{\polylog}{polylog}
\DeclareMathOperator{\Span}{Span}
\DeclareMathOperator{\Vol}{Vol}
\DeclareMathOperator{\Var}{Var}
\DeclareMathOperator{\spec}{spec}

\usepackage{amsthm}
\usepackage{thmtools} 
\usepackage[bookmarksnumbered,linktocpage,hypertexnames=false,colorlinks=true,linkcolor=blue,urlcolor=blue,citecolor=blue,anchorcolor=green,breaklinks=true,pdfusetitle]{hyperref}
\usepackage[capitalize]{cleveref}
\newtheorem{theorem}{Theorem}[section]
\newtheorem{lemma}[theorem]{Lemma}
\newtheorem{definition}[theorem]{Definition}
\newtheorem{corollary}[theorem]{Corollary}
\newtheorem{remark}[theorem]{Remark}

\usepackage{algorithm}
\makeatletter
\newenvironment{breakablealgorithm}
  {
    \vspace{.7em}
     \refstepcounter{algorithm}
     \hrule height.8pt depth0pt \kern2pt
     \renewcommand{\caption}[2][\relax]{
       {\raggedright\textbf{\fname@algorithm~\thealgorithm} ##2\par}%
       \ifx\relax##1\relax 
         \addcontentsline{loa}{algorithm}{\protect\numberline{\thealgorithm}##2}%
       \else 
         \addcontentsline{loa}{algorithm}{\protect\numberline{\thealgorithm}##1}%
       \fi
       \kern2pt\hrule\kern2pt
     }
  }{
     \kern2pt\hrule\relax
    \vspace{.7em}
  }
\makeatother

\usepackage{amssymb}
\newcommand{\C}{\ensuremath{\mathbb{C}}}
\newcommand{\R}{\ensuremath{\mathbb{R}}}
\newcommand{\N}{\ensuremath{\mathbb{N}}}
\newcommand{\D}{\ensuremath{\mathcal{D}}}
\newcommand{\E}{\ensuremath{\mathbb{E}}}
\renewcommand{\P}{\ensuremath{\mathbb{P}}}

\newcommand{\Z}{\ensuremath{\mathbb{Z}}}

\newcommand{\ket}[1]{\ensuremath{\left|#1\right\rangle}}
\newcommand{\norm}[1]{\ensuremath{\left\|#1\right\|}}
\newcommand{\bra}[1]{\ensuremath{\left\langle#1\right|}}
\newcommand{\braket}[2]{\ensuremath{\left\langle#1\middle|#2\right\rangle}}

\newcommand{\tO}{\widetilde{O}}
\newcommand{\tTh}{\widetilde{\Theta}}
\newcommand{\tOm}{\widetilde{\Omega}}
\newcommand{\eps}{\varepsilon}

\usepackage{tikz}

\usepackage{xcolor}
\usepackage{pdfcolfoot} 

\usepackage{enumitem}
\usepackage{bbm}
\usepackage{multirow}
\usepackage{authblk}

\title{Amortized quantum walks and volume estimation algorithms}
\title{Improved quantum volume estimation with transducers and amortized quantum walks}
\author[1]{Arjan Cornelissen}
\author[2]{Simon Apers}
\author[3]{Sander Gribling}
\affil[1]{Department of Computer Science, Cambridge University, UK}
\affil[2]{IRIF, CNRS, Universit\'e Paris Cit\'e, France}
\affil[3]{Department of Econometrics \& Operations Research, Tilburg University, Netherlands}
\begin{document}
    \maketitle

    \begin{abstract}
        The volume estimation problem is a classic task in computational geometry. The development of randomized algorithms for this problem spurred the development of many influential algorithmic techniques related to Markov Chain Monte Carlo and simulated annealing, and the problem connects to several important geometrical results, like the KLS conjecture. In this work, we quantize the state-of-the-art $\widetilde{O}(d^{3.5}+d^3/\varepsilon^2)$-query randomized algorithm developed by Cousins and Vempala, and obtain a $\widetilde{O}(d^{3.5} + d^{1.75}/\varepsilon)$-query quantum algorithm, improving over the $\widetilde{O}(d^{3.5} + d^{2.25}/\varepsilon)$ state-of-the-art bound. Our key technical contribution is a framework for amortizing the cost of a quantum walk. The framework is based on the recent transducer toolkit introduced by Belovs, Jeffery and Yolcu. It is this amortized quantum walk framework that allows us to exploit the amortized analysis of the ball walk by Cousins and Vempala, thus overcoming the key barrier that previously barred its quantum implementation.
    \end{abstract}
    \thispagestyle{empty}

    {\small \tableofcontents}
    \thispagestyle{empty}

    \clearpage
    \pagenumbering{arabic}
    \section{Introduction}

The problem of estimating the volume of a convex body is a canonical problem that has received a lot of attention in computational geometry and theoretical computer science. The problem statement is as follows. Let $d$ be an integer and $R \geq 1$. We are given a convex body $K \subseteq \R^d$ that contains the unit ball and is contained in a ball of radius $R$ centered at the origin, that is, $B_d \subseteq K \subseteq RB_d$ where $B_d$ is the unit ball in $d$ dimensions. Now, given some relative error $0 < \varepsilon < 1$, the goal is to output an $\varepsilon$-precise relative estimate of $\Vol(K)$, i.e., we are asked to output a positive real $\widetilde{V} > 0$, such that
\[\left|\frac{\widetilde{V}}{\Vol(K)} - 1\right| \leq \varepsilon.\]
We assume to have access to $K$ through membership queries, i.e., we can take a point in space $x \in \R^d$, and ask an oracle $O_K$ whether $x$ is contained in $K$.
The randomized (resp.~quantum) query complexity of the volume estimation problem corresponds to the number of calls to $O_K$ made by a randomized (resp.~quantum) algorithm.

Following a long line of work, the state-of-the-art bound on the randomized query complexity is
\[
\tO(d^{3.5} + d^3/\varepsilon^2),
\]
where the tildes hide polylogarithmic dependencies on $d$, $1/\varepsilon$ and $R$.
This follows from combining a $\tO(d^{3.5})$-query algorithm by Jia, Laddha, Lee and Vempala \cite{jia2026reducing} for rounding the body with the $\tO(d^3/\varepsilon^2)$-query algorithm for volume estimation by Cousins and Vempala \cite{CV18}.
In line with earlier works, this last algorithm rephrases the problem as that of estimating a partition function, and then solves this problem using Markov chain Monte Carlo and simulated annealing.
A critical component of the algorithm in \cite{CV18} is an amortized analysis of the so-called \emph{speedy walk} on~$K$.

State-of-the-art on the quantum side is the algorithm by Cornelissen and Hamoudi \cite{cornelissen2023sublinear} with quantum query complexity
\[
\tO(d^{3.5} + d^{2.25}/\varepsilon).
\]
The algorithm combines quantum walks with quantum annealing and mean estimation, and it improves on the earlier quantum algorithm by Chakrabarti, Childs, Hung, Li, Wang and Wu \cite{chakrabarti2023quantum}.
Rather than building on \cite{CV18}, these algorithms effectively quantize an earlier and slower algorithm based on the hit-and-run walk \cite{lovasz2006simulated}.
The main obstacle preventing a quantization of \cite{CV18} has been the absence of a suitable amortized analysis in the quantum setting.

In this work, we overcome this obstacle and prove the following main theorem.
\begin{theorem}[Main result, informal version] \label{thm:intro-main}
There is a bounded-error quantum algorithm that estimates the volume of a convex body $K \subseteq \R^d$, satisfying $B_d \subseteq K \subseteq RB_d$, with a number of queries
\[\widetilde{O}\left(d^{3.5} + \frac{d^{1.75}}{\varepsilon}\right).\]
\end{theorem}
The algorithm is based on an amortized quantum walk framework that we develop to quantize~\cite{CV18}.
The key to this framework is a recent new concept known as a \textit{transducer}, which was introduced by Belovs, Jeffery and Yolcu~\cite{belovs2024taming} as a useful, idealized proxy for a quantum algorithm.
We summarize our new framework in \cref{sec:intro-transducers}, and we sketch the corresponding quantum volume estimation algorithm in \cref{sec:intro-volume}.

\subsection{Main results: Transducers and amortized quantum walks} \label{sec:intro-transducers}

As the key mathematical object throughout this work, we next introduce transducers.
A transducer is a unitary operation that
maps a \textit{source state} $\ket{\sigma}$ to a \textit{target state} $\ket{\tau}$, with the help of a \textit{catalyst state} $\ket{w}$ that is restored at the end of the computation, i.e.,
\[\ket{\sigma} \oplus \ket{w} \overset{U}{\mapsto} \ket{\tau} \oplus \ket{w}.\]
We say $U$ transduces $\ket{\sigma}$ to $\ket{\tau}$, and that $\norm{\ket{w}}^2$ is the \textit{transduction complexity} of this operation.
Modulo some technicality, $\ket{\tau}$ and $\ket{w}$ are uniquely determined by $U$ and $\ket{\sigma}$.
We can hence write $\ket{\sigma} \raisebox{-.2em}{$\overset{U}{\rightsquigarrow}$} \ket{\tau}$ and we can denote the complexity $W(U,\ket{\sigma}) := \norm{\ket{w}}^2$.

One can advocate that such a transducer $U$ effectively behaves as a classical Las Vegas algorithm, which in a sense is also an idealized proxy for the usual notion of bounded-error randomized algorithm:
\begin{itemize}
\item
First, we can \emph{turn any transducer into a bounded-error quantum algorithm} in the circuit model.
More precisely, we can map $\ket{\sigma}$ to a state $\varepsilon$-close in $\ell_2$-norm to $\ket{\tau}$ with $O(W(U,\ket{\sigma})/\varepsilon^2)$ calls to~$U$.
\item
Second, \emph{transducers are more powerful than regular quantum algorithms:} many bounded-error quantum subroutines can be implemented without error as a transducer.
As a consequence, if we do algorithmic composition on the transducer level, we can avoid the pileup of errors that is typical for composition with bounded-error quantum subroutines.
\item
Finally, and crucially to this work, the \emph{complexity of a transducer is input-dependent:} a single transducer $U$ when applied to different input states $\ket{\sigma_1}$ and $\ket{\sigma_2}$ may have different complexities $W(U,\ket{\sigma_1}) \neq W(U,\ket{\sigma_2})$.
This contrasts with usual complexity measures such as the circuit complexity of a quantum subroutine, and it is precisely this that opens the door for amortization.
\end{itemize}

\noindent
In this work we develop novel transducers for several key quantum algorithmic primitives, and we expect these to be of independent interest.

\subsubsection{Amplitude amplification transducer}

The first such construction that we present is an error-free transducer for amplitude amplification.
The construction generalizes an earlier transducer for Grover's algorithm by Cornelissen, Edenhofer, Schaeffer and Szab\'o \cite{cornelissen2026quantum}, and it was independently and concurrently discovered by Dubus, Ladeuze and Roland~\cite{dubus2026transducer}.
The amplitude amplification transducer assumes access to a subroutine $U$ that reflects through the initial state.

\begin{theorem}[Amplitude amplification transducer, informal version of \Cref{cor:ampl-ampl}] \label{thm:intro-AA}
Let
\[\ket{\psi} = \ket{\psi_0}\ket{0} + \ket{\psi_1}\ket{1}, \qquad \text{and} \qquad \ket{\overline{\psi}} = \frac{\norm{\ket{\psi_1}}}{\norm{\ket{\psi_0}}} \ket{\psi_0}\ket{0} - \frac{\norm{\ket{\psi_0}}}{\norm{\ket{\psi_1}}} \ket{\psi_1}\ket{1},\]
and suppose that we have a transducer $U$ that reflects through $\ket{\psi}$, i.e., $\ket{\psi} \raisebox{-.2em}{$\overset{U}{\rightsquigarrow}$} \ket{\psi}$ and $\ket{\perp} \raisebox{-.2em}{$\overset{U}{\rightsquigarrow}$} -\ket{\perp}$ for all $\ket{\perp}$ satisfying $\braket{\perp}{\psi} = 0$.
Then we can implement a transducer $C$ such that
\[
\ket{\psi} \overset{C}{\rightsquigarrow} \frac{\ket{\psi_1}\ket{1}}{\norm{\ket{\psi_1}}}, \qquad \text{with complexity} \qquad O\left(\frac{W(U,\ket{\psi}) + \norm{\ket{\psi_0}}^2W(U,\ket{\overline{\psi}})}{\norm{\ket{\psi_1}}}\right).
\]
\end{theorem}

We also describe a transducer $V$ that reflects through the good state, rather than preparing it:

\begin{align} \label{eq:AA-refl}
\begin{split}
\frac{\ket{\psi_1}}{\norm{\ket{\psi_1}}} \overset{V}{\rightsquigarrow} \frac{\ket{\psi_1}}{\norm{\ket{\psi_1}}}, & \qquad \text{with complexity} \qquad O\left(\frac{W(U,\ket{\psi})}{\norm{\ket{\psi_1}}^2}\right), \\
\ket{\perp} \overset{V}{\rightsquigarrow} -\ket{\perp}, & \qquad \text{with complexity} \qquad O\left(W(U,\ket{\perp}\ket{1})\right),
\end{split}
\end{align}
for all vectors $\ket{\perp}$ that satisfy $\braket{\perp}{\psi_1} = 0$.
Notably, the transducers $C$ and $V$ construct and reflect around the amplified state $\ket{\psi_1}/\norm{\ket{\psi_1}}$ without error.
This contrasts with the bounded-error amplitude amplification routines in the circuit model. 

Since amplitude amplification is a ubiquitous tool in quantum algorithm design, we can use this transducer construction to port many algorithmic primitives to the transducer setting.
In this work, we specifically work this out for quantum annealing (\Cref{thm:two-state-ampl}), quantum rejection sampling (\Cref{thm:rejection-sampling}), and quantum mean estimation (\Cref{thm:transducer-mean-est}), all of which rely on the amplitude amplification construction above, and all of which will be used in our quantum volume estimation algorithm.
We believe that the technique will lead to the porting of many more algorithmic primitives to the transducer framework in the near future.

It is interesting to contrast our amplitude amplification transducer with an earlier approach in~\cite{apers2026elfs}.
That work built an amplitude amplification transducer essentially by considering the history state of the zero-error Las Vegas algorithm for amplitude amplification from \cite{boyer1998tight}.
Our approach builds a transducer more directly from a solution to the adversary bound for a corresponding state conversion problem.
While both approaches yield a similar query complexity, the downside of our approach is that we lose time efficiency as compared to \cite{apers2026elfs}.
The important upside is that our approach does not produce garbage, in contrast to \cite{apers2026elfs}.
We will see that this last feature is necessary for our quantum walk applications.

\subsubsection{Quantum walk reflection transducer} \label{sec:intro-QW}

We now turn to quantum walks, which were already shown to have an intricate connection to transducers in \cite{belovs2024taming,belovs2026space,apers2026elfs}.
We start from an ergodic, reversible Markov chain $P$ on a finite state space~$\Omega$ with stationary distribution $\pi$ and spectral gap $\gamma(P) > 0$.
For any state $\omega \in \Omega$, we can then define the star state
\[
\ket{*_\omega}
= \sum_{\nu \in \Omega} \sqrt{P_{\omega,\nu}} \ket{\nu}.
\]
Writing $\ket{\omega,*_\omega} := \ket{\omega}\ket{*_\omega}$, the quantum walk operator is then
\[
\mathrm{SWAP} \cdot \bigoplus_{\omega \in \Omega} (2\ket{\omega,*_\omega}\bra{\omega,*_\omega}-I),
\]
where $\mathrm{SWAP} \ket{\omega,\nu} = \ket{\nu,\omega}$.

In this work, we use the quantum walk operator to construct a transducer that reflects around a quantum version of the stationary distribution,
\[
\ket{\pi}
:= \sum_{\omega \in \Omega} \sqrt{\pi(\omega)}\ket{\omega}.
\]
To understand the corresponding transduction complexity, we exploit the particularly elegant connection between a quantum walk and the electrical network associated to the Markov chain $P$.
In this context, we build on a key observation from Cornelissen~\cite{cornelissen2025quantum} that interprets a state
\[
\ket{\overline{\pi}}
=\sum_{\omega \in \Omega} \frac{\overline{\pi}(\omega)}{\sqrt{\pi(\omega)}}\ket{\omega}
\]
that is orthogonal to $\ket{\pi}$ as a demand vector on $\Omega$.
This demand vector induces an associated electric flow on the electrical network, and we show that its edge-flows dictate the catalyst state of the resulting transducer.
This yields the following statement.

\begin{theorem}[Quantum walk reflection transducer, informal version of \Cref{thm:refl-stationary}]
    \label{thm:refl-stationary-intro}
    For every~$\omega \in \Omega$, let $C_{\omega}$ and $U_{\omega}$ be transducers that construct and reflect through $\ket{*_{\omega}}$, respectively.
    We construct a transducer $U$ such that
    \begin{align*}
        \ket{\pi} \overset{U}{\rightsquigarrow} \ket{\pi}, & \quad \text{with complexity} \quad O\left(\sum_{\omega \in \Omega} \pi(\omega) S_\omega \right), \\
        \ket{\overline{\pi}} \overset{U}{\rightsquigarrow} -\ket{\overline{\pi}}, & \quad \text{with complexity} \quad O\left(\sum_{\omega \in \Omega} \frac{|\overline{\pi}(\omega)|^2}{\pi(\omega)} \left(\frac{S_\perp}{\gamma(P)} + S_\omega \right)\right),
    \end{align*}
    where $S_\omega = W(C_{\omega},\ket{0}) + W(U_{\omega}, \ket{*_{\omega}})$ and
    \[
        S_\perp = \sup \{W(U_{\omega}, \ket{\perp}) : \omega \in \Omega, \norm{\ket{\perp}} = 1, \braket{\perp}{*_{\omega}} = 0\}.
    \]
\end{theorem}

The crucial point here is that the complexity of the reflection transducer is \emph{amortized}: it charges the cost $S_\omega$ of preparing or reflecting around $\ket{*_\omega}$ averaged over the relevant distributions of the input state.
Notably, the transduction complexity $S_\perp$ of reflecting an orthogonal state $\ket{\perp}$ is not amortized, but this will be a trivial cost in our algorithms (see e.g.~\cref{eq:AA-refl} or \cref{thm:speedy-intro}).

To finalize this section, we again compare to a relevant prior result in \cite{apers2026elfs}.
There it was argued that a slightly altered quantum walk operator is a transducer that reflects around an electric flow state, rather than $\ket{\pi}$, and thus can be used to prepare that electric flow state.
The authors argued that this was a particular instance of a more ``effective gap transducer'' for reflecting around the intersection of two subspaces, and effectively our \cref{thm:refl-stationary-intro} is a different instantiation of that same transducer.

\subsection{Main results: Quantum volume estimation} \label{sec:intro-volume}

We now turn to our main application, which is an improved quantum walk algorithm for volume estimation.
Before going into its details, we give a more detailed exposition of the current state-of-the-art.

\subsubsection{State of the art}

In the deterministic model, Elekes \cite{elekes1986geometric} and independently B\'ar\'any and F\"uredi \cite{barany1987computing} showed that the query complexity is exponential in $d$ even in the regime where $\eps$ is constant.  
Conversely, if the dimension $d$ is constant, the query complexity is $\Theta(\varepsilon^{-(d-1)/2})$, where the big-$\Theta$-notation hides a constant factor that depends on $d$~\cite{cornelissen2025compute}.

With the additional power of the randomized model, the situation changes drastically. In a landmark result, Dyer, Frieze and Kannan developed a randomized algorithm for the volume estimation problem that makes $\widetilde{O}(d^{23}/\varepsilon^2)$ queries, thus proving a polynomial vs.\ exponential separation between the deterministic and randomized models~\cite{dyer1991random}.
A long line of research subsequently improved the polynomial dependence on $d$, which we represent in \Cref{tab:vol-est}.
Many of the algorithmic innovations in this line of work have become key techniques in Gibbs sampling and Markov Chain Monte Carlo methods, and it also led to the development of several deep geometrical results, like the recently progress on the KLS-conjecture~\cite{chen2021almost,klartag2022bourgain,jambulapati2022slightly,klartag2023logarithmic,letwin2026kls}.
This has motivated the continued search for algorithmic improvements for the volume estimation problem.

\begin{table}[!ht]
\centering
\begin{tabular}{r|ccl}
Model & Source & Result & New idea \\\hline
Randomized & \cite{dyer1991random} & $\widetilde{O}(d^{23}/\varepsilon^2)$ & Markov chain Monte Carlo and polynomial mixing \\
& \cite{lovasz1990mixing} & $\widetilde{O}(d^{16}/\varepsilon^4)$ & Better isoperimetry \\
& \cite{applegate1991sampling} & $\widetilde{O}(d^{10}/\varepsilon^2)$ & Logconcave sampling \\
& \cite{lovasz1991compute} & $\widetilde{O}(d^{10}\poly(1/\varepsilon))$ & Ball walk \\
& \cite{dyer1991computing} & $\widetilde{O}(d^8/\varepsilon^2)$ & Better error analysis \\
& \cite{lovasz1993random} & $\widetilde{O}(d^7/\varepsilon^2)$ & Localization lemma \\
& \cite{kannan1997random} & $\widetilde{O}(d^5/\varepsilon^2)$ & Speedy walk, isotropy \\
& \cite{lovasz2006simulated} & $\widetilde{O}(d^4/\varepsilon^2)$ & Annealing, hit-and-run walk \\
& \cite{CV18} & $\widetilde{O}(d^4 + d^3/\varepsilon^2)$ & Gaussian cooling in well-rounded setting \\
& \cite{jia2026reducing} & $\widetilde{O}(d^{3.5} + d^3/\varepsilon^2)$ & Improved rounding 
\\\hline
Quantum & \cite{chakrabarti2023quantum} & $\widetilde{O}(d^{3.5} + d^{2.5}/\varepsilon)$ & Quantization of hit-and-run walk  \\
& \cite{cornelissen2023sublinear} & $\widetilde{O}(d^{3.5} + d^{2.25}/\varepsilon)$ & Unbiased, non-destructive mean estimation \\
& \textbf{This work} & $\widetilde{O}(d^{3.5} + d^{1.75}/\varepsilon)$ & Amortized quantization of ball walk 
\end{tabular}
\caption{Historical overview of complexity bounds for estimating the volume of convex bodies, based on the tables presented in \cite{lovasz1993random,jia2026reducing}. The tildes hide factors that are polylogarithmic in $d$, $R$ and $1/\varepsilon$. 
}
\label{tab:vol-est}
\end{table}

Since the randomized setting allows one to substantially improve over deterministic approaches, a natural follow-up question is whether having \textit{quantum} access to the oracle could reduce the query complexity even further. In that direction, Chakrabarti, Childs, Hung, Li, Wang and Wu \cite{chakrabarti2023quantum} quantized the approach outlined by Lov\'asz and Vempala~\cite{lovasz2006simulated} based on the hit-and-run-walk, obtaining a quadratic speed-up in terms of its mixing time. Combined with a quantum mean estimation routine developed by Hamoudi and Magniez~\cite{hamoudi2019quantum} and an annealing procedure designed by Wocjan and Abeyesinghe~\cite{wocjan2008speedup}, this yielded a quantum algorithm for the volume estimation problem that runs in $\widetilde{O}(d^{3.5} + d^{2.5}/\varepsilon)$ queries.\footnote{The original version of the construction in \cite{chakrabarti2023quantum} contained an issue in the argument that ensures non-destructiveness of the mean-estimation routine~\cite[Lemma~4.3]{chakrabarti2023quantum}. This issue was partially resolved by subsequent developments of Cornelissen and Hamoudi~\cite{cornelissen2023sublinear}, but only in the univariate setting. It is our understanding that the correctness of the rounding part of the construction in \cite{chakrabarti2023quantum}, which relies on multivariate mean estimation, is not established as of now. However, this only concerns the $\widetilde{O}(d^3)$ additive term in the complexity statement, which we therefore replace with the classical $\widetilde{O}(d^{3.5})$ complexity of the rounding algorithm of \cite{jia2026reducing}. 

In the published version of their paper, the authors of \cite{chakrabarti2023quantum} moreover argue that the resolution to the KLS-conjecture implies an $\widetilde{O}(d^{2.5} + d^2/\varepsilon)$-query quantum algorithm. As noted before in \cite[Footnote~1]{cornelissen2025compute}, this claim seems to be unjustified as of now.\label{footnote:chakrabarti}} Cornelissen and Hamoudi subsequently improved this construction by making the mean-estimation routine unbiased, and obtained an $\widetilde{O}(d^{3.5} + d^{2.25}/\varepsilon)$-query quantum algorithm~\cite{cornelissen2023sublinear}. 
Prior to this work, this was the state-of-the-art quantum algorithm for the volume estimation problem.
Our work improves over these bounds by quantizing the algorithm by Cousins and Vempala \cite{CV18}, which we summarize next.

\paragraph{Cousins-Vempala volume estimation.}

Most randomized algorithms that estimate the volume of a convex body consist of two main steps: rounding the convex body, and estimating the volume of a well-rounded convex body.
The rounding step aims to find an affine transformation that places the convex body in a \textit{well-rounded} position. We call a convex body $K$ well-rounded if it satisfies $B_d \subseteq T(K)$ and $\E_{K}[\norm{x}^2] = \widetilde O(d)$. A well-rounded convex body has the property that all but an $\eps$-fraction of its volume is contained in  $R'B_d$, where $R' \in O(\sqrt{d}\log(1/\eps))$ (e.g.~via Markov's inequality on the random variable $\|x\|^2$ and~\cite[Corollary~2.9]{lovasz1993random}). For the volume estimation problem this allows us to assume $B_d \subseteq K \subseteq R' B_d$ for $R'\in \widetilde O(\sqrt{d})$. Recent work by Jia, Laddha, Lee and Vempala~\cite{jia2026reducing} provides a randomized rounding algorithm that finds such an affine transformation for convex bodies that satisfy $B_d \subseteq K \subseteq R B_d$.
Their algorithm uses $\tO(d^{3.5} \, \psi_d^2)$ queries, where $\psi_d$ is the KLS-constant in $d$ dimensions.
This constant was subsequently shown to satisfy $\psi_d \in O(\sqrt{\log(d)})$~\cite{klartag2023logarithmic}, thus yielding a rounding algorithm that makes $\widetilde{O}(d^{3.5})$ queries.\footnote{An earlier conference version \cite{jia2021reducing} incorrectly claimed a runtime of $\widetilde O(d^3\psi_d^2)$. The quantum algorithms in \cref{tab:vol-est} all use this rounding algorithm as a preprocessing step; we report here the corrected complexities. To the best of our knowledge, improving the $\widetilde O(d^{3.5})$ query complexity of \cite{jia2026reducing} to $\widetilde O(d^3)$ remains an open problem.}

The remaining step is to estimate the volume of the well-rounded convex body.
The core idea in all of the aforementioned works is to phrase this problem as a partition function estimation problem, and then use Markov chain Monte Carlo methods to estimate this partition function by sampling from the associated Gibbs distributions.
In more detail, the algorithm by Cousins and Vempala~\cite{CV18} considers the Gibbs distribution $\pi_\beta$ on convex body $K$ with energy Hamiltonian the norm squared:
\[
\pi_\beta(x)
= \frac{1}{Z(\beta)} e^{-\beta \|x\|^2} \mathbbm{1}_{x \in K}.
\]
Here $\beta$ denotes the inverse temperature, and $Z(\beta) = \int_K e^{-\beta \|x\|^2} \;\mathrm{d}x$ the partition function.
Noting that $Z(0) = \Vol(K)$, it suffices to estimate this partition function.

To do so, first note that if $\beta^* \in \widetilde\Omega(d)$ then $\int_K e^{-\beta \|x\|^2} \;\mathrm{d}x \approx \int_{\R^d} e^{-\beta \|x\|^2} \;\mathrm{d}x$ and so $Z(\beta^*)$ can be estimated without queries.
It thus suffices to estimate $Z(0)/Z(\beta^*)$ to relative precision $\varepsilon$.
This is done by defining an inverse cooling schedule $(\beta_j)_{j=0}^\ell$ with $\beta_0 = \beta^*$ and $\beta_\ell = 0$, and using it to rewrite
\[
\frac{Z(0)}{Z(\beta^*)}
= \prod_{j=0}^{\ell-1} \frac{Z(\beta_{j+1})}{Z(\beta_j)}.
\]
The Markov chain Monte Carlo angle then comes from noting that
\[
\frac{Z(\beta')}{Z(\beta)}
= \E_{\pi_\beta}\left[ e^{-(\beta'-\beta) \|x\|^2} \right],
\]
and so it is possible to estimate each of these ratios by sampling from the Gibbs distribution at the appropriate temperatures.
\cite{CV18} subsequently uses the \textit{speedy walk} to mix rapidly to these Gibbs distributions.

In the following sections, we detail our quantum walk implementation of the speedy walk and how we use it to speed up the preparation of the Gibbs states.

\subsubsection{Quantum speedy walk}

The key primitive in \cite{CV18} is to sample from the Gibbs distribution $\pi_\beta$, which they do using the so-called speedy walk:

\begin{definition}[Speedy walk] \label{def:speedy_intro}
    Write $f_{\beta}(x) = \exp(-\beta \|x\|^2)\mathbbm{1}_{x \in K}$. At $x \in K$, do:
    \begin{enumerate}[nosep]
        \item Select $y \in (x + \delta B_d) \cap K$ uniformly at random, for $\delta \in \widetilde\Theta(1/\sqrt{(1+\beta) d})$.
        \item Transition to $y$ with probability $\min\{1,f_{\beta}(y)/f_{\beta}(x)\}$.
    \end{enumerate}
\end{definition}

Under appropriate conditions, this walk has a large spectral gap $\widetilde\Omega(\min\{1,\beta\}/d^2)$.
A first downside however is that its stationary distribution is slightly off: it is proportional to $f(x) \ell_\delta(x)$ rather than to~$f(x)$, where~$\ell_\delta(x)$ is the so-called local conductance
\[
\ell_\delta(x)
= \frac{\Vol((x+\delta B_d) \cap K)}{\Vol(\delta B_d)}.
\]
\cite{CV18} however show that a simple rejection sampling routine can turn this stationary into the appropriate one.
A second downside is that it might be hard to implement the speedy walk.
Indeed, the naive algorithm of sampling uniformly from $x+\delta B_d$ and postselecting on being inside $K$ will have success probability $\ell_\delta(x)$, which can be exponentially small, see \Cref{fig:local-conductance} for an illustration.

\begin{figure}
    \centering
    \includegraphics[width=0.5\linewidth]{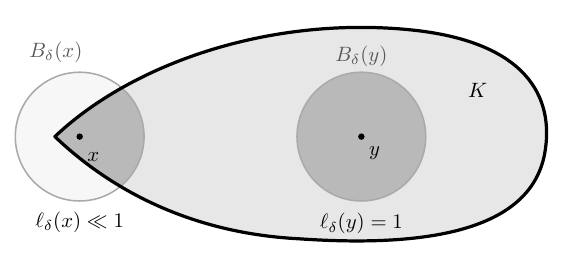}
    \caption{Illustration of local conductance: at points $y$ in the interior, the local conductance equals $1$, whereas $\ell_\delta(x) \ll 1$ at points $x$ close to corners of the convex body $K$.}
    \label{fig:local-conductance}
\end{figure}

The way around this is precisely where amortization comes in: averaged over any initial distribution that is sufficiently ``warm'' with respect to the stationary, the mean local conductance is constant.
In the classical setting, by a concentration argument, this easily suffices to argue that the total algorithm can effectively charge a constant cost per step of the speedy walk.

\paragraph{Our contribution.}
In the quantum case, we use our amortized quantum walk framework to argue likewise.
However, since the speedy walk is defined on a continuous state space, we first introduce an appropriate discretization.
We show that there is a \emph{discretized speedy walk} that walks on the grid $\zeta\mathbb{Z}^d \cap K$ for a sufficiently small grid spacing $\zeta=\zeta(d,R,\eps)>0$, and whose spectral gap and stationary distribution we can relate to the original speedy walk. We remark that we only require the existence of a sufficiently small $\zeta$, its value does not appear in the \textit{query} complexity of the final algorithm; we merely use it to make the state space discrete and finite.  
While related discretizations had been studied before, to the best of our knowledge none of these captured the precise quantities that we need.

After this discretization, we can describe how we implement a transducer for the \emph{quantum speedy walk}, a quantum walk based on the (discretized) speedy walk.
As discussed in \cref{sec:intro-QW}, the key operation is the preparation and reflection around the star state $\ket{*_x}$ for $x \in K$, which is a superposition over the steps that the discrete speedy walk would take from $x$, weighted by their probabilities.
For step 1 of the speedy walk (\cref{def:speedy_intro}), we invoke our amplitude amplification transducer (\cref{thm:intro-AA}), where it is important that this transducer does not produce garbage.
For the Metropolis filter in step 2, however, we run into the difficulty that quantumly we cannot simply reject and undo an operation after conditioning on some outcome (see e.g.~\cite{temme2011quantum,lemieux2020efficient} for some discussion on this). 
To remedy this, we modify the speedy walk to act as a two-step process on a bipartite graph, where one step is to select the proposal point, and the second is to do the transition. 
We prove that the spectral gap of this bipartite speedy walk is the same as the spectral gap of the original speedy walk, up to constants, and we show that we can implement the corresponding star-state construction and reflection transducers efficiently.
The same procedure can be applied to any Markov chain and removes the challenges posed by the Metropolis filter (a related idea appeared in e.g.~\cite{apers23ustcon}).
This yields the following theorem.

\begin{theorem}[Quantum speedy walk transducer, Informal version of \Cref{lem:star-state-transducers}] \label{thm:speedy-intro}
For the (bipartite, discretized) speedy walk, we can implement the star-state construction and reflection transducers $C_x$ and $U_x$ at $x \in K$, each making a single membership oracle query, such that
\begin{align*}
\ket{0} \overset{C_x}{\rightsquigarrow} \ket{*_x}, & \qquad \text{with complexity} \qquad O\left(1/\sqrt{\ell_{\delta}(x)}\right), \\
\ket{*_x} \overset{U_x}{\rightsquigarrow} \ket{*_x}, & \qquad \text{with complexity} \qquad O\left(1/\ell_{\delta}(x)\right), \\
\ket{\perp} \overset{U_x}{\rightsquigarrow} -\ket{\perp}, & \qquad \text{with complexity} \qquad O\left(\norm{\ket{\perp}}^2\right),
\end{align*}
for all $\ket{\perp}$ such that $\braket{\perp}{*_x} = 0$.
\end{theorem}

We can then combine this with our quantum walk reflection transducer (\cref{thm:refl-stationary-intro}) to yield a transducer for reflection around the stationary of the discretized speedy walk, with a cost that is appropriately amortized.

\begin{theorem}[{Quantum speedy walk reflection transducer, informal version of \cref{thm:speedy-walk-reflection}}]
\label{thm:intro-speedy-walk-reflection}
Let $\sigma_\beta$ be the stationary distribution of the discretized speedy walk.
We construct a transducer $U$ such that
\begin{align*}
\ket{\sigma_\beta} \overset{U}{\rightsquigarrow} \ket{\sigma_\beta}, & \qquad \text{with complexity} \qquad O\left(1\right), \\
\ket{\overline{\sigma}} \overset{U}{\rightsquigarrow} -\ket{\overline{\sigma}}, & \qquad \text{with complexity} \qquad \tO\left(\frac{d^2 \, \norm{\ket{\overline{\sigma}}}^2}{\min\{1,\beta\}} + \int_K \frac{|\overline{\sigma}(x)|^2}{\sigma_\beta(x)\ell_\delta(x)} \;\mathrm{d}x\right),
\end{align*}
for all $\ket{\overline{\sigma}}$ such that $\braket{\overline{\sigma}}{\sigma_\beta} = 0$.
\end{theorem}

\subsubsection{Quantum annealing transducer}

In the quantum case, in order to actually prepare $\ket{\pi_\beta}$ (or rather its approximation $\ket{\sigma_\beta}$) instead of just reflecting through it, we take a similar approach to earlier works \cite{wocjan2008speedup,chakrabarti2023quantum,cornelissen2023sublinear}.
We build a quantum simulated annealing algorithm to prepare $\ket{\pi_\beta}$ by first preparing $\ket{\pi_{\beta^*}}$ for sufficiently large $\beta^*$ (which can be done without queries), and then use reflections around $\ket{\pi_{\beta_j}}$ for inverse temperatures $\beta_j$ to iteratively map $\ket{\pi_{\beta_j}}$ to $\ket{\pi_{\beta_{j+1}}}$ for increasing $j$.
The key point is that, so long as $|\braket{\pi_\beta}{\pi_{\beta'}}| \in \Omega(1)$, we only need a constant number of reflections through $\ket{\pi_\beta}$ and $\ket{\pi_{\beta'}}$ to map $\ket{\pi_\beta}$ to $\ket{\pi_{\beta'}}$, and we can do this exactly with our amplitude amplification transducer from \cref{thm:intro-AA}.
Since all of these reflections act on states in the subspace spanned by $\ket{\pi_\beta}$ and $\ket{\pi_{\beta'}}$, we can invoke our quantum speedy walk reflection transducer from \cref{thm:refl-stationary-intro} to analyze and amortize their costs.
Crucially it allows us to show that, similarly to \cite{CV18}, the amortized cost of our quantum speedy walk is constant.

Aside from the application to volume estimation, this quantum walk annealing algorithm already yields an improved algorithm for generating a uniform superposition over the convex body.
Notably, our result involves amortization and is ideally stated in the transducer framework, as we do in the following theorem.

\begin{theorem}[Informal version of \Cref{thm:uniform-sampling}]
\label{thm:sampling}
Let $B_d \subseteq K \subseteq R B_d$ be a convex body.
There exists a transducer that generates a quantum sample of the (discretized) uniform distribution over $K$ up to total-variation distance $\eta > 0$, with transduction complexity $\widetilde{O}\left(dR + d^2\right)$.
\end{theorem}

This is a quantum analogue to \cite[Theorem~1.2]{CV18}, improving over their $\tO(d^2 R^2 + d^3)$ bound on the randomized query complexity.

While the transducer in \cref{thm:sampling} can of course be turned into a quantum algorithm in the usual circuit model, this would incur a naive $O(1/\varepsilon^2)$ multiplicative overhead, where $\varepsilon$ is the operator-norm error.
Better is to use the nice composition properties of transducers, which imply for instance that any algorithm for a decision problem that requires $N$ samples of this uniform superposition over $K$, can be run with $\widetilde{O}(N(dR + d^2)\log(1/\delta))$ queries with failure probability at most $\delta$. Thus, \Cref{thm:sampling} is to be interpreted as a quantum subroutine, in which it is more efficient than its classical counterpart developed in~\cite{CV18}.

\subsubsection{Quantum volume estimation}

For the well-rounded case, $R \in \tO(\sqrt{d})$, the volume estimation algorithm in \cite{CV18} has complexity $\tO(d^3/\varepsilon^2)$.
This stems essentially from a length-$\ell$ cooling schedule for $\ell$ between $\sqrt{d}$ and $d$, multiplied by a $\tO(d^3/(\ell\varepsilon^2))$ cost-per-temperature for mean estimation and speedy walk annealing subroutines.
Combining our quantum speedy walk, the quantum annealing transducer and our quantum mean estimation transducers, we can put a square root over the cost-per-temperature, but not over the schedule length.
Picking schedule length $\ell \in \tO(\sqrt{d})$ then yields a $\tO(d^{1.25}/\varepsilon)$ quantum cost-per-temperature.\footnote{\cite{CV18} do not establish warmness between subsequent Gibbs distributions for this shorter schedule, but our quantum annealing approach only requires a weaker notion of warmness (namely $|\braket{\pi_\beta}{\pi_{\beta'}}| \in \Omega(1)$) that is easier to prove.}
For general~$R$, the resulting complexity of our quantum volume estimation algorithm is summarized in the following theorem.
There is an additive $\widetilde{O}(d^2)$ overhead because we have to use the longer $\tO(d)$-length schedule for the annealing task.

\begin{theorem}[Informal version of \Cref{thm:vol-est-algo}]
\label{thm:result-vol-est-algo}
There is a bounded-error quantum algorithm that estimates the volume of a convex body $K \subseteq \R^d$ satisfying $B_d \subseteq K \subseteq RB_d$ with a number of queries that satisfies
\[\widetilde{O}\left(\frac{d^{1.75} + d^{1.25}R}{\varepsilon} + d^2\right).\]
\end{theorem}

Combined with the $\tO(d^{3.5})$ rounding procedure from \cite{jia2026reducing} that yields $R \in \widetilde O(\sqrt{d})$ this implies our main result, \cref{thm:intro-main}.

\subsection{Open questions}

\paragraph{Lower bounds.}
Both in the randomized and quantum setting there is a gap between the best known upper and lower bounds for volume estimation.
In the randomized setting, the best-known lower bound for volume estimation is by Rademacher and Vempala~\cite{rademacher2008dispersion}, who showed that it requires at least $\widetilde{\Omega}(d^2)$ membership queries for constant precision.
In very recent work, Vempala \cite{vempala2026nearly} showed that the same lower bound already holds for the problem of sampling uniformly from the convex body.
In the slightly different regime where $d$ is kept constant, the randomized query complexity was exactly characterized to be~$\Theta_d(1/\varepsilon^{2(d-1)/(d+3)})$~\cite{cornelissen2025compute}.

On the quantum side, in the constant-error regime, the current state-of-the-art lower bound is $\Omega(\sqrt{d})$, proved in \cite[Theorem~6.1]{chakrabarti2023quantum}.
They also show a lower bound of $\Omega(1/\varepsilon)$ in the restricted regime where $\varepsilon > 1/d$.
In the regime where $d$ is fixed, the quantum query complexity was completely characterized to be $\Theta_d(1/\varepsilon^{(d-1)/(d+1)})$~\cite{cornelissen2025compute}.

\paragraph{Hit-and-run walk.}
Very recently, Zhang~\cite{zhang2026hit} and Kook and Vempala~\cite{kook2026spectral} improved the analysis of the hit-and-run walk.
They showed that it mixes in $\widetilde{O}(d^2)$ steps to the uniform distribution, improving on the earlier $\tO(d^3)$ bound and matching the mixing time of the speedy walk.
If this bound can be extended to mixing to more general logconcave distributions, then this would likely improve the randomized volume estimation by Lov\'asz and Vempala~\cite{lovasz2006simulated} to match the current state-of-the-art randomized algorithm by Cousins and Vempala~\cite{CV18}, making $\widetilde{O}(d^{3.5} + d^3/\varepsilon^2)$ queries.
A simple back-of-the-envelope calculation then suggests that quantizing the resulting algorithm (essentially following \cite{chakrabarti2023quantum}) could also yield a $\widetilde{O}(d^{3.5} + d^{1.75}/\varepsilon)$-query quantum algorithm for volume estimation, matching the bound we obtain in this work.
While this approach might thus avoid the need for amortization, potentially yielding a conceptual simplification, it seems unlikely to further improve over our results without introducing additional ideas.

\paragraph{Improving our analysis.}
We expect that it should be possible to tighten the upper bound on the annealing cost for bigger updates of the step size $\delta$. To that end, one would have to improve the analysis of \Cref{lem:annealing-average-local-conductance}. If it were possible to replace the upper bound $(\delta'/\delta)^d$ by a constant, this would remove the additive $d^2$-overhead from the complexity statement in~\Cref{thm:result-vol-est-algo}.

To improve further, it seems one would have to additionally quadratically speed up the dependence on the schedule length. This way, one would obtain a full quadratic speed-up over the classical algorithm for partition function estimation, and in the context of volume estimation, this could potentially bring the complexity down to $d^{1.5}/\varepsilon$. Even though we don't see any inherent barriers to this approach, it also seems to require several non-trivial new ideas, which is why we leave this approach for future work.

Another interesting direction for future work would be to analyze the time complexity of this algorithm. In principle, it should be possible to analyze the time-efficiency of transducer compositions, as it was one of the motivating ideas behind developing the framework~\cite{belovs2024taming}. However, we derive some composition results from the adversary bound directly, and this seems to forgo any time complexity considerations.

\paragraph{Quantum rounding.}
Finally, it would be interesting to see if the rounding part of the complexity can be sped up by quantizing the algorithm developed by Jia, Laddha, Lee and Vempala~\cite{jia2026reducing}.
While this was already considered in \cite{chakrabarti2023quantum} (see \cref{footnote:chakrabarti}), it seems largely open.
The most immediate technical bottleneck seems to be to develop a non-destructive \textit{multivariate} mean estimation routine, generalizing the non-destructive univariate quantum mean estimation used in this and prior works. We leave this for future work as well. 

\paragraph{Quantum subroutine composition.}
It is interesting to compare our amortized quantum walk framework with that of Jeffery \cite{jeffery2022quantum}.
The main reason that we cannot invoke that work is that it assumes unit cost for generating a superposition over neighboring states (i.e., for genenerating the \emph{star state}), while it amortizes the cost of making an edge transition.
For the speedy walk, however, the transition cost is trivial while the cost of generating the start state is precisely the one to be amortized.
Our framework allows us to amortize the cost of generating the star state.
While it seems likely that it can also amortize edge transition costs, we leave this for future work.

\paragraph{Continuous-space quantum walks and transducers.}
Similar to earlier work \cite{chakrabarti2023quantum}, we put a significant amount of effort into discretizing the state space, and bounding the impact on the relevant Markov chains.
We do this discretization \emph{before} the quantization so that we can keep our quantum walks and transducers finite-dimensional.
An alternative approach would be to do discretization \emph{after} quantization.
I.e., develop the theory of quantum walks, transducers and electrical networks directly in the continuous-space setting,\footnote{For quantum walks this was already done to some extent in \cite{chakrabarti2023quantum}.} and then discretize the resulting continuous-variable quantum algorithm into a finite-dimensional one.
While this is bound to yield an elegant quantization, it seems less clear how to properly handle the discretization and its impact.

\subsection{Organization}

\cref{part:QWs} of the paper develops our amortized quantum walk toolkit based on transducers.
In its \Cref{sec:preliminaries}, we fix the notation and we introduce the necessary background on quantum walks and transducers. In \Cref{sec:ampl-ampl}, we derive amplitude amplification on the level of transducers. Then, in \Cref{sec:quantum-walks-transducers}, we construct a transducer that reflects through the quantum state that encodes the stationary distribution of a random walk.

\cref{part:volume} of the paper applies our new toolkit to the volume estimation problem.
In \cref{sec:CV-overview} we describe the volume estimation problem and the algorithm of Cousins and Vempala.
In \cref{sec:MC-measurable,sec:discretization,sec:ball-walks} we develop tools for discretizing the algorithm.
In \cref{sec:transducer-ball-walk} our amortized quantum walk version of the ball walk in \Cref{sec:ball-walks}, and finally in \Cref{sec:vol-est} we develop our new quantum algorithm for volume estimation.
    
    \newpage
    \part{Transducers and amortized quantum walks} \label{part:QWs}

    \section{Preliminaries: transducers and random walks}
    \label{sec:preliminaries}

    \subsection{Notation}

    We take $\N = \{1,2,\dots\}$. Whenever we have an expression involving the symbol $\pm$, we interpret this as a pair of expressions where we put replace it with $+$ and $-$, respectively. If we use the symbol $\pm$ multiple times in the same expressions, we replace each with the same sign, and if we use $\mp$ instead, we use the opposite sign.
    For all $d \in \N$, let $B_d := \{x \in \R^d : \norm{x} \leq 1\}$ be the unit ball.

    Let $d \in \N$ and $f,g : \R_{\geq0}^d \supseteq \D \to \R$. We say that $f \in O(g)$ if there exists a $C,M > 0$ such that for all $x \in \D$ with $\norm{x} \geq M$, then $|f(x)| \leq C \cdot |g(x)|$. We write $f \in \Omega(g)$ if and only if $g \in O(f)$, and we write $\Theta(f) = O(f) \cap \Omega(f)$. Furthermore, we write $f \in \widetilde{O}(g)$, if there exists a $k \in \N$ such that $f \in O(g \cdot \prod_{j=1}^d\log^k(x_j))$. Similarly, we write $f \in \widetilde{\Omega}(g)$ if and only if $g \in \widetilde{O}(f)$, and $\widetilde{\Theta}(f) = \widetilde{O}(f) \cap \widetilde{\Omega}(f)$. If the dimension of the domain is not clear from context, we write $f \in O(g \cdot \polylog((x_j)_{j \in S}))$, for $S \subseteq [d]$, if there exist $k_j \in \N$, for every $j \in S$, such that $f \in O(g(x) \cdot \prod_{j \in S} \log^{k_j}(x_j))$.

    When $\mathcal{H}$ is a Hilbert space, let $\mathcal{U}(\mathcal{H})$ be the set of unitary operators on $\mathcal{H}$. We write $\mathbbm{1}$ for the all-ones vector. For a measureable space $\Omega$, we write $\Delta_{\Omega}$ for the set of probability distributions over it.

    An undirected graph $G = (V,E)$ consists of a vertex set $V$ and an edge set $E \subseteq \binom{V}{2}$ of unordered pairs (in particular, $E$ contains no double edges). For all $v \in V$, we let $N(v) = \{w \in V : \{v,w\} \in E\}$ and we refer to $N(v)$ as the neighborhood of $v$.

    \subsection{Quantum adversary bound for state conversion}

    The state-conversion problem was first introduced by Lee, Mittal, Reichardt, \v{S}palek and Szegedy~\cite{lee2011quantum}, and captures the problem of converting one state into another, given access to an oracle. We consider a generalized version of the problem for arbitrary input oracles introduced by Belovs~\cite[Definition~7]{belovs2015variations}.

    \begin{definition}[\cite{lee2011quantum,belovs2015variations}]
        Let $\mathcal{V}_I$, $\mathcal{V}_O$ and $\mathcal{H}$ be Hilbert spaces, and $\D$ a set. For all $x \in \D$, let $\ket{\sigma_x} \in \mathcal{V}_I$, $\ket{\tau_x} \in \mathcal{V}_O$ and $O_x \in \mathcal{U}(\mathcal{H})$. Then, $P = \{(\ket{\sigma_x},\ket{\tau_x},O_x)\}_{x \in \D}$ is a state-conversion problem. Let $\mathcal{A}$ be a quantum query algorithm with state space $\mathcal{V}_I \oplus \mathcal{V}_O \oplus (\mathcal{H} \otimes \mathcal{W})$, where $\mathcal{W}$ is a Hilbert space, starting with an initial state that does not depend on $x$ and making queries to $O_x$. Then, $\mathcal{A}$ is said to solve the state-conversion problem with precision $\eta > 0$ if for all $x \in \D$, it maps $\ket{\sigma_x}$ to $\ket{\widetilde{\tau}_x}$, such that $\norm{\ket{\widetilde{\tau}_x} - \ket{\tau_x}} \leq \eta$.\footnote{Here, $\ket{\tau_x}$ is embedded in the state space in the canonical way, i.e., as $0 \oplus \ket{\tau_x} \oplus (0 \otimes 0)$.} We write $\mathsf{Q}_{\eta}(P)$ for the smallest number of queries to $O_x$ required for a quantum query algorithm to solve the state-conversion problem up to precision $\eta$, and $\mathsf{Q}(P) := \mathsf{Q}_{1/3}(P)$.
    \end{definition}

    A long line of work characterized the query complexity of state conversion in terms of a semi-definite program, referred to as the adversary bound, see~\cite{Amb00,spalek2006all,hoyer2007negative,reichardt2011reflections,belovs2015variations,belovs2023one} and the references therein. We use two formulations of the adversary bound, recently introduced by Belovs and Yolcu~\cite{belovs2023one}.

    \begin{definition}[\cite{belovs2023one}]
        \label{def:adversary-bounds}
        Let $P = \{(\ket{\sigma_x}, \ket{\tau_x}, O_x)\}_{x \in \D}$ be a state-conversion problem, with $O_x \in \mathcal{U}(\mathcal{H})$. We define the unidirectional adversary bound for $P$ as
        \begin{align*}
            \overrightarrow{\ADV}(P) := \min\quad & \max_{x \in \D}\norm{\ket{w_x}}^2, \\
            \text{s.t.}\quad & \bra{w_x}((I_{\mathcal{H}} - O_x^{\dagger}O_y) \otimes I_{\mathcal{W}})\ket{w_y} = \braket{\sigma_x}{\sigma_y} - \braket{\tau_x}{\tau_y}, & \forall x,y \in \D, \\
            & \ket{w_x} \in \mathcal{H} \otimes \mathcal{W}, \quad \mathcal{W} \text{ Hilbert space}, & \forall x \in \D,
        \end{align*}
        and we define the bidirectional adversary bound for $P$ as
        \begin{align*}
            \overleftrightarrow{\ADV}(P) := \min\quad & \max_{x \in \D} \max\left\{\norm{\ket{w_x^+}}^2, \norm{\ket{w_x^-}}^2\right\}, \\
            \text{s.t.}\quad & \bra{w_x^+}((I_{\mathcal{H}} - O_x^{\dagger}O_y) \otimes I_{\mathcal{W}})\ket{w_y^-} = \braket{\sigma_x}{\sigma_y} - \braket{\tau_x}{\tau_y}, & \forall x,y \in \D, \\
            & \ket{w_x^{\pm}} \in \mathcal{H} \otimes \mathcal{W}, \quad \mathcal{W} \text{ Hilbert space}, & \forall x \in \D.
        \end{align*}
    \end{definition}

    Belovs and Yolcu observed that the bidirectional version of the adversary bound captures having access to both forward and backward queries to the oracle, motivating the nomenclature for the unidirectional and bidirectional adversary bound. Specifically, they proved the following statement, and we provide the proof for convenience. 

    \begin{theorem}[{Based on \cite[Proposition~9.4]{belovs2023one}}]
        \label{thm:unidirectional-bidirectional}
        Let $P = \{(\ket{\sigma_x}, \ket{\tau_x}, O_x)\}_{x \in \D}$ be a state-conversion problem {with $O_x \in \mathcal{U}(\mathcal{H})$}, and let $P' = \{(\ket{\sigma_x}, \ket{\tau_x}, O_x \oplus O_x^{\dagger})\}_{x \in \D}$. Then, $\overleftrightarrow{\ADV}(P) = \overrightarrow{\ADV}(P')$.
    \end{theorem}

    Cornelissen studied a specific type of state-conversion problems, referred to as \textit{reflection problems}, where we either act as identity or minus identity on the input states~\cite{cornelissen2025quantum}, i.e., where $\ket{\tau_x} \in \{\pm \ket{\sigma_x}\}$ for all $x \in \D$. For such state-conversion problems, constructing a feasible solution to the adversary bound can be simplified, requiring only a single vector per input to create a feasible solution for the bidirectional adversary bound. 

    \begin{theorem}[{Based on \cite[Theorem~4.4]{cornelissen2025quantum}}]
        \label{thm:balancing-plus-minus}
        Let $P = \{(\ket{\psi_x}, \ket{\psi_x}, O_x)\}_{x \in \D_+} \cup \{(\ket{\psi_x},-\ket{\psi_x},O_x)\}_{x \in \D_-}$ be a state-conversion problem. Suppose that we have $\{\ket{w_x^+}\}_{x \in \D_+} \cup \{\ket{w_x^-}\}_{x \in \D_-} \subseteq \mathcal{H} \otimes \mathcal{W}$, such that for all $(x,y) \in \D_+ \times \D_-$,
        \[\bra{w_x^+}(I - O_x^{\dagger}O_y)\ket{w_y^-} = 2\braket{\psi_x}{\psi_y}.\]
        Then, $\overleftrightarrow{\ADV}(P) \leq \sqrt{\max_{x \in \D_+} \norm{\ket{w_x^+}}^2 \cdot \max_{x \in \D_-} \norm{\ket{w_x^-}}^2}$. 
    \end{theorem}

    \begin{proof}
        For any $\alpha > 0$, we let
        \begin{align*}
            \ket{v_x^+} := \alpha\ket{w_x^+} \oplus 0, & \qquad \text{and} \qquad \ket{v_x^-} := 0 \oplus \alpha\ket{w_x^+}, \qquad \forall x \in \D_+, \\
            \ket{v_x^+} := 0 \oplus \frac{1}{\alpha}\ket{w_x^-}, & \qquad \text{and} \qquad \ket{v_x^-} := \frac{1}{\alpha}\ket{w_x^-} \oplus 0, \qquad \forall x \in \D_-.
        \end{align*}
        Then, we elementarily check that 
        \begin{align*}
            \forall x,y \in \D_+, \bra{v_x^+}(I - O_x^{\dagger}O_y)\otimes I_2\ket{v_y^-} &= 0 = \braket{\psi_x}{\psi_y} - \braket{\psi_x}{\psi_y}, \\
            \forall x,y \in \D_-, \bra{v_x^+}(I - O_x^{\dagger}O_y)\otimes I_2\ket{v_y^-} &= 0 = \braket{\psi_x}{\psi_y} - (-\bra{\psi_x})(-\ket{\psi_y}), \\
            \forall x \in \D_+, y \in \D_-, \bra{v_x^+}(I - O_x^{\dagger}O_y)\otimes I_2\ket{v_y^-} &= \bra{w_x^+}(I - O_x^{\dagger}O_y)\ket{w_y^-} = 2\braket{\psi_x}{\psi_y} = \braket{\psi_x}{\psi_y} - \bra{\psi_x}(-\ket{\psi_y}), \\
            \forall x \in \D_-, y \in \D_+, \bra{v_x^+}(I - O_x^{\dagger}O_y)\otimes I_2\ket{v_y^-} &= \bra{w_x^-}(I - O_x^{\dagger}O_y)\ket{w_y^+} = 2\braket{\psi_x}{\psi_y} = \braket{\psi_x}{\psi_y} - (-\bra{\psi_x})\ket{\psi_y}.
        \end{align*}
        Thus, $\{\ket{v_x^{\pm}}\}_{x \in \D_+ \cup \D_-}$ is a feasible solution to $\overleftrightarrow{\ADV}(P)$, for any choice $\alpha>0$. We balance the maximal costs of inputs in $\mathcal D_+$ and $\mathcal D_-$ by choosing
        \[\alpha^2 := \sqrt{\frac{\max_{x \in \D_-} \norm{\ket{w_x^-}}^2}{\max_{x \in \D_+} \norm{\ket{w_x^+}}^2}}.\qedhere\]
    \end{proof}

    We observe that in the proof of the above theorem we could freely rebalance the norms of the vectors. This is a well-known property of the adversary bound of Boolean functions $f:\mathcal D \subseteq \{0,1\}^n \to \{0,1\}$, where the costs of $f^{-1}(0)$ and $f^{-1}(1)$ can be balanced. Below we emphasize the similarity through a rebalancing lemma that captures both the reflection problem and the Boolean query complexity problem. 
    
    Indeed, we can rebalance the cost whenever the set $\mathcal D$ can be partitioned into two sets $\mathcal D_+$ and $\mathcal D_-$ such that the following relation holds:
    \begin{equation}\label{eq:balance-condition}
            \braket{\sigma_x}{\sigma_y} - \braket{\tau_x}{\tau_y} = 0 \qquad x,y \in \D_+ \text{ or } x,y \in D_-.
    \end{equation}
    
    \begin{lemma}[Rebalancing]
        \label{lem:rebalancing}
        Let $P = \{(\ket{\sigma_x}, \ket{\tau_x}, O_x)\}_{x \in \D}$ be a state-conversion problem, with $O_x \in \mathcal{U}(\mathcal{H})$. Assume $\mathcal D$ can be partitioned into two sets $\mathcal D_+$ and $\mathcal D_-$ such that \cref{eq:balance-condition} holds. Then, if $\{\ket{w_x^\pm}\}_{x \in D}$ is a feasible solution to  $\overleftrightarrow{\ADV}(P)$ and $\alpha>0$, then $\{\ket{v_x^\pm}\}_{x \in D}$ is also a feasible solution to $\overleftrightarrow{\ADV}(P)$ where we set $\ket{v_x^\pm} = \alpha \ket{w_x^\pm}$ for $x \in \D_+$ and $\ket{v_x^\pm} = \alpha^{-1} \ket{w_x^\pm}$ for $x \in \D_-$.
    \end{lemma}
    
    Note that reflection problems clearly satisfy the premise of the rebalancing lemma, where $\D_{\pm} = \{x \in \D: \ket{\tau_x} = \pm \ket{\sigma_x}\}$. The usual query complexity adversary bound of Boolean functions $f:\mathcal D \subseteq \{0,1\}^n \to \{0,1\}$ corresponds to the setting where $\ket{\sigma_x} = \ket{0}$, $\ket{\tau_x} = \ket{f(x)}$ and $\D_+ = \{x : f(x)=1\}$ and $\D_- = \{x: f(x)=0\}$.

    Finally, we provide a lemma that allows us to compose two reflection operations so that we can reflect around a larger subspace.

    \begin{lemma}
        \label{lem:combine-reflections}
        Let $\D$ be a set with two partitions $\D = \D_+^{(1)} \sqcup \D_-^{(1)} = \D_+^{(2)} \sqcup \D_-^{(2)}$, such that $\D_+^{(1)} \cap \D_+^{(2)} = \varnothing$. Suppose that we have two state-conversion problems $P^{(j)} = \{(\ket{\psi_x}, \ket{\psi_x}, O_x)\}_{x \in \D_+^{(j)}} \cup \{(\ket{\psi_x},-\ket{\psi_x},O_x)\}_{x \in \D_-^{(j)}}$, for $j \in \{1,2\}$, with corresponding feasible solutions $\{\ket{(w^{(j)})_x^+}\}_{x \in \D_+^{(j)}} \cup \{\ket{(w^{(j)})_x^-}\}_{x \in \D_-^{(j)}}$ in the sense of \Cref{thm:balancing-plus-minus}. Then, we can generate a feasible solution $\{\ket{w_x^+}\}_{x \in \D_+^{(1)} \cup \D_+^{(2)}} \cup \{\ket{w_x^-}\}_{x \in \D_-^{(1)} \cap \D_-^{(2)}}$, again in the sense of \Cref{thm:balancing-plus-minus}, such that
        \[\norm{\ket{w_x^+}}^2 = \begin{cases}
            \norm{\ket{(w^{(1)})^+_x}}^2, & \text{if } x \in \D_+^{(1)}, \\
            \norm{\ket{(w^{(2)})^+_x}}^2, & \text{if } x \in \D_+^{(2)},
        \end{cases} \qquad \text{and} \qquad \norm{\ket{w_x^-}}^2 = \norm{\ket{(w^{(1)})_x^-}}^2 + \norm{\ket{(w^{(2)})_x^-}}^2.\]
    \end{lemma}

    \begin{proof}
        We construct the solution directly. If $x \in \D_+^{(1)} \cup \D_+^{(2)}$, then
        \[\ket{w_x^+} = \begin{cases}
            \ket{(w^{(1)})_x^+} \oplus 0, & \text{if } x \in \D_+^{(1)}, \\
            0 \oplus \ket{(w^{(2)})_x^+}, & \text{if } x \in \D_+^{(2)},
        \end{cases}\]
        and if $x \in \D_-^{(1)} \cap \D_-^{(2)}$, then
        \[\ket{w_x^-} = \ket{(w^{(1)})_x^-} \oplus \ket{(w^{(2)})_x^-}.\]
        The constraints and claimed complexities are easily verified.
    \end{proof}

    \subsection{Transducers}

    Belovs, Jeffery and Yolcu~\cite{belovs2024taming} introduced transducers, a new and exciting way to turn solutions to the adversary bound into quantum algorithms. 

    \begin{definition}[{\cite[Section~5.1]{belovs2024taming}}]
        \label{def:transducer}
        Let $\mathcal{V}, \mathcal{H}$ be Hilbert spaces. A transducer is a unitary $U \in \mathcal{U}(\mathcal{V} \oplus \mathcal{H})$, with public space $\mathcal{V}$ and private space $\mathcal{H}$. If for a state $\ket{\sigma} \in \mathcal{V}$ and $\ket{w} \in \mathcal{H}$, we have $U(\ket{\sigma} \oplus \ket{w}) = \ket{\tau} \oplus \ket{w}$, then we write
        \[\ket{\sigma} \overset{U}{\rightsquigarrow} \ket{\tau} \qquad \text{with complexity} \qquad W(U,\ket{\sigma}) := \min \{\norm{\ket{w}}^2 : \ket{w} \in \mathcal{H}, \ket{\sigma} \oplus \ket{w} \overset{U}{\mapsto} \ket{\tau} \oplus \ket{w}\}.\]
        We refer to $W(U,\ket{\sigma})$ as the transduction complexity of $U$ applied to $\ket{\sigma}$.
    \end{definition}

    Transducers exhibit particularly nice linear-algebraic properties, as stated in the following theorem.

    \begin{theorem}[{\cite[Theorem~5.1]{belovs2024taming}}]
        \label{thm:transducer-catalyst}
        Let $\mathcal{V},\mathcal{H}$ be finite-dimensional Hilbert spaces. Let $U \in \mathcal{U}(\mathcal{V} \oplus \mathcal{H})$ be a transducer. Then, for every $\ket{\sigma} \in \mathcal{V}$, there exists a unique $\ket{\tau} \in \mathcal{V}$ such that $\ket{\sigma} \raisebox{-.2em}{$\overset{U}{\rightsquigarrow}$} \ket{\tau}$, and this mapping is unitary. Moreover, the state $\ket{w} \in \mathcal{H}$ such that $\ket{\sigma} \oplus \ket{w} \raisebox{-.2em}{$\overset{U}{\mapsto}$} \ket{\tau} \oplus \ket{w}$ is uniquely defined if we require it to be orthogonal to the $1$-eigenspace of $\Pi_{\mathcal{H}}U\Pi_{\mathcal{H}}$, where $\Pi_{\mathcal{H}}$ is the projection onto $\mathcal{H}$, and the operation mapping $\ket{\sigma}$ to $\ket{w}$ is linear.
    \end{theorem}

   Transducers can be turned into quantum algorithms, as witnessed by the following result.

    \begin{theorem}[{\cite[Theorem~5.5]{belovs2024taming}}]
        \label{thm:transducer-to-alg}
        Let $U \in \mathcal{U}(\mathcal{V} \oplus \mathcal{H})$ be a transducer and $K \in \N$. Then, there is a quantum quantum algorithm $\mathcal{A}_K$ that makes $K$ calls to $U$ such that for all $\ket{\sigma} \in \mathcal{V}$,
        \[\ket{\sigma} \overset{U}{\rightsquigarrow} \ket{\tau} \qquad \Rightarrow \qquad \norm{\mathcal{A}_K\ket{\sigma} - \ket{\tau}} \leq 2\sqrt{\frac{W(U,\ket{\sigma})}{K}}.\]
    \end{theorem}

    In the query setting, transducers can make calls to a query oracle $O_x \in \mathcal{U}(\mathcal{H})$. Belovs, Jeffery and Yolcu show that such 
    transducers can always be cast into a canonical form: a     \textit{canonical transducer}.\footnote{We note that the definition used in \cite{belovs2024taming} is slightly more general, as it allows for the subdivision of the catalyst space into two parts. We do not 
    need this additional property for our purposes. Additionally, we deviate from \cite{belovs2024taming} by restricting the space of oracles to be some subset of $\mathcal{U}(\mathcal{H})$, rather than all such unitary operators. This will be more convenient for the presentation of our results later on.}

    \begin{definition}[Canonical transducer, {\cite[Section~7.1]{belovs2024taming}}]
        Let $\mathcal{V}, \mathcal{H}, \mathcal{W}$ be Hilbert spaces, and let $O := \{O_x\}_{x \in \D} \subseteq \mathcal{U}(\mathcal{H})$ be a subset of unitary operations referred to as oracles.
        \begin{enumerate}[nosep]
            \item A unidirectional canonical transducer $T$ is a function mapping $O$ to transducers as $T(O_x) := U(I_{\mathcal{V}} \oplus (O_x \otimes I_{\mathcal{W}}))$, where $U \in \mathcal{U}(\mathcal{V} \oplus (\mathcal{H} \otimes \mathcal{W}))$ does not depend on the input $x$.
            \item A bidirectional canonical transducer $T$ is defined as $T(O_x) := U(I_{\mathcal{V}} \oplus ((O_x \oplus O_x^{\dagger}) \otimes I_{\mathcal{W}}))$, where $U \in \mathcal{U}(\mathcal{V} \oplus (\mathcal{H}^{\oplus 2} \otimes \mathcal{W}))$ does not depend on the input $x$.
        \end{enumerate}
    \end{definition}

    From the above definition, it is apparent that a canonical transducer requires an oracle $O_x$ to be a well-defined transducer. Typically, the input oracle $O_x$ is easily derived from context (e.g., in the volume estimation problem, it is always the membership oracle $O_K$). Therefore, we will often leave the oracle that we use in the definition of a canonical transducer implicit.
    Similarly, we often also leave implicit the dependence on the set $\D$ on which the source and target states depend. That is, if we let $T$ be a canonical transducer that maps $\ket{\psi}$ to $\ket{\varphi}$, we often mean that for all $x \in \D$, where $\D$ is some set to be inferred from context, $T(O_x)$ is a canonical transducer with respect to an implicit oracle $O_x$, that maps $\ket{\psi_x}$ to $\ket{\varphi_x}$.

    Note that canonical transducers make an explicit distinction between the public space $\mathcal{V}$, and the private space $\mathcal{H} \otimes \mathcal{W}$ or $\mathcal{H}^{\oplus 2} \otimes \mathcal{W}$. These spaces are not allowed to depend on the implicit input $x \in \D$. 

    Finally, observe that all bidirectional canonical transducers are also unidirectional canonical transducers, but with respect to a different set of oracles, i.e., $\{O_x \oplus O_x^{\dagger}\}_{x \in \D}$ instead of $\{O_x\}_{x \in \D}$. Thus, all results that hold for unidirectional canonical transducers also hold for bidirectional canonical transducers.

    Belovs, Jeffery and Yolcu showed that feasible solutions to the unidirectional adversary bound can be turned into canonical transducers~\cite[Section~8.1]{belovs2024taming}.

    \begin{theorem}[{\cite[Section~8.1]{belovs2024taming}}]
        \label{thm:adv-to-transducer}
        Let $P = \{(\ket{\sigma_x}, \ket{\tau_x}, O_x)\}_{x \in \D}$ be a state-conversion problem, where for all $x \in \D$, $\ket{\sigma_x},\ket{\tau_x} \in \mathcal{V}$ and $O_x \in \mathcal{U}(\mathcal{H})$. Let $\{\ket{w_x}\}_{x \in \D} \subseteq \mathcal{H} \otimes \mathcal{W}$ be a feasible solution to the unidirectional adversary bound for $P$. Then, there exists a unidirectional canonical transducer $T$ with input-independent unitary $U \in \mathcal{U}(\mathcal{V} \oplus (\mathcal{H} \otimes \mathcal{W}))$, such that for all $x \in \D$, $T(O_x) := U(I_{\mathcal{V}} \oplus (O_x \otimes I_{\mathcal{W}}))$ transduces $\ket{\sigma_x}$ to $\ket{\tau_x}$ with catalyst state $\ket{w_x}$. In other words, for all $x \in \D$,
        \[\ket{\sigma_x} \overset{T_x}{\rightsquigarrow} \ket{\tau_x} \qquad \text{with complexity} \qquad W(T_x, \ket{\sigma_x}) \leq \norm{\ket{w_x}}^2.\]
    \end{theorem}

    The above theorem strictly speaking requires $\ket{\sigma_x}$ and $\ket{\tau_x}$ to be elements of the same Hilbert space~$\mathcal{V}$. In general, whenever $\ket{\sigma_x} \in \mathcal{V}_I$ and $\ket{\tau_x} \in \mathcal{V}_O$, we take $\mathcal{V} := \mathcal{V}_I \oplus \mathcal{V}_O$, and embed $\ket{\sigma_x}$ and $\ket{\tau_x}$ into $\mathcal{V}$ in the canonical way.

    We note here that there is also a reverse connection between canonical transducers and feasible solutions to the adversary bound. As far as we know, this has not been explicitly stated before, but it follows immediately from the observations in \cite[Section~8.1]{belovs2024taming}.

    \begin{theorem}
        \label{thm:transducer-to-adv}
        Let $\D$ be a set, $U \in \mathcal{U}(\mathcal{V} \oplus (\mathcal{H} \otimes \mathcal{W}))$, and $\{O_x\}_{x \in \D} \subseteq \mathcal{U}(\mathcal{H})$ be a set of oracles. Suppose that we have a unidirectional canonical transducer $T$ such that for all $x \in \D$, $T(O_x) := U(I_{\mathcal{V}} \oplus (O_x \otimes I_{\mathcal{W}}))$ implements the mapping $\ket{\sigma_x} \oplus \ket{w_x} \mapsto \ket{\tau_x} \oplus \ket{w_x}$. Then, $\{\ket{w_x}\}_{x \in \D}$ forms a feasible solution to the unidirectional adversary bound for the state conversion problem $\{(\ket{\sigma_x}, \ket{\tau_x}, O_x)\}_{x \in \D}$.
    \end{theorem}

    \begin{proof}
        For all $x \in \D$, we observe that
        \[\ket{\sigma_x} \oplus \ket{w_x} \overset{T(O_x)}{\mapsto} \ket{\sigma_x} \oplus O_x\ket{w_x} \overset{U}{\mapsto} \ket{\tau_x} \oplus \ket{w_x},\]
        and as such, for any $x,y \in \D$, we have
        \[\braket{\tau_x}{\tau_y} + \braket{w_x}{w_y} = (\bra{\sigma_x} \oplus \bra{w_x}(O_x^{\dagger} \otimes I_{\mathcal{W}}))U^{\dagger}U(\ket{\sigma_y} \oplus (O_y \otimes I_{\mathcal{W}})\ket{w_y}) = \braket{\sigma_x}{\sigma_y} + \bra{w_x}(O_x^{\dagger}O_y \otimes I_{\mathcal{W}})\ket{w_y}.\]
        By rearranging the left- and right-hand sides, we recover the constraint in the unidirectional adversary bound, as defined in \Cref{def:adversary-bounds}, from which we conclude that $\{\ket{w_x}\}_{x \in \D}$ is indeed a feasible solution.
    \end{proof}

    We make use of several results about transduction complexity from \cite{belovs2024taming}. We start by a characterization of the transduction complexity of the inverse of a transducer.

    \begin{theorem}[{Transducer inverse \cite[Proposition~9.1]{belovs2024taming}}]
        \label{thm:inverse-transducer}
        Let $T$ be a bidirectional canonical transducer, mapping $\ket{\sigma}$ to $\ket{\tau}$. Then, we can construct a bidirectional canonical transducer $T^{\dagger}$ that acts as
        \[\ket{\tau} \overset{T^{\dagger}}{\rightsquigarrow} \ket{\sigma}, \qquad \text{with complexity} \qquad W(T^{\dagger},\ket{\tau}) = W(T, \ket{\sigma}).\]
    \end{theorem}

    We will also make use of two composition theorems for canonical transducers.

    \begin{theorem}[Parallel composition~{\cite[Proposition~9.4]{belovs2024taming}}]
        \label{thm:parallel-composition}
        Let $T^{(1)}, \dots, T^{(n)}$ be undirectional canonical transducers, transducing $|\sigma_x^{(j)}\rangle$ into $|\tau_x^{(j)}\rangle$, for all $j \in [n]$. Then, we can construct a undirectional canonical transducer $T$ that acts as
        \[\bigoplus_{j=1}^n \ket{\sigma^{(j)}} \overset{T}{\rightsquigarrow} \bigoplus_{j=1}^n \ket{\tau^{(j)}}, \qquad \text{with complexity} \qquad W\left(T, \bigoplus_{j=1}^n \ket{\sigma^{(j)}}\right) \leq \sum_{j=1}^n W\left(T^{(j)}, \ket{\sigma^{(j)}}\right).\]
    \end{theorem}

    \begin{theorem}[Sequential composition~{\cite[Proposition~9.9]{belovs2024taming}}]
        \label{thm:sequential-composition}
        Let $T^{(1)}, \dots, T^{(n)}$ be undirectional canonical transducers with the same public and private space, transducing $|\psi^{(j-1)}\rangle$ into $|\psi^{(j)}\rangle$, for all $j \in [n]$. Then, we can construct a unidirectional canonical transducer $T$ that acts as
        \[\ket{\psi^{(0)}} \overset{T}{\rightsquigarrow} \ket{\psi^{(n)}}, \qquad \text{with complexity} \qquad W\left(T, \ket{\psi^{(0)}}\right) \leq \sum_{j=1}^n W\left(T^{(j)}, \ket{\psi^{(j-1)}}\right).\]
    \end{theorem}

    We note that the referenced theorems suggest a full pipeline for converting feasible solutions to the bidirectional adversary bound into quantum algorithms. Indeed, we first turn it into a feasible solution of the unidirectional adversary bound with both forward and backward oracle access, using \Cref{thm:unidirectional-bidirectional}. Then, we turn the resulting feasible solution into a set of transducers using \Cref{thm:adv-to-transducer}. Finally, we use \Cref{thm:transducer-to-alg} to turn these transducers into quantum algorithms. Note that each of the transducers produced in \Cref{thm:adv-to-transducer} makes exactly one query to the input oracles, and therefore the final algorithm makes $K$ queries. Finally, observe that it suffices to take $K$ as $\Theta(\ADV(P)/\eta^2)$ to obtain an approximation of the final state $\ket{\tau_x}$ with precision $\eta > 0$.
    In particular, we have $\mathsf{Q}(P) \in \Theta(\ADV(P))$.

    \subsection{Two-subspace transducer}

    We state one of the main technical results from the recent work by Apers, Roland and Zhang~\cite{apers2026elfs}.

    \begin{theorem}[{\cite[Theorem~1.1]{apers2026elfs}}]
        \label{thm:two-subspace-transducer}
        Let $\Pi$ and $\Delta$ be projectors on subspaces of a Hilbert space $\mathcal{H}$. Then, $U = (2\Pi - I)(2\Delta - I)$ is a transducer on $\mathrm{Ker}(\Pi)$, with catalyst space $\mathrm{Im}(\Pi)$, acting as
        \begin{align*}
            \ket{\psi} \overset{U}{\rightsquigarrow} \ket{\psi}, & \qquad \text{with catalyst} \qquad 0, \\
            \ket{\overline{\psi}} \overset{U}{\rightsquigarrow} -\ket{\overline{\psi}}, & \qquad \text{with catalyst} \qquad (\Pi - \Pi\Delta\Pi)^+\Delta\ket{\overline{\psi}},
        \end{align*}
        for any $\ket{\psi} \in \mathrm{Ker}(\Pi) \cap \mathrm{Ker}(\Delta)$ and $\ket{\overline{\psi}} \in \mathrm{Ker}(\Pi) \cap (\mathrm{Ker}(\Pi) \cap \mathrm{Ker}(\Delta))^{\perp}$.
    \end{theorem}

    Apers, Roland and Zhang explore this transducer in the context of Markov chains and quantum walks.
    In \Cref{sec:quantum-walks-transducers}, we build on this intuition and show how one can construct a transducer that reflects through the stationary distribution of a random walk.

    \subsection{Markov chains, random walks and electrical networks}
    \label{subsec:prelims-markov-chains}

    We briefly recall the definition of Markov processes and their relation to random walks on graphs. We base our exposition on the book by Levin and Peres~\cite{levin2017markov}. To simplify notation, we restrict to the setting where the state space is finite, but this is not a fundamental restriction.

    \begin{definition}[{See e.g.~\cite[Sections~1.1-1.5]{levin2017markov}}]
        \label{def:Markov-process}
        Let $\Omega$ be a finite set and $P \in \R_{\geq 0}^{\Omega \times \Omega}$ such that $P\mathbbm{1} = \mathbbm{1}$. Then, $(\Omega,P)$ is a Markov process with state space $\Omega$ and probability-transition matrix $P$. A stationary distribution is a probability distribution $\pi \in \Delta_{\Omega}$ such that $\pi P = \pi$. We say that the Markov process is irreducible, if for any two states $\omega_1, \omega_2 \in \Omega$, there exits a $t \in \N$ such that $(P^t)(\omega_1, \omega_2) > 0$.
    \end{definition}

    Irreducible Markov processes have the property that they have a unique stationary distribution.

    \begin{theorem}[{See e.g.~\cite[Corollary~1.17]{levin2017markov}}]
        \label{thm:irreducible-implies-unique-stationary}
        An irreducible Markov process has a unique stationary distribution with positive probability everywhere.
    \end{theorem}

    We will only be considering irreducible Markov processes in this work, and so we often implicitly use this fact when we refer to the stationary distribution of the Markov process. We proceed by looking at a particularly elegant subset of Markov processes, known as reversible Markov chains.

    \begin{definition}[{See e.g.~\cite[Section~1.6]{levin2017markov}}]
        \label{def:reversible-random-walk}
        Let $(\Omega,P)$ be an irreducible Markov process. It is reversible with respect to a distribution $\pi$ if for any two states $\omega_1, \omega_2 \in \Omega$, we have
        \begin{equation}
            \label{eq:detailed balance}
            \pi(\omega_1)P(\omega_1,\omega_2) = \pi(\omega_2)P(\omega_2,\omega_1).
        \end{equation}
        We refer to this condition as the detailed balance condition.  \end{definition}
        If $P$ is reversible with respect to a distribution $\pi$, then $\pi$ is a stationary distribution of $P$.

    We consider a particular set of Markov processes known as random walks. In particular, we consider the following class of random walks over electric networks. 

    \begin{definition}[{See e.g.~\cite[Section~9.1]{levin2017markov}}]
        \label{def:random-walk-electrical-network}
        Let $G = (\Omega,E)$ be a connected undirected graph, and let $r : E \to \R_{>0}$ be a resistance function. We refer to $(G,r)$ as an electrical network. We define $M = (\Omega,P)$ as
        \[P(\omega_1,\omega_2) := \begin{cases}
            \frac{r_{\omega_1}}{r_{\omega_1\omega_2}}, & \text{if } \omega_1\omega_2 \in E, \\
            0, & \text{otherwise},
        \end{cases} \quad \text{where} \quad r_{\omega_1} := \left[\sum_{\omega_2 \in N(\omega_1)} \frac{1}{r_{\omega_1\omega_2}}\right]^{-1}, \quad \text{and} \quad R := \left[\sum_{\omega \in \Omega} \frac{1}{r_{\omega}}\right]^{-1}.\]
    \end{definition}

    Random walks over connected electrical networks are always irreducible and reversible.

    \begin{theorem}[{See e.g.~\cite[Section~9.1]{levin2017markov}}]
        \label{thm:electrical-network-walk}
        Let $G = (\Omega,E)$ be a connected undirected graph, and $r : E \to \R_{>0}$ be a resistance function. The random walk over $G$ with resistance function $r$ is irreducible, reversible, and the stationary distribution is $\pi_{\omega} = R/r_{\omega}$.
    \end{theorem}

    In fact, the previous implication is indeed an equivalence.

    \begin{theorem}[{See e.g.~\cite[Section~9.1]{levin2017markov}}]
        \label{thm:walk-electrical-network}
        Let $M = (\Omega, P)$ be an irreducible, reversible Markov process. Then, there is an electrical network $(G,r)$ such that $M$ is the random walk on $G$ with resistance function $r$, and $R = 1$.
    \end{theorem}

    For electrical networks, we have a linear-algebraic formulation of the effective resistance. 
    We base our exposition on \cite{cornelissen2025quantum}, but a similar treatment can be found in \cite{levin2017markov}, and sources mentioned therein. 

    \begin{definition}[{\cite[Definition~2.8]{cornelissen2025quantum}}]
        \label{def:effective-resistance}
        Let $(G = (V,E), r)$ be an electrical network. We arbitrarily associate a direction to every edge $E$, and we write $N_+(v)$ for the outgoing edges at $v$ and $N_-(v)$ for the incoming ones. Let $f : E \to \C$ be a flow, and let $\delta_f \in \C^V$ be its net-flow, defined as $\delta_f(v) = \sum_{e \in N_+(v)} f_e - \sum_{e \in N_-(v)} f_e$. Let $\C^E$ be the flow space, and we write $\ket{f} = \sum_{e \in E} f_e\sqrt{r_e}\ket{e}$. The energy of a flow is defined as $\norm{\ket{f}}^2$, and for any net-flow $\delta \in \C^V$, we define the effective resistance as
        \[R_{\mathrm{eff}}(G,r;\delta) := \min_{\substack{f : E \to \C \\ \delta_f = \delta}} \norm{\ket{f}}^2.\]
    \end{definition}

    Let us make the linear-algebraic formulation more concrete. Based on the arbitrarily assigned direction to each edge, let $B \in \R^{V \times E}$ be the resistance weighted vertex-edge incidence matrix: $B_{v,(u,v)} = r_{u,v}^{-1/2}$, $B_{v,(v,u)}=-r_{u,v}^{-1/2}$, and zero elsewhere. We define the weighted Laplacian $L = B B^T$. The effective resistance of a netflow $\delta \in \R^V$ is then $\delta^T L^+ \delta$, where $(\cdot)^+$ denotes the Moore-Penrose pseudo-inverse.

    Finally, we define the notions of spectral gap and conductance of Markov processes.

    \begin{definition}[{See e.g.~\cite[Section~7.2]{levin2017markov}}] \label{def:conductance-gap}
        Let $(\Omega,P)$ be an irreducible, reversible Markov process, and let $\lambda$ be the second-largest eigenvalue of $P$. Then we write $\gamma(P) := 1 - \lambda$ for the spectral gap of the Markov process. We also write $\varphi(P)$ for the conductance, i.e.,
        \[\varphi(P) := \min_{S \subseteq \Omega} \frac{\sum_{(x,y) \in S \times (\Omega \setminus S)} \pi(x)P(x,y)}{\min\{\pi(S), \pi(\Omega \setminus S)\}}, \qquad \text{where} \qquad \pi(A) := \sum_{x \in A} \pi(x).\]
    \end{definition}
    We record the following variational characterization of $\gamma(P)$ using the Dirichlet form associated to $P$ and its stationary distribution. 
    \begin{lemma}[{See e.g.~\cite[Lemma~13.7]{levin2017markov}}]
        \label{lem:gap-dirichlet}
        Let $(\Omega,P)$ be an irreducible, reversible Markov process with stationary distribution $\pi$, then 
        \[
        \gamma(P) = \frac12 \min_{\substack{f:\Omega \to \R,\\ \Var_\pi(f) \neq 0}} \frac{\sum_{x,y \in \Omega} (f(x)-f(y))^2 \pi(x) P(x,y)}{\Var_\pi(f)}.
        \]
    \end{lemma}

    The spectral gap and conductance are related according to the well-known Cheeger inequalities.

    \begin{theorem}[Cheeger's inequalities~(see e.g.~{\cite[Equation~(13.6)]{levin2017markov})}]
        \label{thm:conductance-spectral-gap}
        For any reversible Markov process, we have
        \[\frac{\varphi(P)^2}{2} \leq \gamma(P) \leq 2\varphi(P).\]
    \end{theorem}

    The linear-algebraic perspective again helps us to relate the effective resistance of a netflow to the spectral gap of the random walk. Indeed, first observe that the random walk transition matrix $P$ satisfies the identity $L = D(I-P)$, where $D=\diag(\pi)$ and $\pi$ is the stationary distribution (cf.~\cite[Theorem 2.11]{cornelissen2025quantum}). This implies that the spectrum of $I-P$ satisfies $\spec(I-P) = \spec(D^{-1}L) = \spec(D^{-1/2} L D^{-1/2})$. The spectrum of $I-P$ also takes the form $\{0,\gamma(P),...\}$. We can thus characterize $\gamma(P)^{-1}$ by considering the largest eigenvalue of the pseudoinverse of $D^{-1/2} L D^{-1/2}$. That is, when $\gamma(P) > 0$, $D^{-1/2}LD^{-1/2}$ is invertible on the subspace $\Span\{D^{1/2}\mathbbm{1}\}^{\perp}$, and so
    \[\gamma(P)^{-1} = \max_{\substack{0 \neq u \in \C^{\Omega} \\ u^*D^{1/2}\mathbbm{1} = 0}} \frac{u^* (D^{-1/2} L D^{-1/2})^+ u}{\norm{u}^2} = \max_{\substack{0 \neq \delta \in \C^{\Omega} \\ \delta^*\mathbbm{1} = 0}} \frac{\delta^* L^+ \delta}{\norm{D^{-1/2}\delta}^2} = \max_{\substack{0 \neq \delta \in \R^{\Omega} \\ \delta^*\mathbbm{1} = 0}} \frac{R_{\mathrm{eff}}(G,r;\delta)}{\norm{D^{-1/2}\delta}^2}.\]
    This allows us to relate effective resistances to the spectral gap of the walk, as shown in the following theorem.
    \begin{theorem}[{\cite[Lemma~3.2]{cornelissen2025quantum}}]
        \label{thm:effective-resistance-spectral-gap}
        Let $(\Omega,P)$ be a Markov process, and let $(G,r)$ be the associated electrical network. Let $D = \diag(\pi)$. Then, for any net-flow $\delta \in \C^{\Omega}$, we have
        \[R_{\mathrm{eff}}(G,r;\delta) \leq \frac{\norm{D^{-1/2}\delta}^2}{\gamma(P)}.\]
    \end{theorem}

    \section{Amplitude amplification with transducers}
    \label{sec:ampl-ampl}

    In this section, we give an amplitude amplification construction on the level of transducers. We start by introducing a basic version of the construction in \Cref{subsec:ampl-ampl-base-case}. Then, we provide two applications in \Cref{subsec:ampl-ampl-two-state-rotation,subsec:ampl-ampl-rejection-sampling}.

    Our constructions are to be contrasted with \cite[Theorem~6.1]{apers2026elfs}, where Apers, Roland and Zhang also develop an amplitude amplification construction on the level of transducers. The benefit of the construction presented in this section over theirs is that this construction does not produce a garbage state. This is necessary for the implementation of the quantum walk operator, which we will see in \Cref{sec:quantum-walks-transducers}.
    The downside of our construction is that we do not currently know if it is possible to implement this transducer time-efficiently, in contrast with the transducer constructed in \cite{apers2026elfs}. We leave this for future work.

    \subsection{Base case}
    \label{subsec:ampl-ampl-base-case}

    The simplest instantiation of the amplitude amplification problem is where we start in the state $\ket{\psi_0}\ket{0} + \ket{\psi_1}\ket{1}$, we can reflect through this state, and we wish to prepare $\ket{\psi_1}$. The construction we present in this section is a generalization of the routine that prepares a uniform superposition over the solutions in Grover's search, as introduced by Cornelissen, Edenhofer, Schaeffer and Szab\'o~\cite{cornelissen2026quantum}. Concurrently, it was also independently discovered by Dubus, Ladeuze and Roland~\cite{dubus2026transducer}.

    We start by recalling a particular class of positive semidefinite matrices  which we then use to construct a feasible solution to the bidirectional adversary bound for an amplitude-amplification state-conversion problem. 

    \begin{lemma}
        \label{lem:ratio}
        Let $\D$ be a set, and $(b_x)_{x \in \D} \subseteq \R_{>0}$. Let $A \in \R^{\D \times \D}$ be defined by $A[x,y] = 1/(b_x + b_y)$. Then $A \succeq 0$.
    \end{lemma}

    \begin{proof}
        For each $x \in \D$, let $\ket{\chi_x} =  e^{-b_x t} \in L^2(0,\infty)$. Then $A[x,y] = \braket{\chi_x}{\chi_y} =  \int_{0}^\infty e^{-(b_x+b_y)t}  \;\mathrm{d}t = \frac{1}{b_x+b_y}$, implying that $A$ is a Gram matrix and hence positive semidefinite.
    \end{proof}

    \begin{theorem}[Amplitude-amplification adversary solution]
        \label{thm:ampl-ampl-adv}
        Let $\D$ be a set, and for all $x \in \D$, we write $\ket{\psi_x^0},\ket{\psi_x^1} \in \mathcal{H}$, such that for all $x \in \D$, $\norm{\ket{\psi_x^0}}^2 + \norm{\ket{\psi_x^1}}^2 = 1$ and $\norm{\ket{\psi_x^0}},\norm{\ket{\psi_x^1}} > 0$. Then, we write
        \[\ket{\psi_x} = \ket{\psi_x^0}\ket{0} + \ket{\psi_x^1}\ket{1}, \qquad \text{and} \qquad \ket{\overline{\psi}_x} = \frac{\norm{\ket{\psi_x^1}}}{\norm{\ket{\psi_x^0}}}\ket{\psi_x^0}\ket{0} - \frac{\norm{\ket{\psi_x^0}}}{\norm{\ket{\psi_x^1}}}\ket{\psi_x^1}\ket{1}.\]
        Suppose we have a feasible solution $\{\ket{v_x^{\pm}}, \ket{\overline{v}_x^{\pm}}\}_{x \in \D}$ to the bidirectional adversary bound for the state-conversion problem $\{(\ket{\psi_x},\ket{\psi_x},O_x),(\ket{\overline{\psi_x}},-\ket{\overline{\psi}_x}, O_x)\}_{x \in \D}$. Then, we can construct a feasible solution $\{\ket{w_x^{\pm}},\ket{\overline w_x^{\pm}}\}_{x \in \D}$ to the bidirectional adversary bound for the state-conversion problem
        \[\left\{\left(\ket{\psi_x}, \frac{\ket{\psi_x^1}}{\norm{\ket{\psi_x^1}}}\ket{1}, O_x\right),\left(\ket{\overline\psi_x}, -\frac{\ket{\psi_x^0}}{\norm{\ket{\psi_x^0}}}\ket{0}, O_x\right)\right\}_{x \in \D}\]
        with
        \begin{align*}
            \norm{\ket{w_x^{\pm}}}^2 &= \frac{(1+2\norm{\ket{\psi_x^1}}) \norm{\ket{v_x^{\pm}}}^2 + \norm{\ket{\psi_0^x}}^2\norm{\ket{\overline{v}_x^{\pm}}}^2}{4\norm{\ket{\psi_x^1}}} \leq \frac{\norm{\ket{v_x^{\pm}}}^2 + \norm{\ket{\overline{v}_x^{\pm}}}^2}{4\norm{\ket{\psi_x^1}}} + \frac{\norm{\ket{v_x^\pm}}^2}2, \\
            \norm{\ket{\overline w_x^{\pm}}}^2 &= \frac{(1+2\norm{\ket{\psi_x^1}}) \norm{\ket{\overline v_x^{\pm}}}^2 + \norm{\ket{\psi_0^x}}^2\norm{\ket{v_x^{\pm}}}^2}{4\norm{\ket{\psi_x^1}}} \leq \frac{\norm{\ket{v_x^{\pm}}}^2 + \norm{\ket{\overline{v}_x^{\pm}}}^2}{4\norm{\ket{\psi_x^1}}} + \frac{\norm{\ket{\overline v_x^\pm}}^2}2.
        \end{align*}
    \end{theorem}

    \begin{proof}
    
	We first introduce some notation. For each $x \in \D$, let $a_x := \norm{\ket{\psi_x^0}}$ and $b_x := \norm{\ket{\psi_x^1}}$. 
    Note that $a_x^2+b_x^2=1$. We then have 
	\[
	\ket{\overline{\psi}_x} = \frac{b_x}{a_x} \ket{\psi_x^0}\ket{0} - \frac{a_x}{b_x}\ket{\psi_x^1}\ket{1}.
	\]
    
    Now let $\{\ket{v_x^{\pm}}, \ket{\overline{v}_x^{\pm}}\}_{x \in \D}$ be a feasible solution to the bidirectional adversary bound for the state-conversion problem $\{(\ket{\psi_x},\ket{\psi_x},O_x),(\ket{\overline{\psi_x}},-\ket{\overline{\psi}_x}, O_x)\}_{x \in \D}$. Then for any pair $x,y \in \D$ we have
    \begin{subequations} \label{eq:solprops}
    \begin{align}
    	\bra{v_x^+}(I-O_x^* O_y) \ket{v_y^-} &= 0,\\
    	\bra{v_x^+}(I-O_x^* O_y) \ket{\overline v_y^-} &= 2 \braket{\psi_x}{\overline\psi_y},\\
    	\bra{\overline v_x^+}(I-O_x^* O_y) \ket{\overline v_y^-} &= 0,\\
    	\bra{\overline v_x^+}(I-O_x^* O_y) \ket{v_y^-} &= 2 \braket{\overline\psi_x}{\psi_y}.
    \end{align}
    \end{subequations}
    We consider solutions of the form 
    \begin{align*} 
  	\sqrt{2}\ket{w_x^+} = \ket{v_x^+} \otimes \ket{\alpha_x^+} \oplus \ket{\overline{v}_x^+} \otimes \ket{\beta_x^+} \oplus  \ket{v_x^+}, \quad &\text{and} \quad \sqrt{2}\ket{w_x^-} = \ket{\overline{v}_x^-} \otimes \ket{\alpha_x^-} \oplus \ket{v_x^-} \otimes \ket{\beta_x^-} \oplus  \ket{v_x^-}, \\
  	\sqrt{2}\ket{\overline w_x^+} = \ket{\overline v_x^+} \otimes \ket{\overline \beta_x^+} \oplus \ket{v_x^+} \otimes \ket{\overline \alpha_x^+} \oplus \ket{\overline v_x^+}, \quad &\text{and} \quad \sqrt{2}\ket{\overline w_x^-} = \ket{v_x^-} \otimes \ket{\overline \beta_x^-} \oplus \ket{\overline v_x^-} \otimes \ket{\alpha_x^-} \oplus  \ket{\overline v_x^-},
    \end{align*}
        for vectors $\ket{\alpha_x^{\pm}}, \ket{\beta_x^{\pm}},\ket{\overline\alpha_x^{\pm}}, \ket{\overline\beta_x^{\pm}}$ that we will determine later. The relations in \Cref{eq:solprops} imply the following:
    \begin{subequations} \label{eq:solprops2}
    \begin{align}
    	\bra{w_x^+}(I-O_x^* O_y) \ket{w_y^-} &= \braket{\alpha_x^+}{\alpha_y^-} \braket{\psi_x}{\overline\psi_y} + \braket{\beta_x^+}{\beta_y^-}\braket{\overline\psi_x}{\psi_y}, \label{eq:a} \\
    	\bra{w_x^+}(I-O_x^* O_y) \ket{\overline w_y^-} &=  \braket{\psi_x}{\overline \psi_y}, \label{eq:b} \\
    	\bra{\overline w_x^+}(I-O_x^* O_y) \ket{\overline w_y^-} &= \braket{\overline \alpha_x^+}{\overline \alpha_y^-} \braket{\psi_x}{\overline\psi_y} + \braket{\overline \beta_x^+}{\overline \beta_y^-}\braket{\overline\psi_x}{\psi_y}, \label{eq:c} \\
    	\bra{\overline w_x^+}(I-O_x^* O_y) \ket{w_y^-} &= \braket{\overline \psi_x}{\psi_y}.\label{eq:d} 
    \end{align}
    \end{subequations}
	The constraints of the bidirectional adversary bound for the state-conversion problem 
\[
        \{(\ket{\psi_x}, \frac{\ket{\psi_x^1}}{\norm{\ket{\psi_x^1}}}\ket{1}, O_x),(\ket{\overline\psi_x}, -\frac{\ket{\psi_x^0}}{\norm{\ket{\psi_x^0}}}\ket{0}, O_x)\}_{x \in \D}\]
        	allow us to determine the vectors $\ket{\alpha_x^{\pm}}, \ket{\beta_x^{\pm}},\ket{\overline\alpha_x^{\pm}}, \ket{\overline\beta_x^{\pm}}$. Note that \cref{eq:b,eq:d} precisely correspond to the inner product of the input states minus the inner product of the target states, since the latter are orthogonal. The right hand side of \cref{eq:a} should equal \[
        	\braket{\psi_x}{\psi_y} - \frac{1}{b_x b_y} \braket{\psi_x^1}{\psi_y^1} = \braket{\psi_x^0}{\psi_y^0} + \left(1-\frac{1}{b_x b_y}\right) \braket{\psi_x^1}{\psi_y^1}. 
\]
Expressing $\braket{\psi_x}{\overline \psi_y}$ and $\braket{\overline \psi_x}{\psi_y}$ in terms of $\braket{\psi_x^0}{\psi_y^0}$ and $\braket{\psi_x^1}{\psi_y^1}$ gives rise to the linear system of equations
\[
\begin{bmatrix}
\frac{b_y}{a_y} & \frac{b_x}{a_x} \\ -\frac{a_y}{b_y} & -\frac{a_x}{b_x} 
\end{bmatrix} \begin{bmatrix} \braket{\alpha_x^+}{\alpha_y^-} \\ \braket{\beta_x^+}{\beta_y^-}
\end{bmatrix} = \begin{bmatrix}
1 \\ 1 - \frac{1}{b_x b_y}
\end{bmatrix}
\]
whose solution is $\braket{\alpha_x^+}{\alpha_y^-} = \frac{a_y}{b_x+b_y}$ and $\braket{\beta_x^+}{\beta_y^-} = \frac{a_x}{b_x+b_y}$. Indeed, the first equation is easy to verify and for the second we have 
\[
-\frac{a_y}{b_y} \frac{a_y}{b_x+b_y}  -\frac{a_x}{b_x} \frac{a_x}{b_x+b_y} = - \frac{\frac{1-b_y^2}{b_y} + \frac{1-b_x^2}{b_x}}{b_x+b_y} = 1 - \frac{\frac{1}{b_y} + \frac{1}{b_x}}{b_x+b_y} = 1 - \frac{1}{b_x b_y}.
\]
To construct $\ket{\alpha_x^\pm}$ and $\ket{\beta_x^\pm}$ with the desired inner products, we use \Cref{lem:ratio} to define $\ket{\chi_x}$ and set 
\[
\ket{\alpha_x^+} = \ket{\chi_x}, \quad \ket{\alpha_x^-} = a_x \ket{\chi_x}, \quad \ket{\beta_x^+} = a_x \ket{\chi_x}, \quad \ket{\beta_x^-} = \ket{\chi_x}.
\]
For \cref{eq:c} we proceed similarly. We require 
\[
\braket{\overline \alpha_x^+}{\overline \alpha_y^-} \braket{\psi_x}{\overline\psi_y} + \braket{\overline \beta_x^+}{\overline \beta_y^-}\braket{\overline\psi_x}{\psi_y} = \braket{\overline \psi_x}{\overline \psi_y} - \frac{1}{a_x a_y} \braket{\psi_x^0}{\psi_y^0} = \left(\frac{b_x b_y-1}{a_x a_y}\right) \braket{\psi_x^0}{\psi_y^0}+ \left(\frac{a_x a_y}{b_x b_y}\right) \braket{\psi_x^1}{\psi_y^1},
\]
which results in the linear system 
\[
\begin{bmatrix}
\frac{b_y}{a_y} & \frac{b_x}{a_x} \\ -\frac{a_y}{b_y} & -\frac{a_x}{b_x} 
\end{bmatrix} \begin{bmatrix} \braket{\overline\alpha_x^+}{\overline\alpha_y^-} \\ \braket{\overline \beta_x^+}{\overline \beta_y^-}
\end{bmatrix} = \begin{bmatrix}
\frac{b_x b_y-1}{a_x a_y} \\ \frac{a_x a_y}{b_x b_y}
\end{bmatrix}
\]
whose solution is $\braket{\overline \alpha_x^+}{\overline\alpha_y^-} = \frac{-a_x}{b_x+b_y}$ and $\braket{\overline\beta_x^+}{\overline\beta_y^-} = \frac{-a_y}{b_x+b_y}$. Indeed, let us verify both equations. First, 
\[
\frac{b_y}{a_y} \frac{-a_x}{b_x+b_y} + \frac{b_x}{a_x}\frac{-a_y}{b_x+b_y} = \frac{1}{a_x a_y} \frac{b_y(b_x^2-1) + b_x(b_y^2-1)}{b_x+b_y} = \frac{b_x b_y-1}{a_x a_y},  
\]
and second, 
\[
\frac{-a_y}{b_y} \frac{-a_x}{b_x+b_y} + \frac{-a_x}{b_x}\frac{-a_y}{b_x+b_y} = a_x a_y \left(\frac{\frac{1}{b_y} + \frac{1}{b_x}}{b_x+b_y}\right) = \frac{a_x a_y}{b_x b_y}.
\]
The construction of vectors with the desired inner products is similar to before:
\[
\ket{\overline \alpha_x^+} = -a_x \ket{\chi_x}, \quad \ket{\overline \alpha_x^-} = \ket{\chi_x}, \quad \ket{\overline\beta_x^+} = \ket{\chi_x}, \quad \ket{\overline \beta_x^-} = -a_x\ket{\chi_x}.
\]

We finally observe that the norms of the constructed solution agree with the conclusion of the theorem.
    \end{proof}

    Now that we have constructed a transducer that prepares the amplified state, we can additionally wonder how hard it is to reflect around it. In the circuit model, a typical approach is to run the amplification circuit in reverse, reflect around the original state, and then run the amplification circuit again. In the transducer picture, though, this would require analyzing the transduction complexity of the amplification transducer on all the states that are not in the two-dimensional subspace spanned by $\ket{\psi}$ and $\ket{\overline{\psi}}$ in the previous theorem, and this appears to be a tall order.
    Instead, we provide a different construction directly on the level of adversary-bound solutions, without relying on the intuitive approach from the circuit model, in the following theorem.

    \begin{theorem}[Reflection through amplified state]
        \label{thm:refl-amplified-state}
        Let $\D$ be a set, and $x \in \D$. We take the vectors $\ket{\psi_x^0},\ket{\psi_x^1} \in \mathcal{H}$, such that $\norm{\ket{\psi_x^0}}^2 + \norm{\ket{\psi_x^1}}^2 = 1$ and $\norm{\ket{\psi_x^1}} > 0$. Let
        \[\ket{\psi_x} = \ket{\psi_x^0}\ket{0} + \ket{\psi_x^1}\ket{1}, \qquad \text{and} \qquad \ket{\overline{\psi}_x} = \frac{\norm{\ket{\psi_x^1}}}{\norm{\ket{\psi_x^0}}} \ket{\psi_x^0}\ket{0} - \frac{\norm{\ket{\psi_x^0}}}{\norm{\ket{\psi_x^1}}} \ket{\psi_x^1}\ket{1}.\]
        We also let $|\widetilde{\psi}_x\rangle:= |\widetilde{\psi}_x^0\rangle\ket{0}$ and $|\widehat{\psi}_x\rangle := |\widehat{\psi}_x^1\rangle\ket{1}$ with $\braket{\widehat{\psi}_x^1}{\psi^1_x} = 0$, and we write $\varepsilon_x := \norm{\ket{\psi_x^1}}$.
        
        Suppose we have a feasible solution $\{\ket{v_x^{\pm}}, \ket{\overline{v}_x^{\pm}}, \ket{\widetilde{v}_x^{\pm}}, \ket{\widehat{v}_x^{\pm}}\}_{x \in \D}$ to the bidirectional adversary bound for the state-conversion problem
        \[P := \{(\ket{\psi_x},\ket{\psi_x},O_x), (\ket{\overline{\psi}_x},-\ket{\overline{\psi}_x},O_x),(|\widetilde{\psi}_x\rangle,-|\widetilde{\psi}_x\rangle,O_x),(|\widehat{\psi}_x\rangle,-|\widehat{\psi}_x\rangle,O_x)\}_{x \in \D}.\]
        Then, we can construct a feasible solution $\{\ket{(w_x^0)^{\pm}}, \ket{(w_x^1)^{\pm}}, \ket{\widetilde{w}_x^{\pm}}, \ket{\widehat{w}_x^{\pm}}\}_{x \in \D}$ to the state-conversion problem 
        \[P' := \left\{\left(\frac{\ket{\psi_x^0}\ket{0}}{\sqrt{1-\varepsilon_x^2}},-\frac{\ket{\psi_x^0}\ket{0}}{\sqrt{1-\varepsilon_x^2}},O_x\right), \left(\frac{\ket{\psi_x^1}\ket{1}}{\varepsilon_x}, \frac{\ket{\psi_x^1}\ket{1}}{\varepsilon_x}, O_x\right), (|\widetilde{\psi}_x\rangle, -|\widetilde{\psi}_x\rangle, O_x), (|\widehat{\psi}_x\rangle, -|\widehat{\psi}_x\rangle, O_x)\right\}_{x \in \D},\]
        such that
        \[\norm{\ket{(w_x^0)^{\pm}}}^2 = 0, \quad \norm{\ket{(w_x^1)^{\pm}}}^2 = \frac{\norm{\ket{v_x^{\pm}}}^2}{\varepsilon_x^2}, \quad \norm{\ket{\widetilde{w}_x^{\pm}}}^2 = 0, \quad \text{and} \quad \norm{\ket{\widehat{w}_x^{\pm}}}^2 = \norm{\ket{\widehat{v}_x^{\pm}}}^2.\]
    \end{theorem}

    \begin{proof}
        We define the vectors
        \[\ket{(w_x^0)^{\pm}} := 0, \quad \ket{(w_x^1)^{\pm}} := \frac{\ket{v_x^{\pm}}}{\varepsilon_x}, \quad \ket{\widetilde{w}_x^{\pm}} := 0, \quad \text{and} \quad \ket{\widehat{w}_x^{\pm}} := \ket{\widehat{v}_x^{\pm}}.\]
        The claims on the norms are immediate, so it remains to prove that this is a feasible solution to the bidirectional adversary bound for $P'$. For any $x,y \in \D$, we can now compute the inner products of $\bra{\cdot}(I - O_x^{\dagger}O_y)\ket{\cdot}$ as in the following table.
        \begin{center}
            \begin{tabular}{r|cccc}
                & $\ket{(w_y^0)^-}$ & $\ket{(w_y^1)^-}$ & $\ket{\widetilde{w}_y^-}$ & $\ket{\widehat{w}_y^-}$ \\\hline
                $\bra{(w_x^0)^+}$ & $0$ & $0$ & $0$ & $0$ \\
                $\bra{(w_x^1)^+}$ & $0$ & $A_1$ & $0$ & $A_2$ \\
                $\bra{\widetilde{w}_x^+}$ & $0$ & $0$ & $0$ & $0$ \\
                $\bra{\widehat{w}_x^+}$ & $0$ & $A_3$ & $0$ & $A_4$
            \end{tabular}
        \end{center}
        We compute the four remaining entries as
        \begin{align*}
            A_1 &= \frac{\bra{v_x^+}(I - O_x^{\dagger}O_y)\ket{v_y^-}}{\varepsilon_x\varepsilon_y} = \frac{\braket{\psi_x}{\psi_y} - \braket{\psi_x}{\psi_y}}{\varepsilon_x\varepsilon_y} = 0 = \frac{\braket{\psi_x^1}{\psi_y^1} - \braket{\psi_x^1}{\psi_y^1}}{\varepsilon_x\varepsilon_y}, \\
            A_2 &= \frac{\bra{v_x^+}(I - O_x^{\dagger}O_y)\ket{\widehat{v}_y^-}}{\varepsilon_x} = \frac{2\langle\psi_x|\widehat{\psi}_y\rangle}{\varepsilon_x} = \frac{2\langle\psi_x^1|\widehat{\psi}_y\rangle}{\varepsilon_x}, \\
            A_3 &= \frac{\bra{\widehat{v}_x^+}(I - O_x^{\dagger}O_y)\ket{v_y^-}}{\varepsilon_y} = \frac{2\langle\widehat{\psi}_x|\psi_y\rangle}{\varepsilon_y} = \frac{2\langle\widehat{\psi}_x|\psi_y^1\rangle}{\varepsilon_y}, \\
            A_4 &= \bra{\widehat{v}_x^+}(I - O_x^{\dagger}O_y)\ket{\widehat{v}_y^-} = \langle\widehat{\psi}_x|\widehat{\psi}_y\rangle - \langle\widehat{\psi}_x|\widehat{\psi}_y\rangle.\qedhere
        \end{align*}
    \end{proof}

    As before, we can use \Cref{thm:unidirectional-bidirectional,thm:adv-to-transducer,thm:transducer-to-adv} to convert this feasible solution to the bidirectional adversary bound into a feasible solution to the unidirectional adversary bound with inverse queries, and then into a canonical transducer. We formalize this in the following corollary.

    \begin{corollary}[Amplitude amplification and reflection transducer]
        \label{cor:ampl-ampl}
        Let $\mathcal{V}$ be a Hilbert space and $\ket{\psi_0}, \ket{\psi_1} \in \mathcal{V}$, such that $\norm{\ket{\psi_0}}^2 + \norm{\ket{\psi_1}}^2 = 1$ and $\norm{\ket{\psi_0}},\norm{\ket{\psi_1}} > 0$. We write
        \[\ket{\psi} = \ket{\psi_0}\ket{0} + \ket{\psi_1}\ket{1}, \qquad \text{and} \qquad \ket{\overline{\psi}} = \frac{\norm{\ket{\psi_1}}}{\norm{\ket{\psi_0}}}\ket{\psi_0}\ket{0} - \frac{\norm{\ket{\psi_0}}}{\norm{\ket{\psi_1}}}\ket{\psi_1}\ket{1}.\]
        Suppose we have a bidirectional canonical transducer $R$ that acts as a reflection through $\ket{\psi}$ on $\mathcal{V} \otimes \C^2$. Then, we can implement bidirectional canonical transducers $C$ and $U$ that act as
        \begin{align*}
            \ket{\psi} \overset{C}{\rightsquigarrow} \frac{\ket{\psi_1}\ket{1}}{\norm{\ket{\psi_1}}}, & \qquad \text{with complexity} \qquad O\left(\frac{W(R,\ket{\psi}) + \norm{\ket{\psi_0}}^2W(R,\ket{\overline{\psi}})}{\norm{\ket{\psi_1}}}\right), \\
            \frac{\ket{\psi_1}}{\norm{\ket{\psi_1}}} \overset{U}{\rightsquigarrow} \frac{\ket{\psi_1}}{\norm{\ket{\psi_1}}}, & \qquad \text{with complexity} \qquad O\left(\frac{W(R,\ket{\psi})}{\norm{\ket{\psi_1}}^2}\right), \\
            \ket{\perp} \overset{U}{\rightsquigarrow} -\ket{\perp}, & \qquad \text{with complexity} \qquad O\left(W(R, \ket{\perp}\ket{1})\right),
        \end{align*}
        for all $\ket{\perp} \in \mathcal{V}$ satisfying $\braket{\perp}{\psi_1} = 0$.
    \end{corollary}
    
    \begin{proof}

        Let $\D$ be the implicit set of input labels. Thus, for every $x \in \D$, we have states $\ket{\psi_x^0},\ket{\psi_x^1} \in \mathcal{V}$ such that $\norm{\ket{\psi_x^0}}^2 + \norm{\ket{\psi_x^1}}^2 = 1$ and $\norm{\ket{\psi_x^1}} > 0$, and we have an oracle $O_x \in \mathcal{U}(\mathcal{H})$. We write
        \[\ket{\psi_x} = \ket{\psi_x^0}\ket{0} + \ket{\psi_x^1}\ket{1}, \qquad \text{and} \qquad \ket{\overline{\psi}_x} = \frac{\norm{\ket{\psi_x^1}}}{\norm{\ket{\psi_x^0}}}\ket{\psi_x^0}\ket{0} - \frac{\norm{\ket{\psi_x^0}}}{\norm{\ket{\psi_x^1}}}\ket{\psi_x^1}\ket{1}.\]
        Now, by assumption, we have a canonical transducer $R$, such that for all $x \in \D$,
        \[\ket{\psi_x} \overset{R(O_x \oplus O_x^{\dagger})}{\rightsquigarrow} \ket{\psi_x}, \qquad \text{and} \qquad \ket{\overline{\psi}_x} \overset{R(O_x \oplus O_x^{\dagger})}{\rightsquigarrow} -\ket{\overline{\psi}_x}.\]

        We now use \Cref{thm:transducer-to-adv} to turn this into a feasible solution $\{\ket{w_x}, \ket{\overline{w}_x}\}_{x \in \D}$ to the unidirectional adversary bound for the state-conversion problem
        \[\{(\ket{\psi_x},\ket{\psi_x},O_x \oplus O_x^{\dagger}), (\ket{\overline{\psi}_x}, -\ket{\overline{\psi}_x}, O_x \oplus O_x^{\dagger})\}_{x \in \D}.\]
        Next, we use \Cref{thm:unidirectional-bidirectional} to turn this into a feasible solution $\{\ket{w_x^{\pm}}, \ket{\overline{w}_x^{\pm}}\}_{x \in \D}$ for the bidirectional adversary bound of the state-conversion problem
        \[\{(\ket{\psi_x}, \ket{\psi_x}, O_x), (\ket{\overline{\psi}_x}, -\ket{\overline{\psi}_x}, O_x)\}_{x \in \D}.\]
        Then, we use \Cref{thm:ampl-ampl-adv} to construct a feasible solution $\{\ket{v_x^{\pm}}, \ket{\overline{v}_x^{\pm}}\}_{x \in \D}$ for the bidirectional adversary bound of the state-conversion problem
        \[\left\{\left(\ket{\psi_x}, \frac{\ket{\psi_x^1}\ket{1}}{\norm{\ket{\psi_x^1}}}, O_x\right), \left(\ket{\overline{\psi}_x}, -\frac{\ket{\psi_x^0}\ket{0}}{\norm{\ket{\psi_x^0}}}, O_x\right)\right\}.\]
        Now, we consider the simpler state-conversion problem where we only consider the states $\ket{\psi_x}$, and not the states $\ket{\overline{\psi}_x}$. We now obtain the feasible solution $\{\ket{v_x^{\pm}}\}_{x \in \D}$ for its bidirectional adversary bound. Moreover, we can remove the extra register $\ket{1}$ from the output, since it does not impact the Gram-matrix in the constraints of the bidirectional adversary bound, and so we obtain that $\{\ket{v_x^{\pm}}\}_{x \in \D}$ is a feasible solution to the bidirectional adversary bound of the state-conversion problem
        \[\left\{\left(\ket{\psi_x}, \frac{\ket{\psi_x^1}}{\norm{\ket{\psi_x^1}}}, O_x\right)\right\}_{x \in \D}.\]
        We now use \Cref{thm:unidirectional-bidirectional} to convert this into a feasible solution $\{\ket{v_x}\}_{x \in \D}$ to the unidirectional adversary bound for the state-conversion problem
        \[\left\{\left(\ket{\psi_x}, \frac{\ket{\psi_x^1}}{\norm{\ket{\psi_x^1}}}, O_x \oplus O_x^{\dagger}\right)\right\}_{x \in \D}.\]
        Finally, we use \Cref{thm:adv-to-transducer} to turn this feasible solution into a bidirectional canonical transducer $C$ that acts as claimed in the theorem statement. The complexity statement follows from checking the resulting norms, as for all $x \in \D$,
        \begin{align*}
            W(C,\ket{\psi_x}) &\leq \norm{\ket{v_x}}^2 \leq \max_{\pm} \norm{\ket{v_x^{\pm}}}^2 \in O\left(\max_{\pm} \frac{\norm{\ket{w_x^{\pm}}}^2 + \norm{\ket{\psi_x^0}}^2\norm{\ket{\overline{w}_x^{\pm}}}^2}{\norm{\ket{\psi_x^1}}}\right) \\
            &\subseteq O\left(\frac{\norm{\ket{w_x}}^2 + \norm{\ket{\psi_x^0}}^2\norm{\ket{\overline{w}_x}}^2}{\norm{\ket{\psi_x^1}}}\right) \subseteq O\left(\frac{W(R,\ket{\psi_x}) + \norm{\ket{\psi_x^0}}^2W(R, \ket{\overline{\psi}_x})}{\norm{\ket{\psi_x^1}}}\right).
        \end{align*}
        The construction for $U$ follows along similar lines, but relying on \Cref{thm:refl-amplified-state} instead of \Cref{thm:ampl-ampl-adv}.
    \end{proof}

    \subsection{Annealing transducer from two state-reflection transducers}
    \label{subsec:ampl-ampl-two-state-rotation}

    We can interpret the transducer we constructed in the previous section as Grover's iterate, in the setting where we assume that one of the two reflections is trivial, i.e., requires no queries and can therefore be implemented for free. In this section, we generalize the setting and consider a situation where both reflections are costly.

    To that end, suppose that we are given a state $\ket{\varphi}$, and we want to turn it into $\ket{\chi}$, given access to reflections through both states. Assume that $\braket{\varphi}{\chi} \neq 0$, i.e., they have non-zero overlap. We define
    \[\ket{\overline{\chi}} = \frac{\ket{\varphi} - \braket{\chi}{\varphi}\ket{\chi}}{\norm{\ket{\varphi} - \braket{\chi}{\varphi}\ket{\chi}}}, \qquad \text{and} \qquad \ket{\overline{\varphi}} = -\frac{\ket{\chi} - \braket{\varphi}{\chi}\ket{\varphi}}{\norm{\ket{\chi} - \braket{\varphi}{\chi}\ket{\varphi}}} \cdot \frac{\braket{\chi}{\varphi}}{|\braket{\chi}{\varphi}|}.\]
    See \Cref{fig:two-refl-ampl-ampl} for a pictorial representation of these states.
    Our goal is to construct a transducer that maps $\ket{\varphi}$ to $\ket{\chi}$.

    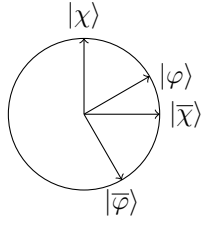
\begin{figure}[!ht]
        \centering
        \begin{tikzpicture}
            \draw (0,0) circle[radius=1];
            \draw[->] (0,0) to (0,1) node[above] {$\ket{\chi}$};
            \draw[->] (0,0) to (1,0) node[right] {$\ket{\overline{\chi}}$};
            \draw[->] (0,0) to ({cos(30)},{sin(30)}) node[right] {$\ket{\varphi}$};
            \draw[->] (0,0) to ({cos(-60)},{sin(-60)}) node[below] {$\ket{\overline{\varphi}}$};
        \end{tikzpicture}
        \caption{Pictorial overview of the states. We assume to start in $\ket{\varphi}$, and we wish to prepare $\ket{\chi}$.}
        \label{fig:two-refl-ampl-ampl}
    \end{figure}

    \begin{theorem}
        \label{thm:two-state-ampl}
        Let $\ket{\varphi},\ket{\chi} \in \mathcal{H}$ such that $\braket{\chi}{\varphi} \neq 0$, and suppose that we have bidirectional canonical transducers $R_{\ket{\varphi}}$ and $R_{\ket{\chi}}$ that reflect through $\ket{\varphi}$ and $\ket{\chi}$, respectively. Then, we can construct a bidirectional canonical transducer $U$ that maps $\ket{\varphi}$ to $\braket{\chi}{\varphi}\ket{\chi}/|\braket{\chi}{\varphi}|$, with transduction complexity
        \begin{align*}
            W(U,\ket{\varphi}) &\in O\left(\frac{W(R_{\ket{\chi}}, \ket{\chi}) + W(R_{\ket{\varphi}},\ket{\varphi}) + (1 - |\braket{\chi}{\varphi}|^4)W(R_{\ket{\chi}},\ket{\overline{\chi}}) + (1 - |\braket{\chi}{\varphi}|^2)W(R_{\ket{\varphi}},\ket{\overline{\varphi}})}{|\braket{\chi}{\varphi}|}\right) \\
            &\subseteq O\left(\frac{W(R_{\ket{\chi}},\ket{\chi}) + W(R_{\ket{\chi}},\ket{\overline{\chi}}) + W(R_{\ket{\varphi}},\ket{\varphi}) + W(R_{\ket{\varphi}},\ket{\overline{\varphi}})}{|\braket{\chi}{\varphi}|}\right).
        \end{align*}
    \end{theorem}

    \begin{proof}
        First, we consider the operation $A := (H \otimes I)(I \oplus (-R_{\ket{\chi}}))(H \otimes I)$, and we observe that it acts as
        \begin{align*}
            A : \begin{bmatrix}
                \ket{\varphi} \\ 0
            \end{bmatrix} \overset{H \otimes I}{\mapsto} \frac{1}{\sqrt{2}}\begin{bmatrix}
                \ket{\varphi} \\
                \ket{\varphi}
            \end{bmatrix} &\overset{I \oplus (-R_{\ket{\chi}})}{\mapsto} \frac{1}{\sqrt{2}} \begin{bmatrix}
                \braket{\chi}{\varphi}\ket{\chi} + \braket{\overline{\chi}}{\varphi}\ket{\overline{\chi}} \\
                -\braket{\chi}{\varphi}\ket{\chi} + \braket{\overline{\chi}}{\varphi}\ket{\overline{\chi}}
            \end{bmatrix} \overset{H \otimes I}{\mapsto} \begin{bmatrix}
                \braket{\overline{\chi}}{\varphi}\ket{\overline{\chi}} \\
                \braket{\chi}{\varphi}\ket{\chi}
            \end{bmatrix}, \\
            A : \begin{bmatrix}
                \ket{\overline{\varphi}} \\ 0
            \end{bmatrix} \overset{H \otimes I}{\mapsto} \frac{1}{\sqrt{2}}\begin{bmatrix}
                \ket{\overline{\varphi}} \\
                \ket{\overline{\varphi}}
            \end{bmatrix} &\overset{I \oplus (-R_{\ket{\chi}})}{\mapsto} \frac{1}{\sqrt{2}} \begin{bmatrix}
                \braket{\chi}{\overline{\varphi}}\ket{\chi} + \braket{\overline{\chi}}{\overline{\varphi}}\ket{\overline{\chi}} \\
                -\braket{\chi}{\overline{\varphi}}\ket{\chi} + \braket{\overline{\chi}}{\overline{\varphi}}\ket{\overline{\chi}}
            \end{bmatrix} \overset{H \otimes I}{\mapsto} \begin{bmatrix}
                \braket{\overline{\chi}}{\overline{\varphi}}\ket{\overline{\chi}} \\
                \braket{\chi}{\overline{\varphi}}\ket{\chi}
            \end{bmatrix}.
        \end{align*}
        We write $\ket{\psi_0} = \braket{\overline{\chi}}{\varphi}\ket{\overline{\chi}}$ and $\ket{\psi_1} = \braket{\chi}{\varphi}\ket{\chi}$, and we observe that $\norm{\ket{\psi_1}} = |\braket{\chi}{\varphi}|$ and $\norm{\ket{\psi_0}} = \sqrt{1 - |\braket{\chi}{\varphi}|^2}$. Then, plugging in the definitions of $\ket{\overline{\varphi}}$ yields
        \[\braket{\chi}{\overline{\varphi}}\ket{\chi} = -\frac{1 - |\braket{\varphi}{\chi}|^2}{\norm{\ket{\chi} - \braket{\varphi}{\chi}\ket{\varphi}}} \cdot \frac{\braket{\chi}{\varphi}}{|\braket{\chi}{\varphi}|}\ket{\chi} = -\frac{\norm{\ket{\psi_0}}}{\norm{\ket{\psi_1}}}\ket{\psi_1},\]
        and similarly,
        \[\braket{\overline{\chi}}{\overline{\varphi}}\ket{\overline{\chi}} = \frac{\braket{\varphi}{\chi}\braket{\overline{\chi}}{\varphi}}{\norm{\ket{\chi} - \braket{\varphi}{\chi}\ket{\varphi}}} \cdot \frac{\braket{\chi}{\varphi}}{|\braket{\chi}{\varphi}|}\ket{\overline{\chi}} = \frac{|\braket{\chi}{\varphi}|}{\sqrt{1 - |\braket{\chi}{\varphi}|^2}} \cdot \braket{\overline{\chi}}{\varphi}\ket{\overline{\chi}} = \frac{\norm{\ket{\psi_1}}}{\norm{\ket{\psi_0}}} \ket{\psi_0}.\]
        Thus, $A$ implements the mapping
        \[A : \ket{\varphi}\ket{0} \mapsto \underbrace{\ket{\psi_0}\ket{0} + \ket{\psi_1}\ket{1}}_{=: \ket{\psi}}, \qquad \text{and} \qquad A : \ket{\overline{\varphi}}\ket{0} \mapsto \underbrace{\frac{\norm{\ket{\psi_1}}}{\norm{\ket{\psi_0}}} \ket{\psi_0}\ket{0} - \frac{\norm{\ket{\psi_0}}}{\norm{\ket{\psi_1}}}\ket{\psi_1}\ket{1}}_{=: \ket{\overline{\psi}}}.\]
        
        We observe that $A$ is the sequential composition of $H \otimes I$, $I \oplus (-R_{\ket{\chi}})$ and $H \otimes I$, so we use \Cref{thm:sequential-composition} to compute the transduction complexity of $A$. Since $H \otimes I$ is an input-independent unitary operation, its transduction complexity is $0$, and so it remains to analyze the transduction complexity of $I \oplus (-R_{\ket{\chi}})$. This is a parallel composition of $I$ and $-R_{\ket{\chi}}$, and so we use \Cref{thm:parallel-composition} to analyze its transduction complexity. Thus, we find that
        \begin{align*}
            W(A,\ket{\varphi}) &\leq W\left(R_{\ket{\chi}}, \frac{1}{\sqrt{2}}\ket{\varphi}\right) = \frac12W(R_{\ket{\chi}}, \ket{\varphi}) \leq |\braket{\chi}{\varphi}|^2 W(R_{\ket{\chi}}, \ket{\chi}) + (1-|\braket{\chi}{\varphi}|^2)W(R_{\ket{\chi}}, \ket{\overline{\chi}}), \\
            W(A,\ket{\overline{\varphi}}) &\leq W\left(R_{\ket{\chi}}, \frac{1}{\sqrt{2}} \ket{\overline{\varphi}}\right) = \frac12 W(R_{\ket{\chi}}, \ket{\overline{\varphi}}) \leq (1-|\braket{\chi}{\varphi}|^2)W(R_{\ket{\chi}}, \ket{\chi}) + |\braket{\chi}{\varphi}|^2 W(R_{\ket{\chi}}, \ket{\overline{\chi}}),
        \end{align*}
        where we used the linearity of the mapping from initial state to witness vector, and $\norm{\ket{v} + \ket{w}}^2 \leq 2(\norm{\ket{v}}^2 + \norm{\ket{w}}^2)$ for any two vectors $\ket{v}$ and $\ket{w}$.
        
        Next, we observe that the operation $R := AR_{\ket{\varphi}}A^{\dagger}$ acts as
        \[R : \ket{\psi} \overset{A^{\dagger}}{\mapsto} \ket{\varphi} \overset{R_{\ket{\varphi}}}{\mapsto} \ket{\varphi} \overset{A}{\mapsto} \ket{\psi}, \qquad \text{and} \qquad R : \ket{\overline{\psi}} \overset{A^{\dagger}}{\mapsto} \ket{\overline{\varphi}} \overset{R_{\ket{\varphi}}}{\mapsto} -\ket{\overline{\varphi}} \overset{A}{\mapsto} -\ket{\overline{\psi}}.\]
        This is a sequential composition of three transducers, so we can analyze its total transduction complexity using \Cref{thm:sequential-composition}. We use \Cref{thm:inverse-transducer} to express the transduction complexity of $A^{\dagger}$ in terms of that of $A$. Thus, we obtain that
        \begin{align*}
            W(R,\ket{\psi}) &\leq 2W(A,\ket{\varphi}) + W(R_{\ket{\varphi}}, \ket{\varphi}) \\
            &\leq 2|\braket{\chi}{\varphi}|^2 W(R_{\ket{\chi}}, \ket{\chi}) + 2(1-|\braket{\chi}{\varphi}|^2)W(R_{\ket{\chi}}, \ket{\overline{\chi}}) + W(R_{\ket{\varphi}},\ket{\varphi}), \\
            W(R,\ket{\overline{\psi}}) &\leq 2W(A,\ket{\overline{\varphi}}) + W(R_{\ket{\varphi}}, \ket{\overline{\varphi}}) \\
            &\leq 2(1-|\braket{\chi}{\varphi}|^2) W(R_{\ket{\chi}}, \ket{\chi}) + 2|\braket{\chi}{\varphi}|^2 W(R_{\ket{\chi}}, \ket{\overline{\chi}}) + W(R_{\ket{\varphi}},\ket{\overline{\varphi}}).
        \end{align*}

        Finally, note that $R$ satisfies the premise of \Cref{cor:ampl-ampl}, and so we can construct a bidirectional canonical transducer $C$ that maps $\ket{\psi}$ to $\braket{\chi}{\varphi}\ket{\chi}/|\braket{\chi}{\varphi}|$. Finally, we let $U := CA$, i.e., a sequential composition of $A$ and $C$. We now observe that $U$'s transduction action is
        \[\ket{\varphi} \overset{A}{\rightsquigarrow} \ket{\psi} \overset{C}{\rightsquigarrow} \frac{\braket{\chi}{\varphi}}{|\braket{\chi}{\varphi}|}\ket{\chi},\]
        and using \Cref{thm:sequential-composition}, the transduction complexity can be upper bounded by $W(U, \ket{\varphi}) \leq W(A, \ket{\varphi}) + W(C, \ket{\psi})$, which by plugging in the upper bounds for both terms derived earlier yields the claimed transduction complexity from the theorem statement.
    \end{proof}

    We remark here that there is another way to obtain a similar result. 
    One can take the probabilistic annealing algorithm from \cite[Algorithm~7]{cornelissen2023sublinear}, and observe that it is a Las Vegas algorithm, in the sense introduced in \cite{belovs2023one}. Thus, using black-box conversion techniques of Las Vegas algorithms into transducers, as developed in \cite{belovs2024taming}, one can obtain a transducer that maps $\ket{\varphi} \ket{0}$ to $\ket{\chi} \ket{\Gamma_{\varphi,\chi}}$, for some garbage state \ket{\Gamma_{\varphi,\chi}}, with $O((1 + 1/|\braket{\varphi}{\chi}|) \cdot \max\{W(R_{\ket{\chi}},\ket{\chi}), W(R_{\ket{\chi}},\ket{\varphi}), W(R_{\ket{\varphi}},\ket{\chi}), W(R_{\ket{\varphi}},\ket{\varphi})\})$ transduction complexity. This approach does not require inverse queries, whereas the one from the previous theorem does. However, it does generate garbage, and in the rejection-sampling routine that we focus on next, the construction presented here is significantly better in the setting where $\ket{\varphi}$ and $\ket{\chi}$ are very close together.

    \subsection{Rejection sampling and mean estimation with transducers}
    \label{subsec:ampl-ampl-rejection-sampling}

    In this subsection, we port the mean estimation results from \cite{cornelissen2023sublinear} to the transducer setting. That is, we substitute the reflections through the q-sample, which is called in several places in \cite[Algorithms~3,4 and 5]{cornelissen2023sublinear}, by a corresponding transducer construction.

    The intricacy with this approach is that the transduction complexity of these reflections can depend on the state we apply the transducer on. Thus, to favorably analyze the cost of these mean estimation routines in the transducer setting, we must understand on which states these reflection transducers are applied throughout the routines in \cite{cornelissen2023sublinear}. This is the objective of this subsection.

    The main new ingredient that allows us to port the mean estimation results of \cite{cornelissen2023sublinear} is a rejection sampling routine in the transducer setting. Suppose that we have a probability distribution $\pi$ over a finite set $\Omega$. Suppose that we have a function $a : \Omega \to [0,1]$ that to every element $\omega \in \Omega$ assigns an \textit{acceptance probability} $a(\omega)$. Now, \textit{rejection sampling} is the canonical procedure that samples from the posterior distribution $\pi^{(a)}(\omega) := \pi(\omega)a(\omega)/A$, where $A = \E[a]$ is the average acceptance probability.

    In the following theorem, we build a transducer that maps $\ket{\pi}$, i.e., a q-sample of the probability distribution $\pi$, to $\ket{\pi^{(a)}}$. We assume that we have bidirectional canonical transducers $U$ and $O^{(a)}$ that reflect through $\ket{\pi}$ and compute the rejection probabilities, respectively, implicitly using some oracle. We then construct new bidirectional canonical transducers that prepare and reflect around the quantum state encoding the distribution generated by the rejection sampling routine, using the same implicit oracles.

    \begin{theorem}[Rejection sampling transducer]
        \label{thm:rejection-sampling}
        Let $\pi$ be a distribution over a finite set $\Omega$, and let $a : \Omega \to [0,1]$. We define the following states in $\C^{\Omega}$ as
        \[\ket{\pi} = \sum_{\omega \in \Omega} \sqrt{\pi(\omega)}\ket{\omega}, \qquad \text{and} \qquad \ket{\pi^{(a)}} = \sum_{\omega \in \Omega} \sqrt{\frac{\pi(\omega)a(\omega)}{A}}\ket{\omega}, \qquad \text{where} \qquad A = \underset{\omega \sim \pi}{\E} [a(\omega)].\]
        Now, we define the vectors
        \[\ket{\overline{\pi}} := \sum_{\omega \in \Omega} \sqrt{\pi(\omega)} \cdot \frac{A - a(\omega)}{\sqrt{A(1-A)}} \cdot \ket{\omega}, \qquad \text{and} \qquad \ket{\overline{\pi}^{(a)}} = \sum_{\omega \in \Omega} \frac{\overline{\pi}^{(a)}(\omega)}{\sqrt{\pi^{(a)}(\omega)}} \ket{\omega}.\]
        Suppose that we have a bidirectional canonical transducer $U$ that reflects through $\ket{\pi}$, and a bidirectional canonical transducer $O^{(a)}$ with uniform cost bound $C^{(a)}$, i.e., $W(O^{(a)},\ket{\psi}) \leq C^{(a)}\norm{\ket{\psi}}^2$ for all vectors $\ket{\psi}$, that acts on $\C^{\Omega} \otimes \C^2$ and for every $\omega \in \Omega$ maps
        \[\ket{\omega}\ket{0} \overset{O^{(a)}}{\rightsquigarrow} \ket{\omega}(\sqrt{a(\omega)}\ket{1} + \sqrt{1-a(\omega)}\ket{0}).\]
        Then, we can build bidirectional canonical transducers $C$ and $V$ that act as
        \begin{align*}
            \ket{\pi} \overset{C}{\rightsquigarrow} \ket{\pi^{(a)}}, & \qquad \text{with complexity} \qquad O\left(\frac{C^{(a)} + W(U,\ket{\pi}) + (1-A) \cdot W(U,\ket{\overline{\pi}})}{\sqrt{A}}\right), \\
            \ket{\pi^{(a)}} \overset{V}{\rightsquigarrow} \ket{\pi^{(a)}}, & \qquad \text{with complexity} \qquad O\left(\frac{C^{(a)} + W(U,\ket{\pi})}{A}\right), \\
            \ket{\overline{\pi}^{(a)}} \overset{V}{\rightsquigarrow} -\ket{\overline{\pi}^{(a)}}, & \qquad \text{with complexity} \qquad O\left(C^{(a)}\norm{\ket{\overline{\pi}^{(a)}}}^2 + A \cdot W\left(U,\sum_{\omega \in \Omega} \frac{\overline{\pi}^{(a)}(\omega)}{\sqrt{\pi(\omega)}}\right)\right),
        \end{align*}
        for all states $\ket{\overline{\pi}^{(a)}} \in \C^{\Omega}$ satisfying $\braket{\overline{\pi}^{(a)}}{\pi^{(a)}} = 0$.
    \end{theorem}

    \begin{proof}
        We start by applying $O^{(a)}$ to the initial state $\ket{\pi}$, with cost $C^{(a)}$, to obtain
        \[\ket{\pi}\ket{0} \overset{O^{(a)}}{\rightsquigarrow} \underbrace{\sum_{\omega \in \Omega} \sqrt{\pi(\omega)(1-a(\omega))}\ket{\omega}}_{=: \ket{\pi_0}}\ket{0} + \underbrace{\sum_{\omega \in \Omega} \sqrt{\pi(\omega)a(\omega)}\ket{\omega}}_{=: \ket{\pi_1}}\ket{1},\]
        and so we find that $\norm{\ket{\pi_0}}^2 = 1 - A$ and $\norm{\ket{\pi_1}}^2 = A$. We now map these states to the premise of \Cref{cor:ampl-ampl}, and so we write $\ket{\psi} = \ket{\pi_0}\ket{0} + \ket{\pi_1}\ket{1}$, and $\ket{\overline{\psi}} = \sqrt{A/(1-A)}\ket{\pi_0}\ket{0} - \sqrt{(1-A)/A}\ket{\pi_1}\ket{1}$. Then, observe that
        \begin{align*}
            (I \otimes \bra{0})(O^{(a)})^{\dagger}\ket{\overline{\psi}} &= \sum_{\omega \in \Omega} \sqrt{\pi(\omega)} \cdot \left(\sqrt{\frac{A}{1-A}}(1-a(\omega)) - \sqrt{\frac{1-A}{A}}a(\omega)\right) \ket{\omega} \\
            &= \sum_{\omega \in \Omega} \sqrt{\pi(\omega)} \cdot \frac{A - a(\omega)}{\sqrt{A(1-A)}} \cdot \ket{\omega} = \ket{\overline{\pi}}.
        \end{align*}

        Now, we can construct a transducer $R$ that reflects through $\ket{\psi}$. The idea is to first apply $O^{(a)}$ in reverse, then use $U$ to reflect through $\ket{\pi}$ on the $\ket{0}$-branch, and then apply $O^{(a)}$ again. Using sequential and parallel composition (\Cref{thm:sequential-composition,thm:parallel-composition}), we then obtain that
        \[W(R,\ket{\psi}) \leq 2C^{(a)} + W(U,\ket{\pi}), \qquad \text{and} \qquad W(R,\ket{\overline{\psi}}) \leq 2C^{(a)} + W(U,\ket{\overline{\pi}}).\]
        
        Now, we obtain transducers $C$ and $V$ from \Cref{cor:ampl-ampl}. For the transduction complexity of $V$ applied to the state $\ket{\overline{\pi}^{(a)}}$, we observe that the state that $U$ acts on is
        \[(I \otimes \bra{0})(O^{(a)})^{\dagger}\ket{\overline{\pi}^{(a)}}\ket{1} = \sum_{\omega \in \Omega} \frac{\overline{\pi}^{(a)}(\omega)\sqrt{a(\omega)}}{\sqrt{\pi^{(a)}(\omega)}} \ket{\omega} = \sqrt{A} \cdot \sum_{\omega \in \Omega} \frac{\overline{\pi}^{(a)}(\omega)}{\sqrt{\pi(\omega)}}.\qedhere\]
    \end{proof}

    Note in particular that if the average acceptance probability $A$ is known beforehand, we can use the rebalancing construction from \Cref{lem:rebalancing} to rebalance the transduction complexities in $V$ to get rid of the dependence on $A$ altogether. However, if we do not know $A$ (or a good approximation thereof) beforehand, then we cannot use the same trick, because we cannot choose a rebalancing parameter that depends on the input.

    This construction allows us to port the unbiased, non-destructive mean estimation routine of a bounded variable, as presented in \cite[Algorithm~3]{cornelissen2023sublinear} into the transducer framework. 

    \begin{theorem}[Bounded mean estimation transducer]
        \label{thm:transducer-bdd-mean-est}
        Let $\pi$ be a probability distribution over a finite set $\Omega$, and let $X : \Omega \to [0,1]$ be a bounded random variable. We write
        \[\ket{\pi} = \sum_{\omega \in \Omega} \sqrt{\pi(\omega)}\ket{\omega}, \qquad \ket{\overline{\pi}} = \sum_{\omega \in \Omega} \sqrt{\pi(\omega)} \cdot \frac{\mu - X(\omega)}{\sqrt{\mu(1-\mu)}} \ket{\omega}, \qquad \text{where} \qquad \mu = \underset{\omega \sim \pi}{\E} \left[X(\omega)\right].\]
        Suppose we have a bidirectional canonical transducer $U$ that reflects through $\ket{\pi}$, and a bidirectional canonical transducer $O_X$ with uniform cost bound $C_X$ that implements the mapping
        \[\ket{\omega}\ket{0} \overset{O_X}{\rightsquigarrow} \ket{\omega}(\sqrt{X(\omega)}\ket{1} + \sqrt{1-X(\omega)}\ket{0}).\]
        Let $t \geq 1$ and $\epsilon \in (0,1)$. Then, we can implement a bidirectional canonical transducer denoted by $V := \mathtt{BddMeanEst}(X,U,t,\epsilon)$ that maps
        \[\ket{\pi}\ket{0} \overset{V}{\rightsquigarrow} \ket{\pi}\ket{\widetilde{\mu}}, \qquad \text{with complexity} \qquad O\left((C_X + W(U,\ket{\pi}) + W(U,\ket{\overline{\pi}})) \cdot t \cdot \polylog(t,1/\epsilon)\right),\]
        such that measuring a subset of the qubits of $\ket{\widetilde{\mu}}$ in the computational basis yields a random variable $\widetilde{\mu}$ that satisfies
        \[|\E[\widetilde{\mu}] - \mu| \leq \epsilon, \qquad \text{and} \qquad \Var[\widetilde{\mu}] \leq \frac{91\mu}{t^2} + \epsilon.\]
    \end{theorem}

    \begin{proof}
        We describe here how one adapts the approach taken in \cite[Algorithm~3]{cornelissen2023sublinear}. First, we consider the non-destructiveness of each of its components. For the non-destructive amplitude estimation part called in \cite[Algorithm~3; Line~1]{cornelissen2023sublinear}, we observe that the analysis relies on \cite[Lemma~B.1]{cornelissen2023sublinear}. This lemma uses fixed-point amplitude amplification to post-select on the uncomputation being successful. We realize that this is exactly an instance of amplitude amplification that we can make exact in the transducer setting using \Cref{cor:ampl-ampl}. Thus, we can remove the non-exactness of this lemma in the transducer setting (and additionally remove a $\log(1/\epsilon)$-factor in this part), and make the non-destructive amplitude amplification step exact.

        Next, we focus on the non-destructive coin flip called in \cite[Algorithm~3; Line~4]{cornelissen2023sublinear}. Observe that this routine is a Las Vegas algorithm with constant cost, since the number of queries it makes is a random variable that is constant in expectation. We use the observation from Belovs and Yolcu~\cite{belovs2023one} that such algorithms can be turned into feasible solutions to the adversary bound, and hence into transducers using the work by Belovs, Jeffery and Yolcu~\cite{belovs2024taming}. Thus, we obtain a transducer that produces a measurement of the coin flip, some additional garbage, and restores the initial state with constant transduction complexity. In the theorem statement, the garbage is absorbed in the resulting state $\ket{\widetilde{\mu}}$.

        The amplitude-to-phase conversion step in \cite[Algorithm~3; Line~7]{cornelissen2023sublinear} is already non-destructive by design, and we observe that the unitary we use for the unbiased phase estimation algorithm used in \cite[Algorithm~3; Line~8]{cornelissen2023sublinear} is also non-destructive by design. Moreover, the linear amplitude amplification part is undone at the end of the computation, so we don't need it to be non-destructive. Thus, by a sequential composition of these components, the entire algorithm is now implemented as a non-destructive transducer.
        
        Consequently, similar to \Cref{thm:rejection-sampling}, we realize that the algorithm acts exclusively on the 2D-subspace spanned by $\ket{\pi}$ and $\ket{\overline{\pi}}$. As such, the total transduction complexity incurred per call to the reflection oracle through $\ket{\pi}$ is $O(C_X + W(U,\ket{\pi}) + W(U,\ket{\overline{\pi}}))$. Thus, we deduce the complexity statement and the resulting claims on the expectation and variance of $\widetilde{\mu}$ from \cite[Theorem~2.4]{cornelissen2023sublinear}.
    \end{proof}

    Using the bounded mean-estimation routine as a building block, we 
    port several more estimation routines to the transducer framework. We start by designing a quantile estimator, inspired by \cite[Algorithm~5]{cornelissen2023sublinear}.

    \begin{theorem}[Quantile estimator]
        \label{thm:quantile-est}
        Let $\pi$ be a probability distribution over a finite set $\Omega$, and let $X : \Omega \to [0,M]$ be a (univariate) random variable, with $M > 0$. We define the quantile function $Q_X : [0,1] \to [0,M]$ as
        \[Q_X(p) := \inf\left\{t \in [0,M] : \underset{x \sim \pi}{\P} \left[X(x) \leq t\right] \geq p\right\}.\]
        Now, let $p \in (0,1)$, $\epsilon \in (0,\min\{p,1-p\})$ and $\Delta > 0$, such that $Q_X(p + \epsilon) - Q_X(p - \epsilon) \geq \Delta$. Suppose we have a bidirectional canonical transducer $U$ that reflects through $\ket{\pi}$. For any $t \in [0,M]$, we write $X_t := 1/4 + \mathbbm{1}(X \leq t)/2$, and suppose that we have a bidirectional canonical transducer $O_X^{(t)}$ with uniform cost bound $C_X$ that acts as
        \[\ket{\omega}\ket{0} \overset{O_X^{(t)}}{\rightsquigarrow} \ket{\omega}(\sqrt{X_t(\omega)}\ket{1} + \sqrt{1-X_t(\omega)}\ket{0}).\]
        Then we can construct a bidirectional canonical transducer $V := \mathtt{QuantileEst}(X,U,p,\epsilon,M,\Delta,\rho)$, for $\rho > 0$, that acts as $\ket{\pi}\ket{0} \overset{V}{\rightsquigarrow} \ket{\pi}\ket{\widetilde{q}_p}$, with complexity
        \[\widetilde{O}\left(\left(C_X + W(U,\ket{\pi} + \sup_{t \in [0,M]} W\left(U,\ket{\overline{\pi}\left(X_t\right)}\right)\right) \cdot \frac{1}{\epsilon}\left(1 + \log\frac{M}{\Delta}\right) \log\frac{1}{\rho}\right),\]
        where for any random variable $Y : \Omega \to [0,1]$, we write
        \[\ket{\overline{\pi}(Y)} = \sum_{\omega \in \Omega} \sqrt{\pi(\omega)} \cdot \frac{\mu_Y - Y(\omega)}{\sqrt{\mu_Y(1 - \mu_Y)}} \ket{\omega}, \qquad \text{where} \qquad \mu_Y = \underset{\omega \sim \pi}{\E} \left[Y(\omega)\right].\]
        such that a computational basis measurement on a subset of the qubits of $\ket{\widetilde{q}_p}$ yields a value $q_p$ such that with probability at least $1-\rho$, $\P_{x \sim \pi}[X(x) \leq q_p] \in [p-3\epsilon,p+3\epsilon]$.
    \end{theorem}

    \begin{proof}
        We follow the approach outlined in \cite[Algorithm~5]{cornelissen2023sublinear}. We run binary search over the interval $[0,M]$, where at every value $t \in [0,M]$ we compute an $\epsilon$-approximation of $\P_{x \sim \pi} \left[X(x) \leq t\right]$. We stop whenever the estimated value is in the interval $[p-2\epsilon,p+2\epsilon]$. Since the interval of $t$'s for which we are guaranteed to obtain such an outcome is at least $\Delta$, we can find such a value in at most $T \in O(2 + \log(M/\Delta))$ iterations of binary search.

        It remains to describe the approximation routine. We use \Cref{thm:transducer-bdd-mean-est} with the random variable $X_t(x) = 1/4 + \mathbbm{1}(X(x) \leq t)/2$, repetition parameter $t' = 60/\epsilon$ and precision parameter $\epsilon' = \epsilon^2/32$. The resulting bidirectional canonical transducer $\mathtt{BddMeanEst}(X_t, U, t', \epsilon')$ now produces a random variable $\widetilde{\mu}_t$ that satisfies
        \[|\E[\widetilde{\mu}_t] - \mu_t| \leq \epsilon' \leq \frac{\epsilon}{32}, \quad \text{and} \quad \Var[\widetilde{\mu}_t] \leq \frac{91}{(t')^2} + \epsilon' \leq \frac{\epsilon^2}{16}, \quad \text{where} \quad \mu_t = \E[X_t] = \frac14 + \frac{\underset{x \sim \pi}{\P}[X(x) \leq t]}{2}.\]
        Thus, using Chebyshev's inequality, we obtain that
        \[\P\left[\left|\left(2\widetilde{\mu}_t - \frac12\right) - \underset{x \sim \pi}{\P}[X(x) \leq t]\right| \geq \epsilon\right] = \P\left[|\widetilde{\mu}_t - \mu_t| \geq \frac{\epsilon}{2}\right] \leq \P\left[|\widetilde{\mu}_t - \E[\widetilde{\mu}_t]| \geq \frac{15\epsilon}{32}\right] < \frac13.\]
        
        We repeat the above process $O(\log(T/\rho))$ times and take the median of the estimated probabilities to obtain an $\epsilon$-precise estimate of $\P_{x \sim \pi}[X(x) \leq t]$ with probability at least $1-\rho/T$. By a union bound, the probability that each iteration succeeds is at least $1-\rho$. Finally, observe from \Cref{thm:transducer-bdd-mean-est} that the cost of one estimation procedure is $\widetilde{O}((C_X + W(U,\ket{\pi}) + W(U,\ket{\overline{\pi}(X_t)}))/\epsilon)$. Thus, this procedure produces an output $\widetilde{q}_p \in [0,M]$ such that $\P_{x \sim \pi} [X(x) \leq \widetilde{q}_p] \in [p-3\epsilon,p+3\epsilon]$ with probability at least $1-\rho$, and with the claimed complexity.
    \end{proof}

    We note that it is at the moment not obvious why we use $X_t = 1/4 + \mathbbm{1}(X \leq t)/2$ rather than simply $X_t = \mathbbm{1}(X \leq t)$. However, it does not hurt the above approach more than a constant factor that is hidden in the big-$O$-notation, and it turns out that we will need this modification for technical reasons in the volume estimation algorithm in~\Cref{sec:vol-est}.

    Finally, we also port the unbiased, non-destructive mean-estimation routine, presented in \cite[Algorithm~4]{cornelissen2023sublinear}. We generalize the statement slightly to the setting where we can only reflect through a q-sample that merely approximates the target probability distribution up to some small total-variation distance $\eta \geq 0$.  
    We will need this more general version in the volume estimation algorithm, which we present in \Cref{sec:vol-est}.

    \begin{theorem}[Mean estimation]
        \label{thm:transducer-mean-est}
        Let $\pi,\widetilde{\pi}$ be distributions over a finite state space $\Omega$, with $\norm{\pi-\widetilde{\pi}}_{\mathrm{TV}} \leq \eta$ for $\eta \geq 0$. Let $X : \Omega \to \R$ be a random variable. We write
        \[\ket{\widetilde{\pi}} = \sum_{\omega \in \Omega} \sqrt{\widetilde{\pi}(\omega)} \ket{\omega}, \qquad \mu_X = \underset{\omega \sim \pi}{\E} \left[X(\omega)\right], \qquad \text{and} \qquad \widetilde{\mu}_X = \underset{\omega \sim \widetilde{\pi}}{\E} \left[X(\omega)\right].\]
        For any random variable $Y : \Omega \to [0,1]$, we define the state
        \[\ket{\overline{\pi}(Y)} = \sum_{\omega \in \Omega} \sqrt{\widetilde{\pi}(\omega)} \cdot \frac{\widetilde{\mu}_Y - Y(\omega)}{\sqrt{\widetilde{\mu}_Y(1-\widetilde{\mu}_Y)}} \ket{\omega}.\]
        Suppose we have a bidirectional canonical transducer $U$ that reflects through $\ket{\widetilde{\pi}}$. Suppose that $\Var_{\omega \sim \pi}[X] \leq \widetilde{\sigma}^2$, and $\widetilde{m} \in [0,M]$ such that $|\widetilde{m} - \mu_X| \leq 17\widetilde{\sigma}$. Let $t \geq 1$ and $\epsilon \in (0,1)$. We write $k = \lceil \log(580/\epsilon)\rceil$, and for all $j \in \{0,\dots,k\}$, $a_j = 2^j\widetilde{\sigma}$. We also let $a_{-1} = 0$, $B_j = \min\{1,580/2^{2j}\}$, and for all $j \in \{0, \dots, k\}$,
        \[Y_j^{\pm}(\omega) := \frac{B_j}{4} + \left(1 - \frac{B_j}{2}\right) \cdot \pm\frac{X(\omega) - \widetilde{m}}{a_j} \cdot \mathbbm{1}_{\pm(X(\omega) - \widetilde{m}) \in (a_{j-1},a_j]}.\]
        Suppose that for every $j \in \{0,\dots,k\}$ and sign $\pm$, we have a transudcer $O_X^{(j,\pm)}$ with uniform cost bound $C_X$ that implements
        \[\ket{\omega}\ket{0} \overset{O_X^{(j,\pm)}}{\rightsquigarrow} \ket{\omega}\left(\sqrt{Y_j^{\pm}(\omega)}\ket{1} + \sqrt{1-Y_j^{\pm}(\omega)}\ket{0}\right).\]
        Then, we can construct a bidirectional canonical transducer $V := \mathtt{MeanEst}(X,U,\widetilde{m},\widetilde{\sigma},t,\epsilon)$ that acts as
        \[\ket{\widetilde{\pi}}\ket{0} \overset{V}{\rightsquigarrow} \ket{\widetilde{\pi}}\ket{\overline{\mu}_X}, \quad \text{with complexity} \quad \widetilde{O}\left(\sum_{j=0}^k (C_X + W(U,\ket{\widetilde{\pi}}) + W(U,\ket{\overline{\pi}(Y_j^{\pm})}) \cdot t \cdot \polylog(t,1/\epsilon)\right),\]
        such that measuring $\ket{\overline{\mu}_X}$ on a subset of the qubits produces a random variable $\overline{\mu}_X$ that satisfies
        \[|\E[\overline{\mu}_X] - \mu_X| \leq \left(\epsilon + 2^{14} \cdot \frac{\eta}{\epsilon}\right)\widetilde{\sigma}, \qquad \text{and} \qquad \Var[\overline{\mu}_X] \leq \left(\frac{1}{t^2} + 2^{32} \cdot \frac{\eta}{\epsilon^2}\right)\widetilde{\sigma}^2.\]
    \end{theorem}

    \begin{proof}
        We follow the construction in \cite[Algorithm~4]{cornelissen2023sublinear}. For all $j \in \{0, \dots, k\}$ and $\omega \in \Omega$, we easily observe that $0 \leq Y_j^{\pm}(\omega) \leq 1$. For all $j \in [k]$, we have
        \begin{align*}
            \frac{1}{2^{2j}} &\leq \frac{B_j}{4} \leq \mu_{Y_j^{\pm}} \leq \frac{B_j}{4} + \underset{\omega \sim \pi}{\E} \left[\pm\frac{X(\omega) - \widetilde{m}}{a_j} \cdot \mathbbm{1}_{\pm(X(\omega) - \widetilde{m}) \in (a_{j-1},a_j]}\right] \\
            &\leq \frac{580}{2^{2j}} + 2\underset{\omega \sim \pi}{\E} \left[\frac{(X(\omega) - \widetilde{\mu})^2}{a_j^2}\right] \leq \frac{580}{2^{2j}} + \frac{2\underset{\omega \sim \pi}{\Var}[X(\omega)] + 2|\mu_X - \widetilde{m}|^2}{a_j^2} \leq \frac{580}{2^{2j}} + \frac{580\widetilde{\sigma}^2}{2^{2j}\widetilde{\sigma}^2} = \frac{1160}{2^{2j}}.
        \end{align*}
        Consequently, we have that $1/2^{2j} \leq \widetilde{\mu}_{Y_j^{\pm}} \leq 1160/2^{2j} + \eta$.
        
        We now use \Cref{thm:transducer-bdd-mean-est} to obtain a bidirectional canonical transducers $V_j^{\pm} := \mathtt{BddMeanEst}(Y_j^{\pm}, U, t', \epsilon')$, with $t' = 2^{11}t\sqrt{k+1}$ and $\epsilon' = \epsilon^2/(2^{27}t^2)$. Then, the transduction complexity is as claimed. We obtain outcome random variables $\overline{\mu}_j^{\pm}$, and we define
        \[\widetilde{\mu} := \widetilde{m} + \sum_{j=0}^k a_j\left(\widetilde{\mu}_j^+ - \widetilde{\mu}_j^-\right), \qquad \text{where} \qquad \widetilde{\mu}_j^{\pm} = \frac{\overline{\mu}_j^{\pm} - \frac{B_j}{4}}{1 - \frac{B_j}{2}}.\]
        We observe that
        \begin{align*}
            |\E[\widetilde{\mu}] - \mu| &\leq \sum_{j=0}^k a_j\left(\frac{\left|\E[\overline{\mu}_j^+] - \mu_{Y_j^+}\right|}{1 - \frac{B_j}{2}} + \frac{\left|\E[\overline{\mu}_j^-] - \mu_{Y_j^-}\right|}{1 - \frac{B_j}{2}}\right) + \underset{\omega \sim \pi}{\E}[|X(\omega) - \widetilde{m}| \cdot \mathbbm{1}_{|X(\omega) - \widetilde{m}| \geq a_k}] \\
            &\leq \sum_{j=0}^k 4a_j(\epsilon' + \eta) + \underset{\omega \sim \pi}{\E} \left[\frac{(X(\omega) - \widetilde{m})^2}{a_k}\right] \leq \frac{9280}{\epsilon} \left(\frac{\epsilon^2}{2^{27}} + \eta\right)\widetilde{\sigma} + \frac{290\epsilon}{580} \widetilde{\sigma} < \left(2^{14} \cdot \frac{\eta}{\epsilon} + \epsilon\right)\widetilde{\sigma}.
        \end{align*}
        Finally, since all the random variables are independent, we have
        \begin{align*}
            \Var[\widetilde{\mu}] &= \sum_{j=0}^k a_j^2\left(\Var[\widetilde{\mu}_j^+] + \Var[\widetilde{\mu}_j^-]\right) \leq 4\sum_{j=0}^k a_j^2\left(\Var[\overline{\mu}_j^+] + \Var[\overline{\mu}_j^-]\right) \leq 4\sum_{j=0}^k a_j^2 \cdot \left(\frac{91(\widetilde{\mu}_{Y_j^+} + \widetilde{\mu}_{Y_j^-})}{(t')^2} + 2\epsilon'\right) \\
            &\leq 8\sum_{j=0}^k 2^{2j}\widetilde{\sigma}^2\left(\frac{105560}{2^{2j}(t')^2} + \frac{91\eta}{(t')^2} + \frac{\epsilon^2}{18560}\right) = 8\widetilde{\sigma}^2 \cdot \left(\frac{105560(k+1)}{(t')^2} + \frac{91 \cdot 4^{k+1}\eta}{(t')^2} + 4^{k+1} \cdot \frac{\epsilon^2}{2^{27}t^2}\right) \\
            &< \left(\frac{1}{t^2} + 2^{32} \cdot \frac{\eta}{\epsilon^2}\right)\widetilde{\sigma}^2.\qedhere
        \end{align*}
    \end{proof}

    \section{Reflection through the stationary state with transducers}
    \label{sec:quantum-walks-transducers}

    In this section, we construct a transducer that reflects through the stationary distribution of a random walk on an electrical network, given transducers that construct and reflect around the star states at every vertex. The transducer we build is based on the quantum walk construction introduced by Apers, Roland and Zhang~\cite{apers2026elfs}.

    \subsection{Transducers for electrical networks}

    We start by defining several Hilbert spaces associated to a walk on an electrical network.

    \begin{definition}
        Let $(G = (V,E),r)$ be an electrical network with $R = 1$. For all $v \in V$, we write
        \[\ket{*_v} := \sum_{w \in N(v)} \sqrt{\frac{r_v}{r_{vw}}}\ket{w} \in \C^V.\]
        Next, we define the spaces $\mathcal{H}_1, \mathcal{H}_2 \subseteq (C^V)^{\otimes 2}$ as
        \[\mathcal{H}_1 := \Span\{\ket{v}\ket{*_v}\}, \qquad \text{and} \qquad \mathcal{H}_2 := \Span\{\ket{v,w} + \ket{w,v} : v,w \in V\}.\]
        We also define the isometry $\mathcal{N} : \C^V \to (\C^{V})^{ \otimes 2}$ as $\mathcal{N} : \ket{v} \mapsto \ket{v}\ket{*_v}$ so that $\mathcal{H}_1 = \mathrm{Im}(\mathcal{N})$.
    \end{definition}
    
    Note that measuring $\ket{*_v}$ in the computational basis amounts to sampling a step of the random walk at vertex $v \in V$, as defined in \Cref{def:random-walk-electrical-network}. Moreover, observe from \Cref{thm:electrical-network-walk} that this random walk is reversible, with $\pi_v = 1/r_v$. It turns out that the stationary distribution exactly captures the intersection between the two spaces $\mathcal{H}_1$ and $\mathcal{H}_2$ we just defined, as shown in the following lemma.
    
    \begin{lemma}
        \label{lem:electrical-network-properties}
        Let $(G = (V,E), r)$ be an electrical network with $R = 1$. For every edge $vw \in E$ we associate an arbitrary positive direction, and we write $\C^E := \mathrm{span}\{\ket{v,w},\ket{w,v} : vw \in E\}$. For all $f : E \to \C$, we define
        \[\ket{f} := \sum_{vw \in E} f_{vw}\sqrt{r_{vw}}(\ket{v,w} - \ket{w,v}) \in \C^E.\]
        For every net-flow $\overline{\pi} \in \C^V$, i.e., $\sum_{v \in V} \overline{\pi}_v = 0$, we let $f^{\overline{\pi}}$ be the minimum-energy flow that has $\delta_f = \overline{\pi}$, and we write
        \[\ket{\overline{\pi}} := \sum_{v \in V} \frac{\overline{\pi}_v}{\sqrt{\pi_v}}\ket{v}.\]
        Then,
        \begin{enumerate}[nosep]
            \item $\mathcal{N}(\C^V) = \mathcal{H}_1$.
            \item $\mathcal{H}_1 \cap \mathcal{H}_2 = \Span\{\mathcal{N}\ket{\pi}\}$.
            \item $\mathcal{H}_2^{\perp} \cap \C^E = \Span\{\ket{f} : f : E \to \C\}$.
            \item If $f : E \to \C$, then $\Pi_{\mathcal{H}_1}\ket{f} = \sum_{v \in V} \delta_f(v)/\sqrt{\pi_v}\ket{v}\ket{*_v}$.
            \item $\mathcal{H}_1^{\perp} \cap \mathcal{H}_2^{\perp} \cap \C^E$ is the circulation space, corresponding to all flows $\ket{f}$ with $\delta_f = 0$ (zero net-flow).
            \item $\mathcal{H}_2^{\perp} \cap (\mathcal{H}_1^{\perp} \cap \mathcal{H}_2^{\perp})^{\perp} \cap \C^E = \Span\{\ket{f^{\overline{\pi}}} : \overline{\pi} \in \C^V \text{ a net-flow}\}$.
            \item $\Pi_{\mathcal{H}_1}\ket{f^{\overline{\pi}}} = \mathcal{N}\ket{\overline{\pi}}$.
        \end{enumerate}
    \end{lemma}

    \begin{proof}
        The first claim is clear. For the second one, suppose that $\mathcal{N}\ket{\psi} \in \mathcal{H}_2$. We write $\ket{\psi} = \sum_{v \in V} \psi_v\ket{v}$, and let $p_{vw} = r_v/r_{vw}$ be the transition probability from $v$ to $w$. We use detailed balance to show for all $v,w \in V$ that
        \[\psi_v \cdot \sqrt{p_{vw}} = \psi_v \cdot \sqrt{\frac{r_v}{r_{vw}}} = \psi_w \cdot \sqrt{\frac{r_w}{r_{vw}}} = \psi_w \cdot \sqrt{p_{wv}} = \psi_w \cdot \sqrt{\frac{\pi_w}{\pi_v}} \cdot \sqrt{p_{vw}}.\]
        Here the second equality follows from containment of $\mathcal{N}\ket{\psi} \in \mathcal{H}_2$, implying that $\bra{v,w}\mathcal{N}\ket{\psi} = \bra{w,v}\mathcal{N}\ket{\psi}$.
        Hence, we find that $\psi_v/\psi_w = \sqrt{\pi_w/\pi_v}$, from which we deduce that $\psi_v \propto \sqrt{\pi_v}$.

        For the third item, observe that $\mathcal{H}_2^{\perp}$ is the anti-symmetric subspace, so restricting that to $\C^E$ restricts the space to embeddings of flows $f : E \to \C$. 
        
        For the fourth item, we observe that for any $f : E \to \C$, we have for all $v \in V$,
        \[(\bra{v}\bra{*_v})\ket{f} = \sum_{w \in N^+(v)} \sqrt{\frac{r_v}{r_{vw}}} \cdot f_{vw}\sqrt{r_{vw}} - \sum_{w \in N^-(v)} \sqrt{\frac{r_v}{r_{vw}}} \cdot f_{vw}\sqrt{r_{vw}} = \frac{\delta_f(v)}{\sqrt{\pi_v}}.\]

        For the fifth item, observe that a flow $f : E \to \C$ is a circulation if and only if $\delta_f \equiv 0$. Using items 3 and 4 yields the statement.

        For the sixth item, observe that a flow is a minimum-energy flow for some net-flow $\overline{\pi}$ if and only if it is orthogonal to all circulations. Combining items 3 and 5 yields the result.

        Item 7 follows directly from item 4.
    \end{proof}

    As the stationary distribution is found by taking the intersection of two subspaces, we can construct a transducer that reflects through it using the transducer designed by Apers, Roland and Zhang, see~\Cref{thm:two-subspace-transducer}. We analyze the resulting catalyst state and transduction complexity in the following lemma.

    \begin{lemma}
        \label{lem:net-flow-catalyst}
        Let $(G = (V,E),r)$ be an electrical network with $R = 1$. The operator $U = (2\Pi_{\mathcal{H}_1} - I)(2\Pi_{\mathcal{H}_2} - I)$ is a transducer with public and private spaces $\mathcal{H}_1$ and $\mathcal{H}_1^{\perp}$, respectively, that acts as
        \begin{align*}
            \mathcal{N}\ket{\pi} \overset{U}{\rightsquigarrow} \mathcal{N}\ket{\pi}, & \quad \text{with catalyst} \quad 0, & \text{and complexity} \quad & W(U,\ket{\pi}) = 0, \\
            \mathcal{N}\ket{\overline{\pi}} \overset{U}{\rightsquigarrow} -\mathcal{N}\ket{\overline{\pi}}, & \quad \text{with catalyst} \quad \Pi_{\mathcal{H}_1^{\perp}}\ket{f^{\overline{\pi}}}, & \text{and complexity} \quad & W(U,\ket{\overline{\pi}}) = 2R_{\mathrm{eff}} (G,r;\overline{\pi}) - \norm{\ket{\overline{\pi}}}^2,
        \end{align*}
        for all net-flows $\overline{\pi} \in \C^V$.
       \end{lemma}

    \begin{proof}
    The lemma follows from the general two-subspace transducer in \Cref{thm:two-subspace-transducer} with $\Pi = \Pi_{\mathcal{H}_1^{\perp}}$ and $\Delta = \Pi_{\mathcal{H}_2^{\perp}}$.
    For this particular case however, we can also provide a simpler, self-contained proof.
    First, we trivially have that
    \[
    U \mathcal{N} \ket{\pi}
    = \mathcal{N} \ket{\pi},
    \]
    and this proves the first equation.
    Then, for $\ket{\bar\pi}$ that satisfies $\braket{\bar\pi}{\pi}=0$ and has corresponding electric flow vector $\ket{f^{\bar\pi}}$, we have
        \[
    U(\mathcal{N}\ket{\bar\pi} + \Pi_{\mathcal{H}_1^\perp} \ket{f^{\bar\pi}})
    = U \ket{f^{\bar\pi}}
    = -(2\Pi_{\mathcal{H}_1}-I) \ket{f^{\bar\pi}}
    = -\mathcal{N}\ket{\bar\pi} + \Pi_{\mathcal{H}_1^\perp} \ket{f^{\bar\pi}}.
    \]
    It remains to note that
    \[
    \| \Pi_{\mathcal{H}_1^\perp} \ket{f^{\bar\pi}} \|^2
    = \| \ket{f^{\bar\pi}} \|^2 - \| \mathcal{N}\ket{\bar\pi} \|^2
    = 2R_{\mathrm{eff}} (G,r;\overline{\pi}) - \norm{\ket{\overline{\pi}}}^2. \qedhere
    \]
    \end{proof}

    The construction in \Cref{lem:net-flow-catalyst} generates a transducer that reflects through the stationary state. However, similarly to the construction in \cite{apers2026elfs}, the public and private spaces, $\mathcal{H}_1$ and $\mathcal{H}_1^{\perp}$, respectively, depend on the star-states $\ket{*_v}$. Thus, if we want to implement a canonical transducer that only accesses the graph structure through an oracle, then we must make sure that the public and private spaces are input-independent. This is the objective of the following theorem.

    \begin{theorem}
        \label{thm:transducer-reflection-stationary}
        Let $(G = (V,E), r)$ be an electrical network with $R = 1$. Let $\mathcal{H} = (\C^V)^{\otimes 2}$ and $U = (2\Pi_{\mathcal{H}_1} - I)(2\Pi_{\mathcal{H}_2} - I) \in \mathcal{U}(\mathcal{H})$ be the transducer constructed in \Cref{lem:net-flow-catalyst}. Now, let $A \in \mathcal{U}(\mathcal{H}^{\oplus 2})$ be defined as
        \[A = (H \otimes I_{\mathcal{H}})(I_{\mathcal{H}} \oplus (2\Pi_{\mathcal{H}_1} - I_{\mathcal{H}}))(H \otimes I_{\mathcal{H}}), \qquad \text{where} \qquad H = \frac{1}{\sqrt{2}}\begin{bmatrix}
            1 & 1 \\ 1 & -1
        \end{bmatrix}.\]
        Next, let $T = A^{\dagger}(U \oplus I_{\mathcal{H}})A$. Then, $T$ is a transducer with public space $\mathcal{H}$ and private space $\mathcal{H}$, and acts as
        \begin{align*}
            \mathcal{N}\ket{\pi} \overset{T}{\rightsquigarrow} \mathcal{N}\ket{\pi}, & \quad \text{with catalyst} \quad 0, & \text{and complexity} \quad & W(U,\ket{\pi}) = 0, \\
            \mathcal{N}\ket{\overline{\pi}} \overset{T}{\rightsquigarrow} -\mathcal{N}\ket{\overline{\pi}}, & \quad \text{with catalyst} \quad \Pi_{\mathcal{H}_1^{\perp}}\ket{f^{\overline{\pi}}}, & \text{and complexity} \quad & W(U,\ket{\overline{\pi}}) = 2R_{\mathrm{eff}} (G,r;\overline{\pi}) - \norm{\ket{\overline{\pi}}}^2,
        \end{align*}
        for all net-flows $\overline{\pi} \in \C^V$.
    \end{theorem}

    \begin{proof}
        We first analyze the action of $A$. We write $X = I_{\mathcal{H}} \oplus (2\Pi_{\mathcal{H}_1} - I_{\mathcal{H}})$, and observe that
        \begin{align*}
            A : \begin{bmatrix}
                \mathcal{N}\ket{\pi} \\
                0
            \end{bmatrix} \overset{H \otimes I}{\mapsto} \frac{1}{\sqrt{2}}\begin{bmatrix}
                \mathcal{N}\ket{\pi} \\
                \mathcal{N}\ket{\pi}
            \end{bmatrix} &\overset{X}{\mapsto} \frac{1}{\sqrt{2}}\begin{bmatrix}
                \mathcal{N}\ket{\pi} \\
                \mathcal{N}\ket{\pi}
            \end{bmatrix} \overset{H \otimes I}{\mapsto} \begin{bmatrix}
                \mathcal{N}\ket{\pi} \\
                0
            \end{bmatrix}, \qquad \text{and} \\
            A : \begin{bmatrix}
                \pm\mathcal{N}\ket{\overline{\pi}} \\
                \Pi_{\mathcal{H}_1^{\perp}}\ket{f^{\overline{\pi}}}
            \end{bmatrix} \overset{H \otimes I}{\mapsto} \frac{1}{\sqrt{2}}\begin{bmatrix}
                \pm\mathcal{N}\ket{\overline{\pi}} + \Pi_{\mathcal{H}_1^{\perp}}\ket{f^{\overline{\pi}}} \\
                \pm\mathcal{N}\ket{\overline{\pi}} - \Pi_{\mathcal{H}_1^{\perp}}\ket{f^{\overline{\pi}}}
            \end{bmatrix} &\overset{X}{\mapsto} \frac{1}{\sqrt{2}}\begin{bmatrix}
                \pm\mathcal{N}\ket{\overline{\pi}} + \Pi_{\mathcal{H}_1^{\perp}}\ket{f^{\overline{\pi}}} \\
                \pm\mathcal{N}\ket{\overline{\pi}} + \Pi_{\mathcal{H}_1^{\perp}}\ket{f^{\overline{\pi}}}
            \end{bmatrix} \overset{H \otimes I}{\mapsto} \begin{bmatrix}
                \pm\mathcal{N}\ket{\overline{\pi}} + \Pi_{\mathcal{H}_1^{\perp}}\ket{f^{\overline{\pi}}} \\
                0
            \end{bmatrix}.
        \end{align*}
        Next, we use \Cref{lem:net-flow-catalyst} to analyze the action of $T$, to find that
        \begin{align*}
            T: \begin{bmatrix}
                \mathcal{N}\ket{\pi} \\
                0
            \end{bmatrix} \overset{A}{\mapsto} \begin{bmatrix}
                \mathcal{N}\ket{\pi} \\
                0
            \end{bmatrix} &\overset{U \oplus I_{\mathcal{H}}}{\mapsto} \begin{bmatrix}
                \mathcal{N}\ket{\pi} \\
                0
            \end{bmatrix} \overset{A^{\dagger}}{\mapsto} \begin{bmatrix}
                \mathcal{N}\ket{\pi} \\
                0
            \end{bmatrix}, \qquad \text{and} \\
            T : \begin{bmatrix}
                \mathcal{N}\ket{\overline{\pi}} \\
                \Pi_{\mathcal{H}_1^{\perp}}\ket{f^{\overline{\pi}}} 
            \end{bmatrix} \overset{A}{\mapsto} \begin{bmatrix}
                \mathcal{N}\ket{\overline{\pi}} + \Pi_{\mathcal{H}_1^{\perp}}\ket{f^{\overline{\pi}}} \\
                0
            \end{bmatrix} &\overset{U \oplus I_{\mathcal{H}}}{\mapsto} \begin{bmatrix}
                -\mathcal{N}\ket{\overline{\pi}} + \Pi_{\mathcal{H}_1^{\perp}}\ket{f^{\overline{\pi}}} \\
                0
            \end{bmatrix} \overset{A^{\dagger}}{\mapsto} \begin{bmatrix}
                -\mathcal{N}\ket{\overline{\pi}} \\
                \Pi_{\mathcal{H}_1^{\perp}}\ket{f^{\overline{\pi}}}
            \end{bmatrix}.\qedhere
        \end{align*}
    \end{proof}

    We note that $T$ only depends on the star states of the electrical network through the three reflections through $\mathcal{H}_1$ that it makes in total. We see in the next section how we can construct a canonical transducer that performs this reflection.

    \subsection{Amortized reflection through the stationary state}

    The construction from the previous section produces a transducer that reflects through the stationary distribution, making three calls to the reflection through $\mathcal{H}_1$. Remember that $\mathcal{H}_1$ is the space spanned by the states $\ket{v}\ket{*_v}$, for all $v \in V$, and as such, can be written as
    \begin{equation}
        \label{eq:star-state-reflection=decomposition}
        \mathcal{H}_1 = \bigoplus_{v \in V} \Span\{\ket{*_v}\}, \qquad \text{and so} \qquad 2\Pi_{\mathcal{H}_1} - I = \bigoplus_{v \in V} (2\ket{*_v}\bra{*_v} - I).
    \end{equation}
    As such, we can construct a transducer that performs the reflection through $\mathcal{H}_1$ as a parallel composition~\cite[Proposition~9.4]{belovs2024taming}, amortizing the cost over all the vertices. We make this construction explicit in the following lemma.

    Throughout, we will assume that electrical networks can depend on the implicit input only through the edge set and the resistances. That is, in an electrical network $(G = (V,E),r)$, only $E$ and $r$ may depend on the implicit input, but $V$ may not.

    \begin{lemma}
        \label{lem:transducer-reflection-stationary}
        Let $(G = (V,E),r)$ be an electrical network. Suppose that for all $v \in V$, we have a bidirectional canonical transducer $U_v$, acting on $\C^V$, that reflects through $\ket{*_v}$. Then, we can construct a bidirectional canonical transducer $U$, acting on $(\C^V)^{\otimes 2}$ that acts as
        \begin{align*}
            \mathcal{N}\ket{\pi} \overset{U}{\rightsquigarrow} \mathcal{N}\ket{\pi}, & \qquad \text{with complexity} \qquad O\left(\sum_{v \in V} \pi_vW(U_v, \ket{*_v})\right), \\
            \mathcal{N}\ket{\overline{\pi}} \overset{U}{\rightsquigarrow} -\mathcal{N}\ket{\overline{\pi}}, & \qquad \text{with complexity} \qquad O\left(\sum_{v \in V} \left(\frac{|\overline{\pi}_v|^2}{\pi_v} W(U_v, \ket{*_v}) + W(U_v, (\bra{v} \otimes I_V)\Pi_{\mathcal{H}_1^{\perp}}\ket{f^{\overline{\pi}}})\right)\right),
        \end{align*}
        for all net-flows $\overline{\pi} \in \C^V$.
    \end{lemma}

    \begin{proof}
        Let $\mathcal{H} = (\C^V)^{\otimes 2}$. We start by observing that $2\Pi_{\mathcal{H}_1} - I_{\mathcal{H}}$ can be written as a direct sum over reflections through star states, as in \Cref{eq:star-state-reflection=decomposition}. Thus, whenever we want to apply $2\Pi_{\mathcal{H}_1} - I_{\mathcal{H}}$, we use~\Cref{thm:parallel-composition} to compose the $U_v$'s in parallel and construct a bidirectional canonical transducer $R$ that for all $\ket{\psi} \in \mathcal{H}$ acts as
        \[\ket{\psi} \overset{R}{\rightsquigarrow} (2\Pi_{\mathcal{H}_1} - I_{\mathcal{H}})\ket{\psi}, \qquad \text{with complexity} \qquad \sum_{v \in V} W(U_v, (\bra{v} \otimes I_V)\ket{\psi}).\]
        
        Next, we observe from \Cref{thm:transducer-reflection-stationary} that the operation $T$ makes three calls to $2\Pi_{\mathcal{H}_1} - I_{\mathcal{H}}$. Thus, using sequential composition, i.e., \Cref{thm:sequential-composition}, we construct a bidirectional canonical transducer $T'$ that acts as
        \begin{subequations}
        \begin{align}
            \label{eq:reflection-transduction-a}
            \begin{bmatrix}
                \mathcal{N}\ket{\pi} \\
                0
            \end{bmatrix} \overset{T'}{\rightsquigarrow} \begin{bmatrix}
                \mathcal{N}\ket{\pi} \\
                0
            \end{bmatrix}, & \qquad \text{with complexity} \qquad A_1 \\
            \label{eq:reflection-transduction-b}
            \begin{bmatrix}
                \mathcal{N}\ket{\overline{\pi}} \\
                \Pi_{\mathcal{H}_1^{\perp}}\ket{f^{\overline{\pi}}}
            \end{bmatrix} \overset{T'}{\rightsquigarrow} \begin{bmatrix}
                -\mathcal{N}\ket{\pi} \\
                \Pi_{\mathcal{H}_1^{\perp}}\ket{f^{\overline{\pi}}}
            \end{bmatrix}, & \qquad \text{with complexity} \qquad A_2.
        \end{align}
        \end{subequations}
        By checking what states the reflections through $\mathcal{H}_1$ act on when they are called, we can analyze the witness complexities, to obtain
        \begin{align*}
            A_1 &\leq W\left(R,\frac{1}{\sqrt{2}}\mathcal{N}\ket{\pi}\right) + W\left(R, \mathcal{N}\ket{\pi}\right) + W\left(R, \frac{1}{\sqrt{2}}\mathcal{N}\ket{\pi}\right) = 2W(R, \mathcal{N}\ket{\pi}) \leq \sum_{v \in V} \pi_vW(U_v, \ket{*_v}), \\
            A_2 & \leq W\left(R, \frac{1}{\sqrt{2}}(\mathcal{N}\ket{\overline{\pi}} - \Pi_{\mathcal{H}_1^{\perp}}\ket{f^{\overline{\pi}}})\right) + W(R, \mathcal{N}\ket{\overline{\pi}} + \Pi_{\mathcal{H}_1^{\perp}}\ket{f^{\overline{\pi}}}) + W\left(R, \frac{1}{\sqrt{2}}(-\mathcal{N}\ket{\overline{\pi}} + \Pi_{\mathcal{H}_1^{\perp}}\ket{f^{\overline{\pi}}})\right) \\
            &\in O\left(W(R, \mathcal{N}\ket{\overline{\pi}}) + W(R, \Pi_{\mathcal{H}_1^{\perp}}\ket{f^{\overline{\pi}}})\right) \subseteq O\left(\sum_{v \in V} \left(\frac{|\overline{\pi}_v|^2}{\pi_v} W(U_v, \ket{*_v}) + W(U_v, (\bra{v} \otimes I_V)\Pi_{\mathcal{H}_1^{\perp}}\ket{f^{\overline{\pi}}})\right)\right).
        \end{align*}

        Finally, let $\D$ be the implicit set of inputs. That is, for every input $x \in \D$, $T'$ calls an oracle $O_x$ and its inverse. We observe from \Cref{eq:reflection-transduction-a,eq:reflection-transduction-b} that $T'$ is a bidirectional canonical transducer that solves the state-conversion problem:
        \[P := \left\{\left(\begin{bmatrix}
            \mathcal{N}_x\ket{\pi_x} \\
            0
        \end{bmatrix}, \begin{bmatrix}
            \mathcal{N}_x\ket{\pi_x} \\
            0
        \end{bmatrix}, O_x\right), \left(\begin{bmatrix}
            \mathcal{N}_x\ket{\overline{\pi}}_x \\
            \Pi_{(\mathcal{H}_1^{\perp})_x}\ket{f_x^{\overline{\pi}}}
        \end{bmatrix}, \begin{bmatrix}
            -\mathcal{N}_x\ket{\overline{\pi}}_x \\
            \Pi_{(\mathcal{H}_1^{\perp})_x}\ket{f_x^{\overline{\pi}}}
        \end{bmatrix}, O_x\right)\right\}_{x \in \D}.\]
        Thus, through \Cref{thm:transducer-to-adv,thm:unidirectional-bidirectional}, we can construct a feasible solution to the bidirectional adversary bound for $P$ from $T'$. Now, we consider the following state-conversion problem:
        \[P' := \{(\mathcal{N}_x\ket{\pi_x},\mathcal{N}_x\ket{\pi_x}, O_x), (\mathcal{N}_x\ket{\overline{\pi}}_x, -\mathcal{N}_x\ket{\overline{\pi}}_x, O_x)\}_{x \in \D}.\]
        A straightforward computation shows 
        that $P$ and $P'$ have the same constraints in their adversary bounds. Thus, their feasible regions are identical, and so the feasible solution that we constructed for the bidirectional adversary bound for $P$ is also a feasible solution to the bidirectional adversary bound for $P'$. We now use \Cref{thm:unidirectional-bidirectional,thm:adv-to-transducer} to turn this back into the bidirectional canonical transducer $U$ from the theorem statement.
    \end{proof}

    The transducer we constructed in the previous lemma acts on $\mathcal{N}\ket{\pi}$, rather than on $\ket{\pi}$ directly. This is typically fine if we just consider a single quantum walk, since $\mathcal{N}$ is merely an isometry, and for most intents and purposes one can simply ignore the second register storing the star state corresponding to the vertex in the first register. However, in the specific application we consider in this work, i.e., the volume estimation problem, we consider several different quantum walks, which consequently have different star states, preventing us from simply ignoring them.
    
    The solution we present here is to use a transducer that acts as $\mathcal{N}$, alongside the transducers that reflect through the star states. This is an additional assumption we make on the access model to our quantum walk, which might not always be easy to justify. However, we will be able to justify it for the volume estimation application later, in~\Cref{sec:ball-walks}.

    \begin{lemma}[Stationary state reflection]
        \label{lem:refl-stationary-electrical-network}
        Let $(G = (V,E),r)$ be an electrical network. For all $v \in V$, let $C_v$ and $U_v$ be bidirectional canonical transducers that construct and reflect through $\ket{*_v}$. Then, we can construct a bidirectional canonical transducer $V$ that reflects through $\ket{\pi}$ acting, for all net-flows $\overline{\pi} \in \C^V$, as
        \begin{align*}
            \ket{\pi} \overset{V}{\rightsquigarrow} \ket{\pi}, & \quad \text{with compl.} \quad O\left(\sum_{v \in V} \pi_v (W(C_v,\ket{0}) + W(U_v,\ket{*_v}))\right), \\
            \ket{\overline{\pi}} \overset{V}{\rightsquigarrow} -\ket{\overline{\pi}}, & \quad \text{with compl.} \quad O\left(\sum_{v \in V} \left(\frac{|\overline{\pi}_v|^2}{\pi_v} (W(C_v,\ket{0}) + W(U_v, \ket{*_v})) + W(U_v, (\bra{v} \otimes I_V)\Pi_{\mathcal{H}_1^{\perp}}\ket{f^{\overline{\pi}}})\right)\right).
        \end{align*}
    \end{lemma}

    \begin{proof}
        For all $v \in V$, we say that $C_v$ transduces $\ket{0}$ to $\ket{*_v}$, where $\ket{0}$ is an arbitrary reference state. We take a parallel composition of all the construction transducers $C_v$, to to obtain a bidirectional canonical transducer $C$ that acts for all $\ket{\psi} \in \C^V$ as
        \[\ket{\psi} \overset{C}{\rightsquigarrow} \mathcal{N}\ket{\psi}, \qquad \text{with complexity} \qquad \sum_{v \in V} |\braket{v}{\psi}|^2W(C_v, \ket{0}).\]
        
        Now, we let $U$ be the bidirectional canonical transducer constructed in \Cref{lem:transducer-reflection-stationary}, and we let $U$ be the sequential composition $V = C^{\dagger}UC$. Then, by \Cref{thm:sequential-composition}, we obtain that
        \begin{align*}
            \ket{\pi} \overset{U}{\rightsquigarrow} \ket{\pi}, & \qquad \text{with complexity} \qquad \sum_{v \in V} \pi_vW(C_v,\ket{0}) + W(U,\mathcal{N}\ket{\pi}), \\
            \ket{\overline{\pi}} \overset{U}{\rightsquigarrow} -\ket{\overline{\pi}}, & \qquad \text{with complexity} \qquad \sum_{v \in V} \frac{|\overline{\pi}_v|^2}{\pi_v}W(C_v, \ket{0}) + W(U, \mathcal{N}\ket{\overline{\pi}}).
        \end{align*}
        The result now follows by plugging in the complexity expressions from \Cref{lem:transducer-reflection-stationary}.
    \end{proof}

    Next, we can use the relation between the electrical network and irreducible, reversible random walks to rephrase the above result in terms of random walks only. Similarly as before, we assume that the state-transition matrix $P$ depends on the implicit input, but the state space $\Omega$ does not.

    \begin{theorem}
        \label{thm:refl-stationary}
        Let $P$ be an irreducible, reversible Markov chain over a finite state space $\Omega$ with stationary distribution $\pi$ and spectral gap $\gamma(P)$. We write for any net-flow $\overline{\pi} \in \C^{\Omega}$,
        \[\ket{\pi} := \sum_{\omega \in \Omega} \sqrt{\pi(\omega)}\ket{\omega}, \qquad \text{and} \qquad \ket{\overline{\pi}} := \sum_{\omega \in \Omega} \frac{\overline{\pi}(\omega)}{\sqrt{\pi(\omega)}}\ket{\omega}.\]
        For every $\omega \in \Omega$, let $\ket{*_{\omega}} \in \C^{\Omega}$ be the star state, and let $C_{\omega}$ and $U_{\omega}$ be bidirectional canonical transducers that construct and reflect through $\ket{*_{\omega}}$, respectively. Then, we can construct a bidirectional canonical transducer $U$ that reflects through $\ket{\pi}$ with complexities
        \begin{align*}
            \ket{\pi} \overset{U}{\rightsquigarrow} \ket{\pi}, & \quad \text{with complexity} \quad O\left(\sum_{\omega \in \Omega} \pi(\omega)(W(C_{\omega},\ket{0}) + W(U_{\omega}, \ket{*_{\omega}}))\right), \\
            \ket{\overline{\pi}} \overset{U}{\rightsquigarrow} -\ket{\overline{\pi}}, & \quad \text{with complexity} \quad O\left(\frac{M\norm{\ket{\overline{\pi}}}^2}{\gamma(P)} + \sum_{\omega \in \Omega} \frac{|\overline{\pi}(\omega)|^2}{\pi(\omega)} (W(C_{\omega},\ket{0}) + W(U_{\omega},\ket{*_{\omega}}))\right),
        \end{align*}
        where
        \begin{equation}
            \label{eq:M}
            M := \sup \{W(U_{\omega}, \ket{\perp}) : \omega \in \Omega, \ket{\perp} \in \C^{\Omega} : \norm{\ket{\perp}} = 1, \braket{\perp}{*_{\omega}} = 0\}.
        \end{equation}
    \end{theorem}

    \begin{proof}
        We start by observing from \Cref{thm:walk-electrical-network} that we can turn the walk into an electrical network $(G = (\Omega,E),r)$, with $R = 1$. Thus, we can use~\Cref{lem:refl-stationary-electrical-network} to construct a bidirectional canonical transducer $U$ that reflects through $\ket{\pi}$. The transduction complexity of $U$ acting on $\ket{\pi}$ follows immediately. For the complexity on $\ket{\overline{\pi}}$, it suffices to analyze that
        \begin{align*}
            \sum_{\omega \in \Omega} W(U_{\omega}, (\bra{\omega} \otimes I_{\Omega})\Pi_{\mathcal{H}_1^{\perp}}\ket{f^{\overline{\pi}}}) &\leq M\sum_{\omega \in \Omega} \norm{(\bra{\omega} \otimes I_{\Omega})\Pi_{\mathcal{H}_1^{\perp}}\ket{f^{\overline{\pi}}}}^2 \leq M\norm{\ket{f^{\overline{\pi}}}}^2 = 2MR_{\mathrm{eff}}(G,r;\overline{\pi}) \\
            &\leq \frac{2M\norm{\ket{\overline{\pi}}}^2}{\gamma(P)}.
        \end{align*}
        In the final steps, we used \Cref{def:effective-resistance,thm:effective-resistance-spectral-gap}.
    \end{proof}

    \subsection{Metropolis filters via bipartite walks} \label{sec:bipartite}

    Let $(\Omega,P)$ be a Markov process on a finite state space $\Omega$. Given a target distribution $\pi$ on $\Omega$, the Markov process can be modified so that the new process has stationary distribution $\pi$. The Metropolis-Hastings framework provides a principled way to do so \cite{metropolis53,hastings70}. At a state $x \in \Omega$, it uses $P$ to generate a proposal $y$, and then accepts the proposal with a probability $a(x,y)$. If the proposal is rejected, the chain remains at $x$. This leads to the Markov process $(\Omega,P')$ where 
    \[
    P'(x,y) = \begin{cases} P(x,y)a(x,y) & \text{ if } y \neq x, \\
    1-\sum_{y \neq x} P(x,y) a(x,y) & \text{ if } y=x. \end{cases}
    \]
    We refer to $a$ as the \textit{filter}. A quantum walk based on $P'$ uses reflections through the star states $\ket{*_x}$ defined as 
    \[
    \ket{*_x} = \sum_{y \in \Omega, y \neq x} \sqrt{P(x,y)a(x,y)}\ket{y} + \sqrt{1-\sum_{y \neq x} P(x,y) a(x,y)}\ket{x}.
    \]
    Preparing or reflecting through such a state is challenging due to the amplitude of $\ket{x}$ in $\ket{*_x}$. Informally, it requires us to be able to forget rejected proposals. 

    In the remainder of this section we discuss two canonical filters, the Metropolis and Glauber  filters. These filters behave similarly in terms of parameters such as the spectral gap or conductance of the Markov chain. We argue, however, that the Glauber filter is particularly suitable for quantization. Indeed, we show how to turn a Markov chain with a Glauber filter into a bipartite walk with similar spectral properties, easy to prepare star states, but which does not require a filter. In later sections, we use this construction to quantize the speedy walk over a convex body.

    We note that \cite{apers23ustcon} (based on \cite[arXiv v2]{kosowski2012}) constructed a quantum walk based on a similar bipartite construction, for unweighted graphs. They observed a connection to Glauber filters for the associated two-step process. Here we generalize this to arbitrary graphs, and we connect the spectral properties of the bipartite walk to that of the filtered walk.

    For general filters, \cite{Claudon26metropolis} construct a Markov process on the extended state space $\Omega^2$ with desirable spectral properties, and a quantum circuit that implements the associated Szegedy walk operator.

    \subsubsection{Basic properties of Metropolis and Glauber filters} \label{sec:basic properties metropolis-glauber}
    
    In the Metropolis-Hastings framework, two canonical choices are the Metropolis filter $a_M$ and Glauber filter~$a_G$:
    \[
    a_M(x,y) = \min\left\{1,\frac{\pi(y)P(y,x)}{\pi(x) P(x,y)}\right\}, \qquad a_G(x,y) = \frac{1}{1+\frac{\pi(x)P(x,y)}{\pi(y) P(y,x)}} = \frac{\pi(y) P(y,x)}{\pi(y) P(y,x)+\pi(x)P(x,y)}.
    \]
    Both filters are functions of the ratio $r(x,y) = \frac{\pi(y)P(y,x)}{\pi(x) P(x,y)}$: we have $a_M(x,y) = \min\{1,r(x,y)\}$ and $a_G(x,y) = \frac{1}{1+1/r(x,y)}$. More generally, if $g:\R_{>0}\to [0,1]$ is such that $g(r) = r \cdot g(1/r)$, and $a(x,y) = g(r(x,y))$, then $P'$ satisfies the detailed-balance equation with respect to $\pi$ (\cref{eq:detailed balance}) and therefore $\pi$ is the stationary distribution of $P'$. 
    For the Metropolis and Glauber filters, the corresponding functions are 
    \[
    g_M(r) = \min\{1,r\}, \qquad g_G(r) = \frac{r}{r+1}.
    \]
    These functions satisfy 
    \[
    g_G(r) \leq g_M(r) \leq 2 g_G(r). 
    \]
    As a result, if we let $P_M(x,y) = P(x,y) a_M(x,y)$ and $P_G(x,y) = P(x,y) = a_G(x,y)$, then $P_M(x,y)$ and $P_G(x,y)$ are within a constant factor of each other for all $y \neq x$. Combined with the fact that $P_M$ and $P_G$ have the same stationary distribution $\pi$, this implies via \cref{def:conductance-gap,lem:gap-dirichlet} that their conductances and spectral gaps are comparable as well:
    \[
    \varphi(P_G) \leq \varphi(P_M) \leq 2 \varphi(P_G), \qquad \gamma(P_G) \leq \gamma(P_M) \leq 2 \gamma(P_G).
    \]
    
    \subsubsection{Glauber filters via bipartite walks}

    We now focus on the Glauber filter $a_G$. It has the additional property that 
    \[
        a_G(x,y) = 1-a_G(y,x).
    \]
    Hence, a proposal to move from $x$ to $y$ leads to accepting $y$ with probability $a_G(x,y)$, and a move from $y$ to $x$ leads to rejecting $x$ (thus remaining at $y$) with the same probability. In other words, both the proposals ``$x$ to $y$'' and ``$y$ to $x$'' lead to the same distribution on the outcomes $x$ and $y$. Informally, we exploit this property to ``forget'' the proposal. 
    Formally, we consider the following bipartite walk:

    \begin{definition}[Bipartite Glauber walk] \label{def:bipartite-glauber}
        Let $(\Omega,P)$ be an irreducible Markov process, and $\pi$ a distribution whose support is $\Omega$. We then define
        \[
        E_P = \{\{x\}:P(x,x)>0\} \cup \{\{x,y\}: x,y \in \Omega, x \neq y, P(x,y)+P(y,x)>0\},
        \]
        and set $\Omega' = \Omega \sqcup E_P$. We define a bipartite Markov chain $P'$ on $\Omega'$ as follows:\footnote{In both cases, when $x=y$ we have $\{x,y\} = \{x\}$. A proposal to move from $\{x\}$ to $x$ results in $x$ with probability $1$.} \\
        On the one hand, at $x \in \Omega$, do:
        \begin{enumerate}[nosep]
            \item Select a point $y \in \Omega$ according to $P(x,y)$.
            \item Output $\{x,y\} \in E_P$.
        \end{enumerate}
        On the other hand, at $\{x,y\} \in E_P$, do: 
                \begin{enumerate}[nosep]
            \item Output $y$ with probability $a_G(x,y) = \frac{\pi(y) P(y,x)}{\pi(y) P(y,x)+\pi(x)P(x,y)}$, 
            \item Else output $x$.
            \end{enumerate}
    \end{definition}

We note that the star states associated to $P'$ take the following form for $x \in \Omega, \{x\},\{x,y\} \in E_P$:
\begin{align*}
    \ket{*_x} &= \sum_{y \in \Omega} \sqrt{P(x,y)} \ket{\{x,y\}},\\
    \ket{*_{\{x\}}} &= \ket{x},\\
    \ket{*_{\{x,y\}}} &= \sqrt{1-a_G(x,y)} \ket{x} + \sqrt{a_G(x,y)} \ket{y} = \sqrt{a_G(y,x)} \ket{x} + \sqrt{1-a_G(y,x)}\ket{y},
\end{align*}
where for $\ket{*_{\{x,y\}}}$ we emphasize the symmetry due to $a_G(x,y)+a_G(y,x)=1$.

We now relate the spectral properties of $P'$ to those of $P_G$.

    \begin{lemma}
        \label{lem:metropolis-to-glauber-filter} \label{lem:stationary-anti-bipartite} 
        Let $(\Omega,P)$ be an irreducible Markov process, $\pi$ a distribution whose support is $\Omega$, and let $(\Omega',P')$ be the bipartite Glauber walk associated to $P$ and $\pi$ via \cref{def:bipartite-glauber}.  Then, $(P')^2|_{\Omega} = P_G$. We view $\pi$ as a distribution on~$\Omega'$ and let $\pi'= (\pi + \pi P')/2$. Then $P'$ is reversible with respect to $\pi'$. 
        
        Moreover, if $P'$ is irreducible, its spectral gap is at least half of the spectral gap of $P_G$, $\pi'$ is the unique stationary distribution of $P'$, and $\overline{\pi}' = \frac{\pi - \pi P'}{2}$ is the unique $(-1)$-eigenvector of $P'$.
    \end{lemma}
    \begin{proof}
        Let $x,y \in \Omega$, with $x \neq y$. Then,
        \[(P')^2(x,y) = P'({x,\{x,y\}}) P'({\{x,y\},y}) = P(x,y) a_G(x,y) = P_G(x,y),
        \]
        and 
        \begin{align*}
            (P')^2(x,x) &= P'({x,\{x\}}) P'({\{x\},x}) +\sum_{y \neq x} P'({x,\{x,y\}}) P'({\{x,y\},x}) \\
            &= P(x,x) + \sum_{y \neq x} P(x,y) (1-a_G(x,y)) = P_G(x,x).
        \end{align*}
        Hence, the submatrix of $(P')^2$ indexed by rows and columns from $\Omega$ equals $P_G$. 

        We now view $\pi$ as a distribution on $\Omega'$ (by setting $\pi$ equal to zero on $E_P$), then we have
        \[\left(\frac{\pi \pm \pi P'}{2}\right) P' = \frac{\pi P' \pm \pi (P')^2}{2} = \frac{\pi P' \pm \pi P_G}{2} = \frac{\pi P' \pm \pi}{2} = \pm\frac{\pi \pm \pi P'}{2}.\]
        Hence, $\pi' \coloneqq \frac{\pi + \pi P'}{2}$ is a stationary distribution of $P'$, and $\frac{\pi - \pi P'}{2}$ is an eigenvector corresponding to eigenvalue~$-1$. We verify that $P'$ is reversible with respect to $\pi'$.  First, note that for all $x \neq y \in \Omega$, we have 
        \[
        \pi'(x) = \pi(x)/2, \qquad \pi'(\{x\}) = \pi(x)P(x,x)/2, \qquad \pi'(\{x,y\}) = \frac{\pi(x) P(x,y) + \pi(y)P(y,x)}{2}.
        \]
        We then verify detailed balance with respect to $\pi'$. Since $P'$ is a bipartite walk, it suffices to consider the following two combinations of states in $\Omega'$. First,  
        \[
        \pi'(x) P'(x,\{x\}) = \frac{\pi(x) P(x,x)}{2} = \pi'(\{x\}) P'(\{x\},x).
        \]
        Second, 
        \begin{align*}
        \pi'(x) P'(x,\{x,y\}) &= \frac{\pi(x) P(x,y)}{2} \\
        &=  \frac{\pi(x) P(x,y)+\pi(y)P(y,x)}{2} \frac{\pi(x) P(x,y)}{\pi(x) P(x,y)+\pi(y)P(y,x)} \\
        &= \pi'(\{x,y\}) P'(\{x,y\},x).
        \end{align*}
        Hence, $P'$ is reversible with respect to $\pi'$. Now, let $D$ be the diagonal matrix with the entries of the stationary distribution $\pi'$ of $P'$ on the diagonal. Since $P'$ is reversible with respect to $\pi'$, the matrix $D^{1/2} P' D^{-1/2}$ is symmetric and thus of the form $\begin{bmatrix}
            0 & B \\ B^T & 0 
        \end{bmatrix}$ for some matrix $B$. We obtain that $D^{1/2}(P')^2D^{-1/2}$ is block-diagonal with two blocks $BB^T$ and $B^T B$, which have the same non-zero eigenvalues. As the top-left block of $(P')^2$ is similar to $P_G$, so is $BB^T$, from which we obtain $\mathrm{spec}(P') = \{0\} \cup \{\pm \sqrt{\lambda}: \lambda \in \mathrm{spec}(P_G), \lambda>0\}$. Assuming now that $P'$ is irreducible, the eigenvalues $1$ and $-1$ have multiplicity $1$, and we find that the spectral gap of $P'$  satisfies
        \[\gamma(P') = 1 - \sqrt{1-\gamma(P_G)} \geq 1-\left(1-\frac{\gamma(P_G)}{2}\right) = \frac{\gamma(P_G)}{2}. \qedhere\]
    \end{proof}

    \newpage
    \part{Quantum volume estimation} \label{part:volume}

    We now turn to the problem of estimating the volume of a convex body. Our quantum algorithm will closely follow the structure of the prior classical state-of-the-art due to Cousins and Vempala~\cite{CV18}. We therefore first present their approach in more detail in \Cref{sec:CV-overview}. As we will see, their approach uses a random walk on a continuous state space. In contrast, the results in the first part of this work are most conveniently stated for random walks on a discrete and finite state space. The convenience stems from the fact that the quantum walks and transducers we constructed in the previous sections consequently act on a Hilbert space of finite dimension. This makes our results fit neatly within the existing transducer framework, which relies on this property for its structural results, see, e.g., \Cref{thm:transducer-catalyst}. In \Cref{sec:MC-measurable} we state the necessary preliminaries about random walks on measurable state spaces, and in \Cref{sec:discretization} we show how to discretize the speedy walk. We show that discretization is benign in the sense that the discretized speedy walk has a comparable spectral gap and similar stationary distribution (when interpreted as a function on on $\R^d$). We formally state and analyze our (discretized) speedy walk in \Cref{sec:ball-walks}. We develop the necessary transducers to quantize the speedy walk in \Cref{sec:transducer-ball-walk}. We then combine all ingredients into our quantum volume estimation algorithm in \Cref{sec:vol-est}.

    \section{Algorithm overview} \label{sec:CV-overview}

  We now give a more detailed overview of the randomized volume estimation algorithm for convex bodies developed by Cousins and Vempala~\cite{CV18}, and how we modify it in this work. We refer to \Cref{sec:intro-volume} for a high-level overview.

    \subsection{Partition function estimation and cooling schedule.}

    As mentioned before, the approach of Cousins and Vempala fits within the partition function estimation framework. Concretely, they use the following partition function and associated Gibbs distributions. 

    \begin{definition}[Partition functions and Gibbs distributions] \label{def:part-gibbs}
        Let $d \in \N$, and $K \subseteq \R^d$ be a convex body. For any inverse temperature $\beta \in (-\infty,\infty]$, the Gibbs distribution $\pi_{\beta}$ is a distribution over $K$ with density function proportional to $x \mapsto \exp(-\beta\norm{x}^2)$. The partition function $Z : (-\infty,\infty] \to \R_{\geq 0}$ computes the normalization factor of this distribution, i.e.,
        \[Z(\beta) := \int_K \exp(-\beta\norm{x}^2) \;\mathrm{d}x.\]
    \end{definition}

    The Gibbs distribution $\pi_{\beta}$ is thus a Gaussian distribution with covariance matrix $I/(2\beta)$ that is conditioned on being inside~$K$.
    This implies that at infinite-temperature (i.e., $\beta=0$) we have $Z(0) = \Vol(K)$, and so we can solve the volume estimation problem by estimating $Z(0)$. It turns out to be easier to estimate ratios of partition functions and thus instead of estimating $Z(0)$ directly, we will estimate $Z(0)/Z(\beta^*)$ up to relative precision $\epsilon$ for a value of $\beta^*$ for which $Z(\beta^*)$ is easy to compute. 
    Concretely,  whenever $B_d \subseteq K$ and $\beta^* \in \Omega(d\sqrt{\ln(1/\varepsilon)})$, then all but an $\varepsilon$-fraction of the mass of the Gaussian distribution is contained in $K$~\cite[Lemma~7.9]{CV18}, and so we know a good estimate of $Z(\beta^*)$.
    
    Given that estimate, it suffices to estimate the ratio between these two values $Z(0)/Z(\beta^*)$ up to relative precision $\varepsilon$.    
    To estimate this ratio, the typical idea is to define a decreasing sequence of inverse temperatures $(\beta_j)_{j=0}^{\ell}$, referred to as an \textit{inverse cooling schedule}, with $\beta_0 = \beta^*$ and $\beta_{\ell} = 0$. We then use this schedule to write the quantity $Z(0)/Z(\beta^*)$ as a finite telescoping product, as
    \begin{equation}
        \label{eq:telescope}
        \frac{Z(0)}{Z(\beta^*)} = \frac{Z(\beta_{\ell})}{Z(\beta_0)} = \prod_{j=0}^{\ell-1} \frac{Z(\beta_{j+1})}{Z(\beta_j)}.
    \end{equation}
    An elementary argument shows that every factor can 
    be written as an expectation value over the Gibbs distribution at a given inverse temperature. That is, for all $j \in \{0, \dots, \ell-1\}$,
    \[\underset{x \sim \pi_{\beta_j}}{\E} \left[X_{\beta_{j+1},\beta_j}(x)\right] = \frac{Z(\beta_{j+1})}{Z(\beta_j)}, \qquad \text{where} \qquad X_{\beta',\beta}(x) := \exp(-(\beta' - \beta)\norm{x}^2).\]

    We then use Monte Carlo sampling on these random variables to estimate each factor in the telescoping product in \Cref{eq:telescope} individually, up to some relative error $\varepsilon_j$ for the $j$th factor. For this to be efficient, it is important to establish a bound on the relative variance of these random variables.
    Deviating from~\cite{CV18}, in this work we choose the inverse cooling schedule $(\beta_j)_{j=0}^{\ell}$ such that the relative variance of each of these random variables is constant. That is, we set 
    \[\beta_{j+1} = \left(1 - \max\left\{\frac{1}{2\sqrt{d}}, \frac{1}{2R\sqrt{\beta_j}}\right\}\right)\beta_j, \qquad \text{where} \qquad \beta_0 = \beta^*,\]
    and we stop whenever $\beta$ becomes negative, say in the $\ell$th iteration, and then set $\beta_{\ell} = 0$ instead. We then prove (in \Cref{lem:relative-variance}) that for all $j \in \{0, \dots, \ell-1\}$, we have
    \[\frac{\underset{x \sim \pi_{\beta_j}}{\Var} \left[X_{\beta_{j+1},\beta_j}(x)\right]}{\underset{x \sim \pi_{\beta_j}}{\E} \left[X_{\beta_{j+1},\beta_j}(x)\right]^2} \in O(1).\]
    For technical reasons that we will elaborate on later, 
    \cite{CV18} chooses a longer inverse cooling schedule, i.e., with smaller gaps between consecutive $\beta$'s. Our schedule is significantly shorter: our length satisfies $\ell \in \widetilde{O}(\sqrt{d})$, as compared to $\widetilde{O}(d)$ in their work~\cite[Figure~2]{CV18}.

    We now have a high-level recipe for the randomized volume estimation algorithm developed in \cite{CV18}: for each of the factors in the inverse cooling schedule, gather $\Theta(1/\varepsilon_j^2)$ samples from the Gibbs distribution at the given inverse temperature, with precisions $\varepsilon_j$ to be fixed later. By Chebyshev's inequality, we then obtain an unbiased estimate of each of the factors with relative error $\varepsilon_j$, which becomes an estimate of the telescoping product with relative error $\varepsilon$ satisfying $\varepsilon^2 := \sum_{j=0}^{\ell-1} \varepsilon_j^2$.\footnote{Strictly speaking, this equality only holds whenever the right-hand side is at most a constant. However, our choice of $\varepsilon_j$'s will ensure this later on.} The total number of queries made by the entire randomized algorithm is then analyzed as
    \begin{equation}
        \label{eq:cost-formula}
        \sum_{j=0}^{\ell-1} \; \underbrace{\text{no.\ samples}}_{\in O(1/\varepsilon_j^2)} \; \times \; \text{queries per sample}.
    \end{equation}

    \subsection{Continuous space random walks: the ball walk and speedy walk.}
    At this point, the key challenge is to sample from the Gibbs distributions $\pi_\beta$ at the inverse temperature~$\beta$ prescribed by the inverse cooling schedule. To that end, we use a Markov chain $M^{\beta}$ on a continuous state space that converges to it. We will work with the following formulations of the ball and speedy walk on $K$. 

    \begin{definition}[{Ball walk}] \label{def:ball}
        Let $d \in \N$, and $K \subseteq \R^d$ be a convex body. Let $\beta \geq 0$ and $\delta > 0$. We write $f(x) = \exp(-\beta\norm{x}^2)\mathbbm{1}_{x \in K}$. At point $x \in K$, do:
        \begin{enumerate}[nosep]
            \item Select a point $y \in x + \delta B_d$ uniformly at random.
            \item Accept $y$ with probability ${q}_{xy} := \min\{1, f(y)/f(x)\}$. Else stay in $x$.
        \end{enumerate}
    \end{definition}
    The stationary distribution of the ball walk is the distribution $\pi_\beta$ on $K$ whose density is proportional to $f$. To see this, note that the proposal distribution in step 1 is symmetric in $x$ and $y$, and hence step 2 consists of the Metropolis filter for the stationary distribution $\pi_\beta$, see \cref{sec:basic properties metropolis-glauber}. Unfortunately, the ball walk does not mix rapidly: its conductance is poor due to its inability to escape corners of the convex body. 
    
    To overcome this issue, we consider a modification of the ball walk known as the \textit{speedy walk}. The speedy walk was first analyzed in \cite{kannan1997random} in the setting of a uniform target distribution on $K$. Cousins and Vempala extended the analysis to Gaussian target distributions restricted to $K$.  Importantly, for the speedy walk, one \textit{can} prove rapid mixing. Indeed, 
    its {conductance} satisfies $\Omega(\delta\sqrt{\beta/d})$~\cite[Lemma~6.6]{CV18}\footnote{Strictly speaking, this conductance bound only holds as long as $K \subseteq 4\sqrt{d/(2\beta)}B_d$. Thus, to make use of the result, Cousins and Vempala clip the convex body at this radius at every step of the inverse cooling schedule, and then argue that only an inverse exponential probability mass is outside of it. We do the same in our approach.}.

    \begin{definition}[{Speedy walk}] \label{def:speedy}
        Let $d \in \N$, and $K \subseteq \R^d$ be a convex body. Let $\beta \geq 0$ and $\delta > 0$. We write $f(x) = \exp(-\beta\norm{x}^2)$. At point $x \in K$, do:
        \begin{enumerate}[nosep]
            \item Select a point $y \in (x + \delta B_d) \cap K$ uniformly at random.
            \item Accept $y$ with probability $\overline{q}_{xy} := \min\{1, f(y)/f(x)\}$. Else stay in $x$.
        \end{enumerate} 
        Let $P_{\beta,\delta}$ be the Markov transition kernel and $\pi_{\beta,\delta}$ the stationary distribution of this walk.   \end{definition}
 
         Note that when $x$ is at least distance $\delta$ away from the boundary of $K$, the proposal distributions of the ball walk and speedy walk are identical.       
    However, for $x$ close to the boundary, the proposal distributions can differ significantly, which changes the stationary distribution. To characterize the stationary distribution of the speedy walk, we require the following notion of 
    \textit{local conductance}, cf.~\cite[Page~1243]{CV18}.
    
    \begin{definition}[Local conductance]
        \label{def:avg-local-conductance}
        Let $d \in \N$, $K \subseteq \R^d$ be a convex body and $\delta > 0$. Then, for all $x \in K$, we define the \textit{local conductance} as
        \[\ell_{\delta}(x) := \frac{\Vol((x + \delta B_d) \cap K)}{\Vol(\delta B_d)}.\]
    \end{definition}
    By \cite[Lemma 3.1]{CV18}, the stationary distribution $\pi_{\beta,\delta}$ of the speedy walk is supported on $K$ and has a density proportional to $\ell_{\delta,r}(x) f(x)$. One can show this by again interpreting the second step as a Metropolis filter (cf.~\cref{sec:basic properties metropolis-glauber}).

    \paragraph{Rejection sampling to recover the Gibbs distribution from the speedy walk stationary.} First note that the local conductance at $x \in K$ is $1$ whenever $x$ is at least $\delta$ away from the boundary of $K$, and it can only become very small in sharp corners. Cousins and Vempala introduce an elegant and efficient rejection sampling routine that samples ``outward'': 
    for any point $x \in K$ sampled from the speedy walk distribution, they accept the point $y := x/(1-1/(2d))$ with probability $f_{\beta}(y)/f_{\beta}(x)$.
    If $\delta$ is chosen sufficiently small, then the outcome of this rejection sampling routine is sufficiently close to the Gibbs distribution $\pi_\beta$. More precisely, \cite[Lemma~6.13]{CV18} shows that it suffices to choose $\delta \in \tTh(\min\{1,1/\sqrt{2\beta}\}/\sqrt{d})$.
    
    \subsection{Annealing to obtain a warm start} By combining the aforementioned conductance bound of the speedy walk with Cheeger's inequality~(\Cref{thm:conductance-spectral-gap}), the same choice of $\delta \in \tTh(\min\{1,1/\sqrt{2\beta}\}/\sqrt{d})$ 
    ensures that the spectral gap $\gamma$ of the speedy walk satisfies $\gamma \in \tOm(\min\{\beta,1\}/d^2)$.
    Following a result by Lov\'asz and Simonovits~\cite[Corollary~1.5]{lovasz1993random}, the actual mixing time scales as $O(\log(M)/\gamma)$.
    Here $M$ is the \emph{warmness} of the initial distribution, which is defined as the supremum of the pointwise ratio between the initial and target density.
    We call an initial distibution with bounded warmness $M \in O(1)$ a \emph{warm start}.
    Obtaining a sufficiently warm start for a Markov chain is often a challenging task. In our setting we use a 
         process known as \textit{annealing}, where we use the stationary distribution at inverse temperature $\beta_{j-1}$ as a warm start for the walk at inverse temperature $\beta_j$. 
        
    To ensure warmness of two consecutive speedy walk distributions, we require the inverse cooling schedule to be slower compared to what is needed to ensure constant relative variance, as in Cousins and Vempala. Quantitatively, we slowly change $\beta_{j-1}$ to $\beta_j$ in $C_j \in O(\min\{d,\beta_jR^2\}\ln(\beta_{j-1}/\beta_j) + d\ln(\delta_j/\delta_{j-1}))$ intermediate steps, to maintain constant warmness in each step.
    Thus, in the $j$th step of the inverse cooling schedule, we first spend $\widetilde{O}(C_j/\gamma_j)$ Markov chain steps to obtain the first sample from $\pi_{\beta_j}$, after which we spend $\widetilde{O}(1/\gamma_j)$ steps of the Markov chain per subsequent sample. Building on \Cref{eq:cost-formula}, we obtain that the total number of Markov chain steps in the entire algorithm to obtain an $\varepsilon$-error relative approximation to $Z(0)/Z(\beta^*)$ becomes (neglecting polylogarithmic factors)
    \begin{equation}
        \label{eq:randomized-cost}
        \sum_{j=0}^{\ell-1} \left(\min\{d,\beta_jR^2\}\ln\frac{\beta_j}{\beta_{j+1}} + d\ln\frac{\delta_{j+1}}{\delta_j} + \frac{1}{\varepsilon_j^2}\right) \cdot \frac{1}{\gamma_j}, \qquad \text{where} \qquad \varepsilon^2 = \sum_{j=0}^{\ell-1} \varepsilon_j^2.
    \end{equation}
    \noindent
    Optimizing over the choice of the precision parameters $\varepsilon_j$ yields $\varepsilon_j \propto \varepsilon/\gamma_j^{1/4}$. Plugging in all the quantities then leads to a total number of speedy walk steps that satisfies $\tO((d^3 + d^2R^2)/\varepsilon^2)$, matching the complexity obtained by Cousins and Vempala~\cite{CV18}.

    \paragraph{Classical amortization and Las Vegas algorithms.}
    It remains to analyze the actual cost, i.e., number of queries, required to implement a single step of the speedy walk. A naive computation shows that selecting a point from $(x + \delta B_d) \cap K$ costs roughly $1/\ell_{\delta}(x)$ queries, because that is the expected number of random samples from $x + \delta B_d$ one would have to take to hit a point in $K$. Clearly, for points $x \in K$ that lie close to corners of the convex body, this cost can be excessively high (up to exponential in $d$).
    
    To get around this, Cousins and Vempala first argue that \emph{on average} (w.r.t.~the distributions encountered in the algorithm) one will not find itself in one of these corners very often.
    In fact, the expected or \textit{amortized} cost of taking a step of the speedy walk throughout the algorithm is only constant~\cite[Lemmas~6.11 and 6.12]{CV18}.
    Since each step of the speedy walk is effectively a Las Vegas algorithm with constant expected cost, Cousins and Vempala then compose these to obtain a global bound of $\tO(d^3/\varepsilon^2)$ on the total expected cost.
    Finally, a simple Markov inequality allows them to truncate the resulting algorithm to a bounded-error Monte Carlo algorithm with worst-case cost $\tO(d^3/\varepsilon^2)$.

\paragraph{Amortized quantum walks.} 

    Quantumly, we use the transducer framework to perform amortization. Neglecting polylogarithmic factors, we obtain two quadratic speed-ups compared to the expression in \Cref{eq:randomized-cost}, from which we obtain that the quantum query complexity is upper bounded (neglecting polylogarithmic factors) by
\begin{equation}
\sum_{j=0}^{\ell-1} \left(d\ln\frac{\delta_{j+1}}{\delta_j} + \frac{1}{\varepsilon_j}\right) \cdot \frac{1}{\sqrt{\gamma_j}}, \qquad \text{where} \qquad \varepsilon^2 = \sum_{j=0}^{\ell-1} \varepsilon_j^2.
\end{equation}
By setting $\varepsilon_j \in \widetilde{\Theta}(\varepsilon/d^{1/4})$, the total precision indeed becomes $\varepsilon$, and the resulting query complexity can be analyzed to be $\widetilde{O}(d^2 + (d^{7/4} + d^{5/4}R)/\varepsilon)$, obtaining the complexity from \Cref{thm:result-vol-est-algo}.

From the final expression, we can see why it is important that we shortened the inverse cooling schedule by speeding up the inverse cooling. Indeed, we only obtain quadratic speed-ups for the individual steps of the cooling schedule, and not over the number of steps. Thus, by making the steps bigger and doing more work per step, we allow for a bigger quantum speed-up in terms of the number of queries.

    \section{Preliminaries: Markov chains on measurable spaces} \label{sec:MC-measurable}

    We start by formally introducing Markov chains on measurable spaces. We follow the exposition by Meyn and Tweedie~\cite{meyn2012markov}.

    \begin{definition}[see, e.g., {\cite[Section~3.4]{meyn2012markov}}]
        Let $(\Omega,\Sigma)$ be a measurable space, where $\Sigma$ is a countably-generated $\sigma$-algebra, that is, a $\sigma$-algebra generated by a countable set family. A \textit{Markov kernel} is a function $P : \Omega \times \Sigma \to [0,1]$, such that:
        \begin{enumerate}[nosep]
            \item For any $A \in \Sigma$, the map $x \mapsto P(x,A)$ is $\Sigma$-measurable.
            \item For every $x \in \Omega$, $A \mapsto P(x,A)$ is a probability measure on $\Sigma$.
        \end{enumerate}
        For any $t \in \N$, we inductively define the kernel $P^t$ by
        \[
        P^t(x,A)
        = \int_\Omega P(x,\mathrm{d}y) P^{t-1}(y,A),
        \]
        with $P^0(x,A) = \delta_x(A)$ (which equals $1$ if $x \in A$ and $0$ otherwise).
        Finally, for a probability measure $q : \Sigma \to [0,1]$ we define
        \[
        (qP)(A) = \int_{\Omega} P(x,A) \;\mathrm{d}q(x).
        \]
    \end{definition}

    The requirement that $\Sigma$ is countably generated is important for various results in \cite{meyn2012markov}, including for the irreducibility results we make use of in this work. This means that we cannot apply our results in general to Lebesgue-measurable sets, since its $\sigma$-algebra does not satisfy this condition. Instead, we will merely make use of Borel $\sigma$-algebras, which suffices for our purposes.

    \begin{definition}[irreducibility, see, e.g., {\cite[Section~4.2]{meyn2012markov}}]
        Let $(\Omega,\Sigma,\mu)$ be a measure space with $\mu \neq 0$, where $\Sigma$ is countably-generated, and let $P$ be a Markov kernel. We say that $P$ is $\mu$-irreducible if for any $x \in \Omega$ and $A \in \Sigma$ with $\mu(A) > 0$, there is a $t \in \N$ such that $P^t(x,A) > 0$.
    \end{definition}

    Contrary to discrete Markov chains with a finite state space, $\mu$-irreducible random walks do not necessarily have a stationary measure (compare to \Cref{thm:irreducible-implies-unique-stationary}). However, if a stationary measure exists, then irreducibility does guarantee uniqueness. 

    \begin{theorem}[see, e.g., {\cite[Proposition~10.1.1 and Theorem~10.4.9]{meyn2012markov}}]
        Let $(\Omega,\Sigma,\mu)$ be a measure space where $\Sigma$ is countably-generated, and let $P$ be a $\mu$-irreducible Markov kernel. If $\pi$ is a $\Sigma$-measurable stationary probability measure, i.e., $\pi P = \pi$, then it is unique and $\mu \ll \pi$.
    \end{theorem}

    Next, we follow the exposition by Rudolf~\cite{rudolf2011explicit} to generalize the concept of reversibility to the measure-space setting (compare with \Cref{def:reversible-random-walk}).

    \begin{definition}[see, e.g., {\cite[Definition~3.4]{rudolf2011explicit}}]
        Let $(\Omega,\Sigma)$ be a measurable space where $\Sigma$ is countably-generated, and $P$ a Markov kernel. Let $\pi : \Sigma \to [0,1]$ be a probability measure. We say that $P$ is reversible with respect to $\pi$ if for all $A,B \in \Sigma$,
        \[\int_A P(x,B) \;\mathrm{d}\pi(x) = \int_B P(x,A) \;\mathrm{d}\pi(x).\]
    \end{definition}

    As in the discrete setting, reversibility with respect to a probability measure means that the measure is stationary.

    \begin{corollary}[see, e.g., {\cite[Section~3]{rudolf2011explicit}}]
        Let $(\Omega,\Sigma)$ be a measurable space where $\Sigma$ is countably-generated, and $P$ is a reversible Markov kernel w.r.t.\ a probability measure $\pi : \Sigma \to [0,1]$. Then, $\pi$ is a stationary measure for $P$.
    \end{corollary}

    \begin{proof}
        For any $A \in \Sigma$, we have
        $(\pi P)(A) = \int_{\Omega} P(x,A) \;\mathrm{d}\pi(x) = \int_A P(x,\Omega) \mathrm{d}\pi(x) = \int_A \mathrm{d}\pi(x) = \pi(A)$. 
    \end{proof}

    From here, one could proceed with a more careful build-up of the properties of random walks on measure spaces, akin to~\Cref{subsec:prelims-markov-chains}. However, for our purposes, we only require the \textit{conductance} of these types of Markov chains. As such, we forgo setting up the rest of this theory here, and refer the interested reader to \cite[Section~3]{rudolf2011explicit}.

      \begin{definition}[{See, e.g., \cite[Section~3.5]{rudolf2011explicit}}]
        Let $(\Omega,\Sigma,\mu)$ be a measure space where $\Sigma$ is countably-generated, and $P$ a Markov kernel that is $\mu$-irreducible and reversible w.r.t.\ the probability measure $\pi : \Sigma \to [0,1]$. Then, the \textit{conductance} of $P$ is defined to be\footnote{We use a slightly different denominator than \cite{rudolf2011explicit}, to match the definition we use in the discrete setting, i.e., \Cref{def:conductance-gap}. Consequently, our definition only matches Rudolf's definition up to a multiplicative factor of at most $2$.}
        \[\varphi(P) = \inf_{\substack{A \in \Sigma \\ 0 < \pi(A) < 1}} \frac{\int_A P(x,\Omega \setminus A) \;\mathrm{d}\pi(x)}{\min\{\pi(A), \pi(\Omega \setminus A)\}}.\]
    \end{definition}

    \subsection{Induced Markov chain}

    Now that we have defined Markov chains and random walks over measurable spaces, we consider their discretizations.
    One useful approach that we will use is the Markov kernel induced on a finite partition of the state space.

    \begin{definition}[{Induced Markov kernel, see, e.g.~\cite[Section~4.6.1]{aldous2002reversible}}]
        \label{def:canonical-discretization}
        Let $(\Omega, \Sigma, \mu)$ be a measure space where $\Sigma$ is countably-generated, and $P$ be a reversible, $\mu$-irreducible Markov kernel with stationary probability measure $\pi$. Let $\Omega = \Omega_1 \sqcup \cdots \sqcup \Omega_n$ be a $\Sigma$-measurable finite partition of $\Omega$. Let $\Omega' = \{j \in [n] : \pi(\Omega_j) > 0\}$. For all $i,j \in \Omega'$, we define
        \[P'_{ij} := \frac{1}{\pi(\Omega_i)}\int_{\Omega_i} P(x,\Omega_j) \;\mathrm{d}\pi(x).\]
        Then, $(\Omega',P')$ is the induced discretization of $P$ w.r.t.\ $(\Omega_j)_{j=1}^n$.
    \end{definition}

    The benefit of this way of discretizing a reversible, irreducible Markov kernel is that lower bounds on the spectral gap and conductance carry over from the continuous to the discrete setting. In this work, we use this property of the conductance, for which we provide a theorem statement below.
    It is essentially a variation on \cite[Proposition 4.44]{aldous2002reversible}.

    \begin{theorem}
        \label{thm:discretization-conductance}
        Let $(\Omega, \Sigma, \mu)$ be a measure space where $\Sigma$ is countably-generated, $P$ be a reversible, $\mu$-irreducible Markov kernel with stationary probability measure $\pi$. Let $(\Omega',P')$ be the induced discretization w.r.t.\ $(\Omega_j)_{j=1}^n$, where for every $j \in [n]$, $\mu(\Omega_j) > 0$. Then, $P'$ is irreducible, $P'$ is reversible w.r.t.\ $\pi'$, where for all $i \in \Omega'$, $\pi_i' = \pi(\Omega_i)$, and $\varphi(P') \geq \varphi(P)$.
    \end{theorem}

    \begin{proof}
        We start by proving irreducibility by contradiction. Suppose that $P'$ is not irreducible. Then, there exist $i,j \in \Omega'$ such that for all $t \in \N$, $(P')^t_{ij} = 0$. Thus, for all $t \in \N$, we have
        \[\int_{\Omega_i} P^t(x,\Omega_j) \;\mathrm{d}\pi(x) = 0,\]
        and so for all $t \in \N$, the set $S_t = \{x \in \Omega_i : P^t(x,\Omega_j) > 0\}$ satisfies $\pi(S_t) = 0$. This implies that $S := \Omega_i \setminus \cup_{t=1}^{\infty} S_t$ satisfies $\pi(S) > 0$ and so $S$ is non-empty. For $x \in S$, we have $P^t(x,\Omega_j) = 0$ for all $t \in \N$. Since $\mu(\Omega_j) > 0$ by assumption, this contradicts the $\mu$-irreducibility of $P$, and therefore $P'$ must be irreducible.
        
        Next, we prove that $P'$ is reversible w.r.t.\ $\pi'$. Indeed, using the reversibility of $P$, we find for all $i,j \in \Omega'$,
        \[\pi'_iP'_{ij} = \pi(\Omega_i) \cdot \frac{1}{\pi(\Omega_i)} \int_{\Omega_i} P(x,\Omega_j) \;\mathrm{d}\pi(x) = \pi(\Omega_j) \cdot \frac{1}{\pi(\Omega_j)}\int_{\Omega_j} P(x,\Omega_i) \;\mathrm{d}\pi(x) = \pi'_jP'_{ji}.\]
        
        Finally, we relate prove the relation between their conductances. Let $J \subseteq \Omega'$, and $A = \sqcup_{j \in J} \Omega_j$. Then, we obtain that
        \[\frac{\sum_{j \in J, k \in \Omega' \setminus J} \pi'_jP'_{jk}}{\min\{\pi'(J), \pi'(\Omega' \setminus J)\}} = \frac{\int_A P(x,\Omega \setminus A) \;\mathrm{d}\pi(x)}{\min\{\pi(A), \pi(\Omega \setminus A)\}} \geq \varphi(P).\]
        Since this holds for any $J \subseteq \Omega'$, we can take the minimum on the left-hand side and conclude that $\varphi(P') \geq \varphi(P)$.
    \end{proof}

    \subsection{Perturbations of random walk transition probabilities}

    Even though the induced discretization provides a direct link between Markov processes on measurable spaces and their discretizations on finite state spaces, they are typically not very practical to implement directly because the transition probabilities can be hard to calculate. Thus, we resort to approximating them instead, which perturbs the chain, and so we require that key Markov chain properties are robust against these perturbations. We start by characterizing the effect on the stationary distribution.

    \begin{theorem}
        \label{thm:robustness-diameter}
        Suppose that $(\Omega,P)$ and $(\Omega,P')$ are reversible, irreducible random walks on a finite state space with stationary distributions $\pi$, $\pi'$. For all $x,y \in \Omega$, suppose that $P_{xy} \neq 0 \Leftrightarrow P'_{xy} \neq 0$, and let $D$ be the diameter of the underlying graph defining the random walk. Let $0 < \eta < 1/(32D)$, and suppose that for all $x,y \in \Omega$ where $P_{xy} \neq 0$ and $x \neq y$, we have
        \[\left|\frac{P_{xy}'}{P_{xy}} - 1\right| \leq \eta.\]
        Then, for all $x \in \Omega$, we have
        \[\norm{P_{x\cdot} - P'_{x\cdot}}_{\mathrm{TV}} \leq \eta, \qquad \text{and} \qquad \left|\frac{\pi_x'}{\pi_x} - 1\right| < 3D\eta.\]
    \end{theorem}

    \begin{proof}
        First, let $x \in \Omega$. Then,
        \begin{align*}
            \norm{P_{x\cdot} - P'_{x\cdot}}_{\mathrm{TV}} &= \frac12\sum_{y \in \Omega} |P_{xy} - P'_{xy}| = \frac12 \sum_{\substack{y \in \Omega \\ y \neq x}} |P_{xy} - P_{xy}'| + \frac12\left|1 - \sum_{\substack{y \in \Omega \\ y \neq x}} P_{xy} - 1 + \sum_{\substack{y \in \Omega \\ y \neq x}} P_{xy}'\right| \\
            &= \sum_{\substack{y \in \Omega \\ y \neq x}} |P_{xy} - P'_{xy}| = \sum_{\substack{y \in \Omega \\ P_{xy} > 0}} P_{xy} \left|\frac{P'_{xy}}{P_{xy}} - 1\right| \leq \sum_{\substack{y \in \Omega \\ P_{xy} > 0}} P_{xy} \cdot \eta = \eta.
        \end{align*}
        
        Next, let $x,y \in \Omega$ with $x \neq y$ such that $P_{xy} > 0$ and $P'_{xy} > 0$. Then,
        \[\frac{\pi'_x}{\pi_x} \cdot \frac{\pi_y}{\pi'_y} = \frac{\pi_y}{\pi_x} \cdot \frac{\pi'_x}{\pi'_y} = \frac{P_{xy}}{P_{yx}} \cdot \frac{P'_{yx}}{P'_{xy}} = \frac{P_{xy}}{P'_{xy}} \cdot \frac{P'_{yx}}{P_{yx}} \in \left[\frac{1-\eta}{1+\eta}, \frac{1+\eta}{1-\eta}\right].\]
        The above statement also trivially holds whenever $x = y$. Hence, for any pair $x,y \in \Omega$, let $x = x_0, \dots, x_\ell = y$ be the shortest path from $x$ to $y$ such that for any $j \in [\ell]$, $P_{x_{j-1}x_j} > 0$. We find that $\ell \leq D$, and so
        \[\frac{\pi'_x}{\pi_x} \cdot \frac{\pi_y}{\pi'_y} = \prod_{j=1}^{\ell} \left(\frac{\pi'_{x_{j-1}}}{\pi_{x_{j-1}}} \cdot \frac{\pi_{x_j}}{\pi'_{x_j}}\right) \in \left[\left(\frac{1-\eta}{1+\eta}\right)^{\ell}, \left(\frac{1+\eta}{1-\eta}\right)^{\ell}\right] \subseteq \left[\left(\frac{1-\eta}{1+\eta}\right)^D, \left(\frac{1+\eta}{1-\eta}\right)^D\right].\]
        We use the inequalities $((1+\eta)/(1-\eta))^D < 1+3D\eta$ and $((1-\eta)/(1+\eta))^D > 1-3D\eta$, and so for all $x,y \in \Omega$, we have
        \[\left|\frac{\pi'_x}{\pi_x} \cdot \frac{\pi_y}{\pi'_y} - 1\right| < 3D\eta.\]
        Thus, for all $x \in \Omega$,
        \[\left|\frac{\pi'_x}{\pi_x} - 1\right| = \left|\sum_{y \in \Omega} \pi_y' \cdot \frac{\pi'_x}{\pi_x} \cdot \frac{\pi_y}{\pi_y'} - 1\right| \leq \sum_{y \in \Omega} \pi_y \left|\frac{\pi'_x}{\pi_x} \cdot \frac{\pi_y}{\pi_y'} - 1\right| < 3D\eta.\qedhere\]
    \end{proof}

    Next, we analyze how the conductance behaves under perturbations of the transition probabilities. We do this in general, so we can apply it both to the continuous and discrete settings.

    \begin{theorem}
        \label{thm:conductance-robustness}
        Let $(\Omega,\Sigma,\mu)$ be a measure space, where $\Sigma$ is countably generated, and let $(\Omega,P)$ and $(\Omega,P')$ be random walks with stationary measures $\pi$ and $\pi'$, such that $\pi' \ll \pi$. Suppose that for every $x \in \Omega$, we have
        \[\norm{P(x,\cdot) - P'(x,\cdot)}_{\mathrm{TV}} \leq \eta, \qquad \text{and} \qquad \left|\frac{\mathrm{d}\pi'}{\mathrm{d}\pi} - 1\right| \leq \xi < \frac14 \text{ a.s.\ on } \mathrm{supp}(\pi').\]
        Then, $\varphi(P') \geq \varphi(P)(1-\xi)/(2(1+\xi)) - \eta > 3\varphi(P)/10 - \eta$.
    \end{theorem}

    \begin{proof}
        Let $A \in \Sigma$ such that $0 < \pi'(A) \leq 1/2$. Then, we use the Radon-Nikodym theorem to conclude that
        \[|\pi'(A) - \pi(A)| = \left|\int_A \;\mathrm{d}\pi' - \int_A \;\mathrm{d}\pi\right| = \left|\int_A \left(\frac{\mathrm{d}\pi'}{\mathrm{d}\pi}(x) - 1\right) \;\mathrm{d}\pi(x)\right| \leq \int_A \left|\frac{\mathrm{d}\pi'}{\mathrm{d}\pi} - 1\right| \;\mathrm{d}\pi \leq \xi\pi(A).\]
        Consequently, $(1-\xi)\pi(A) \leq \pi'(A) \leq (1+\xi)\pi(A)$, and
        \[\pi(\Omega \setminus A) = 1 - \pi(A) \geq 1 - \frac{\pi'(A)}{1-\xi} \geq 1 - \frac{1}{2(1-\xi)} = \frac{1 - 2\xi}{2(1-\xi)} > \frac{\pi'(A)}{2(1 - \xi)} > \frac{\pi'(A)}{2(1+\xi)}.\]
        Thus, $\pi'(A) \leq 2(1+\xi)\min\{\pi(A), \pi(\Omega \setminus A)\}$, and so we observe that
        \begin{align*}
            &\frac{1}{\pi'(A)} \int_A P'(x,\Omega \setminus A) \;\mathrm{d}\pi'(x) \geq \frac{1}{\pi'(A)} \int_A (P(x,\Omega \setminus A) - \eta) \;\mathrm{d}\pi'(x) = \frac{1}{\pi'(A)} \int_A P(x,\Omega \setminus A) \;\mathrm{d}\pi'(x) - \eta \\
            &\qquad \geq \frac{1}{2(1+\xi)\min\{\pi(A), \pi(\Omega \setminus A)\}} \int_A p(x,\Omega \setminus A) \frac{\mathrm{d}\pi'}{\mathrm{d}\pi}(x) \;\mathrm{d}\pi(x) - \eta \geq \frac{\varphi(P)(1-\xi)}{2(1+\xi)} - \eta.
        \end{align*}
        As this holds for any $A \in \Sigma$ for which $0 < \pi'(A) < 1/2$, it also holds for $\Omega \setminus A \in \Sigma$ by reversibility. Thus, taking the minimum over $A$ on the left-hand side yields the desired inequality.
    \end{proof}

    \section{Discretization of random walks over $\R^d$} \label{sec:discretization}

    When a Markov process is defined on a \textit{metric} space, rather than a measurable space, one can use smoothness properties of the Markov kernel's measures to implement more practical discretizations of the Markov process with bounded perturbations of the transition probabilities. In this work, it suffices to consider only the Euclidean space $\R^d$, but what we do here could be extended to Riemannian manifolds as well.

    \subsection{Smoothness conditions on sets and functions}
    
    We focus on random walks over sets $\Omega \subseteq \R^d$ of a particular type, namely that are Jordan-measurable.

    \begin{definition}[Jordan-measurable sets~{\cite[Section~11.2.3, Definition~5]{zorich2016mathematical}}]
        A set $\Omega \subseteq \R^d$ is Jordan-measurable if it is bounded and its boundary is a measure-zero set.
    \end{definition}

    We also consider functions that satisfy a specific smoothness condition, which we refer to as piecewise Lipschitz-continuity. Slight variations of this notion exist, the closest matching definition seems to be \cite[Assumption~2.2]{mikkelsen2018degrees}. See also \cite[Definition~3]{leobacher2022exception} for a related notion and its various versions studied in the literature.

    \begin{definition}[Piecewise Lipschitz-contiuous functions]
        Let $\Omega \subseteq \R^d$ be Jordan-measurable, and $f : \Omega \to \C$. Let $n \in \N$, and $\Omega = \Omega_1 \sqcup \cdots \sqcup \Omega_n$ be a finite partition, where for every $j \in [n]$, $\Omega_j$ is also Jordan-measurable. Suppose that $f|_{\Omega_j}$ is Lipschitz-continuous, that is, there is an $L > 0$ such that for $x,y \in \Omega_j$, we have $|f(x) - f(y)| \leq L\norm{x-y}$. Moreover, suppose that there is a continuous function $h : \R_{>0} \to \R_{>0}$ with $h(\zeta) \downarrow 0$ as $\zeta \downarrow 0$, such that for all $\zeta > 0$,
        \[\sum_{j=1}^n \Vol(\{x \in \R^d : \mathrm{dist}(x,\partial \Omega_j) < \zeta\}) < h(\zeta).\]
        Then, we say that $f$ is piecewise $L$-Lipschitz-continuous on an $h$-smooth partition $(\Omega_j)_{j=1}^n$.
    \end{definition}

    We prove several useful properties of piecewise Lipschitz-continuous functions.

    \begin{lemma}
        \label{lem:lipschitz-properties}
        Let $\Omega \subseteq \R^d$ be Jordan-measurable, and let $f,g : \Omega \to \C$ be piecewise Lipschitz-continuous. Then, $f+g$ and $fg$ are also piecewise Lipschitz continuous. Moreover, if $f$ is globally bounded, i.e., $f : \Omega \to [a,b]$ for $0 < a < b$, then $1/f$ is also piecewise Lipschitz continuous.
    \end{lemma}

    \begin{proof}
        We observe that both $f$ and $g$ are Lipschitz on the intersections of the elements of their partitions, of which there are finitely many. Thus, the statement follow directly from applying the properties of Lipschitz functions to these restrictions.
    \end{proof}

    Next, we introduce some notation for the rounding of vectors, sets, functions and densities in $\R^d$ to a grid with a given spacing $\zeta > 0$.

    \begin{definition}[Rounding]
        \label{def:grid-rounding}
        Let $\zeta > 0$. For all $x \in \R^d$, we write
        \[x_{\zeta} = \zeta \cdot \mathrm{round}(x/\zeta).\]
        Here, $\mathrm{round}$ rounds each of the entries of a $d$-dimensional vector individually to the nearest integer. Conversely, we write
        \[M_{y,\zeta} = \{x \in \R^d : x_\zeta = y_\zeta\}.\]
        For a set $\Omega \subseteq \R^d$, we define
        \[\Omega_\zeta = \zeta \Z^d \cap \Omega, \qquad \text{and} \qquad M_{\Omega,\zeta} = \bigcup_{x \in \Omega_\zeta} M_{x,\zeta}.\]
        For a function $f : \Omega \to X$, we define $f_\zeta : \Omega_{\zeta} \to X$ as $f_\zeta(x) = f(x)$. If $0 \leq f \in L^1(\Omega)$ is continuous, then we denote by $\pi_f$ the probability distribution on $\Omega$ with density $f$, i.e., $\pi_f(x) = f(x)/N_f$, where $N_f = \int_{\Omega} f(x)\;\mathrm{d}x$. Similarly, we let $\pi_{f,\zeta}$ be the normalized distribution on $\Omega_{\zeta}$ of $f_{\zeta}$, i.e., $(\pi_{f,\zeta})_x = f(x)/N_{f,\zeta}$ where $N_{f,\zeta} = \sum_{x \in \Omega_{\zeta}} f(x)$.
    \end{definition}

    \begin{figure}[htb]
    \centering
    \includegraphics[width=.5\textwidth]{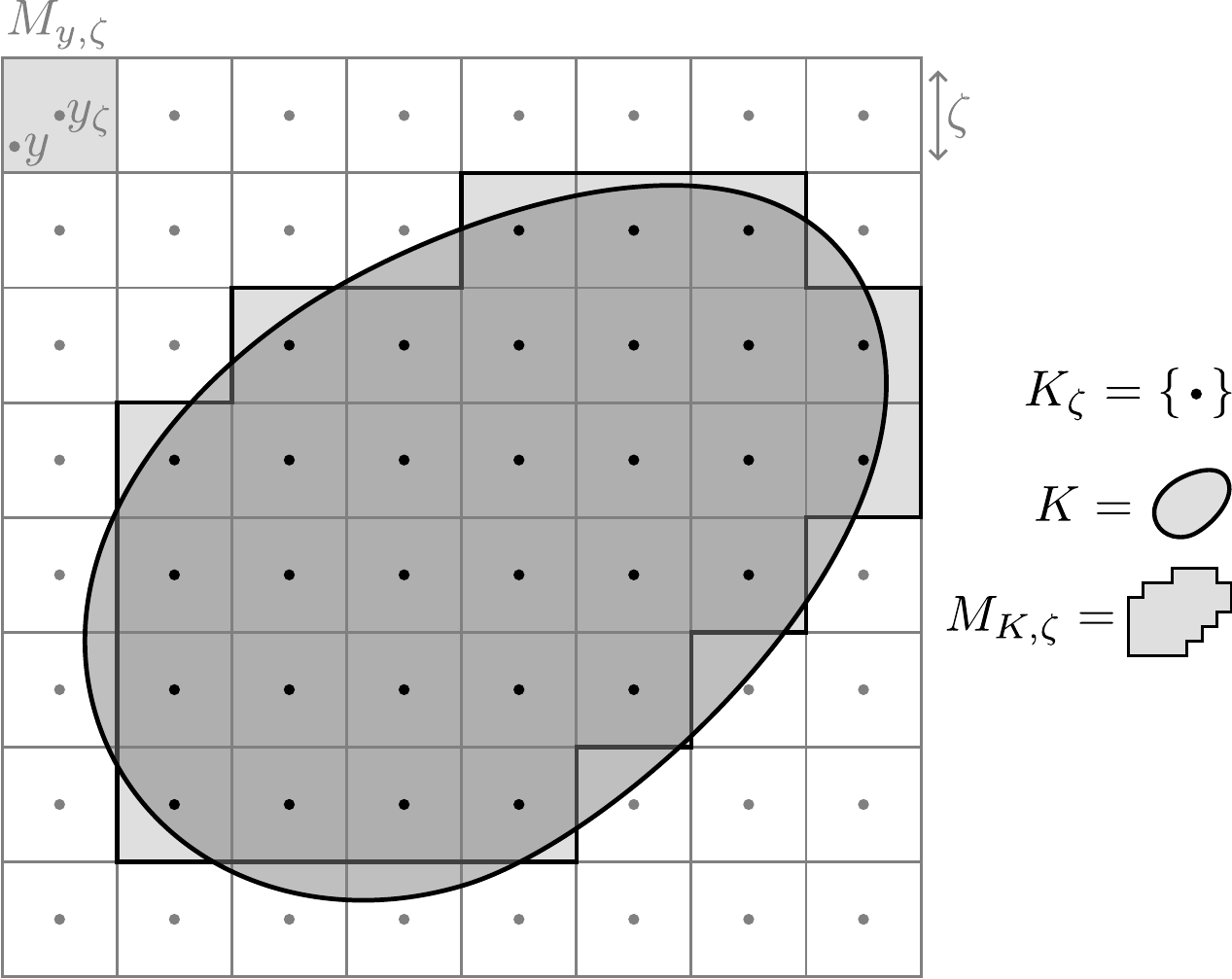}
    \caption{Illustration of the concepts in \Cref{def:grid-rounding}.}
    \label{fig:discretization}
    \end{figure}

    Now, we prove a useful approximation of integrals using finite summations of functions that are piecewise Lipschitz-continuous.

    \begin{lemma}[Approximation of integrals of piecewise Lipschitz continuous functions]
        \label{lem:smoothness}
        Let $\Omega \subseteq \R^d$, $L > 0$ and $f:\Omega \to \C$ piecewise $L$-Lipschitz continuous on an $h$-smooth partition $(\Omega_j)_{j=1}^n$. 
        Then, 
        \begin{align*}
            \left|\int_{\Omega} f(x) \;\mathrm{d}x - \zeta^d \sum_{x \in \Omega_\zeta} f(x)\right| &\leq \int_{\Omega \cup M_{\Omega,\zeta}} |f(x)\mathbbm{1}_{x \in \Omega} - f_\zeta(x)\mathbbm{1}_{x \in M_{\Omega,\zeta}}| \;\mathrm{d}x \\
            &\leq 2h(\sqrt{d}\zeta)\|f\|_\infty + L\sqrt{d}\Vol(\Omega) \zeta \downarrow 0, \qquad (\zeta \downarrow 0).
        \end{align*}
    \end{lemma}

    \begin{proof}
        We have
        \[\left|\int_{\Omega} f(x) \;\mathrm{d}x - \zeta^d \sum_{x \in \Omega_\zeta} f(x)\right| = \left|\int_{\Omega} f(x) \;\mathrm{d}x - \int_{M_{\Omega,\zeta}} f_\zeta(x) \;\mathrm{d}x\right| \leq \int_{\Omega \cup M_{\Omega,\zeta}} |f(x)\mathbbm{1}_{x \in \Omega} - f_\zeta(x)\mathbbm{1}_{x \in M_{\Omega,\zeta}}| \;\mathrm{d}x.\]
        This establishes the first inequality. Next, we write
        \begin{align*}
            &\int_{\Omega \cup M_{\Omega,\zeta}} |f(x)\mathbbm{1}_{x \in \Omega} - f_\zeta(x)\mathbbm{1}_{x \in M_{\Omega,\zeta}}| \;\mathrm{d}x \leq \sum_{j=1}^n \int_{\Omega_j \cup M_{\Omega_j,\zeta}} |f(x)\mathbbm{1}_{x \in \Omega} - f_\zeta(x)\mathbbm{1}_{x \in M_{\Omega,\zeta}}| \;\mathrm{d}x \\
            &\qquad \leq \sum_{j=1}^n \left(\int_{\Omega_j \cap M_{\Omega_j,\zeta}} |f(x) - f_\zeta(x)| \;\mathrm{d}x + 2\norm{f}_{\infty}\Vol(\Omega_j \; \Delta \; M_{\Omega_j,\zeta})\right).
        \end{align*}
        First observe that, by construction, every $x \in \Omega_j \; \Delta \; M_{\Omega_j,\zeta}$ satisfies $\mathrm{dist}(x,\partial\Omega_j) \leq \sqrt{d} \zeta$. We now bound the two terms separately. The previous observation together with the assumption on the sets $\Omega_j$ implies $\sum_{j=1}^n \Vol(\Omega_j \; \Delta \; M_{\Omega_j,\zeta}) \leq h(\sqrt{d}\zeta)$, which bounds the second term. For the first term, we use the Lipschitz property to conclude that for each $j \in [n]$,
        \begin{align*}
            \int_{\Omega_j \cap M_{\Omega_j,\zeta}} |f(x) - f_\zeta(x)| \;\mathrm{d}x &\leq L \int_{\Omega_j \cap M_{\Omega_j,\zeta}} \norm{x - x_\zeta} \;\mathrm{d}x \leq L \sqrt{d}\zeta \Vol(\Omega_j).
        \end{align*}
        Summing over $j$ and using $\Vol(\Omega) = \sum_{j=1}^n \Vol(\Omega_j)$ establishes the inequality.
    \end{proof}

    We note that for a partition to be $h$-smooth, it must consist of Jordan-measurable sets. Indeed, let $j \in [n]$. Since for any decreasing sequence $(\zeta_i)_{i \in \N}$, we have for every $i \in \N$,
    \[\partial \Omega_j = \bigcap_{i=1}^{\infty} \{x \in \R^d : \mathrm{dist}(x,\partial \Omega_j) < \zeta_i\}.\]
    Thus, if we take the volume on both sides above, then the right-hand side is upper bounded by $h(\zeta_i)$ for every $i \in \N$. Thus, if we let $\zeta_i \downarrow 0$, then $h(\zeta_i) \downarrow 0$, and so $\Vol(\partial \Omega_j) = 0$, from which we find that $\Omega_j$ is indeed Jordan-measurable.

    \subsection{Discretization of walks with a Metropolis filter}

    We consider a general form of a random walk with a Metropolis filter.

    \begin{definition}[Naive discretization of Metropolis walks]
        \label{def:naive-discretization}
        Let $\Omega \subseteq \R^d$. Let $f : \Omega \to \R_{>0}$, and for all $x \in \Omega$, let $p_x : \Omega \to \R_{\geq0}$ not identically zero, such that for any $x,y \in \Omega$, we have $p_x(y) = p_y(x)$. Then, we let $(\Omega,P)$ be the Markov process with proposal density $(p_x)_{x \in \Omega}$ and target density $f$ that at $x \in \Omega$ performs the following operations:
        \begin{enumerate}[nosep]
            \item Select a proposal $y$ according to $\pi_{p_x}$.
            \item Transition to $y$ with probability $a_{xy} := \min\{1,f(y)/f(x)\}$, else stay in $x$.
        \end{enumerate}
        We also consider its \textit{naive discretization} $(\Omega_{\zeta},P_{\zeta})$, for all $x,y \in \Omega_{\zeta}$ with $x \neq y$ defined as $(P_{\zeta})_{x,y} = (\pi_{p_x,\zeta})_y \cdot a_{x,y}$.
    \end{definition}

    Intuitively, when the proposal distributions and density are sufficiently smooth, then for decreasing choices of the grid spacing parameter, its naive discretization should introduce a vanishing perturbation. The following theorem makes that intuition precise.

    \begin{theorem}
        \label{thm:naive-discretization}
        Let $\Omega \subseteq \R^d$ be Jordan-measurable, $0 < a < b$, $L > 0$, $h : \R_{>0} \to \R_{>0}$ such that $h(\zeta) \downarrow 0$ as $\zeta \downarrow 0$, and for all $x \in \Omega$, let $p_x : \Omega \to \{0\} \cup [a,b]$ be piecewise $L$-Lipschitz continuous on an $h$-smooth partition $(\Omega^{(x)}_j)_{j=1}^n$. Moreover, suppose that for all $x \in \Omega$, $p_x$ is not identically $0$, and for all $x,y \in \Omega$, $p_x(y) = p_y(x)$. Suppose that $\ell : \Omega \to [a,b]$, defined as $\ell(x) := N_{p_x}$, and $f : \Omega \to [a,b]$ are piecewise $L$-Lipschitz continuous and globally bounded on an $h$-smooth partition too. Let $(\Omega,P)$ be the Markov process with proposal distributions $(p_x)_{x \in \Omega}$ and density $f$, and $(\Omega_{\zeta},P_{\zeta})$ its naive discretization. Then, $(\Omega,P)$ and $(\Omega_{\zeta},P_{\zeta})$ are reversible w.r.t.\ $\pi_{f \cdot \ell}$ and $\pi_{f \cdot \overline{\ell}, \zeta}$, respectively, where $\overline{\ell}(x) = N_{(p_x)_{\zeta}}$. $f \cdot \ell$ is also piecewise Lipschitz continuous and globally bounded on an $h$-smooth partition, hence so is $\pi_{f \cdot \ell}$. Furthermore,
        \[\sup_{x \in \Omega_{\zeta}} \left|\frac{\pi_{f \cdot \overline{\ell},\zeta}(x)}{\zeta^d} - \pi_{f \cdot \ell}(x)\right| \downarrow 0, \qquad (\zeta \downarrow 0).\]
        where the rate of convergence can be upper bounded depending only on $a$, $b$, $L$, $h$ and $\Vol(\Omega)$.
    \end{theorem}

    \begin{proof}
        First, we verify reversibility of the continuous walk. We observe that for any disjoint Borel sets $A,B \subseteq \Omega$,
        \begin{align*}
            &\int_A P(x,B) \;\mathrm{d}\pi_{f \cdot \ell}(x) = \int_{A \times B} \pi_{p_x}(y)a_{x,y} \;\mathrm{d}(\pi_{f \cdot \ell}(x),y) = \frac{1}{N_{f \cdot \ell}} \int_{A \times B} \ell(x)\frac{p_x(y)}{N_{p_x}}\min\{f(x),f(y)\} \;\mathrm{d}(x,y) \\
            &\; = \frac{1}{N_{f \cdot \ell}} \int_{A \times B} \ell(y)\frac{p_y(x)}{N_{p_y}}\min\{f(x),f(y)\} \;\mathrm{d}(x,y) = \int_{A \times B} \pi_{p_y}(x)a_{y,x} \;\mathrm{d}(x,\pi_{f \cdot \ell}(y)) = \int_B P(x,A) \;\mathrm{d}\pi_{f \cdot \ell}(x).
        \end{align*}
        If $A$ and $B$ overlap, then we have an additional term accounting for the probability that we start in their intersection and reject, but this term is also symmetric. Finally, we can prove reversibility of the discrete walk along similar lines.
        
        It remains to check that the stationary distributions of the continuous and discrete Markov processes are close when evaluated on $\Omega_{\zeta}$. To that end, observe that
        \begin{align*}
            &\sup_{x \in \Omega_{\zeta}} \left|\frac{\pi_{f \cdot \overline{\ell},\zeta}(x)}{\zeta^d} - \pi_{f \cdot \ell}(x)\right| = \sup_{x \in \Omega_{\zeta}} \left|\frac{f(x)\overline{\ell}(x)}{\zeta^dN_{f \cdot \overline{\ell},\zeta}} - \frac{f(x)\ell(x)}{N_{f \cdot \ell}}\right| \\
            &\qquad \leq \sup_{x \in \Omega_{\zeta}} \left|\frac{\zeta^df(x)\overline{\ell}(x)}{\zeta^{2d}N_{f \cdot \overline{\ell},\zeta}} - \frac{f(x)\ell(x)}{\zeta^{2d}N_{f \cdot \overline{\ell},\zeta}}\right| + \sup_{x \in \Omega_{\zeta}} \left|\frac{f(x)\ell(x)}{\zeta^{2d}N_{f \cdot \overline{\ell},\zeta}} - \frac{f(x)\ell(x)}{\zeta^dN_{f \cdot \ell,\zeta}}\right| + \sup_{x \in \Omega_{\zeta}} \left|\frac{f(x)\ell(x)}{\zeta^dN_{f \cdot \ell,\zeta}} - \frac{f(x)\ell(x)}{N_{f \cdot \ell}}\right|.
        \end{align*}
        We find from \Cref{lem:smoothness} that
        \[\sup_{x \in \Omega_{\zeta}} \left|\zeta^d\overline{\ell}(x) - \ell(x)\right| = \sup_{x \in \Omega_{\zeta}} \left|\zeta^d N_{(p_x)_{\zeta}} - N_{p_x}\right| \downarrow 0, \qquad (\zeta \downarrow 0).\]
        Since $|f(x)|$ is upper bounded as well, it remains to prove that the denominators are all $\Omega(1)$ and their difference vanishes as $\zeta \downarrow 0$. To that end, since $f \cdot \ell$ is piecewise Lipschitz continuous on an $h$-smooth domain, we immediately find from \Cref{lem:smoothness} that
        \[|\zeta^dN_{f \cdot \ell,\zeta} - N_{f \cdot \ell}| \downarrow 0.\]
        Next, we find that
        \[|\zeta^{2d}N_{f \cdot \overline{\ell},\zeta} - \zeta^dN_{f \cdot \ell,\zeta}| \leq \zeta^d \sum_{x \in \Omega_{\zeta}} |f(x)| \cdot |\zeta^d\overline{\ell}(x) - \ell(x)| \leq b\zeta^d|\Omega_{\zeta}| \cdot \sup_{x \in \Omega_{\zeta}} |\zeta^d\overline{\ell}(x) - \ell(x)| \downarrow 0, \qquad (\zeta \downarrow 0).\]
        Finally, since $N_{f \cdot \ell}$ is globally bounded, we obtain that all denominators are $\Omega(1)$ as $\zeta \downarrow 0$.
    \end{proof}

    The above theorem provides a sufficient condition for Markov process discretizations to have bounded perturbations from their continuous counterpart, and it provides us with the tools we need to analyze the stationary distributions and transition probabilities of the discretizations we consider in the remainder of this work.

    The only additional property we need from our discretizations is that the conductances of our discretizations are similar to their continuous counterparts. The conditions in the above theorem are not enough to guarantee this property, and thus we require a more application-specific approach.

    \section{Speedy walk and discretization}
    \label{sec:ball-walks}

    In this section, we consider the speedy walk on convex bodies as in Cousins and Vempala~\cite{CV18}. We collect some of its properties in \cref{sec:cont-speedy}. As our main technical contribution, we show in \cref{sec:disc-speedy} that the grid-discretization of the speedy walk (cf.~\cref{def:naive-discretization}) yields a good approximation of the continuous speedy walk. We finally state the bipartite discretized speedy walk in \cref{sec:bip-speedy}, which we ultimately use in our quantum walk.

    \subsection{Properties of the continuous speedy walk} \label{sec:cont-speedy}

    Recalling the speedy walk and local conductance as defined in \cref{def:speedy,def:avg-local-conductance}, we observe for future reference that the local conductance is uniformly lower bounded, and Lipschitz.

    \begin{lemma}
        \label{lem:local-conductance-lipschitz}
        Let $d \in \N$, $R \geq 1$ and $B_d \subseteq K \subseteq RB_d$ be a convex body. Let $0 < \delta \leq 2R$. For any $x \in K$, $K \cap (x + \delta B_d)$ contains a ball of radius $\delta/(2R)$, and consequently $\ell_{\delta}$ is globally bounded as $1/(2R)^d \leq \ell_{\delta}(x) \leq 1$. Moreover, $\ell_{\delta}$ is $(2d/\delta)$-Lipschitz on $K$.
    \end{lemma}
    
    \begin{proof}
        We set $\alpha = \delta/(2R)$, and observe that $(1-\alpha)x + \alpha B_d \subseteq K$ by convexity. For any $y \in B_d$, we have
        \[\norm{x - (1-\alpha)x - \alpha y} = \alpha\norm{x - y} \leq \alpha(\norm{x} + \norm{y}) \leq \frac{\delta(R + 1)}{2R} \leq \delta,\]
        and so $(1-\alpha)x + \alpha B_d \subseteq x + \delta B_d$. As such, we have
        \[\ell_{\delta}(x) = \frac{\Vol(K \cap (x + \delta B_d))}{\Vol(\delta B_d)} \geq \frac{\Vol((1-\alpha)x + \alpha B_d)}{\Vol(\delta B_d)} = \left(\frac{\alpha}{\delta}\right)^d = \frac{1}{(2R)^d}.\]
        
        Next, let $x,y \in K$. Observe that
        \[|\ell_{\delta}(x) - \ell_{\delta}(y)| \leq \frac{\Vol(((x + \delta B_d) \Delta (y + \delta B_d)) \cap K)}{\Vol(\delta B_d)} \leq 2 \cdot \frac{\Vol(((x-y) + \delta B_d) \setminus (\delta B_d))}{\Vol(\delta B_d)}.\]
        For the final fraction, observe that for any $a \in \R^d$ with $\norm{a} \leq 1$, we have $(1-\norm{a})B_d \subseteq B_d \cap (a + B_d)$. Setting $a = (x-y)/\delta$, we obtain that
        \[|\ell_{\delta}(x) - \ell_{\delta}(y)| \leq 2 \cdot \frac{\Vol((a + B_d) \setminus B_d)}{\Vol(B_d)} \leq 2(1 - (1-\norm{a})^d) \leq 2d\norm{a} \leq \frac{2d}{\delta} \cdot \norm{x-y}.\qedhere\]
    \end{proof}

    Cousins and Vempala~\cite[Theorem 6.6]{CV18} showed that the conductance of the speedy walk applied to a convex body $B_d \subseteq K \subseteq \sqrt{8d/\beta}B_d$ and Gaussian density $\mathcal N(0,(2\beta)^{-1} I)$ is $\Omega(\delta \sqrt{\beta/d})$ as long as the step-size satisfies $\delta \leq 1/(8\sqrt{2\beta d})$. Here, we record a version that allows for slightly larger convex bodies, contained in $R B_d$. Later, we use radii $R$ of size roughly $\sqrt{d/\beta} (1+\log(1/\eps)/d)$ in our $\eps$-volume estimation algorithm, whereas their version only allows radii of size $O(\sqrt{d/\beta})$.

    \begin{theorem}[Slight generalization of {\cite[Theorem~6.6]{CV18}}]
        \label{thm:speedy-walk-conductance}
                 Let $d \in \N$, and $B_d \subseteq K \subseteq R B_d$ be a convex body. The speedy walk $P_{\beta,\delta}$ on $K$, defined in \Cref{def:speedy}, with parameters $\beta \geq 0$ and $0 < \delta \leq \min\{1/(4\beta R),1/\sqrt{2\beta}\}$
        has conductance 
        $\Omega(\delta\sqrt{\beta/d})$.
    \end{theorem}

    \begin{proof}
    The proof of Theorem~6.6 in \cite{CV18} relies on the inclusion $K \subseteq RB_d$ through their Lemma 6.2 which lower bounds the acceptance probability of the metropolis filter by $1/e$. Our bound on $\delta$ results in the same estimate. Indeed,
        for any $x,y \in K \cap RB_d$ with $\norm{x-y} \leq \delta$, we have
        \begin{align*}
            \frac{\exp(-\beta\norm{x}^2)}{\exp(-\beta\norm{y}^2)} &= \exp(-\beta(\norm{x}^2 - \norm{y}^2) \geq \exp(-\beta((\norm{y} + \delta)^2 - \norm{y}^2) = \exp(-\beta(2\delta\norm{y} + \delta^2)) \\
            &\geq \exp\left(-\beta\left(\frac{2\norm{y}}{4\beta R} + \frac{1}{2\beta}\right)\right) \geq \exp\left(-\left(\frac12 + \frac12\right)\right) = \frac1e.\qedhere
        \end{align*}
    \end{proof}

    Finally, we analyze the total variation distance from the speedy-walk distribution to the Gibbs distribution $\pi_\beta$, similarly as in Cousins and Vempala~\cite{CV18}.

    \begin{lemma}
        \label{lem:avg-local-conductance}
        Let $d \in \N$, and $B_d \subseteq K \subseteq \R^d$ be a convex body. Let $\beta \geq 0$ and $0 < \delta \leq \min\{1,1/\sqrt{2\beta}\}/(2^{18}\sqrt{d})$. Then,
        \[\norm{\pi_{\beta,\delta} - \pi_{\beta}}_{\mathrm{TV}} \leq \frac14.\]
    \end{lemma}

    \begin{proof}
        We use \cite[Lemmas~6.9 and 6.10]{CV18} to observe that
        \[\lambda := \frac{\int_K \ell_\delta(x)\exp(-\beta\norm{x}^2) \;\mathrm{d}x}{\int_K \exp(-\beta\norm{x}^2) \;\mathrm{d}x} \geq 1 - 32\left(\frac{\min\{1,1/\sqrt{2\beta}\}\sqrt{d}}{2^{18}\sqrt{d}\min\{1/1/\sqrt{2\beta}\}}\right)^{1/2} \geq \frac{15}{16}.\]
        Using that $\ell_\delta(x) \in [0,1]$, we find that
        \[\braket{\pi_{\beta,\delta}}{\pi_{\beta}} = \frac{\int_K \sqrt{\ell_{\delta}(x)}\exp(-\beta\norm{x}^2) \;\mathrm{d}x}{\sqrt{\int_K \ell_{\delta}(x)\exp(-\beta\norm{x}^2) \;\mathrm{d}x \cdot \int_K \exp(-\beta\norm{x}^2) \;\mathrm{d}x}} \geq \frac{\lambda}{\sqrt{\lambda}} \geq \sqrt{\frac{15}{16}}.\]
        Thus, we conclude that
        \[\norm{\pi_{\beta,\delta} - \pi_{\beta}}_{\mathrm{TV}} \leq \sqrt{1 - |\braket{\pi_{\beta,\delta}}{\pi_{\beta}}|^2} \leq \sqrt{1 - \frac{15}{16}} = \frac14.\qedhere\]
    \end{proof}

    \subsection{Discretization of the speedy walk} \label{sec:disc-speedy}

    Now that we have introduced the speedy walk in the continuous setting, we turn to its discretization.
    We will ultimately implement the naive discretization from \Cref{def:naive-discretization} that yields a walk on $K_\zeta$
    that picks a proposal from $(x+\delta B_d) \cap K_\zeta$ rather than from $(x+\delta B_d) \cap K$, and then applies the same filter as the speedy walk. The black-box result in \Cref{thm:naive-discretization} implies that the resulting stationary distributions are close when evaluated on the grid points $K_{\zeta}$.

    Proving a lower bound on the conductance of the discretized speedy walk is more subtle, though. To that end, we start from a walk on a slightly enlarged version of the convex body, i.e., $\hat{K} := (1+\zeta\sqrt{d})K$. We know a lower bound on its conductance from \Cref{thm:speedy-walk-conductance}. Next, we introduce a continuous-space approximation of the speedy walk on $M_{K,\zeta}$, which we call the cube speedy walk, in \Cref{subsubsec:cube-speedy-walk}.
    We show that this walk induces a discrete walk on $K_\zeta$ in \Cref{subsubsec:induced-cube-speedy-walk}, which in turn approximates our final grid speedy walk in \Cref{subsubsec:discrete-speedy-walk}.
       This yields the following chain of approximations of the conductances:
    \[
    \begin{gathered}
        \text{Speedy walk} \\
        \hat{P}_{\beta,\hat{\delta}} \text{ on } \hat{K} \\
        \varphi(\hat{P}_{\beta,\delta})
    \end{gathered}
    \quad\raisebox{-1.8em}{$\lesssim$}\quad
    \begin{gathered}
        \text{Cube speedy walk} \\
        P^C_{\beta,\delta,\zeta} \text{ on } M_{K,\zeta} \\
        \varphi(P^C_{\beta,\delta,\zeta})
    \end{gathered}
    \quad\raisebox{-1.8em}{$\lesssim$}\quad
    \begin{gathered}
        \text{Induced cube speedy walk} \\
        P^I_{\beta,\delta,\zeta} \text{ on } K_\zeta \\
        \varphi(P^I_{\beta,\delta,\zeta})
    \end{gathered}
    \quad\raisebox{-1.8em}{$\lesssim$}\quad
    \begin{gathered}
        \text{Grid speedy walk} \\
        P_{\beta,\delta,\zeta} \text{ on } K_\zeta \\
        \varphi(P_{\beta,\delta,\zeta})
    \end{gathered}
    \]

    \subsubsection{Cube speedy walk}
    \label{subsubsec:cube-speedy-walk}

    First, we consider a version of the speedy walk where the steps are taken based on the nearest grid point: we transition from $x$ to $y$ only if $\norm{x_\zeta - y_\zeta} \leq \delta$ and $x_{\zeta},y_{\zeta} \in K_{\zeta}$. The resulting chain walks over $M_{K,\zeta}$, rather than over~$K$. Since $M_{K,\zeta}$ lives in a world that comprises exclusively of small $d$-dimensional cubes, we refer to this walk as the cube speedy walk.

    \begin{definition}[Cube speedy walk]
        Let $d \in \N$, $K \subseteq \R^d$ a convex body. Let $\beta \geq 0$, $\delta > 0$ and $f(x) = \exp(-\beta\norm{x}^2)$. At $x \in M_{K,\zeta}$, do:
        \begin{enumerate}[nosep]
            \item Select $y \in (x_\zeta + M_{\delta B_d,\zeta}) \cap M_{K,\zeta}$ uniformly at random.
            \item Transition from $x$ to $y$ with probability $a_{x,y} = \min\{1,f(y)/f(x)\}$. Else, stay in $x$.
        \end{enumerate}
        Let $P^C_{\beta,\delta,\zeta}$ be the Markov transition kernel, and $\pi^C_{\beta,\delta,\zeta}$ be the stationary distribution of this walk.
    \end{definition}

    We can define the \textit{cube local conductance} $\ell^C_{\delta,\zeta}(x)$ on $M_{K,\zeta}$ in the same way as the conductance on a convex body.

    \begin{definition}[Cube local conductance]
        Let $d \in \N$, $K \subseteq \R^d$ a convex body. Let $\delta > 0$ and $\zeta > 0$. Then, for all $x \in M_{K, \zeta}$,
        \[\ell^C_{\delta,\zeta}(x) = \frac{\Vol((x_{\zeta} + M_{\delta B_d,\zeta}) \cap M_{K,\zeta})}{\Vol(M_{\delta B_d,\zeta})}.\]
    \end{definition}

    Note that the cube local conductance is the normalization constant of the proposal distribution of the cube speedy walk. Thus, it follows directly from \Cref{thm:naive-discretization} that $\pi^C_{\beta,\delta,\zeta}(x)
    \propto \ell^C_{\delta,\zeta}(x) f(x)$. 

    We now relate the properties of the cube speedy walk to the regular speedy walk. To that end, we first of all observe that cube versions of convex sets are contained in tiny dilations.

    \begin{lemma}
        \label{lem:minecraft-dilation-containment}
        Let $d \in \N$, $\eta > 0$, $c \in \R^d$ and $c + B_d/\eta \subseteq K$ be a convex body. Let $0 < \zeta < 1/(\eta\sqrt{d})$. Then, $$(1-\eta\zeta\sqrt{d})(K-c) \subseteq M_{K,\zeta}-c \subseteq (1+\eta\zeta\sqrt{d})(K-c).$$ Consequently,
        \[\left|\Vol(M_{K,\zeta}) - \Vol(K)\right| \leq \left(\left(1 + \eta\zeta\sqrt{d}\right)^d - 1\right) \cdot \Vol(K) \downarrow 0, \qquad (\zeta \downarrow 0).\]
    \end{lemma}
    
    \begin{proof}
        Let $A = M_{K,\zeta} \Delta K$ be the symmetric difference of the sets $M_{K,\zeta}$ and $K$. By construction, for any $x \in A$ there exists a $y \in \partial K$ with $\|x-y\|_\infty \leq \zeta$, and hence by H\"older's inequality $\|x-y\| \leq \zeta \sqrt{d}$. This implies that $M_{K,\zeta} - c \subseteq (K-c) + \zeta \sqrt{d}B_d \subseteq (1+\eta\zeta \sqrt{d})(K-c)$, where the last inclusion uses $B_d/\eta \subseteq K-c$ and the convexity of $K$.
        
        We now show that $(1-\eta\zeta \sqrt{d})(K-c) \subseteq M_{K,\zeta} - c$. Let $x \in K$. Using H\"older's inequality, we obtain that
        \[\norm{(c+(1-\eta\zeta\sqrt{d})(x-c))_{\zeta} - (c+(1-\eta\zeta\sqrt{d})(x-c))} \leq \zeta\sqrt{d}.\]
        Therefore, we obtain that
        \[(c+(1-\eta\zeta\sqrt{d})(x-c))_{\zeta} \in c+(1-\eta\zeta\sqrt{d})(x-c) + \zeta\sqrt{d}B_d = (1-\eta\zeta\sqrt{d})x + \eta\zeta\sqrt{d}(c + B_d/\eta)\subseteq K,\]
        where the last inclusion uses $c+B_d/\eta \subseteq K$ and the convexity of $K$. This shows that $c+(1-\eta\zeta\sqrt{d})(x-c) \in M_{K,\zeta}$, and so $(1-\eta\zeta\sqrt{d})(K-c) \subseteq M_{K,\zeta}-c$.
    \end{proof}

    We     use this lemma to bound the difference between the local conductances of the regular and cube speedy walks, respectively.

    \begin{lemma}
        \label{lem:local-conductance-lb}
        Let $d \in \N$, $R \geq 1$ and $B_d \subseteq K \subseteq RB_d$ be a convex body. Let $0 < \delta < 2R$. Then,
        \[\sup_{K, \; x \in K \cap M_{K,\zeta}} \left|\ell^C_{\delta,\zeta}(x) - \ell_\delta(x)\right| \downarrow 0, \qquad (\zeta \downarrow 0).\]
        Consequently, $\ell^C_{\delta,\zeta}$ is globally bounded, for sufficiently small choices of $\zeta = \zeta(d,R,\delta)$.
    \end{lemma}
    
    \begin{proof}
        First, we observe from \Cref{lem:local-conductance-lipschitz} that $\ell_{\delta}$ is globally bounded and Lipschitz with constants that only depends on $d$, $R$ and $\delta$, and so
        \[\sup_{K, \; x \in K} |\ell_{\delta}(x) - \ell_{\delta}(x_{\zeta})| \downarrow 0, \qquad (\zeta \downarrow 0).\]
        Next, observe that for all $x \in M_{K,\zeta}$, $\ell^C_{\delta,\zeta}(x) = \ell^C_{\delta,\zeta}(x_{\zeta})$, so it remains to prove that the statement from the lemma holds for all $x \in K_{\zeta}$. To that end, observe that the indicator functions on $x + \delta B_d$ and $(x + \delta B_d) \cap K$ are piecewise Lipschitz on an $h$-smooth partition, where for the $h$-smoothness we have to analyze their boundaries. Since they are both a bounded convex body contained in $3RB_d$, we obtain that its perimeter is bounded by a function of $d$ and $R$, and hence we can upper bound its smoothness by a function $h$ that only depends on $d$, $R$ and $\zeta$. Thus, by \Cref{lem:smoothness}, as $\zeta \downarrow 0$, we have
        \[\sup_{K, \; x \in K_{\zeta}} \left|\zeta^d \cdot \sum_{y \in ((x + \delta B_d) \cap K)_{\zeta}} 1 - \int_{(x + \delta B_d) \cap K} \;\mathrm{d}y\right| \downarrow 0, \quad \text{and} \quad \sup_{K, \; x \in K_{\zeta}} \left|\zeta^d \cdot \sum_{y \in (x + \delta B_d)_{\zeta}} 1 - \int_{x + \delta B_d} \;\mathrm{d}y\right| \downarrow 0.\]
        Thus, by considering their ratio and observing that the denominators are uniformly lower bounded, we obtain that
        \[\sup_{K, \; x \in K_{\zeta}} |\ell^C_{\delta,\zeta}(x) - \ell_{\delta}(x)| \downarrow 0, \qquad (\zeta \downarrow 0).\qedhere\]
    \end{proof}

    We can now use these observations to relate the stationary distributions to one another.

    \begin{lemma}
        \label{lem:warmness-discretization}
        Let $d \in \N$, $R \geq 1$ and $B_d \subseteq K \subseteq RB_d$ be a convex body. Let $\beta \geq 0$, $0 < \delta < 2R$ and $\zeta > 0$. Then,
        \[\sup_{K, \; x \in K \cap M_{K,\zeta}} \left|\pi^C_{\beta,\delta,\zeta}(x) - \pi_{\beta,\delta}(x)\right| \downarrow 0, \qquad (\zeta \downarrow 0).\]
        Consequently, $\pi^C_{\beta,\delta,\zeta}$ is globally bounded for sufficiently small choices of $\zeta = \zeta(d,R,\delta)$, and $\pi^C_{\beta,\delta,\zeta}(M_{K,\zeta} \setminus K) \downarrow 0$ and $\pi_{\beta,\delta}(K \setminus M_{K,\zeta}) \downarrow 0$ as $\zeta \downarrow 0$.
    \end{lemma}

    \begin{proof}
        We observe from \Cref{thm:naive-discretization} that $\pi_{\beta,\delta}(x)$ and $\pi^C_{\beta,\delta,\zeta}(x)$ are proportional to the densities given by $\exp(-\beta\norm{x}^2)\ell_\delta(x)$ and $\exp(-\beta\norm{x}^2)\ell^C_{\delta,\zeta}(x)$, respectively. We find from \Cref{lem:local-conductance-lb} that they are both globally bounded, and their difference vanishes uniformly on $K \cap M_{K,\zeta}$ as $\zeta \downarrow 0$, uniformly over all choices of $K$. By \Cref{lem:minecraft-dilation-containment}, we find that the volume of the symmetric difference between $K$ and $M_{K,\zeta}$ vanishes as well, similarly uniformly over $K$ as $K \subseteq RB_d$. Thus, the normalization factors of both probability densities are $\Theta(1)$ as $\zeta \downarrow 0$, and their difference vanishes. Consequently, both stationaries are globally bounded, and thus the probability on the symmetric difference of $K$ and $M_{K,\zeta}$ vanishes as well.
            \end{proof}

    Now, we consider the convex body $\hat K = (1+\zeta\sqrt{d})K$. Let $\hat P_{\beta,\delta}$ and $\hat{P}^C_{\beta,\delta,\zeta}$ be the speedy walk and cube speedy walk on $\hat K$, with stationary distributions $\hat{\pi}_{\beta,\delta,\zeta}$ and $\hat{\pi}^C_{\beta,\delta,\zeta}$, respectively. We first prove warmness of the stationary distribution and closeness of the Markov transition kernels of the cube speedy walk to this dilated speedy walk.

    \begin{lemma}
        \label{lem:enlarged-body}
        Let $d \in \N$, $R \geq 1$ and $B_d \subseteq K \subseteq RB_d$ be a convex body. Let $\beta \geq 0$, $0 < \delta < 2R$ and $\zeta > 0$. Then,
        \[\sup_{K, \; x \in M_{K,\zeta}} \left|\pi^C_{\beta,\delta,\zeta}(x) - \hat\pi_{\beta,\delta}(x)\right| \downarrow 0, \qquad \text{and} \qquad \sup_{K, \; x \in M_{K,\zeta}} \norm{P^C_{\beta,\delta,\zeta}(x,\cdot) - \hat P_{\beta,\delta}(x,\cdot)}_{\mathrm{TV}} \downarrow 0, \qquad (\zeta \downarrow 0).\]
    \end{lemma}

    \begin{proof}
        Observe from \Cref{lem:minecraft-dilation-containment} that $M_{K,\zeta} \subseteq \hat K$. We also trivially find that $M_{K,\zeta} \subseteq M_{\hat{K},\zeta}$. Thus, we deduce from \Cref{lem:warmness-discretization} that
        \[\sup_{K, \; x \in M_{K,\zeta}} |\hat{\pi}^C_{\beta,\delta,\zeta}(x) - \hat{\pi}_{\beta,\delta}(x)| \downarrow 0, \qquad (\zeta \downarrow 0).\]
        Thus, it remains to relate the cube speedy walk distributions to one another on $M_{K,\zeta}$. To that end, observe from \Cref{thm:naive-discretization} that they are proportional to the densities $\exp(-\beta\norm{x}^2)\ell^C_{\delta,\zeta}(x)$ and $\exp(-\beta\norm{x}^2)\hat{\ell}^C_{\delta,\zeta}(x)$, respectively. We observe that
        \begin{align*}
            \sup_{K, \; x \in M_{K,\zeta}} \left|\hat\ell^C_{\delta,\zeta}(x) - \ell^C_{\delta,\zeta}(x)\right| &= \sup_{K, \; x \in M_{K,\zeta}} \frac{\Vol(M_{(x_{\zeta} + \delta B_d) \cap (\hat K \setminus K),\zeta})}{\Vol(\delta B_d)} \leq \sup_K \frac{\Vol(M_{\hat{K} \setminus K,\zeta})}{\Vol(\delta B_d)} \\
            &= \sup_K \frac{\Vol(M_{\hat{K},\zeta}) - \Vol(M_{K,\zeta})}{\Vol(\delta B_d)} \downarrow 0, \qquad (\zeta \downarrow 0),
        \end{align*}
        where we used \Cref{lem:minecraft-dilation-containment} in the final step. As such, both densities are globally bounded with constants that only depend on $d$, $R$, $\beta$ and $\delta$, their difference vanishes on $M_{K,\zeta}$, and the volume of the symmetric difference between their state spaces does too. Thus, the difference between their normalization constants vanishes too, and we establish the first of the statements in the lemma.

        For the second statement, let $x \in M_{K,\zeta}$. We observe that the proposal distributions for both walks select a point uniformly from $(x + \delta B_d) \cap K$ and $M_{(x_{\zeta} + \delta B_d) \cap K,\zeta}$, respectively. We observe from \Cref{thm:naive-discretization} that their symmetric difference vanishes uniformly, and from \Cref{lem:local-conductance-lb}, we see that the difference between their normalization constants vanishes uniformly too. Hence their total-variation distance vanishes uniformly as well. Since the acceptance probabilities are globally bounded and the same for both walks, the total-variation distance bound extends to $P^C_{\beta,\delta,\zeta}(x,\cdot)$ and $\hat{P}_{\beta,\delta}(x,\cdot)$, completing the proof.
    \end{proof}
    
    Finally, we use the slightly enlarged convex body to prove a lower bound on the conductance of the cube speedy walk.
    
    \begin{theorem}
        \label{thm:conductance-minecraft-walk}
        Let $d \in \N$, $R \geq 1$, $\beta \geq 0$, $0 < \delta < \min\{1/(4\beta R), 1/\sqrt{2\beta}\}$, $\zeta > 0$, and $B_d \subseteq K \subseteq RB_d$ be a convex body. Then, for a sufficiently small choice of $\zeta = \zeta(d,R,\beta,\delta)$, the conductance of the cube speedy walk satisfies $\varphi(P^C_{\beta,\delta,\zeta}) \in \Omega(\delta\sqrt{\beta/d})$.
    \end{theorem}

    \begin{proof}
        Let $\zeta > 0$. We know from \Cref{lem:minecraft-dilation-containment} that $M_{K,\zeta} \subseteq \hat K$, and so the stationary measures satisfy $\pi_{\beta,\delta,\zeta}^C \ll \hat{\pi}_{\beta,\delta}$. We recall from \Cref{thm:naive-discretization} that $\hat{\pi}_{\beta,\delta}$ is globally bounded, and hence we observe from \Cref{lem:enlarged-body}, that
        \[\sup_{K, \; x \in M_{K,\zeta}} \left|\frac{\pi^C_{\beta,\delta,r,\zeta}(x)}{\hat\pi_{\beta,\delta}(x)} - 1\right| \downarrow 0, \qquad \text{and} \qquad \sup_{K, \; x \in M_{K,\zeta}} \norm{P^C_{\beta,\delta,\zeta}(x,\cdot) - \hat P_{\beta,\delta}(x,\cdot)}_{\mathrm{TV}} \downarrow 0, \qquad (\zeta \downarrow 0).\]
        From \Cref{thm:conductance-robustness}, we then obtain that for a sufficiently small choice of $\zeta = \zeta(d,R,\beta,\delta)$, $\varphi(P^C_{\beta,\delta,\zeta}) \in \Omega(\varphi(\hat P_{\beta,\zeta})) \subseteq \Omega(\delta\sqrt{\beta/d})$, where in the final step we used \Cref{thm:speedy-walk-conductance}.
    \end{proof}

    \subsubsection{Induced cube speedy walk}
    \label{subsubsec:induced-cube-speedy-walk}

    Next, we consider the induced cube speedy walk on $(M_{K,\zeta})_\zeta = K_\zeta$, obtained by taking the chain induced by $P^C_{\beta,\delta,\zeta}$ on the cubes $(M_{x,\zeta})_{x \in \R^d}$ (cf.~\cref{def:canonical-discretization}).

    \begin{definition}[Induced cube speedy walk]
        Let $d \in \N$ and $K \subseteq \R^d$ be a convex body. Let $\beta \geq 0$, $\delta > 0$, $\zeta > 0$ and $f(x) = \exp(-\beta\norm{x}^2)$. We define the induced cube speedy walk on $K$ with inverse temperature $\beta$, step size $\delta$ and grid spacing $\zeta$, that at $x \in K_{\zeta}$ does:
        \begin{enumerate}[nosep]
            \item Select an element from $y \in (x + \delta B_d) \cap K_\zeta$ uniformly at random.
            \item Transition to $y$ with probability
            \[a^I_{x,y} := \frac{1}{\pi^C_{\beta,\delta,\zeta}(M_{x,\zeta})} \int_{M_{x,\zeta}} \frac{1}{\zeta^d} \int_{M_{y,\zeta}} \min\left\{1, \frac{f(y')}{f(x')}\right\} \;\mathrm{d}y' \cdot \pi^C_{\beta,\delta,\zeta}(x') \;\mathrm{d}x'.\]
            Else stay in $x$.
        \end{enumerate}
        We refer to the stationary distribution as $\pi^I_{\beta,\delta,\zeta}$.
    \end{definition}

    This walk's conductance can be canonically lower bounded using the results from the previous section.

    \begin{corollary}
        \label{cor:induced-conductance-bound}
        Let $d \in \N$, $\beta \geq 0$, $R > 0$ $0 < \delta < \min\{1/(4\beta R), 1/\sqrt{2\beta}\}$, $\zeta > 0$, and $B_d \subseteq K \subseteq R B_d$ be a convex body. The stationary distribution of the induced cube speedy walk satisfies for all $x \in K_\zeta$, $\pi^I_{\beta,\delta,\zeta}(x) = \hat{\pi}_{\beta,\delta,\zeta}(M_{x,\zeta})$. For a sufficiently small choice of $\zeta = \zeta(d,R,\beta,\delta)$, the conductance of the induced cube speedy walk is $\Omega(\delta\sqrt{\beta/d})$.
    \end{corollary}

    \begin{proof}
        This follows from applying \Cref{thm:discretization-conductance} to \Cref{thm:conductance-minecraft-walk}.
    \end{proof}

    \subsubsection{Discretized speedy walk}
    \label{subsubsec:discrete-speedy-walk}

    The last remaining issue is that we cannot efficiently implement the induced cube speedy walk. Instead, we can merely efficiently approximate its transition probabilities. To that end, we consider the discretized speedy walk.

    \begin{definition}[Discretized speedy walk]
        \label{def:discretized-speedy}
        Let $d \in \N$ and $K \subseteq \R^d$ be a convex body. Let $\beta \geq 0$, $\delta > 0$, $\zeta > 0$ and $f(x) = \exp(-\beta\norm{x}^2)$.
        We define the discretized speedy walk on $K$ with inverse temperature~$\beta$, step size $\delta$ and grid spacing $\zeta$ as the walk that at $x \in K_\zeta$ does:
        \begin{enumerate}[nosep]
            \item Select an element from $y \in (K \cap (x + \delta B_d))_\zeta$ uniformly at random.
            \item Transition to $y$ with probability $a_{x,y} = \min\{1, f(y)/f(x)\}$. Else, stay in $x$.
        \end{enumerate}
        We refer to the probability transition matrix as $P_{\beta,\delta,\zeta}$ and to the stationary distribution as $\pi_{\beta,\delta,\zeta}$.
    \end{definition}

    This particular walk again introduces a notion of local conductance.

    \begin{definition}
        Let $d \in \N$, $\delta > 0$, $\zeta > 0$, $K \subseteq \R^d$ be a convex body. Then, for all $x \in K_\zeta$,
        \[\ell_{\delta,\zeta}(x) := \frac{|(x + \delta B_d) \cap K_\zeta|}{|(\delta B_d)_\zeta|}.\]
    \end{definition}
    
    The stationary distribution $\pi_{\beta,\delta,\zeta}$ is proportional to $\exp(-\beta\norm{x}^2)\ell_{\delta,\zeta}(x)$. We characterize how much this notion of local conductance differs from the regular notion. 
    
    \begin{lemma}[Local conductance approximation]
        \label{lem:discrete-local-conductance-proximity}
        Let $d \in \N$, $R \geq 1$, $0 < \delta < 2R$, $B_d \subseteq K \subseteq RB_d$. Then,
        \[\sup_{K, \; x \in K_\zeta} \left|\ell_\delta(x) - \ell_{\delta,\zeta}(x)\right| \downarrow 0, \qquad (\zeta \downarrow 0).\]
    \end{lemma}

    \begin{proof}
        We note that $\ell_{\delta,\zeta}(x) = \ell^C_{\delta,\zeta}(x)$, and so the result follows directly from \Cref{lem:local-conductance-lb}.
    \end{proof}

    To lower bound the conductance of the discretized speedy walk, we relate it to the conductance of the induced cube speedy walk.

    \begin{theorem}[Conductance approximation]
        \label{thm:discretized-speedy-conductance}
        Let $d \in \N$, $\beta \geq 0$, $R \geq 1$, $0 < \delta < \min\{1/(4\beta R), 1/\sqrt{2\beta}\}$, $\zeta > 0$ and $B_d \subseteq K \subseteq R B_d$ be a convex body. For a sufficiently small choice of $\zeta = \zeta(d,R,\beta,\delta)$, the conductance of the discretized speedy walk on $K$ with step size $\zeta$ satisfies $\varphi(P_{\beta,\delta,\zeta}) \in \Omega(\delta\sqrt{\beta/d})$.
    \end{theorem}

    \begin{proof}
        Let $x \in K_\zeta$ and $y \in (x + \delta B_d) \cap K_\zeta$ with $x \neq y$. Recall that $a_{x,y} = \min\{1, f(y)/f(x)\}$ and 
        \[a^I_{x,y} := \frac{1}{\pi^C_{\beta,\delta,\zeta}(M_{x,\zeta})} \int_{M_{x,\zeta}} \frac{1}{\zeta^d}\int_{M_{y,\zeta}} \min\left\{1, \frac{f(y')}{f(x')}\right\} \;\mathrm{d}y' \cdot \pi^C_{\beta,\delta,\zeta}(x') \;\mathrm{d}x',\]
        where $f(x) = \exp(-\beta\norm{x}^2)$. We observe that for every $x \in RB_d$,
        \[\nabla f(x) = -2\beta\exp(-\beta\norm{x}^2)x, \qquad \text{and so} \qquad \norm{\nabla f(x)} \leq 2\beta\exp(-\beta\norm{x}^2)\norm{x} \leq 2\beta R.\]
        Thus, as $f$ is $2\beta R$-Lipschitz continuous on $RB_d$, we obtain that the inner integral converges uniformly over the choice of $x'$ and $y$ to $\min\{1,f(y)/f(x')\}$ as $\zeta \downarrow 0$. Similarly, as both $\pi^C_{\beta,\delta,\zeta}$ and $1/f$ are piecewise Lipschitz continuous, the entire integral converges uniformly to $\min\{1,f(y)/f(x)\}$. As such,
        \[\sup_{\substack{K, \; x,y \in K_{\zeta} \\ 0 < \norm{x-y} \leq \delta}} |a^I_{x,y} - a_{x,y}| \downarrow 0, \qquad (\zeta \downarrow 0).\]
        Since for all $x \in K_{\zeta}$, $\exp(-\beta R^2) \leq f(x) \leq 1$, we have $a_{x,y} \in \Omega(1)$, and so
        \[\sup_{\substack{K, \; x,y \in K_{\zeta} \\ 0 < \norm{x-y} \leq \delta}} \left|\frac{a^I_{x,y}}{a_{x,y}} - 1\right| \downarrow 0, \qquad (\zeta \downarrow 0).\]

        Next, observe that the diameter of the walk satisfies $D \in O(R/\delta)$ for sufficiently small $\zeta$, and so constant when $\zeta \downarrow 0$, uniformly over the choice of $K$. Thus, from \Cref{thm:robustness-diameter}, we observe that
        \[\sup_{K, \; x \in K_{\zeta}} \left|\frac{\pi^I_{\beta,\delta,\zeta}(x)}{\pi_{\beta,\delta,\zeta}(x)} - 1\right| \downarrow 0, \qquad \text{and} \qquad \sup_{K, \; x \in K_{\zeta}} \norm{P^I_{\beta,\delta,\zeta}(x,\cdot) - P_{\beta,\delta,\zeta}(x,\cdot)}_{\mathrm{TV}} \downarrow 0, \qquad (\zeta \downarrow 0).\]
        Therefore, by \Cref{thm:conductance-robustness}, we obtain that the conductances of both walks are related for small enough choices of $\zeta = \zeta(d,R,\beta,\delta)$, i.e., $\varphi(P_{\beta,\delta,\zeta}) \geq 3\varphi(P^I_{\beta,\delta,\zeta})/10 - \eta$, where $\eta \downarrow 0$ as $\zeta \downarrow 0$. Finally, using \Cref{cor:induced-conductance-bound}, we conclude that for sufficiently small $\zeta = \zeta(d,R,\beta,\delta)$, the conductance satisfies $\varphi(P_{\beta,\delta,\zeta}) \in \Omega(\delta\sqrt{\beta/d})$.
    \end{proof}

    \subsection{The bipartite speedy walk} \label{sec:bip-speedy}

    To quantize the speedy walk, we will use the machinery developed in \cref{sec:bipartite}. As observed before, the discretized speedy walk from \cref{def:discretized-speedy} employs a Metropolis filter. In \cref{sec:bipartite} we have seen that using a Glauber filter instead changes the walk's conductance and spectral gap by at most a factor~$2$. We then instantiate the bipartite Glauber walk from \cref{def:bipartite-glauber} to the speedy walk with a Glauber filter, leading to the bipartite speedy walk.

    \begin{definition}[Bipartite discretized speedy walk]
        \label{def:bipartite-speedy}
        Let $d \in \N$ and $K \subseteq \R^d$ be a convex body. Let $\beta \geq 0$, $\delta > 0$ and $\zeta > 0$. We write $f(x) = \exp(-\beta\norm{x}^2)$. Define 
        \[
        E = \{\{x,y\}: x,y \in K_\zeta, \|y-x\| \leq \delta\},
        \]
        We then define a bipartite Markov chain on $K_\zeta \sqcup E$ as follows: \\ 
        On the one hand, at point $x \in K_\zeta$, do:
        \begin{enumerate}[nosep]
            \item Select a point $y \in (x + \delta B_d) \cap K_\zeta$ uniformly at random.
            \item Output $\{x,y\} \in E$.
        \end{enumerate}
        On the other hand, at $\{x,y \} \in E$, do:
        \begin{enumerate}[nosep]
            \item Output $x$ with probability $\frac{f(x)}{f(x)+f(y)}$ and $y$ with probability $\frac{f(y)}{f(x)+f(y)}$.
        \end{enumerate}
    \end{definition}

    We record the following observations. First, from now on we let $P$ denote the transition kernel of the bipartite discretized speedy walk. Second, \Cref{lem:metropolis-to-glauber-filter} shows that the spectral gap of $P$ is at most a factor two smaller than that of the discretized speedy walk with Glauber filter. Hence, from \Cref{thm:discretized-speedy-conductance} and Cheeger's inequality~(\Cref{thm:conductance-spectral-gap}) it follows that for a sufficiently small choice of the spacing parameter~$\zeta$, $\gamma(P) \in \Omega(\delta^2\beta/d)$ whenever $0 < \delta < \min\{1/(4\beta R), 1/\sqrt{2\beta}\}$ provided $K \subseteq R B_d$.
    Third, in order to quantize the bipartite speedy walk, we show in the next section that we can reflect through the star states associated to $P$, which are as follows. 
    For all $x \in K_\zeta$, we have
    \[\ket{*_x} := \frac{1}{\sqrt{|(x + \delta B_d) \cap K_\zeta|}} \sum_{y \in (x + \delta B_d) \cap K_\zeta} \ket{\{x,y\}}, \qquad \ket{*_{\{x\}}} = \ket{x},\]
    and for all $x,y \in K_\zeta$ with $0 < \norm{x-y} \leq \delta$, we have
    \[\ket{*_{\{x,y\}}} := \frac{\sqrt{\exp(-\beta\norm{x}^2)}\ket{x} + \sqrt{\exp(-\beta\norm{y}^2)}\ket{y}}{\sqrt{\exp(-\beta\norm{x}^2) + \exp(-\beta\norm{y}^2)}}.\]
    Finally, for all $x,y \in K_\zeta$ with $\norm{x-y} \leq \delta$, we write
    \[\mathcal{N}_{\beta,\delta,\zeta} : \ket{x}\ket{\perp} \mapsto \ket{x}\ket{*_x}, \qquad \text{and} \qquad \mathcal{N}_{\beta,\delta,\zeta} : \ket{\{x,y\}}\ket{\perp} \mapsto \ket{\{x,y\}}\ket{*_{\{x,y\}}}.\]

    \begin{remark}
        We observe that $\ket{*_{\{x,y\}}}$ does not depend on the specific convex body $K$ that forms the input to our problem. As such, we can implement operations that prepare and reflect around $\ket{*_{\{x,y\}}}$ without making any queries to $K$, and consequently with transduction complexity $0$. Thus, it remains to construct transducers that prepare and reflect around the star states $\ket{*_x}$, for all $x \in K_\zeta$.
    \end{remark}

    \section{Implementation of the speedy walk with transducers} \label{sec:transducer-ball-walk}

    In this section, we provide transducers that construct and reflect through the stationary states of the discretized bipartite speedy walk from the previous section.
    
    We assume to have access to the convex body through a (strong) membership oracle, i.e., a subroutine that given a point $x \in \R^d$ answers whether $x \in K$. To deal with discretization considerations, we concretely assume to have access to such an oracle evaluated on a finite grid with spacing $\zeta > 0$ that we can choose arbitrarily. That is, we assume to have access to $O_{K,\zeta}$, that for all $x \in (RB_d)_\zeta$ acts as
    \[O_{K,\zeta} : \ket{x} \mapsto (-1)^{x \in K}\ket{x}.\]

    For a probability distribution $p$ over $(RB_d)_{\zeta}$, we write its corresponding quantum state as
    \[\ket{p} = \sum_{x \in (RB_d)_{\zeta}} \sqrt{p_x}\ket{x} \in \C^{(RB_d)_{\zeta}}.\]
    Whenever we have a distribution that is defined on a subset of $(RB_d)_{\zeta}$, we generically extend it with $0$'s on its complement.

    \subsection{Reflections through stationary distributions of the speedy walk}

    In this section, we provide transducers that construct and reflect through the star states of the bipartite discretized speedy walk. We start with a state-preparation and -reflection transducer for $\ket{*_x}$.

    \begin{lemma}
        \label{lem:star-state-transducers}
        Let $d \in \N$, $R \geq 1$, $\beta \geq 0$, $\delta > 0$ and $K \subseteq RB_d$ be a convex body, accessible through a membership oracle on a regular grid with sufficiently small spacing parameter $\zeta = \zeta(d,R,\beta,\delta) > 0$. Let $x \in K_\zeta$. Then, we can construct bidirectional canonical transducers $C_x$ and $U_x$ that act as
        \begin{align*}
            \ket{0} \overset{C_x}{\rightsquigarrow} \ket{*_x}, & \qquad \text{with complexity} \qquad O\left(\frac{1}{\sqrt{\ell_{\delta}(x)}}\right), \\
            \ket{*_x} \overset{U_x}{\rightsquigarrow} \ket{*_x}, & \qquad \text{with complexity} \qquad O\left(\frac{1}{\ell_{\delta}(x)}\right), \\
            \ket{\perp} \overset{U_x}{\rightsquigarrow} -\ket{\perp}, & \qquad \text{with complexity} \qquad O\left(\norm{\ket{\perp}}^2\right),
        \end{align*}
        for all $\ket{\perp} \in \C^{K_\zeta}$ satisfying $\braket{\perp}{*_x} = 0$.
    \end{lemma}

    \begin{proof}
        Let $E$ as in \Cref{def:bipartite-speedy}. We start with the state $\ket{N_x} \in \C^E$, defined as
        \[\ket{N_x} = \frac{1}{\sqrt{|(\delta B_d)_\zeta|}} \sum_{y \in (x + \delta B_d)_\zeta} \ket{\{x,y\}}.\]
        We note that $\ket{N_x}$ is independent of $K$, so we can prepare it without making any queries, i.e., with transduction complexity $0$. Now, observe that with one call to $O_{K,\zeta}$, we can implement the following operation:
        \[O_{K,\zeta} : \ket{N_x}\ket{0} \mapsto \underbrace{\frac{1}{\sqrt{|(\delta B_d)_\zeta|}} \sum_{y \in ((x + \delta B_d) \setminus K)_\zeta} \ket{\{x,y\}}}_{=: \ket{\psi_0^x}}\ket{0} + \underbrace{\frac{1}{\sqrt{|(\delta B_d)_\zeta|}} \sum_{y \in ((x + \delta B_d) \cap K)_\zeta} \ket{\{x,y\}}}_{=: \ket{\psi_1^x}}\ket{1} =: \ket{\psi_x}.\]
        Thus, with one call to $O_{K,\zeta}$ and one to its inverse, we can reflect through $\ket{\psi_x}$. Since this operation is exact, we can use \cite[Theorem~10.3]{belovs2024taming} and turn it back into a bidirectional canonical transducer $R_x$ that reflects through $\ket{\psi_x}$, i.e., that acts as
        \begin{align*}
            \ket{\psi_x} \overset{R_x}{\rightsquigarrow} \ket{\psi_x}, & \qquad \text{with complexity} \qquad 2, \\
            \ket{\overline{\psi}_x} \overset{R_x}{\rightsquigarrow} -\ket{\overline{\psi}_x}, & \qquad \text{with complexity} \qquad 2,
        \end{align*}
        where we write
        \[\ket{\overline{\psi}_x} := \frac{\norm{\ket{\psi_1^x}}}{\norm{\ket{\psi_0^x}}} \ket{\psi_0^x}\ket{0} - \frac{\norm{\ket{\psi_0^x}}}{\norm{\ket{\psi_1^x}}} \ket{\psi_1^x}\ket{1}.\]
        We check that $\norm{\ket{\psi_1^x}} = \sqrt{\ell_{\delta,\zeta}(x)}$, and
        \[\frac{\ket{\psi_1^x}}{\norm{\ket{\psi_1^x}}} = \frac{1}{\sqrt{|((x + \delta B_d) \cap K)_\zeta|}} \sum_{y \in ((x + \delta B_d) \cap K)_\zeta} \ket{y} = \ket{*_x}.\]
        
        Next, we use \Cref{cor:ampl-ampl} to obtain bidirectional canonical transducers $C_x'$ and $U_x$ from $R_x$. $U_x$ has the right transduction action already. We obtain $C_x$ from $C_x'$ by sequential composition, as
        \[C_x : \ket{0} \mapsto \ket{N_x}\ket{0} \overset{O_K}{\rightsquigarrow} \ket{\psi_x} \overset{C_x'}{\rightsquigarrow} \ket{*_x}.\]

        Finally, for the transduction complexities, we observe that
        \[W(C_x,\ket{0}) \in O(1 + 1/\norm{\ket{\psi_1^x}}) = O(1 + 1/\sqrt{\ell_{\delta,\zeta}(x)}) \subseteq O(1/\sqrt{\ell_{\delta}(x)}),\]
        where we used \Cref{lem:discrete-local-conductance-proximity}. Using the same lemma, we obtain similarly for $U_x$ that
        \[W(U_x,\ket{*_x}) \in O(1/\norm{\psi_1^x}^2) = O(1 / \ell_{\delta,\zeta}(x)) \subseteq O(1 / \ell_{\delta}(x)).\]
        Finally, the transduction complexity for $\ket{\perp}$ follows directly from \Cref{cor:ampl-ampl}.
    \end{proof}

    The above construction provides all the ingredients to construct a transducer that reflects through the stationary distribution of the bipartite discretized speedy walk.

    \begin{theorem}[Speedy walk distribution reflection]
        \label{thm:speedy-walk-reflection}
        Let $d \in \N$, $R \geq 1$, $\beta \geq 0$, a step size parameter $0 < \delta \leq \min\{1/(4\beta R), 1/(4096\sqrt{2d\beta}), 1/(4096\sqrt{d})\}$ and $B_d \subseteq K \subseteq RB_d$ be a convex body, accessible through a membership oracle on a regular grid with sufficiently small spacing parameter $\zeta = \zeta(d,R,\beta,\delta) > 0$. Let $\pi_{\beta,\delta,\zeta}$ be the stationary distribution of the discretized speedy walk on $K_\zeta$. Then, we can implement a bidirectional canonical transducer $U$ that acts as
        \begin{align*}
            \ket{\pi_{\beta,\delta,\zeta}} \overset{U}{\rightsquigarrow} \ket{\pi_{\beta,\delta,\zeta}}, & \qquad \text{with complexity} \qquad O\left(1\right), \\
            \ket{\overline{\pi}} \overset{U}{\rightsquigarrow} -\ket{\overline{\pi}}, & \qquad \text{with complexity} \qquad O\left(\frac{d\norm{\ket{\overline{\pi}}}^2}{\beta\delta^2} + \sum_{x \in K_{\zeta}} \frac{|\overline{\pi}(x)|^2}{\pi_{\beta,\delta,\zeta}(x)\ell_{\delta,\zeta}(x)} \;\mathrm{d}x\right),
        \end{align*}
        for all $\ket{\overline{\pi}} \in \C^{K_\zeta}$ such that $\braket{\overline{\pi}}{\pi_{\beta,\delta}} = 0$, where we write $\overline{\pi} \in \C^{K_{\zeta}}$ such that
        \[\ket{\overline{\pi}} := \sum_{x \in K_{\zeta}} \frac{\overline{\pi}(x)}{\sqrt{\pi_{\beta,\delta,\zeta}(x)}} \ket{x}.\]
    \end{theorem}

    \begin{proof}
        Let $P$ be the probability transition matrix of the discretized speedy walk, defined in~\Cref{def:discretized-speedy}, and $P'$ be the probability transition matrix of the bipartite speedy walk, defined in~\Cref{def:bipartite-speedy}. \Cref{lem:metropolis-to-glauber-filter} shows that the stationary distribution of the bipartite version is $\pi' := (\pi_{\beta,\delta,\zeta} + \pi_{\beta,\delta,\zeta}P')/2$. 
        \Cref{lem:star-state-transducers} provides 
        transducers $C_x$ and $U_x$ that construct and reflect through the star state $\ket{*_x}$ at $x \in K_{\zeta}$. Thus, we can use \Cref{thm:refl-stationary} to construct a bidirectional canonical transducer $V$ that acts as
        \begin{align*}
            \ket{\pi'} \overset{V}{\rightsquigarrow} \ket{\pi'}, & \qquad \text{with complexity} \qquad O\left(\underset{x \sim \pi_{\beta,\delta,\zeta}}{\E} \left[\frac{1}{\ell_{\delta,\zeta}(x)}\right]\right), \\
            \ket{\overline{\pi}} \overset{V}{\rightsquigarrow} -\ket{\overline{\pi}}, & \qquad \text{with complexity} \qquad O\left(\frac{\norm{\ket{\overline{\pi}}}^2}{\gamma(P')} + \sum_{x \in K_{\zeta}} \frac{|\overline{\pi}(x)|^2}{\pi_{\beta,\delta,\zeta}(x)\ell_{\delta,\zeta}(x)}\right).
        \end{align*}
        Here, we used that $M = 1$, and $1/\sqrt{\ell_{\delta,\zeta}(x)} \leq 1/\ell_{\delta,\zeta}(x)$, and the fact that $\ket{\overline{\pi}}$ being orthogonal to $\ket{\pi}$ guarantees it is orthogonal to $\ket{\pi'}$ too.

        Next, we let $V' = ZVZ$, where $Z$ acts on as $I$ on $K$ and $-I$ on $E$. This conjugation now reflects through $\pi'' = (\pi_{\beta,\delta,\zeta} - \pi_{\beta,\delta,\zeta}P')/2$, with the same transduction complexities.

        Next, we observe that $V$ and $V'$ reflect around $\ket{\pi'}$ and $\ket{\pi''} = Z\ket{\pi'}$, respectively, which are orthogonal states. Thus, we can use \Cref{lem:combine-reflections} to combine these reflections into an operation $U$ that reflects through the vector sum of these subspaces. Since $\ket{\pi_{\beta,\delta,\zeta}} = (\ket{\pi'} + \ket{\pi''})/\sqrt{2}$, the transduction complexity of $\ket{\pi_{\beta,\delta,\zeta}}$ is upper bounded by the average of $(W(V, \ket{\pi'}) + W(V',\ket{\pi''}))/2 = W(V,\ket{\pi'})$. On the other hand, any state $\ket{\overline{\pi}} \in \C^{K_{\zeta}}$ that is orthogonal to $\ket{\pi_{\beta,\delta,\zeta}}$ is also orthogonal to both $\ket{\pi'}$ and $\ket{\pi''}$, and so the transduction complexity for the composition simply becomes $W(V,\ket{\overline{\pi}}) + W(V',\ket{\overline{\pi}})$. Thus,
        \begin{align*}
            \ket{\pi_{\beta,\delta,\zeta}} \overset{U}{\rightsquigarrow} \ket{\pi_{\beta,\delta,\zeta}}, & \qquad \text{with complexity} \qquad O\left(\underset{x \sim \pi_{\beta,\delta,\zeta}}{\E} \left[\frac{1}{\ell_{\delta,\zeta}(x)}\right]\right), \\
            \ket{\overline{\pi}} \overset{U}{\rightsquigarrow} -\ket{\overline{\pi}}, & \qquad \text{with complexity} \qquad O\left(\frac{\norm{\ket{\overline{\pi}}}^2}{\gamma(P')} + \sum_{x \in K_{\zeta}} \frac{|\overline{\pi}(x)|^2}{\pi_{\beta,\delta,\zeta}(x)\ell_{\delta,\zeta}(x)}\right).
        \end{align*}

        It remains to analyze the transduction complexities. To that end, recall from \Cref{thm:naive-discretization} and \Cref{lem:local-conductance-lipschitz} that both $\pi_{\beta,\delta}$ and $\ell_{\delta}$ are globally bounded and Lipschitz continuous on $K$. By \Cref{lem:lipschitz-properties}, so is their ratio. By the triangle inequality, we observe that
        \[\left|\underset{\pi_{\beta,\delta}}{\E}\left[\frac{1}{\ell_{\delta}(x)}\right] - \underset{\pi_{\beta,\delta,\zeta}}{\E}\left[\frac{1}{\ell_{\delta,\zeta}(x)}\right]\right| \leq \left|\int_K \frac{\pi_{\beta,\delta}(x)}{\ell_{\delta}(x)} \;\mathrm{d}x - \zeta^d\sum_{x \in K_{\zeta}} \frac{\pi_{\beta,\delta}(x)}{\ell_{\delta}(x)}\right| + \sum_{x \in K_{\zeta}} \left|\frac{\zeta^d\pi_{\beta,\delta}(x)}{\ell_{\delta}(x)} - \frac{\pi_{\beta,\delta,\zeta}(x)}{\ell_{\delta,\zeta}(x)}\right|.\]
        By \Cref{lem:smoothness}, the first term vanishes as $\zeta \downarrow 0$. Moreover, recall from \Cref{thm:naive-discretization} and \Cref{lem:discrete-local-conductance-proximity} that 
        \[\sup_{x \in K_{\zeta}} |\zeta^d\pi_{\beta,\delta}(x) - \pi_{\beta,\delta,\zeta}(x)| \downarrow 0, \qquad \text{and} \qquad \sup_{x \in K_{\zeta}} |\ell_{\delta,\zeta}(x) - \ell_{\delta}(x)| \downarrow 0, \qquad (\zeta \downarrow 0),\]
        respectively. Thus, for sufficiently small choices of $\zeta$, the expectation values are equal up to constants. We observe from \cite[Lemma~6.11]{CV18} that the continuous expectation value is constant, completing the proof of the first transduction complexity.
        
        Finally, 
        \Cref{lem:metropolis-to-glauber-filter} shows that $\gamma(P') \geq \gamma(P)/2$, which by the Cheeger inequalities (\Cref{thm:conductance-spectral-gap}) is at least $\Omega(\varphi(P)^2)$. The latter is 
        $\Omega(\delta^2\beta/d)$ by~\Cref{thm:discretized-speedy-conductance}. Plugging in these expressions yields the second transduction complexity.
    \end{proof}

    We finally apply this construction to specific orthogonal states $\ket{\overline{\pi}}$ that we will use later, namely those that come from probability distributions that are different from $\pi_{\beta,\delta}$.
    
    \begin{corollary}[Speedy walk distribution reflection for distributions]
        \label{cor:speedy-walk-refl-distrs}
        Let $d \in \N$, $R \geq 1$, $\beta \geq 0$, $0 < \delta \leq \min\{1/(4\beta R), 1/(4096\sqrt{2d\beta}), 1/(4096\sqrt{d})\}$, and $B_d \subseteq K \subseteq RB_d$ be a convex body accessible through a membership oracle on a regular grid with spacing parameter $\zeta > 0$. Let $M > 0$, $f : K \to [0,M]$ be a piecewise $L$-Lipschitz continuous density function of a probability distribution on an $h$-smooth partition of $K$, and let $\psi_\zeta := \pi_{f,\zeta}$ on $K_{\zeta}$. We write
        \[\ket{\overline{\pi}} := \ket{\psi_{\zeta}} - \braket{\psi_{\zeta}}{\pi_{\beta,\delta,\zeta}}\ket{\pi_{\beta,\delta,\zeta}}.\]
        Then, for a sufficiently small spacing parameter $\zeta = \zeta(d,R,\beta,\delta,M,L,h)$, we can implement a bidirectional canonical transducer $U$ that implements
        \begin{align*}
            \ket{\pi_{\beta,\delta,\zeta}} \overset{U}{\rightsquigarrow} \ket{\pi_{\beta,\delta,\zeta}}, \quad \text{with complexity} \quad & O\left(1\right), \\
            \ket{\overline{\pi}} \overset{U}{\rightsquigarrow} -\ket{\overline{\pi}}, \quad \text{with complexity} \quad & O\left(\frac{d}{\beta\delta^2} + \underset{x \sim \psi}{\E} \left[\frac{1}{\ell_{\delta}(x)}\right]\right).
        \end{align*}
    \end{corollary}

    \begin{proof}
        For all $x \in K_{\zeta}$, we have
        \[\overline{\pi}(x) := \sqrt{\psi_{\zeta}(x)\pi_{\beta,\delta,\zeta}(x)} - \braket{\psi_{\zeta}}{\pi_{\beta,\delta,\zeta}}\pi_{\beta,\delta,\zeta}(x), \qquad \text{and therefore} \qquad \frac{|\overline{\pi}(x)|^2}{\pi_{\beta,\delta,\zeta}(x)} \leq \psi_{\zeta}(x) + \pi_{\beta,\delta,\zeta}(x).\]
        Thus, when we plug these quantities into \Cref{thm:speedy-walk-reflection}, we readily obtain the constant transduction complexity on $\ket{\pi_{\beta,\delta,\zeta}}$. For the transduction complexity on $\ket{\overline{\pi}}$, we use the same argument as in \Cref{thm:speedy-walk-reflection} to show that we can replace the discrete summation over $\psi_{\zeta}$ into an expectation over its continuous counterpart.
    \end{proof}

    \subsection{Annealing of speedy walk distributions} \label{sec:annealing-speedy}

    In this section, we consider clipped versions of the convex body $K$, i.e., we consider $K \cap rB_d$ for some $r \geq 1$. Note that $K \cap rB_d$ is again a convex body, so we can use the transducers constructed in the previous section on this convex body directly. To distinguish between their stationary distributions, we write $\pi_{\beta,\delta,r}$ for the speedy walk distribution $\pi_{\beta,\delta}$ on $K \cap rB_d$. Similarly, we denote their discretizations by $\pi_{\beta,\delta,r,\zeta}$ on $(K \cap rB_d)_\zeta$.
    
    We develop a transducer that maps $\ket{\pi_{\beta,\delta,r,\zeta}}$ to $\ket{\pi_{\beta',\delta',r',\zeta}}$, for different choices of inverse temperatures~$\beta$ and $\beta'$, step sizes $\delta$ and $\delta'$, and clipping radii $r$ and $r'$, and for a sufficiently small grid spacing parameter $\zeta > 0$. Informally, if each of the parameters are ``close to each other'', then the distributions~$\pi_{\beta,\delta,r}$ and $\pi_{\beta',\delta',r'}$ are sufficiently similar such that the overlap of these states is at least constant. This allows us to use the reflection oracles constructed in the previous section to map one to the other.

    To characterize what it means for values $\beta$ and $\beta'$ to be considered ``close'', we investigate how much the partition function changes between $\beta$ and $\beta'$. To that end, we will use the log-partition function. 

    \begin{definition}[Log-partition function] \label{def:log-part}
        Let $d \in \N$, and $K \subseteq \R^d$ be a convex body. Then, the \textit{log-partition function} is $L : (-\infty,\infty] \to \R$, defined as $L(\beta) = \ln(Z(\beta))$.
    \end{definition}

    The log-partition function is directly connected to the random variable $\norm{x}^2$ where $x \sim \pi_\beta$: we can express its expectation value and variance in terms of derivatives of the log-partition function.

    \begin{lemma}
        \label{lem:log-part-derivatives}
        Let $d \in \N$, $R \geq 1$ and $B_d \subseteq K \subseteq RB_d$ be a convex body. Let $\beta \in \R$ and $H(x) = \norm{x}^2$. Then,
        \[L'(\beta) = -\underset{x \sim \pi_{\beta}}{\E} \left[H(x)\right], \qquad \text{and} \qquad L''(\beta) = \underset{x \sim \pi_{\beta}}{\Var} \left[H(x)\right].\]
    \end{lemma}

    \begin{proof}
        We observe that
        \[-L'(\beta) = -\frac{Z'(\beta)}{Z(\beta)} = \int_K \frac{-\frac{\partial}{\partial \beta} \exp(-\beta H(x))}{Z(\beta)} \;\mathrm{d}x = \int_K H(x) \cdot \frac{\exp(-\beta H(x))}{Z(\beta)} \;\mathrm{d}x = \underset{x \sim \pi_{\beta}}{\E} \left[H(x)\right].\]
        Similarly, we have
        \begin{align*}
            L''(\beta) &= \frac{Z''(\beta)}{Z(\beta)} - \left(\frac{Z'(\beta)}{Z(\beta)}\right)^2 = \int_K H(x)^2 \cdot \frac{\exp(-\beta H(x))}{Z(\beta)} \;\mathrm{d}x - \underset{x \sim \pi_{\beta}}{\E} \left[H(x)\right]^2 \\
            &= \underset{x \sim \pi_{\beta}}{\E} \left[H(x)^2\right] - \underset{x \sim \pi_{\beta}}{\E} \left[H(x)\right]^2 = \underset{x \sim \pi_{\beta}}{\Var} \left[H(x)\right].\qedhere
        \end{align*}
    \end{proof}

    The mean and variance of the norm-squared of a point sampled from the Gibbs distribution will show up at various places in the analysis of the volume estimation algorithm. We therefore prove rigorous bounds on these quantities and record them here for future use.

    \begin{theorem} \label{thm:variance-bounds}
        Let $d \in \N$, $R \geq 1$ and $B_d \subseteq K \subseteq RB_d$ convex. Let $H(x) = \norm{x}^2$. Then, for all $\beta \geq 0$, we have
        \begin{align*}
            \Omega\left(\min\left\{\frac{d}{\beta}, 1\right\}\right) &\ni \underset{x \sim \pi_{\beta}}{\E} \left[H(x)\right] \leq \min\left\{\frac{d}{2\beta}, R^2\right\}. \\
            \Omega\left(\min\left\{\frac{d}{\beta^2},\frac{1}{d^2}\right\}\right) &\ni \underset{x \sim \pi_{\beta}}{\Var} \left[H(x)\right] \leq \min\left\{\frac{d}{\beta^2}, \frac{2R^2}{\beta}, \frac{R^4}{4}\right\}.
        \end{align*}
    \end{theorem}

    \begin{proof}
        See \Cref{sec:app}.
    \end{proof}

    Next, we use the characterization of these statistical quantities from \Cref{thm:variance-bounds} to bound the overlap between the Gibbs distributions at different inverse temperatures.

    \begin{lemma}
        \label{lem:gibbs-overlap}
        Let $d \in \N$, $R \geq 1$ and $B_d \subseteq K \subseteq RB_d$ be a convex body. For all $\beta,\beta' \geq 0$, let $\overline{\beta} = (\beta + \beta')/2$. Then,
        \[|\braket{\pi_{\beta}}{\pi_{\beta'}}| = \frac{Z(\overline{\beta})}{\sqrt{Z(\beta)Z(\beta')}}.\]
        Moreover, if $\beta' \leq \beta$ and $\beta - \beta' \leq \max\{1/(8\sqrt{d}), 1/(8R\sqrt{\beta})\}\beta$, then the latter is at least $15/16$.
    \end{lemma}

    \begin{proof}
        For the first claim, we observe that
        \[|\braket{\pi_{\beta}}{\pi_{\beta'}}| = \frac{\int_K \sqrt{\exp(-\beta\norm{x}^2)} \cdot \sqrt{\exp(-\beta'\norm{x}^2)} \;\mathrm{d}x}{\sqrt{Z(\beta)Z(\beta')}} = \frac{Z(\overline{\beta})}{\sqrt{Z(\beta)Z(\beta')}}.\]
        Next, we use the mean value theorem to rewrite the right-hand side, from which we deduce that there exists a $\xi \in [\beta',\beta]$ such that
        \[\ln\left[\frac{\sqrt{Z(\beta)Z(\beta')}}{Z(\overline{\beta})}\right] = \frac12\left(L(\beta) + L(\beta') - 2L(\overline{\beta})\right) = \frac12L''(\xi)\left(\frac{\beta - \beta'}{2}\right)^2 = \underset{x \sim \pi_{\xi}}{\Var}[\norm{x}^2] \cdot \frac{(\beta - \beta')^2}{8}\]
        Now, if $\beta \geq 1/R^2$, then $\beta/\beta' \leq 2$, and we use the bounds from \Cref{thm:variance-bounds} to find that
        \[\underset{x \sim \pi_{\xi}}{\Var}[\norm{x}^2] \cdot \frac{(\beta - \beta')^2}{8} \leq \min\left\{\frac{4d}{\beta^2}, \frac{4R^2}{\beta}\right\} \cdot \max\left\{\frac{1}{64d}, \frac{1}{64R^2\beta}\right\} \cdot \frac{\beta^2}{8} \leq \frac{1}{16}.\]
        On the other hand, if $\beta \leq 1/R^2$, then we use that $\beta' \geq 0$, and so again using \Cref{thm:variance-bounds}, we find that
        \[\underset{x \sim \pi_{\xi}}{\Var}[\norm{x}^2] \cdot \frac{(\beta - \beta')^2}{8} \leq \frac{R^4}{4} \cdot \frac{\beta^2}{8} \leq \frac{R^4}{4} \cdot \frac{1}{8R^4} < \frac{1}{16}.\]
        Thus, in both cases, we obtain that $|\braket{\pi_{\beta}}{\pi_{\beta'}}| \geq \exp(-1/16) > 15/16$.
    \end{proof}

    We upper bound the total variation distance between the clipped and non-clipped version of the Gibbs distribution. Informally, for any inverse temperature $\beta > 0$, a clipping radius $r \sim \sqrt{d/\beta}$ ensures that the Gibbs distribution is mostly supported on a ball of radius $r$.\footnote{We note that \cite[Corollary~5.2]{CV18} claims that the same result can be obtained with an $\exp(-t^2)$-scaling, but this statement appears too strong for large values of $t$.}

    \begin{lemma}
        \label{lem:clipping-perturbation}
        Let $d \in \N$, and $B_d \subseteq K \subseteq \R^d$ be a convex body. There exists a universal constant $c>0$ such that the following holds. 
        For any inverse temperature $\beta \in [0,\infty]$, parameter $t > 0$ and clipping radius $r \geq \sqrt{(2d+3t)/(2\beta)}$, we have
        \[\norm{\pi_{\beta,r} - \pi_{\beta}}_{\mathrm{TV}} \leq \exp(-t).\]
    \end{lemma}

    \begin{proof}
        We observe that
        \[\norm{\pi_{\beta,r} - \pi_{\beta}}_{\mathrm{TV}} = \int_{K \setminus rB_d} \frac{\exp(-\beta\norm{x}^2)}{Z(\beta)} \;\mathrm{d}x = \underset{x \sim \pi_{\beta}}{\P} \left[\norm{x} \geq r\right].\]
        The latter probability is over $\pi_{\beta}$, i.e., a truncated Gaussian, and so it is upper bounded by its untruncated version. That is,
        \[\underset{x \sim \pi_{\beta}}{\P} [\norm{x} \geq r] \leq \underset{x \sim \mathcal{N}(0,I_d/(2\beta))}{\P} [\norm{x} \geq r] = \underset{x \sim \mathcal{N}(0,I_d)}{\P} [\norm{x}^2 \geq 2\beta r^2] = \underset{s \sim \chi_d^2}{\P} \left[s \geq 2\beta r^2\right],\]
        where $\chi_d^2$ is the $\chi^2$-distribution with $d$ degrees of freedom. Finally, we plug in $r^2 \geq (2d+3t)/(2\beta)$ and use \cite[Equation~4.3]{laurent2000adaptive} to upper bound this as
        \[\underset{s \sim \chi_d^2}{\P} \left[s \geq 2\beta r^2\right] \leq \underset{s \sim \chi_d^2}{\P} \left[s \geq 2d+3t\right] \leq \underset{s \sim \chi_d^2}{\P} \left[s \geq d + 2\sqrt{dt} + 2t\right] \leq \exp(-t).\qedhere\]
    \end{proof}

    We can now combine this proximity result with the perturbations arising from the finite step sizes and clipping parameters. This is the objective of the following lemma.

    \begin{lemma}
        \label{lem:overlap}
        Let $d \in \N$, $R \geq 1$ and $B_d \subseteq K \subseteq RB_d$ be a convex body. Let $0 \leq \beta' \leq \beta$, such that $\beta - \beta' \leq \max\{1/(8\sqrt{d}),1/(8R\sqrt{\beta})\}\beta$. Let $r \geq \sqrt{d/\beta}\xi$, and $r' \geq \sqrt{d/\beta'}\xi$, where $\xi = \sqrt{1 + 3\ln(20)/(2d)}$. Let $0 < \delta \leq \min\{1,1/\sqrt{2\beta}\}/(2^{18}\sqrt{d})$ and $0 < \delta' \leq \min\{1,1/\sqrt{2\beta'}\}/(2^{18}\sqrt{d})$. Then,
        \[|\braket{\pi_{\beta,\delta,r}}{\pi_{\beta',\delta',r'}}| \geq \frac{1}{22}.\]
    \end{lemma}

    \begin{proof}
        First, we observe from \Cref{lem:gibbs-overlap} that 
        \[\norm{\pi_{\beta} - \pi_{\beta'}}_{\mathrm{TV}} \leq \sqrt{1 - |\braket{\pi_{\beta}}{\pi_{\beta'}}|^2} \leq \sqrt{1-\left(\frac{15}{16}\right)^2} \leq \frac{1}{2\sqrt{2}}.\]
        By the triangle inequality, and using \Cref{lem:clipping-perturbation,lem:avg-local-conductance}, we observe that
        \begin{align*}
            \norm{\pi_{\beta,\delta,r} - \pi_{\beta',\delta',r'}}_{\mathrm{TV}} &\leq \norm{\pi_{\beta,\delta,r} - \pi_{\beta,r}}_{\mathrm{TV}} + \norm{\pi_{\beta,r} - \pi_{\beta}}_{\mathrm{TV}} + \norm{\pi_\beta - \pi_{\beta'}}_{\mathrm{TV}} \\
            &\qquad + \norm{\pi_{\beta'} - \pi_{\beta',r'}}_{\mathrm{TV}} + \norm{\pi_{\beta',r'} - \pi_{\beta',\delta',r'}}_{\mathrm{TV}} \\
            &\leq \frac14 + \frac{1}{20} + \frac{1}{2\sqrt{2}} + \frac{1}{20} + \frac14 < \frac{21}{22}.
        \end{align*}
        Thus, we find that
        \[1 - |\braket{\pi_{\beta,\delta,r}}{\pi_{\beta',\delta',r'}}| \leq \norm{\pi_{\beta,\delta,r} - \pi_{\beta',\delta',r'}}_{\mathrm{TV}} \leq \frac{21}{22}.\qedhere\]
    \end{proof}

    Now that we know when $\ket{\pi_{\beta,\delta,r}}$ and $\ket{\pi_{\beta',\delta',r'}}$ have constant overlap, we can use \Cref{thm:two-state-ampl} to map one to the other with constant overhead. For this construction, we need to be able to reflect through both of these states. 
    To analyze the cost of these reflections, we first provide two geometrical lemmas.

    \begin{restatable}{lemma}{volumebound}
        \label{lem:volume-bound}
        Let $d \in \N$, and $B_d \subseteq K \subseteq \R^d$ be a convex body. Let $r \geq 1$, and $\delta \leq 1/\sqrt{d}$. Then, for all $x \in K \cap rB_d$,
        \[\frac{\Vol((x + \delta B_d) \cap K)}{\Vol((x + \delta B_d) \cap K \cap rB_d)} \in O(1).\]
    \end{restatable}

    \begin{proof}
        See \Cref{app:volume-bound}.
    \end{proof}

    \begin{lemma}
        \label{lem:step-size-update-warmness}
        Let $d \in \N$, $K \subseteq \R^d$ be a convex body. Let $0 < \delta < \delta'$ and $0 < r < r'$. Then, for any $x \in K \cap rB_d$,
        \[\frac{\ell_{\delta,r}(x)}{\ell_{\delta',r'}(x)} \leq \left(\frac{\delta'}{\delta}\right)^d.\]
    \end{lemma}

    \begin{proof}
        We observe that
        \[(x + \delta B_d) \cap K \cap rB_d \subseteq (x + \delta'B_d) \cap K \cap r'B_d,\]
        and so
        \[\frac{\ell_{\delta,r}(x)}{\ell_{\delta',r'}(x)} = \frac{\Vol((x + \delta B_d) \cap K \cap rB_d)\Vol(\delta'B_d)}{\Vol((x + \delta'B_d) \cap K \cap r'B_d)\Vol(\delta B_d)} \leq \frac{\Vol(\delta'B_d)}{\Vol(\delta B_d)} = \left(\frac{\delta'}{\delta}\right)^d.\qedhere\]
    \end{proof}

    Note that the latter lemma can be used to bound the warmness of two speedy walk distributions. Indeed, for speedy walk distributions with parameters $\beta \geq \beta' \geq 0$, $0 < \delta < \delta'$ and $0 < r < r'$, we have
    \begin{equation}
        \label{eq:speedy-walk-warmness}
        \sup_{x \in K \cap rB_d} \frac{\ell_{\delta,r}(x)}{\ell_{\delta',r'}(x)} \cdot \underbrace{\frac{\exp(-\beta\norm{x}^2)}{\exp(-\beta'\norm{x}^2)}}_{\leq 1} \leq \left(\frac{\delta'}{\delta}\right)^d.
    \end{equation}
    Thus, the density functions that describe two consecutive speedy walk distributions in the annealing procedure have warmness that scales as $(\delta'/\delta)^d$. Cousins and Vempala~\cite{CV18} implicitly require that the warmness of two consecutive speedy walk distributions is constant. Indeed, in their approach $\delta'/\delta \in 1 + O(1/d)$, so one can justify this claim using \Cref{eq:speedy-walk-warmness} and an appropriate analysis of the normalization factors in these distributions.
    \footnote{Cousins and Vempala use the result from \cite{lovasz1993random} that characterizes mixing time from a warm start and apply it to the speedy walk. However, they only prove warmness of the distributions arising from the ball walk. It is our understanding one would additionally need warmness of consecutive speedy walk distributions, which we prove in \Cref{lem:step-size-update-warmness}.}

    These two geometrical lemmas allow us to upper bound the transduction complexities that show up in annealing from one speedy walk distribution to another.

    \begin{lemma}
        \label{lem:annealing-average-local-conductance}
        Let $d \in \N$, and $B_d \subseteq K \subseteq \R^d$ be a convex body. Let $\beta,\beta' \geq 0$, $1 \leq r \leq r'$, $0 < \delta \leq \min\{1,1/\sqrt{2\beta}\}/(4096\sqrt{d})$, and $\delta \leq \delta' \leq \min\{1,1/\sqrt{2\beta'}\}/(4096\sqrt{d})$. Then,
        \[\underset{x \sim \pi_{\beta,\delta,r}}{\E} \left[\frac{1}{\ell_{\delta',r'}(x)}\right] \in O\left(\left(\frac{\delta'}{\delta}\right)^d\right), \qquad \text{and} \qquad \underset{x \sim \pi_{\beta',\delta',r'}}{\E} \left[\frac{\mathbbm{1}_{x \in rB_d}}{\ell_{\delta,r}(x)}\right] \in O(1).\]
    \end{lemma}

    \begin{proof}
        Let $x \in K \cap rB_d$. For the first expectation value, we use \Cref{lem:step-size-update-warmness} to upper bound it as
        \[\underset{x \sim \pi_{\beta,\delta,r}}{\E} \left[\frac{1}{\ell_{\delta',r'}(x)}\right] = \frac{\int_{K \cap rB_d} \frac{\ell_{\delta,r}(x)}{\ell_{\delta',r'}(x)} \exp(-\beta\norm{x}^2) \;\mathrm{d}x}{\int_{K \cap rB_d} \ell_{\delta,r}(x)\exp(-\beta\norm{x}^2) \;\mathrm{d}x} \leq \left(\frac{\delta'}{\delta}\right)^d \cdot \frac{1}{\lambda} \in O\left(\left(\frac{\delta'}{\delta}\right)^d\right),\]
        where we used \cite[Lemma~6.11]{CV18} to argue that $\lambda \geq 1/2$.

        For the second expectation value, we use the inclusions 
        \[(x + \delta'B_d) \cap K \cap r'B_d \subseteq (x + \delta'B_d) \cap K \subseteq x + \frac{\delta'}{\delta}\left[\delta B_d \cap (K - x)\right],\]
        where the last inclusion uses the convexity of $K$.
        We then combine this with \Cref{lem:volume-bound} to obtain that
        \begin{align*}
            \frac{\ell_{\delta',r'}(x)}{\ell_{\delta,r}(x)} &= \frac{\Vol((x + \delta'B_d) \cap K \cap r'B_d)\Vol(\delta B_d)}{\Vol((x + \delta B_d) \cap K \cap rB_d)\Vol(\delta'B_d)} \leq \frac{\Vol\left(\frac{\delta'}{\delta} [\delta B_d \cap (K-x)]\right)}{\Vol((x + \delta B_d) \cap K \cap rB_d)} \cdot \left(\frac{\delta}{\delta'}\right)^d \\
            &= \frac{\Vol((x + \delta B_d) \cap K)}{\Vol((x + \delta B_d) \cap K \cap rB_d)} \in O(1),
        \end{align*}
        and so, by a similar argument relying on \cite[Lemma~6.11]{CV18}, we obtain that
        \[\underset{x \sim \pi_{\beta',\delta',r'}}{\E} \left[\frac{\mathbbm{1}_{x \in rB_d}}{\ell_{\delta,r}(x)}\right] = \frac{\int_{K \cap rB_d} \frac{\ell_{\delta',r'}(x)}{\ell_{\delta,r}(x)} \exp(-\beta'\norm{x}^2) \;\mathrm{d}x}{\int_{K \cap r'B_d} \ell_{\delta',r'}(x)\exp(-\beta'\norm{x}^2) \;\mathrm{d}x} \in O\left(1\right).\qedhere\]
    \end{proof}

    Now, we can put all the ingredients together and implement a transducer that maps $\ket{\pi_{\beta,\delta,r,\zeta}}$ to $\ket{\pi_{\beta',\delta',r',\zeta}}$.

    \begin{theorem}
        \label{thm:anneal}
        Let $d \in \N$, $R \geq 1$, $0 \leq \beta' \leq \beta$, such that $\beta - \beta' \leq \max\{1/(8\sqrt{d}),1/(8R\sqrt{\beta})\}\beta$. Let $r \geq \max\{1,\sqrt{d/\beta}\xi\}$, and $r' \geq \max\{r,\sqrt{d/\beta'}\xi\}$, where $\xi = \sqrt{1 + 3\ln(20)/(2d)}$. Next, we let $0 < \delta \leq \min\{1/(4\beta r), 1/(2^{18}\sqrt{2d\beta}), 1/(2^{18}\sqrt{d})\}$ and $\delta \leq \delta' \leq \min\{1/(4\beta'r'), 1/(2^{18}\sqrt{2d\beta'}), 1/(2^{18}\sqrt{d})\}$. Let $B_d \subseteq K \subseteq RB_d$ be a convex body, accessible through a membership oracle on a grid with sufficiently small spacing parameter $\zeta = \zeta(d,R,\beta,\beta',\delta,\delta') > 0$. Then, we can construct a bidirectional canonical transducer $A$ that maps
        \[\ket{\pi_{\beta,\delta,r,\zeta}} \overset{A}{\rightsquigarrow} \ket{\pi_{\beta',\delta',r',\zeta}}, \qquad \text{with complexity} \qquad O\left(\sqrt{\frac{d}{\beta'\delta^2} + \left(\frac{\delta'}{\delta}\right)^d}\right).\]
    \end{theorem}

    \begin{proof}
        We define the (unnormalized) states
        \[\ket{\overline{\pi}} := \ket{\pi_{\beta',\delta',r',\zeta}} - \braket{\pi_{\beta,\delta,r,\zeta}}{\pi_{\beta',\delta',r',\zeta}}\ket{\pi_{\beta,\delta,r,\zeta}}, \quad \text{and} \quad \ket{\overline{\pi}'} := \ket{\pi_{\beta,\delta,r,\zeta}} - \braket{\pi_{\beta',\delta',r',\zeta}}{\pi_{\beta,\delta,r,\zeta}}\ket{\pi_{\beta',\delta',r',\zeta}}.\]
        Next, we use \Cref{cor:speedy-walk-refl-distrs} to obtain a transducer $U'$ that acts as
        \begin{align*}
            \ket{\pi_{\beta',\delta',r',\zeta}} \overset{U'}{\rightsquigarrow} \ket{\pi_{\beta',\delta',r',\zeta}}, \quad \text{with complexity} \quad & O(1), \\
            \ket{\overline{\pi}'} \overset{U'}{\rightsquigarrow} -\ket{\overline{\pi}'}, \quad \text{with complexity} \quad & O\left(\frac{d}{\beta'(\delta')^2} + \underset{x \sim \pi_{\beta.\delta,r}}{\E} \left[\frac{1}{\ell_{\delta',r'}(x)}\right]\right).
        \end{align*}
        Finally, we apply \Cref{lem:annealing-average-local-conductance} to the final transduction complexity, to obtain a term of $(\delta'/\delta)^d$.

        Conversely, we observe that the support of $\pi_{\beta,\delta,r,\zeta}$ is contained in $(rB_d)_{\zeta}$, whereas the support of $\ket{\overline \pi}$ can have overlap with $(r' B_d \setminus rB_d)_{\zeta}$. On the part of $\ket{\overline \pi}$ that is supported on $(r' B_d \setminus rB_d)_{\zeta}$, 
        we simply apply $-I$, which can be done with transduction complexity $0$, and we use \Cref{cor:speedy-walk-refl-distrs} to construct a transducer that reflects through $\ket{\pi_{\beta,\delta,r,\zeta}}$ on $(K \cap rB_d)_{\zeta}$. We compose these together using parallel composition (\Cref{thm:parallel-composition}) to obtain a transducer $U$ that acts as
        \begin{align*}
            \ket{\pi_{\beta,\delta,r,\zeta}} \overset{U}{\rightsquigarrow} \ket{\pi_{\beta,\delta,r,\zeta}}, \quad \text{with complexity} \quad & O\left(1\right), \\
            \ket{\overline{\pi}} \overset{U}{\rightsquigarrow} -\ket{\overline{\pi}}, \quad \text{with complexity} \quad & O\left(\frac{d}{\beta\delta^2} + \underset{x \sim \pi_{\beta',\delta',r'}}{\E} \left[\frac{\mathbbm{1}_{x \in rB_d}}{\ell_{\delta,r}(x)}\right]\right).
        \end{align*}
        Now, we use the second statement of \Cref{lem:annealing-average-local-conductance} to upper bound the final term with $O(1)$.

        Next, observe that both $U'$ and $U$ are reflection transducers, i.e., their transduction action is a reflection operator. This means we can use \Cref{thm:balancing-plus-minus} to rebalance their complexities with an input-independent constant. Since the complexities only depend on parameters of the walk, i.e., $\beta$, $\beta'$, $\delta$, $\delta'$, $r$, and $r'$, after rebalancing we obtain $V$ and $V'$ that act as
        \begin{align*}
            \ket{\pi_{\beta',\delta',r',\zeta}} \overset{V'}{\rightsquigarrow} \ket{\pi_{\beta',\delta',r',\zeta}}, \quad \text{with complexity} \quad & O\left(\sqrt{\frac{d}{\beta'(\delta')^2} + \left(\frac{\delta'}{\delta}\right)^d}\right), \\
            \ket{\overline{\pi}'} \overset{V'}{\rightsquigarrow} -\ket{\overline{\pi}'}, \quad \text{with complexity} \quad & O\left(\sqrt{\frac{d}{\beta'(\delta')^2} + \left(\frac{\delta'}{\delta}\right)^d}\right), \\
            \ket{\pi_{\beta,\delta,r,\zeta}} \overset{V}{\rightsquigarrow} \ket{\pi_{\beta,\delta,r,\zeta}}, \quad \text{with complexity} \quad & O\left(\sqrt{\frac{d}{\beta\delta^2} + 1}\right), \\
            \ket{\overline{\pi}} \overset{V}{\rightsquigarrow} -\ket{\overline{\pi}}, \quad \text{with complexity} \quad & O\left(\sqrt{\frac{d}{\beta\delta^2} + 1}\right).
        \end{align*}

        Finally, we plug $V$ and $V'$ into \Cref{thm:two-state-ampl} to obtain the transducer $A$ that implements the mapping $\ket{\pi_{\beta,\delta,r}} \rightsquigarrow \ket{\pi_{\beta',\delta',r'}}$. The final complexity then becomes the sum of the above expressions, divided by the overlap between the two states. To evaluate this, observe that $\pi_{\beta,\delta,r}$ and $\pi_{\beta',\delta',r'}$ are both Lipschitz and globally bounded on $K \cap rB_d$. Thus, we use a similar argument as in \Cref{thm:speedy-walk-reflection} to argue that
        \begin{align*}
            &\sup_K |\braket{\pi_{\beta,\delta,r}}{\pi_{\beta',\delta',r'}} - \braket{\pi_{\beta,\delta,r,\zeta}}{\pi_{\beta',\delta',r',\zeta}}| \\
            &\quad = \sup_K \left|\int_{K \cap rB_d} \sqrt{\pi_{\beta,\delta,r}(x)\pi_{\beta',\delta',r'}(x)} \;\mathrm{d}x - \sum_{x \in (K \cap rB_d)_{\zeta}} \sqrt{\pi_{\beta,\delta,r,\zeta}(x)\pi_{\beta',\delta',r',\zeta}(x)}\right| \downarrow 0, \qquad (\zeta \downarrow 0).
        \end{align*}
        As such, for a sufficiently small choice of $\zeta = \zeta(d,R,\beta,\beta',\delta,\delta')$, we obtain that the two are equal up to constants. Moreover, inner product of the continuous distributions is constant due to \Cref{lem:overlap}. Finally, since $\beta' \leq \beta$, and $\delta' \geq \delta$, the complexity statement follows.
    \end{proof}

    The above theorem 
    essentially tells us how to efficiently map $\ket{\pi_{\beta,\delta,r,\zeta}}$ to $\ket{\pi_{\beta',\delta',r',\zeta}}$ whenever the parameters are close to each other. The core idea of annealing is to combine a sequence of such operations 
    to go from any parameter setting to any other. Especially whenever the gap between $\delta$ and $\delta'$ is big, it is beneficial to do the annealing in multiple steps. This is the objective of the following theorem.

    \begin{theorem}
        \label{cor:anneal-small-steps}
        Let $d \in \N$, $R \geq 1$, $0 \leq \beta' \leq \beta$, such that $\beta - \beta' \leq \max\{1/(8\sqrt{d}),1/(8R\sqrt{\beta})\}\beta$. Let $r \geq \max\{1,\sqrt{d/\beta}\xi\}$, and $r' \geq \max\{r,\sqrt{d/\beta'}\xi\}$, where $\xi = \sqrt{1 + 3\ln(20)/(2d)}$. Next, we let $0 < \delta \leq \min\{1/(4\beta r), 1/(2^{18}\sqrt{2d\beta}), 1/(2^{18}\sqrt{d})\}$ and $\delta \leq \delta' \leq \min\{1/(4\beta'r'), 1/(2^{18}\sqrt{2d\beta'}), 1/(2^{18}\sqrt{d})\}$. Let $B_d \subseteq K \subseteq RB_d$ be a convex body, accessible through a membership oracle on a grid with sufficiently small spacing parameter $\zeta = \zeta(d,R,\beta,\beta',\delta,\delta') > 0$. Then, we can implement a bidirectional canonical transducer $A$ that maps
        \[\ket{\pi_{\beta,\delta,r,\zeta}} \overset{A}{\rightsquigarrow} \ket{\pi_{\beta',\delta',r',\zeta}}, \qquad \text{with complexity} \qquad O\left(\sqrt{\frac{d}{\beta'\delta^2}} \cdot \left(1 + d\ln\left(\frac{\delta'}{\delta}\right)\right)\right).\]
    \end{theorem}

    \begin{proof}
        We write $C = d/\sqrt{\beta'\delta^2}$, and we observe that $C \geq \sqrt{2}$. Next, we set $\delta_0 := \delta$, and we compute a sequence $\delta_{j+1} = C^{1/d}\delta_j$, and we let $k$ be the smallest integer such that $\delta_k \geq \delta'$. Now, we first use \Cref{thm:anneal} to map $\ket{\pi_{\beta,\delta,r,\zeta}}$ to $\ket{\pi_{\beta',\delta_0,r'}}$, and then we employ the same theorem iteratively to map $\ket{\pi_{\beta',\delta_j,r',\zeta}}$ to $\ket{\pi_{\beta',\delta_{j+1},r',\zeta}}$. In the final step, we map $\ket{\pi_{\beta',\delta_{k-1},r',\zeta}}$ to $\ket{\pi_{\beta',\delta',r',\zeta}}$. We observe that the total complexity is
        \[O\left(\sqrt{\frac{d}{\beta'\delta^2} + \left(C^{\frac1d}\right)^d}(1 + k)\right) \subseteq O\left(\sqrt{\frac{d}{\beta'\delta^2}}(1 + k)\right),\]
        and so it remains to upper bound the value of $k$. To that end, observe that
        \[C^{\frac{k-1}{d}} \leq \frac{\delta'}{\delta} \Leftrightarrow \frac{k-1}{d} \leq \frac{\ln\left(\frac{\delta'}{\delta}\right)}{\ln(C)}, \qquad \text{which implies} \qquad k \in O\left(d\ln\left(\frac{\delta'}{\delta}\right)\right).\qedhere\]
    \end{proof}

    \subsection{Mapping the speedy walk distribution to the Gibbs distribution}

    Now that we know how to prepare and reflect through the stationary distribution of the speedy walk, we turn to preparing and reflecting through an approximation of the Gibbs distribution. We follow the approach outlined by Cousins and Vempala~\cite[Section~6.4]{CV18}. The idea is to use a rejection-sampling routine that removes a tiny outer shell of the convex body, and scales up the inner part, on which the local conductance is essentially uniform.

    We introduce the following notion of a clipped version of the Gibbs distribution, where we post-select on the norm of the point $x$ being below some constant $r > 0$.

    \begin{definition}
        Let $d \in \N$ and $K \subseteq \R^d$ be a convex body. For any inverse temperature $\beta \in (-\infty,\infty]$ and clipping radius $r > 0$, we define the Gibbs distribution clipped at radius $r$, denoted by $\pi_{\beta,r}$ to be the Gibbs distribution for the convex body $K \cap rB_d$ at inverse temperature $\beta$.
    \end{definition}

    Next we introduce the rejection sampling routine.

    \begin{theorem}
        \label{thm:speedy-to-ball}
        Let $d \in \N$, $R \geq 1$, $\beta \geq 0$, $r \geq 1$, $\eta > 0$ and a stepsize parameter $0 < \delta \leq \min\{1/(4\beta r), 1/(4096\sqrt{2d\beta}), 1/(4096\sqrt{d\ln(20d/\eta)})\}$. Let $B_d \subseteq K \subseteq RB_d$ be a convex body, accessible through membership oracles $O_{K,\zeta}$ and $O_{K,\zeta/\gamma}$ with a sufficiently small spacing parameter $\zeta = \zeta(d,R,\beta,\delta) > 0$. For all $x \in K \cap rB_d$, let
        \[a(x) := \frac{\exp(-\frac{\beta}{\gamma^2}\norm{x}^2)}{\exp(-\beta\norm{x}^2)}\mathbbm{1}_{\frac{x}{\gamma} \in K \cap rB_d}, \qquad \text{where} \qquad \gamma := 1 - \frac{1}{2d}.\]
        Let $\pi_{\beta,r,\zeta/\gamma}$ be the discretization of $\pi_{\beta,r}$ on $(K \cap rB_d)_{\zeta/\gamma}$, and let $\widetilde{\pi}_{\beta,r,\zeta/\gamma}$ be a distribution on $(K \cap rB_d)_{\zeta/\gamma}$ defined as
        \[\widetilde{\pi}_{\beta,r,\zeta/\gamma}(x) := \frac{\pi_{\beta,\delta,r,\zeta}(\gamma x)a(\gamma x)}{A_{\zeta}}, \qquad \text{and} \qquad A_{\zeta} := \underset{x \sim \pi_{\beta,\delta,r,\zeta}}{\E} [a(x)].\]
        Then, $\norm{\widetilde{\pi}_{\beta,r,\zeta/\gamma} - \pi_{\beta,r,\zeta/\gamma}}_{\mathrm{TV}} \leq \eta$. Let $\ket{\overline{\psi}} \in \C^{(K \cap rB_d)_{\zeta/\gamma}}$ be orthogonal to $\ket{\widetilde{\pi}_{\beta,r,\zeta/\gamma}}$, and we define $\overline{\psi} \in \C^{(K \cap rB_d)_{\zeta/\gamma}}$ such that
        \[\ket{\overline{\psi}} = \sum_{x \in (K \cap rB_d)_{\zeta/\gamma}} \frac{\overline{\psi}(x)}{\sqrt{\widetilde{\pi}_{\beta,r,\zeta/\gamma}(x)}} \ket{x}.\]
        Then, we can construct bidirectional canonical transducers $C$ and $V$ that map
        \begin{align*}
            \ket{\pi_{\beta,\delta,r,\zeta}} \overset{C}{\rightsquigarrow} \ket{\widetilde{\pi}_{\beta,r,\zeta/\gamma}}, & \qquad \text{with complexity} \qquad O\left(\sqrt{\frac{d}{\beta\delta^2}}\right), \\
            \ket{\widetilde{\pi}_{\beta,r,\zeta/\gamma}} \overset{V}{\rightsquigarrow} \ket{\widetilde{\pi}_{\beta,r,\zeta/\gamma}}, & \qquad \text{with complexity} \qquad O\left(1\right), \\
            \ket{\overline{\psi}} \overset{V}{\rightsquigarrow} -\ket{\overline{\psi}}, & \qquad \text{with complexity} \qquad O\left(\frac{d\norm{\ket{\overline{\psi}}}^2}{\beta\delta^2} + \sum_{x \in (\gamma(K \cap rB_d))_{\zeta}} \frac{|\overline{\psi}(x/\gamma)|^2}{\pi_{\beta,\delta,r,\zeta}(x)\ell_{\delta,r,\zeta}(x)}\right).
        \end{align*}
    \end{theorem}

    \begin{proof}
        Following \cite[Section~6.4]{CV18}, we define the distribution $\widetilde{\pi}_{\beta,r}$ on $K \cap rB_d$ as
        \begin{equation}
            \label{eq:pi-br}
            \widetilde{\pi}_{\beta,r}(x) = \frac{\gamma^d\pi_{\beta,\delta,r}(\gamma x)a(\gamma x)}{A}, \qquad \text{where} \qquad A := \underset{x \sim \pi_{\beta,\delta,r}}{\E} [a(x)].
        \end{equation}
        We observe from \cite[Lemma~6.13]{CV18} that the acceptance probability of this rejection sampling routine is at least $9/20$, i.e., $A \geq 9/20$. Moreover, the same lemma guarantees that the bound on the total-variation distance between $\widetilde{\pi}_{\beta,r}$ and $\pi_{\beta,r}$ is at most $\eta/2$.

        We port this result to the discretized setting. Indeed, observe that $\pi_{\beta,\delta,r}$ and $a$ are both piecewise Lipschitz continuous. Moreover, as $\zeta \downarrow 0$, $\pi_{\beta,\delta,r,\zeta}$ converges to $(\pi_{\beta,\delta,r})_{\zeta}$, since the local conductance converges to its continuous counterpart by \Cref{lem:discrete-local-conductance-proximity}. As such, we can use \Cref{lem:smoothness} to argue that the average acceptance probability $A_{\zeta}$ converges to $A$ as $\zeta \downarrow 0$. We can then use the same lemma to relate the total variation distance between $\widetilde{\pi}_{\beta,r,\zeta/\gamma}$ and $\pi_{\beta,r,\zeta/\gamma}$ to their continuous counterpart, implying that it converges to a value at most $\eta/2$ as $\zeta \downarrow 0$. Thus, for a small enough choice of $\zeta$, their total-variation distance is below $\eta$.

        We take the speedy-walk reflection transducer $U$ from \Cref{thm:speedy-walk-reflection}. Since this is a reflection transducer, we can use \Cref{thm:balancing-plus-minus} to rebalance the witness sizes with an input-independent constant $\alpha > 0$. Consequently, we obtain the transducer $U_{\alpha}$ that acts as
        \begin{align*}
            \ket{\pi_{\beta,\delta,r,\zeta}} \overset{U_{\alpha}}{\rightsquigarrow} \ket{\pi_{\beta,\delta,r,\zeta}}, & \qquad \text{with complexity} \qquad O(\alpha), \\
            \ket{\overline{\pi}} \overset{U_{\alpha}}{\rightsquigarrow} -\ket{\overline{\pi}}, & \qquad \text{with complexity} \qquad O\left(\frac{d\norm{\ket{\overline{\pi}}}^2}{\alpha\beta\delta^2} + \sum_{x \in (K \cap rB_d)_{\zeta}} \frac{|\overline{\pi}(x)|^2}{\alpha\pi_{\beta,\delta,r,\zeta}(x)\ell_{\delta,r,\zeta}(x)}\right),
        \end{align*}
        where $\ket{\overline{\pi}} \in \C^{(K \cap rB_d)_{\zeta}}$ is such that $\braket{\overline{\pi}}{\pi_{\beta,\delta,r}} = 0$ and
        \[\ket{\overline{\pi}} = \sum_{x \in (K \cap rB_d)_{\zeta}} \frac{\overline{\pi}(x)}{\sqrt{\pi_{\beta,\delta,r,\zeta}(x)}} \ket{x}.\]

        Next, we use rejection sampling (\Cref{thm:rejection-sampling}) on $\pi_{\beta,\delta,r,\zeta}$, with the rejection probabilities $a(x)$. Note that we can compute $a(x)$ with a single call to $O_{K,\zeta/\gamma}$, and so the transducer $O^{(a)}$ required in \Cref{thm:rejection-sampling} can be implemented with uniform cost bound $O(1)$. The resulting distribution satisfies $\pi_{\zeta}^{(a)}(x) := \pi_{\beta,\delta,r,\zeta}(x)a(x)/A_{\zeta} = \widetilde{\pi}_{\beta,r,\zeta/\gamma}(x/\gamma)$. Thus, we obtain bidirectional canonical transducers $C_{\alpha}$ and $V_{\alpha}$ that act as
        \begin{align*}
            \ket{\pi_{\beta,\delta,r,\zeta}} \overset{C_{\alpha}}{\rightsquigarrow} \ket{\pi_{\zeta}^{(a)}}, & \qquad \text{with complexity} \qquad O\left(1 + W(U_{\alpha}, \ket{\pi_{\beta,\delta,r,\zeta}}) + (1-A_{\zeta}) \cdot W(U_{\alpha}, \ket{\overline{\pi}})\right), \\
            \ket{\pi_{\zeta}^{(a)}} \overset{V_{\alpha}}{\rightsquigarrow} \ket{\pi_{\zeta}^{(a)}}, & \qquad \text{with complexity} \qquad O(1 + W(U_{\alpha}, \ket{\pi_{\beta,\delta,r,\zeta}})), \\
            \ket{\overline{\pi}^{(a)}} \overset{V_{\alpha}}{\rightsquigarrow} -\ket{\overline{\pi}^{(a)}}, & \qquad \text{with complexity} \qquad O\left(\norm{\ket{\overline{\pi}^{(a)}}}^2 + W\left(U_{\alpha}, \sum_{x \in (\gamma(K \cap rB_d))_{\zeta}} \frac{\overline{\pi}^{(a)}(x)}{\sqrt{\pi_{\beta,\delta,r,\zeta}(x)}} \ket{x}\right)\right),
        \end{align*}
        where $\ket{\overline{\pi}^{(a)}} \in \C^{(\gamma(K \cap rB_d))_{\zeta}}$ is such that $\braket{\overline{\pi}^{(a)}}{\pi_{\zeta}^{(a)}} = 0$ and $\overline{\pi}^{(a)} \in \C^{(\gamma(K \cap rB_d))_{\zeta}}$ such that
        \[\ket{\overline{\pi}^{(a)}} = \sum_{x \in (\gamma(K \cap rB_d))_{\zeta}} \frac{\overline{\pi}^{(a)}(x)}{\sqrt{\pi_{\zeta}^{(a)}(x)}} \ket{x}, \qquad \text{and} \qquad \ket{\overline{\pi}} = \sum_{x \in (K \cap rB_d)_{\zeta}} \sqrt{\pi_{\beta,\delta,r,\zeta}(x)} \cdot \frac{A_{\zeta} - a(x)}{\sqrt{A_{\zeta}(1-A_{\zeta})}}.\]

        We now analyze the transduction complexity of $C_{\alpha}$. First, observe that $W(U_{\alpha}, \ket{\pi_{\beta,\delta,r,\zeta}}) \in O(a)$. Next, we have $\overline{\pi}(x) = \pi_{\beta,\delta,r,\zeta}(x) \cdot (A_{\zeta} - a(x))/\sqrt{A_{\zeta}(1-A_{\zeta})}$, and so
        \[\norm{\ket{\overline{\pi}}}^2 = \underset{x \sim \pi_{\beta,\delta,r,\zeta}}{\E} \left[\frac{(A_{\zeta} - a(x))^2}{A_{\zeta}(1-A_{\zeta})}\right] = \frac{\underset{x \sim \pi_{\beta,\delta,r,\zeta}}{\Var} \left[a(x)\right]}{A_{\zeta}(1-A_{\zeta})} \leq 1.\]
        We also observe that
        \[\sum_{x \in (K \cap rB_d)_{\zeta}} \frac{|\overline{\pi}(x)|^2}{\alpha\pi_{\beta,\delta,r,\zeta}(x)\ell_{\delta,r,\zeta}(x)} = \frac{1}{\alpha A_{\zeta}(1-A_{\zeta})} \sum_{x \in (K \cap rB_d)_{\zeta}} \frac{(A_{\zeta} - a(x))^2\pi_{\beta,\delta,r,\zeta}(x)}{\ell_{\delta,r,\zeta}(x)} \leq \frac{\underset{x \sim \pi_{\beta,\delta,r,\zeta}}{\E} \left[\frac{1}{\ell_{\delta,r,\zeta}(x)}\right]}{\alpha A_{\zeta}(1-A_{\zeta})}.\]
        The latter expectation value is $O(1)$ due to \cite[Lemma~6.11]{CV18}, and the same argument as in \Cref{thm:speedy-walk-reflection} relating the discrete setting to the continuous one. Thus,
        \[W(U_{\alpha}, \ket{\overline{\pi}}) \in O\left(1 + \frac{1}{\alpha(1-A_{\zeta})}\right),\]
        and so we let $C' := C_{\alpha}$ where $\alpha = \sqrt{d/(\beta\delta^2)}$, to obtain that
        \[\ket{\pi_{\beta,\delta,r,\zeta}} \overset{C'}{\rightsquigarrow} \ket{\pi_{\zeta}^{(a)}}, \quad \text{with complexity} \quad O\left(\alpha + (1-A_{\zeta}) \cdot \left(\frac{d}{\alpha\beta\delta^2} + 1 + \frac{1}{\alpha(1-A)}\right)\right) \subseteq O\left(\sqrt{\frac{d}{\beta\delta^2}}\right).\]

        Next, we turn to the transduction complexity of $V_{\alpha}$. To that end, observe that
        \[\norm{\sum_{x \in (K \cap rB_d)_{\zeta}} \frac{\overline{\pi}^{(a)}(x)}{\sqrt{\pi_{\beta,\delta,r,\zeta}(x)}}\ket{x}}^2 \leq \frac{1}{A_{\zeta}} \norm{\sum_{x \in (K \cap rB_d)_{\zeta}} \frac{\overline{\pi}^{(a)}(x)}{\sqrt{\pi^{(a)}(x)}} \ket{x}}^2 \in O(1).\]
        Thus, we choose $V' := V_1$ to obtain that
        \begin{align*}
            \ket{\pi_{\zeta}^{(a)}} \overset{V'}{\rightsquigarrow} \ket{\pi_{\zeta}^{(a)}}, & \qquad \text{with complexity} \qquad O(1), \\
            \ket{\overline{\pi}^{(a)}} \overset{V'}{\rightsquigarrow} -\ket{\overline{\pi}^{(a)}}, & \qquad \text{with complexity} \qquad O\left(\frac{d\norm{\ket{\overline{\overline{\pi}^{(a)}}}}^2}{\beta\delta^2} + \sum_{x \in (\gamma(K \cap rB_d))_{\zeta}} \frac{|\overline{\pi}^{(a)}(x)|^2}{\pi_{\beta,\delta,r,\zeta}(x)\ell_{\delta,r,\zeta}(x)}\right).
        \end{align*}
        
        Finally, we consider the isometry that scales the space by $\gamma$, i.e., $S : \C^{(\gamma(K \cap rB_d))_{\zeta}} \to \C^{(K \cap rB_d)_{\zeta/\gamma}}$, defined as $S : \ket{x} \mapsto \ket{x/\gamma}$. It is immediate that $S$ maps $\ket{\pi^{(a)}_{\zeta}}$ to $\ket{\widetilde{\pi}_{\beta,r,\zeta/\gamma}}$. Thus, the transducers $C = SC'$ and $V = SV'S^{\dagger}$ act as in the theorem statement.
        
        It remains to analyze the transduction complexity of $V$ acting on $\ket{\overline{\psi}}$. To that end, observe that $W(V,\ket{\overline{\psi}}) = W(V', S^{\dagger}\ket{\overline{\psi}})$, and so
        \[S^{\dagger}\ket{\overline{\psi}} = \sum_{x \in (\gamma(K \cap rB_d))_{\zeta}} \frac{\overline{\psi}(x/\gamma)}{\sqrt{\widetilde{\pi}_{\beta,r,\zeta/\gamma}(x/\gamma)}} \ket{x} = \sum_{x \in (\gamma(K \cap rB_d))_{\zeta}} \frac{\overline{\psi}(x/\gamma)}{\sqrt{\pi_{\zeta}^{(a)}(x)}} \ket{x},\]
        and so we obtain the complexity by setting $\overline{\pi}^{(a)}(x) = \overline{\psi}(x/\gamma)$.
    \end{proof}

    The transducers we constructed allow us to prepare and reflect around approximations to the Gibbs distribution at a given inverse temperature $\beta$. However, when $\beta = 0$, the complexities blow up. To resolve this, we follow the proof of \cite[Theorem~1.2]{CV18} and use a different rejection sampling routine.

    \begin{corollary}
        \label{cor:speedy-to-uniform}
        Let $d \in \N$, $R \geq 1$, $0 \leq \beta \leq 1/R^2$, $r \geq 1$, $0 < \eta < 1/(2e)$ and $0 < \delta \leq \min\{1/(4\beta r), 1/(4096\sqrt{2d\beta}), 1/(4096\sqrt{d\ln(20d/\eta)})\}$. Let $B_d \subseteq K \subseteq RB_d$ be a convex body, accessible through a membership oracles $O_{K,\zeta}$ and $O_{K,\zeta/\gamma}$ with a sufficiently small spacing parameter $\zeta = \zeta(d,R,\beta,\delta) > 0$. For all $x \in K \cap rB_d$, let $a'(x) := \exp(\beta\norm{x}^2-1)$. Let $\gamma$, and $\widetilde{\pi}_{\beta,r,\zeta/\gamma}$ be as in~\Cref{thm:speedy-to-ball}. Then, we define the distribution $\widetilde{\pi}_{0,r,\zeta/\gamma}$ on $(K \cap rB_d)_{\zeta/\gamma}$ as
        \[\widetilde{\pi}_{0,r,\zeta/\gamma}(x) = \frac{\widetilde{\pi}_{\beta,r,\zeta/\gamma}(x)a'(x)}{A'_{\zeta}}, \qquad \text{where} \qquad A'_{\zeta} := \underset{x \sim \widetilde{\pi}_{\beta,r,\zeta/\gamma}}{\E} [a'(x)].\]
        Then, $\norm{\widetilde{\pi}_{0,r,\zeta/\gamma} - \pi_{0,r,\zeta/\gamma}}_{\mathrm{TV}} < 10\eta$.
                Let $\ket{\overline{\varphi}} \in \C^{(K \cap rB_d)_{\zeta/\gamma}}$ be orthogonal to $\ket{\widetilde{\pi}_{0,r,\zeta/\gamma}}$. We let $\overline{\varphi} \in \C^{(K \cap rB_d)_{\zeta/\gamma}}$ be such that
        \[\ket{\overline{\varphi}} = \int_{K \cap rB_d} \frac{\overline{\varphi}(x)}{\sqrt{\widetilde{\pi}_{0,r,\zeta/\gamma}(x)}} \ket{x} \;\mathrm{d}x.\]
        Then, we can construct bidirectional canonical transducers $C$ and $U$ that act as
        \begin{align*}
            \ket{\widetilde{\pi}_{\beta,r,\zeta/\gamma}} \overset{C}{\rightsquigarrow} \ket{\widetilde{\pi}_{0,r,\zeta/\gamma}}, & \qquad \text{with complexity} \qquad O\left(\sqrt{\frac{d}{\beta\delta^2}}\right), \\
            \ket{\widetilde{\pi}_{0,r,\zeta/\gamma}} \overset{U}{\rightsquigarrow} \ket{\widetilde{\pi}_{0,r,\zeta/\gamma}}, & \qquad \text{with complexity} \qquad O(1), \\
            \ket{\overline{\varphi}} \overset{U}{\rightsquigarrow} -\ket{\overline{\varphi}}, & \qquad \text{with complexity} \qquad O\left(\frac{d\norm{\ket{\overline{\varphi}}}^2}{\beta\delta^2} + \sum_{x \in (\gamma(K \cap rB_d))_{\zeta}} \frac{|\overline{\varphi}(x/\gamma)|^2}{\pi_{\beta,\delta,r,\zeta}(x)\ell_{\delta,r,\zeta}(x)}\right).
        \end{align*}
    \end{corollary}

    \begin{proof}
        First, we define
        \[\overline{A}_{\zeta}' := \underset{x \sim \pi_{\beta,r,\zeta/\gamma}}{\E} [a'(x)].\]
        Since $a'(x) \geq 1/e$ everywhere, we immediately find that both $A'_{\zeta}$ and $\overline{A}_{\zeta}'$ are at least $1/e$. Moreover, since $a'(x) \leq 1$ everywhere, we have $|A_{\zeta}' - \overline{A}_{\zeta}'| \leq \norm{\widetilde{\pi}_{\beta,r,\zeta/\gamma} - \pi_{\beta,r,\zeta/\gamma}}_{\mathrm{TV}} \leq \eta$. Thus,
        \begin{align*}
            &\norm{\widetilde{\pi}_{0,r,\zeta/\gamma} - \pi_{0,r,\zeta/\gamma}}_{\mathrm{TV}} = \frac12 \sum_{x \in (K \cap rB_d)_{\zeta/\gamma}} |\widetilde{\pi}_{0,r,\zeta/\gamma}(x) - \pi_{0,r,\zeta/\gamma}(x)| \\
            &\qquad = \frac12 \sum_{x \in (K \cap rB_d)_{\zeta/\gamma}} \left|\frac{\widetilde{\pi}_{\beta,r,\zeta/\gamma}(x)}{A_{\zeta}'} - \frac{\pi_{\beta,r,\zeta/\gamma}(x)}{\overline{A}_{\zeta}'}\right|a'(x) \\
            &\qquad \leq \frac12 \sum_{x \in (K \cap rB_d)_{\zeta/\gamma}} \left|\frac{\widetilde{\pi}_{\beta,r,\zeta/\gamma}(x)}{A_{\zeta}'} - \frac{\pi_{\beta,r,\zeta/\gamma}(x)}{A_{\zeta}'}\right| + \pi_{\beta,r,\zeta/\gamma}(x)\left|\frac{1}{A_{\zeta}'} - \frac{1}{\overline{A}_{\zeta}'}\right| \\
            &\qquad = \frac{\norm{\widetilde{\pi}_{\beta,r,\zeta/\gamma} - \pi_{\beta,r,\zeta/\gamma}}_{\mathrm{TV}}}{A_{\zeta}'} + \frac{|A_{\zeta}' - \overline{A}_{\zeta}'|}{2A_{\zeta}'\overline{A}_{\zeta}'} \leq e\eta + \frac{e^2\eta}{2} < 10\eta.
        \end{align*}
        
        Next, we run an analogous argument as in \Cref{thm:speedy-to-ball} to generate the transducers $C$ and $U$ from the corollary statement. The only difference is that we base the construction on $U$ from \Cref{thm:speedy-to-ball}, rather than from \Cref{thm:speedy-walk-reflection}. All analyses carry over except for the final complexity. There, we observe that
        \[\ket{\overline{\varphi}} = \sum_{x \in (K \cap rB_d)_{\zeta/\gamma}} \frac{\overline{\varphi}(x)}{\sqrt{\widetilde{\pi}_{0,r,\zeta/\gamma}(x)}} \ket{x} = \sum_{x \in (K \cap rB_d)_{\zeta/\gamma}} \frac{\overline{\varphi}(x)}{\sqrt{\frac{\widetilde{\pi}_{\beta,r,\zeta/\gamma}(x)a'(x)}{A'_{\zeta}}}} \ket{x}.\]
        Thus, we apply the transducer $U$ to the state
        \[(I \otimes \bra{0})(O^{(a')})^{\dagger}\ket{\overline{\varphi}}\ket{1} = \sum_{x \in (K \cap rB_d)_{\zeta/\gamma}} \frac{\sqrt{A'_{\zeta}} \cdot \overline{\varphi}(x)}{\sqrt{\widetilde{\pi}_{\beta,r,\zeta/\gamma}(x)}} \ket{x},\]
        and so we take $\overline{\psi}(x) = \sqrt{A'_{\zeta}} \cdot \overline{\varphi}(x)$ in \Cref{thm:speedy-to-ball}. Since $A_{\zeta}' \leq 1$, this only affects the final term in the complexity statement up to a constant.
    \end{proof}

    \subsection{Uniform sampling from a convex body}
    \label{subsec:sampling}

    We are now ready to present a transducer that anneals through the inverse temperature domain and samples uniformly from a convex body. We start by analyzing the length of the cooling schedule.
    
    \begin{lemma}
        \label{lem:schedule-length}
        Let $d \in \N$, $R \geq 1$. Let $\beta_0 \geq d/2$, and recursively define
        \[\beta_{j+1} = \begin{cases}
            \beta_j \cdot \left(1 - \max\left\{\frac{1}{8\sqrt{d}}, \frac{1}{8R\sqrt{\beta_j}}\right\}\right), & \text{if } \beta_j > \frac{1}{(2R)^2}, \\
            0, & \text{otherwise}.
        \end{cases}\]
        Let $\ell_1 := \min\{j \in \N : \beta_j \leq d/R^2\}$, and $\ell_2 := \min\{j \in \N : \beta_{j+\ell_1} \leq 1/(2R)^2\}$. For all $0 \leq a < b < \ell_1$, we have
        \begin{equation}
            \label{eq:stage-1}
            b-a \leq 16\sqrt{d}\ln\frac{\beta_a}{\beta_b}, \qquad \text{and} \qquad \sum_{j=a}^{b-1} \frac{1}{\sqrt{\beta_j}} \leq \frac{16\sqrt{d}}{\sqrt{\beta_{b-1}}}\ln\frac{\beta_a}{\beta_b}.
        \end{equation}
        For all $\ell_1 \leq a < b < \ell_1 + \ell_2$, we have
        \begin{equation}
            \label{eq:stage-2}
            b - a \leq 16R\sqrt{\beta_a} \ln\frac{\beta_a}{\beta_b}, \qquad \text{and} \qquad \sum_{j=a}^{b-1} \frac{1}{\sqrt{\beta_j}} \leq 16R\ln\frac{\beta_a}{\beta_b}.
        \end{equation}
        We have $\ell_1 \leq 1 + 16\sqrt{d}\ln(\beta_0R^2/d) \in \widetilde{O}(\sqrt{d})$, and $\ell_2 \leq 1 + 16\sqrt{d}\ln(d) \in \widetilde{O}(\sqrt{d})$, and consequently $\ell = \ell_1 + \ell_2 + 1 \in \widetilde{O}(\sqrt{d})$.
    \end{lemma}

    \begin{proof}
        For all $0 \leq a < b < \ell_1$, we have
        \[b-a \leq \sum_{j=a}^{b-1} \frac{16\sqrt{d}}{16\sqrt{d}} \leq 16\sqrt{d} \sum_{j=a}^{b-1} \ln\left(1 + \frac{1}{8\sqrt{d}}\right) \leq 16\sqrt{d}(\ln(\beta_a) - \ln(\beta_b)) \leq 16\sqrt{d}\ln\frac{\beta_a}{\beta_b}.\]
        The second expression follows as $\beta_j$ is decreasing. On the other hand, for all $\ell_1 \leq a < b < \ell-1$, we have
        \[\sum_{j=a}^{b-1} \frac{1}{\sqrt{\beta_j}} \leq \sum_{j=a}^{b-1} \frac{16R}{16R\sqrt{\beta_j}} \leq 16R\sum_{j=a}^{b-1} \ln\left(1 + \frac{1}{8R\sqrt{\beta_j}}\right) = 16R\sum_{j=a}^{b-1} (\ln(\beta_j) - \ln(\beta_{j+1})) = 16R\ln\frac{\beta_a}{\beta_b},\]
        and the first expression follows as $\beta_j$ is decreasing.
    \end{proof}

    Now, we are ready to state the uniform sampling algorithm.

    \begin{breakablealgorithm}
        \caption{Transducer for sampling from a convex body}
        \label{alg:sampling-algo}
        
        \noindent\textbf{Input:}
        \begin{enumerate}[nosep]
            \item $3 \leq d \in \N$, $R \geq 1$, $0 < \nu < 1/4$, and $\zeta = \zeta(d,R,\nu) > 0$ sufficiently small.
            \item A convex body $B_d \subseteq K \subseteq RB_d$, accessible through membership oracles $O_{K,\zeta}$ and $O_{K,\zeta/\gamma}$.
        \end{enumerate}

        \noindent\textbf{Parameters:}
        \begin{enumerate}[nosep]
            \item $\gamma := 1 - 1/(2d)$.
            \item $\xi := \sqrt{1 + 3\ln(2/\nu^2)/(2d)}$.
        \end{enumerate}

        \noindent\textbf{Output:} A $\nu$-precise $\ell_2$-approximation of the uniform superposition over $K$.

        \noindent\textbf{Procedure:}
        
        \begin{enumerate}[nosep]
            \item Initialize $\beta_0 = d\xi^2$, $r_0 = 1$, $\delta_0 = \min\{1/(4\beta_0), 1/(2^{18}\sqrt{2d\beta_0}), 1/(2^{18}\sqrt{d\ln(200d/\nu)})\}$ and $j = 0$.
            \item Initialize a quantum register in $\ket{\pi_{\beta_0,\delta_0,1,\zeta}}$.
            \item While $\beta_j > 1/(2R)^2$:
            \begin{enumerate}[nosep]
                \item Set $\beta_{j+1} = \beta_j \cdot (1 - \max\{1/(8\sqrt{d}), 1/(8R\sqrt{\beta_j})\})$.
                \item Set $r_{j+1} = \sqrt{d/\beta_{j+1}}\xi$.
                \item Set $\delta_{j+1} = \min\{1/(4\beta_{j+1}r_{j+1}), 1/(2^{18}\sqrt{2d\beta_{j+1}}), 1/(2^{18}\sqrt{d\ln(200d/\nu)})\}$.
                \item Use \Cref{cor:anneal-small-steps} to map $\ket{\pi_{\beta_j,\delta_j,r_j,\zeta}}$ to $\ket{\pi_{\beta_{j+1},\delta_{j+1},r_{j+1},\zeta}}$.
                \item Increment $j$.
            \end{enumerate}
            \item Use \Cref{thm:speedy-to-ball,cor:speedy-to-uniform} to map $\ket{\pi_{\beta_{\ell},\delta_{\ell},r_{\ell},\zeta}}$ to $\ket{\widetilde{\pi}_{\beta_j,r_{\ell},\zeta/\gamma}}$ and then to $\ket{\widetilde{\pi}_{0,r_{\ell},\zeta/\gamma}}$, where $\ell$ is the final value of $j$.
        \end{enumerate}
    \end{breakablealgorithm}
    
    \begin{theorem}
        \label{thm:uniform-sampling}
        Let $3 \leq d \in \N$, $R \geq 1$, $\nu > 0$, and $B_d \subseteq K \subseteq RB_d$ be a convex body, accessible through a membership oracle on a grid with sufficiently small spacing parameter $\zeta = \zeta(d,R,\nu) > 0$. The bidirectional canonical transducer in \Cref{alg:sampling-algo} acts as
        \[\ket{0} \overset{C}{\rightsquigarrow} \ket{\widetilde{\pi}_{0,r_{\ell},\zeta/\gamma}}, \qquad \text{where} \qquad \norm{\ket{\widetilde{\pi}_{0,r_{\ell},\zeta/\gamma}} - \frac{1}{\sqrt{|K_{\zeta/\gamma}|}} \sum_{x \in K_{\zeta/\gamma}} \ket{x}} \leq \nu,\]
        with transduction complexity $\widetilde{O}\left(dR + d^2\right)$.
    \end{theorem}

    \begin{proof}
        The approximation is guaranteed by \Cref{cor:speedy-to-uniform}, since $r_{\ell} \geq R$, and $\norm{\widetilde{\pi}_{0,r_{\ell},\zeta/\gamma} - \pi_{0,r_\ell,\zeta/\gamma}}_{\mathrm{TV}} \leq \nu^2/2$. Indeed, since $\pi_{0,r_{\ell},\zeta/\gamma}$ is the uniform superposition over $K_{\zeta}$, we verify that
        \[\norm{\ket{\widetilde{\pi}_{0,r_{\ell},\zeta/\gamma}} - \ket{\pi_{0,r_{\ell},\zeta/\gamma}}} \leq \sqrt{2\norm{\widetilde{\pi}_{0,r_{\ell},\zeta/\gamma} - \pi_{0,r_{\ell},\zeta/\gamma}}_{\mathrm{TV}}} \leq \nu.\]
        
        It remains to analyze the total transduction complexity using sequential composition (\Cref{thm:sequential-composition}). In the $j$th iteration, the total transduction complexity from \Cref{cor:anneal-small-steps} equals
        \[O\left(\sqrt{\frac{d}{\beta_{j+1}\delta_j^2}}\left(1 + d\ln\frac{\delta_{j+1}}{\delta_j}\right)\right).\]
        Let $\ell^* = \min\{j \in \N : \beta_j \leq \ln(200d/\nu)/2\}$. For all $j \geq \ell^*$, we have $\delta_{j+1} = \delta_j$, and for all $j < \ell^*$, we have $\delta_{j+1}/\delta_j = \sqrt{\beta_j/\beta_{j+1}} \in O(1 + \max\{1/(8\sqrt{d}),1/(8R\sqrt{\beta_j})\})$. Thus, we analyze the total cost as
        \[\sum_{j=0}^{\ell^*-1} \sqrt{\frac{d}{\beta_j\delta_j^2}} \cdot \left(1 + d\max\left\{\frac{1}{\sqrt{d}}, \frac{1}{R\sqrt{\beta_j}}\right\}\right) + \sum_{j=\ell^*}^{\ell_1 + \ell_2 - 1} \sqrt{\frac{d}{\beta_j\delta_j^2}}.\]
        We evaluate this expression using \Cref{lem:schedule-length}. We distinguish two cases. If $\ell^* < \ell_1$, the above expression becomes
        \[\sum_{j=0}^{\ell^*-1} d \cdot (1 + \sqrt{d}) + \sum_{j = \ell^*}^{\ell_1 + \ell_2 - 1} \frac{d}{\sqrt{\beta_j}} \in \widetilde{O}\left(d^2 + dR\right).\]
        On the other hand, if $\ell^* \geq \ell_1$, then we have
        \[\sum_{j=0}^{\ell_1-1} d \cdot \left(1 + \sqrt{d}\right) + \sum_{j=\ell_1}^{\ell^*-1} d \cdot \left(1 + \frac{d}{R\sqrt{\beta_j}}\right) + \sum_{j=\ell_*}^{\ell_1+\ell_2-1} \frac{d}{\sqrt{\beta_j}} \in \widetilde{O}\left(d^2 + dR\right).\]
        The remaining contributions to the transduction complexities are insignificant compared to this term.
    \end{proof}

    The above construction produces a transducer that prepares an approximation to the uniform superposition over $K$. This is a weaker result compared to a quantum circuit doing the same thing, as turning a transducer into a quantum circuit up to precision $\varepsilon$ incurs an $O(1/\varepsilon^2)$ overhead. However, since transducers have favorable composition properties, this sampling routine can still be used as a subroutine in more general algorithms without incurring this overhead.

    Finally, we note that a $\widetilde{O}(d^2)$ complexity can also be attained via different means. One can employ the hit-and-run walk~\cite{lovasz2006simulated} combined with the techniques introduced in \cite{cornelissen2023sublinear} to achieve a sampling routine from the uniform distribution over $K$ in a similar complexity. In short, the idea is that the hit-and-run cooling schedule is shorter, i.e., only of length $\widetilde{O}(\sqrt{d})$, but the classical mixing time is longer, i.e., $\widetilde{O}(d^3)$. By using the quadratic speed-up for mixing, one can then anneal in a number of queries that satisfies $\widetilde{O}(\sqrt{d} \cdot d^{3/2}) = \widetilde{O}(d^2)$.

    \section{Quantum volume estimation}
    \label{sec:vol-est}
    
    We now compile all building blocks we developed throughout this work to into a quantum volume estimation algorithm based on the state-of-the-art randomized algorithm presented in \cite[Section~7]{CV18}.
    
    \subsection{Mean estimation over Gibbs distributions}
    
    A core component of the volume estimation algorithm runs mean estimation over a probability space formed by the Gibbs distribution. We start with a lemma that explores the cost of reflecting around the approximation of the ball-walk distribution starting from a random variable.

    \begin{lemma}
        \label{lem:speedy-refl-random-variable}
        Let $d \in \N$, $R \geq 1$, $\beta \geq 0$, $1 \leq r \leq R$, $\eta > 0$ and a stepsize parameter $0 < \delta \leq \min\{1/(4\beta r), 1/(4096\sqrt{2d\beta}), 1/(4096\sqrt{d\ln(20d/\eta)})\}$. Let $\gamma$ and $\widetilde{\pi}_{\beta,r,\zeta/\gamma}$ be as in \Cref{thm:speedy-to-ball}, and $B_d \subseteq K \subseteq RB_d$ be a convex body, accessible through a membership oracles $O_{K,\zeta}$ and $O_{K,\zeta/\gamma}$ with a sufficiently small spacing parameter $\zeta = \zeta(d,R,\beta,\delta) > 0$. Let $X : RB_d \to [0,1]$ be a random variable with mean $0 < \widetilde{\mu} < 1$, and
        \[\ket{\overline{\pi}(X)} = \sum_{x \in (K \cap rB_d)_{\zeta/\gamma}} \sqrt{\widetilde{\pi}_{\beta,r,\zeta/\gamma}(x)} \cdot \frac{\widetilde{\mu} - X(x)}{\sqrt{\widetilde{\mu}(1-\widetilde{\mu})}} \ket{x}, \qquad \text{where} \qquad \widetilde{\mu} := \underset{x \sim \widetilde{\pi}_{\beta,r,\zeta/\gamma}}{\E} [X(x)].\]
        Then, we can construct a bidirectional canonical transducer $V$ that acts as
        \begin{align*}
            \ket{\widetilde{\pi}_{\beta,r,\zeta/\gamma}} \overset{V}{\rightsquigarrow} \ket{\widetilde{\pi}_{\beta,r,\zeta/\gamma}}, & \qquad \text{with complexity} \qquad O\left(1\right), \\
            \ket{\overline{\pi}(X)} \overset{V}{\rightsquigarrow} -\ket{\overline{\pi}(X)}, & \qquad \text{with complexity} \qquad O\left(\frac{d}{\beta\delta^2} + \frac{\eta}{\widetilde{\mu}(1-\widetilde{\mu})}\right).
        \end{align*}
    \end{lemma}

    \begin{proof}
        We take $V$ to be the transducer from \Cref{thm:speedy-to-ball}, and for all $x \in (K \cap rB_d)_{\zeta/\gamma}$, we let
        \[\overline{\psi}(x) = \widetilde{\pi}_{\beta,r,\zeta/\gamma}(x) \cdot \frac{\widetilde{\mu} - X(x)}{\sqrt{\widetilde{\mu}(1-\widetilde{\mu})}}.\]
        Recall from \Cref{thm:speedy-to-ball} that for all $x \in (\gamma(K \cap rB_d))_{\zeta}$, we have $\pi_{\beta,\delta,r,\zeta}(x)a(x)/A_{\zeta} = \widetilde{\pi}_{\beta,r,\zeta/\gamma}(x/\gamma)$, and so using $a(x) \leq 1$,
        \[\frac{\widetilde{\pi}_{\beta,r,\zeta/\gamma}(x/\gamma)^2}{\pi_{\beta,\delta,r,\zeta}(x)\ell_{\delta,r,\zeta}(x)} = \frac{\pi_{\beta,\delta,r,\zeta}(x)a(x)^2}{A_{\zeta}^2\ell_{\delta,r,\zeta}(x)} \leq \frac{\pi_{\beta,\delta,r,\zeta}(x)\exp(-\frac{\beta}{\gamma^2}\norm{x}^2)}{A_{\zeta}^2\ell_{\delta,r,\zeta}(x)\exp(-\beta\norm{x}^2)} = \frac{\exp(-\beta\norm{x/\gamma}^2)}{A_{\zeta}^2N_{\beta,\delta,r,\zeta}}.\]
        On the other hand, let $\pi_{\beta,r,\zeta/\gamma}$ be the discretization of $\pi_{\beta,r}$ on $(K \cap rB_d)_{\zeta/\gamma}$. Then, for any $x \in (\gamma(K \cap rB_d))_{\zeta}$, we obtain that
        \[\pi_{\beta,r,\zeta/\gamma}(x/\gamma) = \frac{\exp(-\beta\norm{x/\gamma}^2)}{N_{\beta,r,\zeta/\gamma}}.\]
        By combining both equations, we observe that
        \[\frac{\widetilde{\pi}_{\beta,r,\zeta/\gamma}(x/\gamma)^2}{\pi_{\beta,\delta,r,\zeta}(x)\ell_{\delta,r,\zeta}(x)} \leq \pi_{\beta,r,\zeta/\gamma}(x/\gamma) \cdot \frac{N_{\beta,r,\zeta/\gamma}}{A_{\zeta}^2N_{\beta,\delta,r,\zeta}}.\]
        We analyze the ratio of the normalization constants. To that end, observe that both distributions are piecewise Lipschitz continuous, and so by \Cref{lem:smoothness}, their fraction is well-approximated in the continuous setting. Hence, we observe from \cite[Lemma~6.11]{CV18} that
        \[\frac{\zeta^dN_{\beta,\delta,r,\zeta}}{(\zeta/\gamma)^dN_{\beta,r,\zeta/\gamma}} \to \frac{\int_{K \cap rB_d} \ell_{\delta,r}(x) \exp(-\beta\norm{x}^2) \;\mathrm{d}x}{\int_{K \cap rB_d} \exp(-\beta\norm{x}^2) \;\mathrm{d}x} \in \Omega(1), \qquad (\zeta \downarrow 0).\]
        As such, for sufficiently small choices of $\zeta$, the expression simplifies to
        \[\frac{\widetilde{\pi}_{\beta,r,\zeta/\gamma}(x/\gamma)^2}{\pi_{\beta,\delta,r,\zeta}(x)\ell_{\delta,r,\zeta}(x)} \leq \frac{C\pi_{\beta,r,\zeta/\gamma}(x/\gamma)}{A_{\zeta}^2\gamma^d},\]
        for a universal constant $C > 0$. Finally, recall from \Cref{thm:speedy-to-ball} that $A_{\zeta} \in \Omega(1)$ for sufficiently small choices of $\zeta$, and $\gamma^d \in \Theta(1)$. Thus,
        \begin{align*}
            &\sum_{x \in (\gamma(K \cap rB_d))_{\zeta}} \frac{|\overline{\psi}(x/\gamma)|^2}{\pi_{\beta,\delta,r,\zeta}(x)\ell_{\delta,r,\zeta}(x)} = \sum_{x \in (\gamma(K \cap rB_d))_{\zeta}} \frac{\widetilde{\pi}_{\beta,r,\zeta/\gamma}(x/\gamma)^2(\widetilde{\mu} - X(x/\gamma))^2}{\pi_{\beta,\delta,r,\zeta}(x)\ell_{\delta,r,\zeta}(x)\widetilde{\mu}(1-\widetilde{\mu})} \\
            &\qquad \in O\left(\sum_{x \in (\gamma(K \cap rB_d))_{\zeta}} \pi_{\beta,r,\zeta/\gamma}\left(\frac{x}{\gamma}\right) \cdot \frac{(\widetilde{\mu} - X(x/\gamma))^2}{\widetilde{\mu}(1-\widetilde{\mu})}\right) = O\left(\underset{x \sim \pi_{\beta,r,\zeta/\gamma}}{\E} \left[\frac{(\widetilde{\mu} - X(x))^2}{\widetilde{\mu}(1 - \widetilde{\mu})}\right]\right).
        \end{align*}
        Next, we evaluate this expectation by
        \[\underset{x \sim \pi_{\beta,r,\zeta/\gamma}}{\E} \left[\frac{(\widetilde{\mu} - X(x))^2}{\widetilde{\mu}(1 - \widetilde{\mu})}\right] \leq \frac{\norm{\pi_{\beta,r,\zeta/\gamma} - \widetilde{\pi}_{\beta,r,\zeta/\gamma}}_{\mathrm{TV}}}{\widetilde{\mu}(1 - \widetilde{\mu})} + \underset{x \sim \widetilde{\pi}_{\beta,r,\zeta/\gamma}}{\E} \left[\frac{(\widetilde{\mu} - X(x))^2}{\widetilde{\mu}(1 - \widetilde{\mu})}\right] \leq \frac{\eta}{\widetilde{\mu}(1-\widetilde{\mu})} + 1,\]
        Thus, putting everything together yields
        \[W(V, \ket{\overline{\pi}(X)}) \in O\left(\frac{d}{\beta\delta^2} + \frac{\eta}{\widetilde{\mu}(1-\widetilde{\mu})}\right).\qedhere\]
    \end{proof}
    
    Next, we follow the approach outlined by Cornelissen and Hamoudi~\cite{cornelissen2023sublinear}, and non-destructively implement a transducer that computes the mean of a random variable.

    \begin{theorem}
        \label{thm:mean-est}
        Let $d \in \N$, $R \geq 1$, $0 < \epsilon < 1$, $t \geq 1$, $\beta \geq 0$, $r \geq 1$, and a stepsize parameter $0 < \delta \leq \min\{1/(4\beta r), 1/(4096\sqrt{2d\beta}), 1/(4096\sqrt{d\ln(20d/\min\{\epsilon^2/2^{20}, \epsilon^2/(2^{41}t^2)\})})\}$. Let $\gamma$ and $\widetilde{\pi}_{\beta,r,\zeta/\gamma}$ be as in \Cref{thm:speedy-to-ball}, and $B_d \subseteq K \subseteq RB_d$ be a convex body, accessible through a membership oracles $O_{K,\zeta}$ and $O_{K,\zeta/\gamma}$ with a spacing parameter $\zeta > 0$. Let $M > 0$ and $X : RB_d \to [0,M]$ be a piecewise $L$-Lipschitz continuous random variable with $L > 0$ on an $h$-smooth partition, mean $\mu$ and variance $\sigma^2$ over $\pi_{\beta,r}$. Let $\widetilde{\sigma} > 0$ such that $\sigma \leq \widetilde{\sigma}$, and let $\widetilde{m} \in \R$ such that $|\mu - \widetilde{m}| \leq 17\widetilde{\sigma}$. Then, for sufficiently small $\zeta = \zeta(d,R,\beta,\delta,\epsilon,M,L,h,\widetilde{\sigma})$, we can construct a bidirectional canonical transducer $V$ that acts as
        \[\ket{\widetilde{\pi}_{\beta,r,\zeta/\gamma}} \overset{V}{\rightsquigarrow} \ket{\widetilde{\pi}_{\beta,r,\zeta/\gamma}}\ket{\overline{\mu}}, \qquad \text{with complexity} \qquad \widetilde{O}\left(\sqrt{\frac{d}{\beta\delta^2}} \cdot t \cdot \polylog(t,1/\epsilon)\right),\]
        such that
        \[|\E[\overline{\mu}] - \mu| \leq \epsilon\widetilde{\sigma}, \qquad \text{and} \qquad \Var[\overline{\mu}] \leq \left(\frac{\widetilde{\sigma}}{t}\right)^2.\]
    \end{theorem}

    \begin{proof}
        We choose $\eta = \min\{\epsilon^2/2^{20}, \epsilon^2/(2^{41}t^2)\}$, and we take the transducer $V$ from \Cref{thm:speedy-to-ball} that reflects around $\ket{\widetilde{\pi}_{\beta,r,\zeta/\gamma}}$. Since this is a reflection transducer, we can use \Cref{lem:rebalancing} to rebalance the complexities with an input-independent constant $\alpha > 0$, to obtain a transducer $V_{\alpha}$ that acts as
        \begin{subequations}
        \begin{align}
            \label{eq:ball-refl-a}
            \ket{\widetilde{\pi}_{\beta,r,\zeta/\gamma}} \overset{V_{\alpha}}{\rightsquigarrow} \ket{\widetilde{\pi}_{\beta,r,\zeta/\gamma}}, & \qquad \text{with complexity} \qquad O(\alpha), \\
            \label{eq:ball-refl-b}
            \ket{\overline{\psi}} \overset{V_{\alpha}}{\rightsquigarrow} -\ket{\overline{\psi}}, & \qquad \text{with complexity} \qquad O\left(\frac{d}{\alpha\beta\delta^2} + \sum_{x \in (\gamma(K \cap rB_d))_{\zeta}} \frac{|\overline{\psi}(x/\gamma)|^2}{\alpha\pi_{\beta,\delta,r,\zeta}(x)\ell_{\delta,r,\zeta}(x)}\right),
        \end{align}
        \end{subequations}
        for all states $\ket{\overline{\psi}} \in \C^{(K \cap rB_d)_{\zeta/\gamma}}$ that satisfy $\braket{\overline{\psi}}{\widetilde{\pi}_{\beta,r,\zeta/\gamma}} = 0$, where we write $\overline{\psi} \in \C^{(K \cap rB_d)_{\zeta/\gamma}}$ such that
        \[\ket{\overline{\psi}} = \sum_{x \in (K \cap rB_d)_{\zeta/\gamma}} \frac{\overline{\psi}(x)}{\sqrt{\widetilde{\pi}_{\beta,r,\zeta/\gamma}(x)}} \ket{x}.\]

        Next, we write
        \[\mu_{\zeta} := \underset{x \sim \pi_{\beta,r,\zeta/\gamma}}{\E} \left[X(x)\right], \qquad \text{and} \qquad \sigma_{\zeta}^2 := \underset{x \sim \pi_{\beta,r,\zeta/\gamma}}{\Var}[X(x)].\]
        Since both $X$ and $\pi_{\beta,r}$ are piecewise Lipschitz continuous, we can use \Cref{lem:smoothness} to argue that the mean and variance over $\pi_{\beta,r,\zeta/\gamma}$ converge to those quantities over $\pi_{\beta,r}$. Hence, we can choose $\zeta$ small enough that $\sigma_{\zeta}^2 \leq (2\widetilde{\sigma})^2$ and $|\mu_{\zeta} - \mu| \leq \epsilon\widetilde{\sigma}/2$, such that $|\mu_{\zeta} - \widetilde{m}| \leq 17 \cdot 2\widetilde{\sigma}$.
        
        Note from \Cref{thm:speedy-to-ball} that $\norm{\widetilde{\pi}_{\beta,r,\zeta/\gamma} - \pi_{\beta,r,\zeta/\gamma}}_{\mathrm{TV}} \leq \eta$. Thus, we use $V_{\alpha}$ in \Cref{thm:transducer-mean-est} to construct a transducer $V = \mathtt{MeanEst}(X,V_{\alpha},\widetilde{m},2\widetilde{\sigma},t',\epsilon')$, where $t' = 4t$ and $\epsilon' = \epsilon/8$. Then, we have
        \[|\E[\overline{\mu}] - \mu| \leq |\E[\overline{\mu}] - \mu_{\zeta}| + |\mu_{\zeta} - \mu| \leq \left(\epsilon' + 2^{14} \cdot \frac{\eta}{\epsilon'}\right) \cdot 2\widetilde{\sigma} + \frac{\epsilon\widetilde{\sigma}}{2} \leq \left(\frac{\epsilon}{8} + \frac{\epsilon}{8}\right) \cdot 2\widetilde{\sigma} + \frac{\epsilon\widetilde{\sigma}}{2} = \epsilon\widetilde{\sigma},\]
        and
        \[\Var[\overline{\mu}] \leq \left(\frac{1}{(t')^2} + 2^{32} \cdot \frac{\eta}{(\epsilon')^2}\right)(2\widetilde{\sigma})^2 \leq \left(\frac{\widetilde{\sigma}}{t}\right)^2.\]
        
        It remains to compute the transduction complexity. To that end, for any $j \in \{0, \dots, k\}$, we immediately observe that $W(V_{\alpha}, \ket{\widetilde{\pi}_{\beta,r,\zeta/\gamma}}) \in O(\alpha)$. For the other term, we have
        \[Y_j^{\pm}(\omega) := \frac{B_j}{4} + \left(1 - \frac{B_j}{2}\right) \cdot \pm\frac{X(\omega) - \widetilde{m}}{a_j} \cdot \mathbbm{1}_{\pm(X(\omega) - \widetilde{m}) \in (a_{j-1},a_j]}, \qquad \text{and} \qquad \widetilde{\mu}_j^{\pm} := \underset{x \sim \widetilde{\pi}_{\beta,r,\zeta/\gamma}}{\E} \left[Y_j^{\pm}(x)\right],\]
        where $B_j = \min\{1,580/2^{2j}\} \in \Omega(\epsilon^2)$. We use \Cref{lem:speedy-refl-random-variable} to compute that
        \[W(V_{\alpha}, \ket{\overline{\pi}(Y_j^{\pm})}) \in O\left(\frac{d}{\alpha\beta\delta^2} + \frac{\eta}{\alpha\widetilde{\mu}_j^{\pm}(1-\widetilde{\mu}_j^{\pm})}\right) \subseteq \left(\frac{d}{\alpha\beta\delta^2}\right).\]
        In the last step, we used that $B_j/4 \leq \widetilde{\mu}_j^{\pm} \leq 1-B_j/4$, and so $\widetilde{\mu}_j^{\pm}(1-\widetilde{\mu}_j^{\pm}) \in \Omega(\epsilon^2)$, and $\eta \in O(\epsilon^2)$.
        Thus, by choosing $\alpha = \sqrt{d/(\beta\delta^2)}$, we obtain that both terms are $O(\sqrt{d/(\beta\delta^2)})$.
    \end{proof}

    We have now essentially recovered \cite[Algorithm~4]{cornelissen2023sublinear} as a transducer. In order to use this algorithm, though, we must supply an initial estimate to the mean, $\widetilde{m}$, satisfying $|\widetilde{m} - \mu| \leq 17\widetilde{\sigma}$. In order to find this initial estimate, we modify \cite[Algorithm~5]{cornelissen2023sublinear} to the continuous setting with a known anti-concentration bound.\footnote{Note that \cite{cornelissen2023sublinear} considers the setting where the energy levels are discrete, which is not the case in the volume estimation problem. This necessitates a modification of \cite[Algorithm~5]{cornelissen2023sublinear}, which is a detail that said paper forgoes. We describe such a procedure here.}
    
    \begin{theorem}
        \label{thm:median-estimation}
        Let $d \in \N$, $R \geq 1$, $\beta \geq 0$, $r \geq 1$, $0 < \rho < 1$ and a stepsize parameter $0 < \delta \leq \min\{1/(4\beta r), 1/(4096\sqrt{2d\beta}), 1/(4096\sqrt{d\ln(5120d)})\}$. Let $\gamma$ and $\widetilde{\pi}_{\beta,r,\zeta/\gamma}$ be as in \Cref{thm:speedy-to-ball}, and $B_d \subseteq K \subseteq RB_d$ be a convex body, accessible through a membership oracles $O_{K,\zeta}$ and $O_{K,\zeta/\gamma}$ with a spacing parameter $\zeta > 0$. Let $0 < \Delta \leq M$, $L > 0$, $\Omega = RB_d$, and $X : \Omega \to [0,M]$ be a piecewise $L$-Lipschitz continuous random variable, such that for every $t \in [0,M]$, the partition $\Omega = \{x \in \Omega : X(x) \leq t\} \sqcup \{x \in \Omega : X(x) > t\}$ is $h$-smooth. Let
        \[Q_X\left(\frac{9}{16}\right) - Q_X\left(\frac{7}{16}\right) \geq \Delta, \qquad \text{where} \qquad Q_X(p) := \inf\left\{t \in [0,M] : \underset{x \sim \pi_{\beta,r}}{\P} \left[X(x) \leq t\right] \geq p\right\}.\]
        Then, for sufficiently small $\zeta = \zeta(d,R,\beta,\delta,L,h,\Delta,M,\rho) > 0$, there is a bidirectional canonical a transducer that when run on $\ket{\widetilde{\pi}_{\beta,r,\zeta/\gamma}}$ restores it and outputs a value $\widetilde{m} \in [0,M]$ with transduction complexity that scales as $\widetilde{O}((1+\log(M/\Delta))\log(1/\rho)\sqrt{d/(\beta\delta^2)})$, such that with probability at least $1-\rho$, $Q_X(1/8) \leq \widetilde{m} \leq Q_X(7/8)$. Consequently,
        $|\widetilde{m} - \mu| \leq \sqrt{7}\sigma$, where
        \[\mu := \underset{x \sim \pi_{\beta,r}}{\E} \left[X(x)\right], \qquad \text{and} \qquad \sigma^2 := \underset{x \sim \pi_{\beta,r}}{\Var} \left[X(x)\right].\]
    \end{theorem}

    \begin{proof}
        Let $\eta = 1/64$. Let $V$ be the bidirectional canonical transducer from \Cref{thm:speedy-to-ball} that reflects around $\ket{\widetilde{\pi}_{\beta,r,\zeta/\gamma}}$. Since $V$ is a reflection transducer, we can use \Cref{thm:balancing-plus-minus} to obtain a transducer $V_{\alpha}$ with rebalanced complexities by an input-independent constant $\alpha > 0$, such that it acts as in \Cref{eq:ball-refl-a,eq:ball-refl-b}.

        Next, for every $p \in [0,1]$, we define
        \[\widetilde{Q}_X(p) := \inf\{t \in [0,M] : \underset{x \sim \widetilde{\pi}_{\beta,r}}{\P} [X(x) \leq t] \geq p\},\]
        and
        \[\widetilde{Q}_{X,\zeta}(p) := \inf\{t \in [0,M] : \underset{x \sim \widetilde{\pi}_{\beta,r,\zeta/\gamma}}{\P} [X(x) \leq t] \geq p\}.\]
        Observe that for any $t \in [0,M]$, $\mathbbm{1}_{X(x) \leq t}$ and $\widetilde{\pi}_{\beta,r}$ are piecewise $L$-Lipschitz continuous on an $h$-smooth partition. Therefore, we obtain by \Cref{lem:smoothness} that
        \[\sup_{t \in [0,M]} \left|\underset{x \sim \widetilde{\pi}_{\beta,r,\zeta/\gamma}}{\P} [X(x) \leq t] - \underset{x \sim \widetilde{\pi}_{\beta,r}}{\P} [X(x) \leq t]\right| \downarrow 0, \qquad (\zeta \downarrow 0).\]
        Thus, for a small enough choice of $\zeta$, this is at most $1/128$. Next, since $\norm{\widetilde{\pi}_{\beta,r} - \pi_{\beta,r}}_{\mathrm{TV}} \leq \eta = 1/64$, we have 
        \[\sup_{t \in [0,M]} \left|\underset{x \sim \widetilde{\pi}_{\beta,r}}{\P} [X(x) \leq t] - \underset{x \sim \pi_{\beta,r}}{\P} [X(x) \leq t]\right| \leq \frac{1}{64}.\]
        As such, we have
        \[\underset{x \sim \widetilde{\pi}_{\beta,r,\zeta/\gamma}}{\P} \left[X(x) \leq Q_X\left(\frac{7}{16}\right)\right] \geq \frac{7}{16} - \frac{3}{128} > \frac{13}{32}, \; \text{and} \; \underset{x \sim \widetilde{\pi}_{\beta,r,\zeta/\gamma}}{\P} \left[X(x) < Q_X\left(\frac{9}{16}\right)\right] \leq \frac{9}{16} + \frac{3}{128} < \frac{19}{32},\]
        and so
        \[\widetilde{Q}_{X,\zeta}\left(\frac{19}{32}\right) - \widetilde{Q}_{X,\zeta}\left(\frac{13}{32}\right) \geq Q_X\left(\frac{9}{16}\right) - Q_X\left(\frac{7}{16}\right) \geq \Delta.\]
        
        Next, we plug $V_{\alpha}$ into the construction from \Cref{thm:quantile-est} to construct a bidirectional canonical transducer $T = \mathtt{QuantileEst}(X,V_{\alpha},1/2,3/32,M,\Delta,\rho)$. Then, with probability at least $1-\rho$, the transducer outputs a value $\widetilde{m}$ that satisfies
        \[\underset{x \sim \widetilde{\pi}_{\beta,r,\zeta/\gamma}}{\P} \left[X(x) \leq \widetilde{m}\right] \in \left[7/32,25/32\right].\]
        Next, since $\norm{\widetilde{\pi}_{\beta,r} - \pi_{\beta,r}}_{\mathrm{TV}} \leq \eta = 1/64$, and switching to the discrete version incurs another overhead of at most $1/128$, we have
        \[\underset{x \sim \pi_{\beta,r}}{\P} \left[X(x) \leq \widetilde{m}\right] \in \left[7/32 - 1/32, 25/32 + 1/32\right] \subseteq \left[1/8,7/8\right],\]
        and so by the quantile inequality (see, e.g., \cite[Section~3.1]{bagui2004glimpses}), we observe that $|\widetilde{m} - \mu| \leq \sqrt{7}\sigma$.

        It remains to compute the transduction complexity. To that end, we have $W(V_{\alpha}, \ket{\widetilde{\pi}_{\beta,r,\zeta/\gamma}}) \in O(\alpha)$, and for any $t \in [0,M]$, we write
        \[X_t := \frac14 + \frac12 \cdot \mathbbm{1}_{X(x) \leq t}, \qquad \text{and} \qquad \widetilde{\mu}_t := \underset{x \sim \widetilde{\pi}_{\beta,r,\zeta/\gamma}}{\E} [X_t(x)].\]
        Thus, we have $1/4 \leq \widetilde{\mu}_t \leq 3/4$, and so by \Cref{lem:speedy-refl-random-variable}, we obtain
        \[W(V_{\alpha}, \overline{\pi}(X_t)) \in O\left(\frac{d}{\alpha\beta\delta^2} + \frac{\eta}{\alpha\mu_t(1-\mu_t)}\right) \subseteq O\left(\frac{d}{\alpha\beta\delta^2}\right).\]
        Setting $\alpha = \sqrt{d/(\beta\delta^2)}$ completes the proof.
    \end{proof}

    \subsection{Volume estimation}

    Finally, we port the volume estimation algorithm from Cousins and Vempala~\cite{CV18} to the quantum setting. To that end, we recall the estimators that they use in their construction. We introduce them here and give a simplified analysis of their relative variance~\cite[Lemma~7.8]{CV18}.

    \begin{lemma}
        \label{lem:relative-variance}
        Let $d \in \N$, $R \geq 1$ and $B_d \subseteq K \subseteq RB_d$. We let $\Omega = K$ and $H : \Omega \to \R_{\geq 0}$, defined as $H(x) = \norm{x}^2$. Let $\beta' < \beta$, and $X_{\beta',\beta}(x) := \exp(-(\beta'-\beta)H(x))$. Then, there exists $\xi \in [\beta',\beta]$ such that
        \[\underset{x \sim \pi_{\beta}}{\E} \left[X_{\beta',\beta}(x)\right] = \frac{Z(\beta')}{Z(\beta)} = \exp\left(\underset{x \sim \pi_{\xi}}{\E} \left[H(x)\right] (\beta - \beta')\right).\]
        Similarly, there exists a $\xi \in [2\beta'-\beta,\beta]$ such that
        \[\frac{\underset{x \sim \pi_{\beta}}{\Var} \left[X_{\beta',\beta}(x)\right]}{\underset{x \sim \pi_{\beta}}{\E} \left[X_{\beta',\beta}(x)\right]^2} = \frac{Z(2\beta'-\beta)Z(\beta)}{Z(\beta')^2} - 1 = \exp\left(\underset{x \sim \pi_{\xi}}{\Var} \left[H(x)\right] (\beta' - \beta)^2\right) - 1.\]
    \end{lemma}

    \begin{proof}
        We use the mean value theorem and \Cref{lem:log-part-derivatives} to obtain that there exists a $\xi$ in between $\beta,\beta'$ such that
        \begin{align*}
            \underset{x \sim \pi_{\beta}}{\E} \left[X_{\beta',\beta}(x)\right] &= \int_K \exp(-(\beta'-\beta)H(x)) \frac{\exp(-\beta H(x))}{Z(\beta)} \;\mathrm{d}x = \int_K \frac{\exp(-\beta' H(x))}{Z(\beta)} \;\mathrm{d}x = \frac{Z(\beta')}{Z(\beta)} \\
            &= \exp(L(\beta') - L(\beta)) = \exp\left(L'(\xi)(\beta' - \beta)\right) = \exp\left(-\underset{x \sim \pi_{\xi}}{\E} \left[H(x)\right] (\beta' - \beta)\right).
        \end{align*}
        Similarly, we use the (second order) mean value theorem again with \Cref{lem:log-part-derivatives} to show that there exists a $\xi$ between $2\beta' - \beta$ and $\beta$ such that
        \begin{align*}
            \frac{\underset{x \sim \pi_{\beta}}{\Var} \left[X_{\beta',\beta}(x)\right]}{\underset{x \sim \pi_{\beta}}{\E} \left[X_{\beta',\beta}(x)\right]^2} &= \frac{\underset{x \sim \pi_{\beta}}{\E}\left[X_{\beta',\beta}(x)^2\right]}{\underset{x \sim \pi_{\beta}}{\E}\left[X_{\beta',\beta}(x)\right]^2} - 1 = \frac{Z(\beta)^2}{Z(\beta')^2} \int_K \exp(-2(\beta'-\beta)H(x)) \cdot \frac{\exp(-\beta H(x))}{Z(\beta)} \;\mathrm{d}x - 1 \\
            &= \frac{Z(2\beta' - \beta)Z(\beta)}{Z(\beta')^2} - 1 = \exp\left(L(2\beta' - \beta) + L(\beta) - 2L(\beta')\right) - 1 \\
            &= \exp\left(L''(\xi)(\beta' - \beta)^2\right) - 1 = \exp\left(\underset{x \sim \pi_{\xi}}{\Var} \left[H(x)\right] (\beta' - \beta)^2\right) - 1.\qedhere
        \end{align*}
    \end{proof}

    Note that we can combine the expressions proved in the above lemma with the bounds on the expectation and variance provided in \Cref{thm:variance-bounds}. We will see later, however, that these bounds are not strong enough, and that we need an anti-concentration inequality on the random variable $\norm{x}^2$ over the Gibbs distribution instead. We provide it in the following lemma, whose proof is provided in \cref{lem:lb}.

    \begin{lemma}
        \label{lem:anti-concentration}
        Let $4 \leq d \in \N$, $B_d \subseteq K \subseteq \R^d$, and $\beta > 0$. Then, there exists a universal constant $C>0$ such that for any $0 \leq a < b$, we have
        \[\underset{x \sim \pi_{\beta}}{\P} \left[a \leq \norm{x}^2 \leq b\right] \leq C \max\{d,\beta/\sqrt{d}\} \cdot (b-a).\]
    \end{lemma}

    Next, we describe our canonical transducer for volume estimation.
    
    \begin{breakablealgorithm}
        \caption{Canonical quantum transducer for volume estimation}
        \label{alg:vol-est-algo}
        
        \noindent\textbf{Input:}
        \begin{enumerate}[nosep]
            \item $4 \leq d \in \N$, $R \geq 1$, $0 < \varepsilon < 1/4$ and a sufficiently small $\zeta = \zeta(d,R,\varepsilon) > 0$.
            \item A convex body $B_d \subseteq K \subseteq RB_d$, accessible through membership oracles $O_{K,\zeta}$ and $O_{K,\zeta/\gamma}$.
        \end{enumerate}

        \noindent\textbf{Parameters:}
        \begin{enumerate}[nosep]
            \item $C$: the universal constant from \Cref{lem:anti-concentration}.
            \item $\gamma = 1-1/(2d)$.
            \item $\xi_0 := \sqrt{1 + 3\ln(8/\varepsilon)/(2d)}$.
            \item $\ell_{\max} := 3 + 16\sqrt{d}(\ln(\xi_0^2R^2) + \ln(d))$.
            \item $\nu := \varepsilon/(24e\ell_{\max})$.
            \item $\xi := \sqrt{1+3\ln(2/\nu^2)/(2d)}$.
            \item $\kappa := \sqrt{d\ln(20d/\min\{\nu/37120,\nu^2/(5824 \cdot 1160^2)\})}$.
            \item $t := 20\sqrt{\ell_{\max}}/\varepsilon$.
        \end{enumerate}

        \noindent\textbf{Output:} An $\varepsilon$-relative estimate of $\Vol(K)$.

        \noindent\textbf{Success probability:} At least $2/3$.

        \noindent\textbf{Procedure:}
        
        \begin{enumerate}[nosep]
            \item Initialize $\beta_0 = d\xi_0^2$, $r_0 = \sqrt{d/\beta_0}\xi$, $\delta_0 = \min\{1/(4\beta_0r_0), 1/(2^{18}\sqrt{2d\beta_0}), 1/(2^{18}\kappa)\}$ and $j = 0$.
            \item Initialize a quantum register in $\ket{\pi_{\beta_0,\delta_0,1,\zeta}}$.
            \item Use \Cref{cor:anneal-small-steps} to map $\ket{\pi_{\beta_0,\delta_0,1,\zeta}}$ to $\ket{\pi_{\beta_0,\delta_0,r_0,\zeta}}$
            \item While $\beta_j > 0$:
            \begin{enumerate}[nosep]
                \item Set $\beta_{j+1} = \beta_j \cdot (1 - \max\{1/(8\sqrt{d}), 1/(8R\sqrt{\beta_j})\})$ if $\beta_j > 1/(2R)^2$, else set $\beta_{j+1} = 0$.
                \item Use \Cref{thm:speedy-to-ball} to map $\ket{\pi_{\beta_j,\delta_j,r_j,\zeta}}$ to $\ket{\widetilde{\pi}_{\beta_j,r_j,\zeta/\gamma}}$.
                \item Use \Cref{thm:median-estimation} to calculate a value $\widetilde{m}$, with random variable $x \mapsto \norm{x}^2$, $M = R^2$, and $\Delta = 1/(8C\max\{d,\beta/\sqrt{d}\})$ and $\rho = 1/(12\ell_{\max})$.
                \item Use \Cref{thm:mean-est} to calculate a value $\overline{X}_j$, where we take the random variable $X_j := X_{\beta_{j+1},\beta_j}$ on $\ket{\widetilde{\pi}_{\beta_j,r_j,\zeta/\gamma}}$, with $\widetilde{m}_j = \exp(-(\beta_{j+1}-\beta_j)\widetilde{m})$, $\widetilde{\sigma}_j = \widetilde{m}_j$, repetition parameter $t$, and precision parameter $\nu$.
                \item Use \Cref{thm:speedy-to-ball} in reverse to map $\ket{\widetilde{\pi}_{\beta_j,r_j,\zeta/\gamma}}$ to $\ket{\pi_{\beta_j,\delta_j,r_j,\zeta}}$.
                \item If $\beta_{j+1} > 0$,
                \begin{enumerate}
                    \item Set $r_{j+1} = \sqrt{d/\beta_{j+1}}\xi$.
                    \item Set $\delta_{j+1} = \min\{1/(4\beta_{j+1}r_{j+1}), 1/(2^{18}\sqrt{2d\beta_{j+1}}), 1/(2^{18}\kappa)\}$.
                    \item Use \Cref{cor:anneal-small-steps} to map $\ket{\pi_{\beta_j,\delta_j,r_j,\zeta}}$ to $\ket{\pi_{\beta_{j+1},\delta_{j+1},r_{j+1},\zeta}}$.
                \end{enumerate}
                \item Increment $j$.
            \end{enumerate}
            \item Output $\overline{X} = \prod_{j=0}^{\ell-1} \overline{X}_j \cdot \int_{B_d} \exp(-\beta_0\norm{x}^2) \;\mathrm{d}x$, where $\ell$ is the final value of $j$.
        \end{enumerate}
    \end{breakablealgorithm}

    We now analyze \Cref{alg:vol-est-algo}.

    \begin{theorem}
        \label{thm:vol-est-algo}
        Let $4 \leq d \in \N$, $R \geq 1$ and $0 < \varepsilon < 1$. Then, there is a bounded-error quantum algorithm that returns an $\varepsilon$-relative estimate of $\Vol(K)$ with a number of membership queries that satisfies
        \[\widetilde{O}\left(\frac{d^{5/4}(\sqrt{d} + R)}{\varepsilon} + d^2\right),\]
        where the tilde hides polylogarithmic factors in $d$, $R$ and $1/\varepsilon$.
    \end{theorem}

    \begin{proof}
    Our algorithm is based on the canonical transducer algorithm shown in \cref{alg:vol-est-algo}.
        We start by analyzing the length of the (inverse) cooling schedule. To that end, we write $\ell_1 := \min\{j \in \N : \beta_j \leq d/R^2\}$, and $\ell_2 := \min\{j \in \N : \beta_{j+\ell_1} \leq 1/(2R)^2\}$. Then, from \Cref{lem:schedule-length}, we obtain that $\ell_1 \leq 1 + 16\sqrt{d}\ln(\xi_0^2R^2)$ and $\ell_2 \leq 1 + 16\sqrt{d}\ln(d)$. Thus, $\ell = \ell_1 + \ell_2 + 1 \leq 3 + 16\sqrt{d}(\ln(\xi_0^2R^2) + \ln(d)) = \ell_{\max}$.
        
        Next, we analyze the variance of the random variables. To that end, we write
        \[\mu_j := \underset{x \sim \pi_{\beta}}{\E} [X_j], \qquad \text{and} \qquad \widetilde{\mu}_j := \underset{x \sim \pi_{\beta,r}}{\E} [X_j].\]
        We observe that $\widetilde{\mu}_j \leq \mu_j$. Let $j \in \{0,\dots,\ell-2\}$. We observe that $\beta_j > 1/(2R)^2$, and consequently we have that $2\beta_{j+1} - \beta_j > 0$. Thus,
        \[\frac{\beta_j}{2\beta_{j+1} - \beta_j} \leq \frac{1}{1 - \max\{\frac{1}{4\sqrt{d}}, \frac{1}{4R\sqrt{\beta_j}}\}} \leq \frac{1}{1-\frac12} = 2,\]
        and so by \Cref{lem:relative-variance,thm:variance-bounds} there exists a $\xi \in [2\beta_{j+1}-\beta_j,\beta_j]$ such that
        \begin{align*}
            \frac{\underset{x \sim \pi_{\beta_j}}{\Var} \left[X_j\right]}{\underset{x \sim \pi_{\beta_j}}{\E} \left[X_j\right]^2} &\leq \exp\left(\underset{x \sim \pi_{\xi}}{\Var} \left[\norm{x}^2\right] \cdot (\beta_j - \beta_{j+1})^2\right) - 1 \\
            &\leq \exp\left(\min\left\{\frac{d}{(2\beta_{j+1}-\beta_j)^2}, \frac{2R^2}{2\beta_{j+1}-\beta_j}\right\} \cdot \beta_j^2 \cdot \max\left\{\frac{1}{8\sqrt{d}}, \frac{1}{8R\sqrt{\beta_j}}\right\}^2\right) - 1 \\
            &\leq \exp\left(\frac{1}{16}\right) - 1 \leq \frac{1}{15}.
        \end{align*}
        On the other hand, observe that $1/(4R)^2 \leq \beta_{\ell-1} \leq 1/(2R)^2$. Then by \Cref{lem:relative-variance} and Popoviciu's inequality, there exists a $\xi \in [-\beta_{\ell-1},\beta_{\ell-1}]$ such that
        \[\frac{\underset{x \sim \pi_{\beta_{\ell-1}}}{\Var} \left[X_{\ell-1}\right]}{\underset{x \sim \pi_{\beta_{\ell-1}}}{\E} \left[X_{\ell-1}\right]^2} \leq \exp\left(\underset{x \sim \pi_{\xi}}{\Var} \left[\norm{x}^2\right] \cdot \beta_{\ell-1}^2\right) - 1 \leq \exp\left(\frac{R^4}{4} \cdot \frac{1}{(2R)^4}\right) - 1 \leq \exp\left(\frac{1}{16}\right) - 1 \leq \frac{1}{15}.\]
        Since $\pi_{\beta_j,r_j}$ is a Gibbs distribution over $K \cap rB_d$, the last two equations also hold for this distribution. Thus, for all $j \in \{0, \dots, \ell-1\}$ that
        \[\underset{x \sim \pi_{\beta_j,r_j}}{\Var} [X_j] \leq \frac{\widetilde{\mu}_j^2}{15} \leq \frac{\widetilde{m}_j^2}{15\left(1 - \frac{1}{\sqrt{2}}\right)^2} < \widetilde{m}_j^2 = \widetilde{\sigma}_j^2.\]
        Finally, observe that $X_j$ is a bounded random variable, taking values in $[1,\exp(-(\beta_{j+1}-\beta_j)R^2)]$. Moreover, for any value $t$ in its range, we obtain that $\{x \in RB_d : X(x) \leq t\}$ is a ball centered at the origin. If we intersect this with $K$, we obtain a convex set, which is $h$-smooth where $h$ only depends on $d$ and $R$. Thus, the premise of \Cref{thm:mean-est} is satisfied.

        Next, we analyze the perturbation of the median estimate. We start with the anti-concentration bound. In the $j$th iteration, we consider the random variable $x \mapsto \norm{x}^2$ over $\pi_{\beta_j,r_j}$. Let $a := Q_{\norm{x}^2}(7/16)$ and $b := Q_{\norm{x}^2}(9/16)$ be the $7/16$- and $9/16$-quantile of this random variable. We observe from \Cref{lem:anti-concentration} that
        \[\frac18 \leq \underset{x \sim \pi_{\beta_j,r_j}}{\P} \left[a \leq \norm{x}^2 \leq b\right] \leq C\max\{d,\beta/\sqrt{d}\} \cdot (b-a).\]
        Hence, we obtain that $b-a \geq 1/(8C\max\{d,\beta/\sqrt{d}\}) = \Delta$, and so the premise of \Cref{thm:median-estimation} is satisfied. This also means that with probability at least $11/12$, every run of the median estimation algorithm succeeds. Next, we observe from \Cref{thm:median-estimation} that $Q_{\norm{x}^2}(1/8) \leq \widetilde{m} \leq Q_{\norm{x}^2}(7/8)$. Since $X_j$ is monotone in $\norm{x}^2$, this implies that $Q_{X_j}(1/8) \leq \widetilde{m}_j \leq Q_{X_j}(7/8)$, and so using the quantile inequality \cite[Section~3.1]{bagui2004glimpses} and the previously developed variance bounds, we obtain that for all $j \in \{0, \dots, \ell-1\}$,
        \[|\widetilde{m}_j - \widetilde{\mu}_j| \leq \sqrt{7\underset{x \sim \pi_{\beta_j,r_j}}{\Var}[X_j]} \leq \sqrt{\frac{7}{15}} \cdot \widetilde{\mu}_j \leq \frac{\widetilde{\mu}_j}{\sqrt{2}},\]
        from which we find that $|\widetilde{m}_j - \widetilde{\mu}_j| \leq \sqrt{7}\widetilde{\sigma}_j < 17\widetilde{\sigma}_j$, as required by the premise of \Cref{thm:mean-est}. Moreover, since $\widetilde{\sigma}_j = \widetilde{m}_j$, we obtain that $(1-1/\sqrt{2})\widetilde{\mu}_j \leq \widetilde{m}_j = \widetilde{\sigma}_j \leq 2\widetilde{\mu}_j$.
        
        Next, we observe from \Cref{lem:clipping-perturbation} that $\norm{\pi_{\beta_j,r_j} - \pi_{\beta_j}}_{\mathrm{TV}} \leq \nu^2/2 =: \eta$. We observe that $\eta < 1/4$. Thus, by combining a perturbation inequality with the variance bounds we just obtained, we deduce that for all $j \in \{0, \dots, \ell-1\}$,
        \[|\mu_j - \widetilde{\mu}_j| \leq \left[\sqrt{\underset{x \sim \pi_{\beta_j,r_j}}{\Var} [X_j]} + \sqrt{\underset{x \sim \pi_{\beta_j}}{\Var} [X_j]}\right] \cdot \sqrt{\frac{\eta}{1 - \eta}} \leq \sqrt{\frac{\eta}{15}} \cdot (\widetilde{\mu}_j + \mu_j) \leq \frac{\nu\mu_j}{2}.\]
        Thus, using \Cref{lem:relative-variance} to obtain that $\mu_j = Z(\beta_{j+1})/Z(\beta_j)$, we employ \Cref{thm:mean-est} to observe that for all $j \in \{0, \dots, \ell-1\}$,
        \[\left|\frac{\E[\overline{X}_j]}{\frac{Z(\beta_{j+1})}{Z(\beta_j)}} - 1\right| = \frac{|\E[\overline{X}_j] - \mu_j|}{\mu_j} \leq \frac{|\E[\overline{X}_j] - \widetilde{\mu}_j|}{\mu_j} + \frac{|\widetilde{\mu}_j - \mu_j|}{\mu_j} \leq \frac{\nu\widetilde{\sigma}_j}{\mu_j} + \frac{\nu\mu_j}{2\mu_j} \leq \frac{2\nu\widetilde{\mu}_j}{\mu_j} + \frac{\nu}{2} \leq 2\nu\left(1 + \frac{\nu}{2}\right) + \frac{\nu}{2} < 3\nu.\]
        We also write $C_0 := \int_{B_d} \exp(-\beta_0\norm{x}^2) \;\mathrm{d}x$ for short, and observe from \Cref{lem:clipping-perturbation} that
        \[\left|\frac{C_0}{Z(\beta_0)} - 1\right| = \underset{x \sim \pi_{\beta_0}}{\P} [\norm{x} \geq 1] \leq \frac{\varepsilon}{8}.\]
        Merging these observation together, we obtain
        \begin{align*}
            \left|\frac{\E[\overline{X}]}{\Vol(K)} - 1\right| &= \left|\prod_{j=0}^{\ell-1} \frac{\E[\overline{X}_j]}{\frac{Z(\beta_{j+1})}{Z(\beta_j)}} \cdot \frac{C_0}{Z(\beta_0)} - 1\right| \leq \left(3\ell\nu + \frac{\varepsilon}{8}\right)\left(1 + 3\nu\right)^{\ell}\left(1 + \frac{\varepsilon}{8}\right) \leq \frac{\varepsilon}{4}.
        \end{align*}
        Next, observe that for all $j \in \{0, \dots, \ell-1\}$, we have
        \[\frac{\widetilde{\mu}_j}{\E[\overline{X}_j]} \leq \frac{\widetilde{\mu}_j}{\widetilde{\mu}_j - \left|\E[\overline{X}_j] - \widetilde{\mu}_j\right|} \leq \frac{1}{1 - \frac{\nu\widetilde{\sigma}_j}{\widetilde{\mu_j}}} \leq \frac{1}{1 - 2\nu} < 2,\]
        and so from \Cref{thm:mean-est}, we find that the variance of the estimator $\overline{X}$ satisfies
        \begin{align*}
            1 + \frac{\Var[\overline{X}]}{\E[\overline{X}]^2} &= \prod_{j=0}^{\ell-1} \left(1 + \frac{\Var[\overline{X}_j]}{\E[\overline{X}_j]^2}\right) \leq \prod_{j=0}^{\ell-1} \left(1 + \frac{\widetilde{\sigma}_j^2}{t^2 \E[\overline{X}_j]^2}\right) \leq \prod_{j=0}^{\ell-1} \left(1 + \frac{4\widetilde{\mu}_j^2}{t^2\E[\overline{X}_j]^2}\right) \\
            &< \left(1 + \frac{16}{t^2}\right)^{\ell} \leq \left(1 + \frac{\varepsilon^2}{24\ell}\right)^{\ell} \leq \exp\left(\frac{\varepsilon^2}{24}\right) \leq 1 + \frac{\varepsilon^2}{16}.
        \end{align*}

        Finally, using Chebyshev's inequality, we obtain that
        \[\P\left[\left|\frac{\overline{X}}{\Vol(K)} - 1\right| > \varepsilon\right] \leq \P\left[\left|\frac{\overline{X}}{\E[\overline{X}]} - 1\right| > \frac{\varepsilon}{2}\right] \leq \frac{4\Var[\overline{X}]}{\E[\overline{X}]^2\varepsilon^2} \leq \frac14,\]
        and so the total success probability is at least $3/4 - 1/12 = 2/3$.
        
        It remains to analyze the total transduction complexity. To that end, we adopt the convention that tildes are hiding polylogarithmic factors in terms of $d$, $R$ and $1/\varepsilon$ in the big-$O$-notation. For the parameters in the algorithm, we obtain that
        \[\xi_0 \in \widetilde{O}(1), \quad \ell_{\max} \in \widetilde{O}(\sqrt{d}), \quad \nu \in \widetilde{\Omega}(\varepsilon/\sqrt{d}), \quad \xi \in \widetilde{O}(1), \quad \kappa \in \widetilde{O}(\sqrt{d}), \quad \text{and} \quad t \in \widetilde{O}(d^{1/4}/\varepsilon).\]

        In the $j$th iteration, the individual steps have the following costs:
        \begin{enumerate}[nosep]
            \item The rejection sampling step costs $O(\sqrt{d/(\beta_j\delta_j^2)})$.
            \item For the median estimation, observe that $\beta_j \leq \beta_0 \in O(d)$, and so $\Delta \in \widetilde{\Omega}(1/d)$. Thus, the total cost becomes $\widetilde{O}(\sqrt{d/(\beta_j\delta_j^2)})$.
            \item The mean estimation costs $\widetilde{O}(\sqrt{d/(\beta_j\delta_j^2)} \cdot d^{1/4}/\varepsilon)$.
            \item The annealing step costs $\widetilde{O}(\sqrt{d/(\beta_j\delta_j^2)} \cdot (1 + d\log(\delta_{j+1}/\delta_j)))$.
        \end{enumerate}
        Thus, the total cost for the $j$th iteration is
        \[\widetilde{O}\left(\sqrt{\frac{d}{\beta_j\delta_j^2}} \cdot \left(\frac{d^{1/4}}{\varepsilon} + d\log\frac{\delta_{j+1}}{\delta_j}\right)\right).\]

        Next, for all $j \in \{0, \dots, \ell-1\}$, observe that
        \[\delta_j = \min\left\{\frac{C_1}{\sqrt{\beta_j}}, C_2\right\}, \qquad \text{where} \qquad C_1 = \frac{\min\{1/2^{18},1/(2\xi)\}}{\sqrt{2d}}, \qquad \text{and} \qquad C_2 = \frac{1}{2^{18}\kappa}.\]
        Suppose that $\beta_j \geq (C_1/C_2)^2$, then
        \[\frac{\delta_{j+1}}{\delta_j} - 1 \leq \sqrt{\frac{\beta_j}{\beta_{j+1}}} - 1 \in O\left(\max\left\{\frac{1}{\sqrt{d}}, \frac{1}{R\sqrt{\beta_j}}\right\}\right).\]
        On the other hand, if $\beta_j \leq (C_1/C_2)^2$, then $\delta_{j+1}/\delta_j = 1$.

        Next, we distinguish between four cases:
        \begin{enumerate}[nosep]
            \item Suppose $\beta_j \geq \max\{(C_1/C_2)^2,d/R^2\}$. Then, we have $\ln(\delta_{j+1}/\delta_j) \in O(1/\sqrt{d})$, and $\delta_j = C_1/\sqrt{\beta_j} \in \widetilde{\Omega}(1/\sqrt{d\beta_j})$, and so $\sqrt{d/(\beta_j\delta_j^2)} \in \widetilde{O}(d)$.
            \item Suppose $(C_1/C_2)^2 \geq \beta_j \geq d/R^2$. Then, we have $\ln(\delta_{j+1}/\delta_j) = 0$ and $\delta_j = C_2 \in \Omega(1/\kappa) \subseteq \widetilde{\Omega}(1/\sqrt{d})$ and so $\sqrt{d/(\beta_j\delta_j^2)} \in O(d/\sqrt{\beta_j})$.
            \item Suppose $d/R^2 \geq \beta_j \geq (C_1/C_2)^2$. Then, we have $\ln(\delta_{j+1}/\delta_j) \in O(1/(R\sqrt{\beta_j}))$, and $\delta_j = C_1/\sqrt{\beta_j} \in \widetilde{\Omega}(1/\sqrt{d\beta_j})$, and so $\sqrt{d/(\beta_j\delta_j^2)} \in \widetilde{O}(d)$.
            \item Suppose $\beta_j \leq \min\{(C_1/C_2)^2,d/R^2\}$. Then, we have $\ln(\delta_{j+1}/\delta_j) = 0$ and $\delta_j = C_2 \in \Omega(1/\kappa) \subseteq \widetilde{\Omega}(1/\sqrt{d})$ and so $\sqrt{d/(\beta_j\delta_j^2)} \in O(d/\sqrt{\beta_j})$.
        \end{enumerate}

        Finally, we analyze the algorithm in two different cases:
        \begin{enumerate}[nosep]
            \item If $(C_1/C_2)^2 \geq d/R^2$, then we write $\ell_1^{(a)} = \min\{j \in [\ell] : \beta_j \leq (C_1/C_2)^2\}$, and $\ell_1^{(b)} = \ell_1 - \ell_1^{(a)} - 1$. Then, we use \Cref{eq:stage-1,eq:stage-2} to observe that the total transduction complexity $T$ is
            \begin{align*}
                T &\in \widetilde{O}\left(\frac{d^{1/4}}{\varepsilon}\sum_{j=0}^{\ell-1} \sqrt{\frac{d}{\beta_j\delta_j^2}} + d\sum_{j=0}^{\ell-2} \sqrt{\frac{d}{\beta_j\delta_j^2}} \log\frac{\delta_{j+1}}{\delta_j}\right) \\
                &\subseteq \widetilde{O}\left(\frac{d^{1/4}}{\varepsilon}\left[\sum_{j=0}^{\ell_1^{(a)}-1} d + \sum_{j=\ell_1^{(a)}}^{\ell_1-1} \frac{d}{\sqrt{\beta_j}} + \sum_{j=\ell_1}^{\ell-2} \frac{d}{\sqrt{\beta_j}} + \frac{d}{\sqrt{\beta_{\ell-1}}}\right] + d\left[\sum_{j=0}^{\ell_1^{(a)}-1} \sqrt{d}\right]\right) \\
                &\subseteq \widetilde{O}\left(\frac{d^{1/4}}{\varepsilon}\left[d \cdot 4\sqrt{d}\ln\frac{\beta_0}{\beta_{\ell_1^{(a)}-1}} + d \cdot \frac{4\sqrt{d}}{\sqrt{\beta_{\ell_1-1}}} \ln\frac{\beta_{\ell_1^{(a)}}}{\beta_{\ell_1}} + d \cdot 4R \ln\frac{\beta_{\ell_1}}{\beta_{\ell-1}} + \frac{d}{\sqrt{\beta_{\ell-1}}}\right]\right. \\
                &\qquad \left. + d\left[\sqrt{d} \cdot 4\sqrt{d}\ln\frac{\beta_0}{\beta_{\ell_1^{(a)}-1}}\right]\right) \\
                &\subseteq \widetilde{O}\left(\frac{d^{1/4}}{\varepsilon}\left[d^{3/2} + dR\right] + d \cdot d\right) \subseteq \widetilde{O}\left(\frac{d^{5/4}(\sqrt{d} + R)}{\varepsilon} + d^2\right).
            \end{align*}
            \item If $(C_1/C_2)^2 \leq d/R^2$, then we write $\ell_2^{(a)} = \min\{j \in [\ell] : \beta_j \leq (C_1/C_2)^2\} - \ell_1$, and $\ell_2^{(b)} = \ell_2 - \ell_2^{(a)}$. Then, we again use \Cref{eq:stage-1,eq:stage-2} to analyze the total transduction complexity $T$, as
            \begin{align*}
                T &\in \widetilde{O}\left(\frac{d^{1/4}}{\varepsilon}\sum_{j=0}^{\ell-1} \sqrt{\frac{d}{\beta_j\delta_j^2}} + d \sum_{j=0}^{\ell-2} \sqrt{\frac{d}{\beta_j\delta_j^2}}\log\frac{\delta_{j+1}}{\delta_j}\right) \\
                &\subseteq \widetilde{O}\left(\frac{d^{1/4}}{\varepsilon}\left[\sum_{j=0}^{\ell_1-1} d + \sum_{j=\ell_1}^{\ell-2} \frac{d}{\sqrt{\beta_j}} + \frac{d}{\sqrt{\beta_{\ell-1}}}\right] + d\left[\sum_{j=0}^{\ell_1-1} \sqrt{d} + \sum_{j=\ell_1}^{\ell_2^{(a)}-1} \frac{d}{R\sqrt{\beta_j}}\right]\right) \\
                &\subseteq \widetilde{O}\left(\frac{d^{1/4}}{\varepsilon}\left[d \cdot 4\sqrt{d} \ln\frac{\beta_0}{\beta_{\ell_1}} + d \cdot 4R\ln\frac{\beta_{\ell_1}}{\beta_{\ell-1}} + \frac{d}{\sqrt{\beta_{\ell-1}}}\right] + d\left[\sqrt{d} \cdot 4\sqrt{d}\ln\frac{\beta_0}{\beta_{\ell_1}} + \frac{d}{R} \cdot 4R\ln\frac{\beta_{\ell_1}}{\beta_{\ell-1}}\right]\right) \\
                &\subseteq \widetilde{O}\left(\frac{d^{1/4}}{\varepsilon}\left[d^{3/2} + dR\right] + d\left[d + d\right]\right) \subseteq \widetilde{O}\left(\frac{d^{5/4}(\sqrt{d} + R)}{\varepsilon} + d^2\right).
            \end{align*}
        \end{enumerate}

        This proves that \cref{alg:vol-est-algo} is a canonical transducer (call it $U$) that correctly approximates the volume with the claimed success probability and transduction complexity $T \in \widetilde{O}\left(\frac{d^{5/4}(\sqrt{d} + R)}{\varepsilon} + d^2\right)$.
        It remains to note that we can now invoke \cref{thm:transducer-to-alg} with $K \in O(W)$ calls to $U$ to turn this into an actual quantum algorithm with constant state error $\eta$, which suffices to return the estimate with constant probability.
        Since $U$ is a canonical transducer, it makes a single query, and hence the resulting quantum algorithm makes the claimed number of queries.
    \end{proof}

    We remark here that the description of $\zeta > 0$ ``sufficiently small'' is rather vague. One could conceivably go through the entire analysis and figure out exactly what choice of $\zeta$ would work. However, since the exact choice of $\zeta$ does not impact the transduction complexity (and hence the query complexity), we neglect this analysis in this work. On the other hand, tracking the time complexity of our algorithm necessitates doing this analysis, but we leave this for future work.
   
    \section*{Acknowledgments}

    AC is supported by ERC Starting Grant 101163189.
    SA is supported by the French ANR project QUOPS (ANR-22-CE47-0003-01) and the French PEPR integrated project HQI (ANR-22-PNCQ-0002).

    The key ideas of this paper were developed over the course of the past year, without the use of AI. When finishing the write-up, we used GPT-6 Astra to survey the relevant literature, complete the proofs of \cref{lem:truncated-Gamma,lem:volume-bound}, and pinpoint technical mistakes and typos.

    \bibliographystyle{alphaurl}
    \bibliography{references}

    \appendix

    \section{Gaussian norm bounds}
    \label{sec:app}

    Here we prove \Cref{thm:variance-bounds}, we establish the upper and lower bounds separately in \cref{lem:ub,lem:lb}. 
    
    \begin{lemma} \label{lem:ub}
        Let $d \in \N$, $R \geq 1$ and $B_d \subseteq K \subseteq RB_d$ convex. Let $\Omega = K$, and $H(x) = \norm{x}^2$. Then, for all $\beta > 0$, we have
        \begin{align*}
             \underset{x \sim \pi_{\beta}}{\E} \left[H(x)\right] \leq \min\left\{\frac{d}{2\beta}, R^2\right\}, \qquad  \underset{x \sim \pi_{\beta}}{\Var} \left[H(x)\right] \leq \min\left\{\frac{d}{\beta^2}, \frac{2R^2}{\beta}, \frac{R^4}{4}\right\}.
        \end{align*}
    \end{lemma}

    \begin{proof} 
        The upper bound of $R^2$ for $\E[\norm{x}^2]$ is trivial, so it remains to prove that $\E[\norm{x}^2] \leq d/(2\beta)$. To that end, consider the function $f : \R_{\geq 0} \times \R^d \mapsto \R_{\geq 0}$, defined by $f(\beta,x) = \exp(-\norm{x}^2) \cdot \mathbbm 1_K(x/\sqrt{\beta})$. Then we can see directly that $f$ is logconcave, and hence, by the Prékopa-Leindler inequality, so are all its marginals. By integrating the last $d$ dimensions, we obtain that $g : \R_{\geq 0} \to \R_{\geq 0}$, defined by $g(\beta) = \beta^{d/2}\int_K \exp(-\beta\norm{x}^2)\;\mathrm{d}x = \beta^{d/2}Z(\beta)$, is logconcave.\footnote{Note that Cousins and Vempala claim on page 1266 that the function $z(a) = a^{d+1}\int_K \exp(-a\norm{x}^2/2) \;\mathrm{d}x$ is logconcave, but this does not follow from their earlier-stated Lemma~7.11. What is true, though, is that $z(a) = a^{d/2}\int_K \exp(-a\norm{x}^2) \;\mathrm{d}x$ is logconcave (as proven above), and this is enough to fix their proof.} From the logconcavity of $g$, we find that
        \[0 \geq \frac{\mathrm{d}^2}{\mathrm{d}\beta^2} \left[\log\left(\beta^{\frac{d}{2}}Z(\beta)\right)\right] = -\frac{d}{2\beta^2} + \frac{\mathrm{d}}{\mathrm{d}\beta} \left[\frac{Z'(\beta)}{Z(\beta)}\right] = -\frac{d}{2\beta^2} - \frac{\mathrm{d}}{\mathrm{d}\beta} \underset{x \sim \pi_{\beta}}{\E} \left[\norm{x}^2\right].\]
        As $\E_{x \sim \pi_{\beta}}[\norm{x}^2] \to 0$ as $\beta \to \infty$, we have
        \[\underset{x \sim \pi_{\beta}}{\E}\left[\norm{x}^2\right] = -\int_\beta^{\infty} \frac{\mathrm{d}}{\mathrm{d}s} \underset{x \sim \pi_s}{\E} \left[\norm{x}^2\right] \;\mathrm{d}s \leq \int_\beta^{\infty} \frac{d}{2s^2} \;\mathrm{d}s = \frac{d}{2\beta}.\]
        This completes the upper bounds for $\E_{x \sim \pi_{\beta}}[\norm{x}^2]$.
    
        For the next claim, we consider the Brascamp-Lieb inequality, \cite[Theorem~4.1]{brascamp1976extensions}. We let $f(x) = \beta\norm{x}^2$ on $K$, and let it go to infinity arbitrarily fast outside of $K$. Let $h(x) = \norm{x}^2$ on $K$, and let it go to $0$ arbitrarily fast outside. Then on $K$, we have $\nabla^2f(x) = 2\beta I$ and $\nabla h(x) = 2x$, and so
        \[\underset{x \sim \pi_{\beta}}{\Var}[\norm{x}^2] \leq 2\underset{x \sim \pi_{\beta}}{\E}\left[x^T \frac{I}{\beta}x\right] = \frac{2}{\beta} \underset{x \sim \pi_{\beta}}{\E} \left[\norm{x}^2\right].\]
        Then, the first two upper bounds on $\Var_{x \sim \pi_{\beta}}[\norm{x}^2]$ follow from the upper bounds on $\E_{x \sim \pi_{\beta}}[\norm{x}^2]$, and the last one follows from Popoviciu's inequality.
    \end{proof}

    To prove the lower bounds, we will use the following properties of the truncated Gamma distribution. 

    \begin{lemma}[Truncated Gamma distribution] \label{lem:truncated-Gamma}
        Let $\alpha \geq 2$, $\beta,L>0$ and consider a random variable $S$ drawn from the truncated Gamma distribution whose density is $p(s) = \frac{s^{\alpha-1}e^{-\beta s}}{\int_0^L t^{\alpha-1}e^{-\beta t} \, dt}$ for $s \in (0,L]$.
        Then there exist universal constants $c,C>0$ such that 
        \begin{enumerate}
            \item $\E[S] \geq c \min\{L,\frac{\alpha}{\beta}\}$.
            \item $\Var[S] \geq c \min\{\frac{L^2}{\alpha^2}, \frac{\alpha}{\beta^2}\}$.
            \item  For all $0 \leq a < b \leq L$ we have $\P(a\leq S \leq b) \leq C \max\{\frac{\alpha}{L}, \frac{\beta}{\sqrt{\alpha}}\} (b-a)$.
        \end{enumerate}
    \end{lemma}
    \begin{proof}
 For the second and third item, we will use the elementary fact that if a density function $p$ is upper bounded by $M$, then $\P[a \leq S \leq b] \leq \|p\|_\infty (b-a)$ and $\Var[S] \in \Omega(1/\|p\|_{\infty}^2)$, where $\|p\|_\infty = \sup_{s \in [0,L]} p(s)$. Throughout, we let $h(s) = s^{\alpha-1}e^{-\beta s}$ so that $p(s) = h(s)/Z$ where $Z = \int_0^L h(s) \ ds$.

We now distinguish two cases: 
\begin{enumerate}
    \item[(i)] $\beta L \leq 2\alpha$: 
    By setting the derivative to zero, we find that the maximum of $h$ on $[0,L]$ is attained at $s^* = \min\{L,\frac{\alpha-1}{\beta}\}$. We note that $s^* \geq L/4$ since $\beta L \leq 2\alpha$ and $\alpha \geq 2$.   
    The function $h$ is non-decreasing on $[0,s^*]$, which implies that 
    \[
    \E[S] \geq \P(S \leq s^*) \cdot \frac{s^*}{2} + \P(S > s^*) s^* \geq s^*/2 \geq L/8.
    \]
    Here we use that the conditional expectation of $S$ when $S \leq s^*$ is at least $\frac{s^*}{2}$ since the density is non-decreasing. We finally upper bound the density. 
    To do so, note that for $\gamma s^* \leq s \leq s^*$ we have 
    \[
    \frac{h(s)}{h(s^*)} = \left(\frac{s}{s^*}\right)^{\alpha-1} \exp(-\beta(s-s^*)) \geq \gamma^{\alpha-1}
    \]
    Setting $\gamma = (1-\frac{1}{2(\alpha-1)})$, we have $h(s) \geq \frac{1}{2} h(s^*)$. Hence, $Z \geq \int_{\gamma s^*}^{s^*} h(s) \geq  h(s^*) \frac{s^*}{4(\alpha-1)}$, and thus $\|p\|_\infty \leq \frac{4(\alpha-1)}{s^*} \leq \frac{16(\alpha-1)}{L}$.
    In conclusion, for all $\alpha \geq 2$ we have $\E[S] = \Omega(L)$ and $\|p\|_\infty \leq C \alpha/L$.
    \item[(ii)] $\beta L> 2\alpha$: 
    We let $G$ be distributed according the the untruncated Gamma distribution with the same parameters. It follows that $\E[G] = \frac{\alpha}{\beta}$. Markov's inequality gives $\P(G \geq L) \leq \frac{\alpha}{\beta L} < \frac{1}{2}$. Hence, $\P(G \leq L) \geq \frac12$. The density of the Gamma distribution can be shown to be at most $C \frac{\beta}{\sqrt{\alpha}}$ for some universal constant $C>0$. Restricting to an event with probability $\geq 1/2$ increases the density by at most a factor $2$ and hence $\|p\|_{\infty} \leq C \frac{\beta}{\sqrt{\alpha}}$. It remains to prove the lower bound on $\E[S]$. To do so, note that the maximizer of the density of $G$ is $s^* = \frac{\alpha-1}{\beta} \in [0,L]$. 
        Using the same argument as in part (i), we obtain $\E[S] \geq s^*/2 \geq \frac{\alpha}{4\beta}$. 
     In conclusion, for all $\alpha \geq 2$ we have $\E[S] = \Omega(\alpha/\beta)$ and $\|p\|_\infty \leq C \beta/\sqrt{\alpha}$.
\end{enumerate}

The lemma then follows from the established lower bounds on $\E[S]$ and the upper bounds on $\|p\|_\infty$. 
    \end{proof}

    \begin{lemma} \label{lem:lb}
        Let $d \in \N$, $d \geq 4$, $R \geq 1$ and $B_d \subseteq K \subseteq RB_d$ convex. Let $\Omega = K$, and $H(x) = \norm{x}^2$. Then, for all $\beta > 0$, we have
        \begin{align*}
             \underset{x \sim \pi_{\beta}}{\E} \left[H(x)\right] \in \Omega\left(\min\left\{1, \frac{d}{\beta}\right\}\right), \qquad  \underset{x \sim \pi_{\beta}}{\Var} \left[H(x)\right] \in \Omega\left(\min\left\{\frac{1}{d^2},\frac{d}{\beta^2}\right\}\right).
        \end{align*}
        Furthermore, there exists a universal constant $C>0$ such that for any $0 \leq a <b$ we have 
        \[
        \underset{x \sim \pi_{\beta}}{\P} \left[a \leq \norm{x}^2 \leq b\right] \leq C \max\{d,\beta/\sqrt{d}\} \cdot (b-a).
        \]
    \end{lemma}

\begin{proof}
    For every $d$-dimensional unit vector $v \in \mathbb S^{d-1}$, let $\rho(v)$ be the length of the line segment from $0$ until the boundary of $K$ in the direction of $v$, i.e., $\rho(v) := \sup\{r \geq 1 : rv \in K\}$. 
    
    Now, for $x \sim \pi_\beta$, let $v = x/\|x\|$ be its direction. Let $q_\beta$ be the induced probability distribution on the sphere $\mathbb S^{d-1}$. Then, conditioned on $x$ being a positive multiple of $v$, the squared norm $s = \norm{x}^2$ has density proportional to  
    \[
    p_v(s) = s^{\frac{d}{2}-1} \exp(-\beta s) \mathbbm{1}_{[0,\rho(v)^2]}.
    \]
    From \Cref{lem:truncated-Gamma}, using that $\rho(v)\geq1$, we obtain that there exist universal constants $c,C>0$ such that
    \begin{align*}
        \E_{s \sim p_v}[s] &\geq c \min\left\{\rho(v)^2, \frac{d}{2\beta}\right\} \geq c \min\left\{1, \frac{d}{2\beta}\right\} \\
        \Var_{s \sim p_v}[s] &\geq c \min\left\{\frac{4\rho(v)^4}{d^2},\frac{d}{2\beta^2}\right\} \geq c \min\left\{\frac{4}{d^2},\frac{d}{2\beta^2}\right\},
    \end{align*}
    and for all $0 \leq a <b \leq \rho(v)^2$ we have 
    \[
    \P_{s \sim p_v}(a\leq s \leq b) \leq C \max\left\{\frac{d}{2\rho(v)^2}, \frac{\beta}{\sqrt{d/2}}\right\} (b-a) \leq C\sqrt{2} \max\left\{d, \frac{\beta}{\sqrt{d}}\right\} (b-a).
    \]
    
    Finally, using the previous lower bounds that are uniform in $v$, we have 
    \begin{align*}
    \E_{x \sim \pi_\beta} [\norm{x}^2] &= \E_{v \sim q_\beta}[\E_{s \sim p_v}s] \geq c \min\left\{1, \frac{d}{2\beta}\right\},
    \end{align*}
    and, using the law of total variance, we have 
    \begin{align*}
    \Var_{x \sim \pi_\beta}[\norm{x}^2] &\geq \E_{v \sim q_\beta}\left[\Var_{s \sim p_v}[s] \right] \geq c \min\left\{\frac{4}{d^2},\frac{d}{2\beta^2}\right\}, 
    \end{align*}
    and finally, for all $0 \leq a <b$, we have $\underset{x \sim \pi_{\beta}}{\P} \left[a \leq \norm{x}^2 \leq b\right] \leq C \max\{d,\beta/\sqrt{d}\} \cdot (b-a)$.    
\end{proof}

\section{Proof of \Cref{lem:volume-bound}}
\label{app:volume-bound}

Here we provide the deferred proof of \cref{lem:volume-bound}, which we restate here for convenience. 

\volumebound*

\begin{proof}
Write
\[
S=(x+\delta B_d)\cap K,
\qquad
G=S\cap rB_d.
\]
If $x=0$, then $S\subseteq rB_d$ and the claim is immediate.
Otherwise, write $x=\rho v$, where $\rho>0$ and $\|v\|=1$, and let
\[
H=\{y\in S:\langle y-x,v\rangle\le 0\}.
\]

We first show that $\operatorname{Vol}(S)\le 2\operatorname{Vol}(H)$.
For any $y=x+u+tv\in S$ with $u\perp v$ and $t>0$, its reflection
$y'=x+u-tv$ lies in $x+\delta B_d$. Moreover,
\[
u-tv\in\delta B_d\subseteq K,
\qquad
y'\in[u-tv,y]\subseteq K.
\]
Thus reflection across the hyperplane through $x$ perpendicular to $v$
maps $S\setminus H$ injectively into $H$ and preserves volume. Hence, $\operatorname{Vol}(S)\le 2\operatorname{Vol}(H)$.

We now bound $\operatorname{Vol}(H)$ in terms of $\operatorname{Vol}(G)$. To do so, set
\[
\alpha=\frac{r}{\sqrt{r^2+\delta^2}}.
\]
We claim that $\alpha(H\setminus G)\subseteq G$.
Indeed, for $y\in H\setminus G$,
\[
\|y\|^2
=\|x\|^2+2\langle x,y-x\rangle+\|y-x\|^2
\le r^2+\delta^2,
\]
so $\alpha y\in rB_d$. Also, $\alpha y\in K$ by convexity and
$0\in K$. Finally, using $\|y\|>r$ and $\|x\|\le r$,
\begin{align*}
\|\alpha y-x\|^2
&=\alpha\|y-x\|^2+(1-\alpha)\|x\|^2
  -\alpha(1-\alpha)\|y\|^2\\
&\le \alpha\delta^2+(1-\alpha)^2r^2\\
&\le \delta^2,
\end{align*}
where the last step follows from
\[
1-\alpha
=1-\left(1+\frac{\delta^2}{r^2}\right)^{-1/2}
\le \frac{\delta^2}{r^2}.
\]
This proves the claim. Consequently,
\[
\operatorname{Vol}(H)
\le \operatorname{Vol}(G)+\operatorname{Vol}(H\setminus G)
\le (1+\alpha^{-d})\operatorname{Vol}(G).
\]
Since
\[
\alpha^{-d}
=\left(1+\frac{\delta^2}{r^2}\right)^{d/2}
\le \exp\left(\frac{d\delta^2}{2r^2}\right)
\le \sqrt{e},
\]
we conclude that
\[
\operatorname{Vol}(S)
\le 2\operatorname{Vol}(H)
\le 2(1+\sqrt{e})\operatorname{Vol}(G)
\in O(\operatorname{Vol}(G)). \qedhere
\]
\end{proof}

\end{document}